\documentclass[11pt,a4paper]{article}
\usepackage[pageanchor,hypertexnames=false]{hyperref}
\usepackage[dvipsnames]{xcolor}
\usepackage{amsmath,amssymb,amsbsy,amsthm,calc,cases,enumerate,url,framed,fancyhdr,lipsum,tikz,enumitem,natbib,etoolbox,bbm}
\usepackage[mathscr]{euscript}
\usepackage{xcolor}
\hypersetup{colorlinks,linkcolor={red!70!black},citecolor={blue!65!black},urlcolor={blue!70!black}}
\usepackage{graphicx}
\usepackage{epstopdf}
\usepackage[font=small,skip=0pt]{caption}
\usepackage[all]{hypcap}

\newcommand{\hfparskip}{\vspace{-0.5\parskip}}
\newcommand{\unparskip}{\vspace{-\parskip}}
\newcommand{\umparskip}{\vspace{-1.5\parskip}}
\newcommand{\deparskip}{\vspace{-2\parskip}}

\newcommand{\abs}[1]{\lvert #1 \rvert}

\newcommand{\norm}[1]{\| #1 \|}

\newcommand{\floor}[1]{\lfloor #1 \rfloor}

\newcommand{\ceil}[1]{\lceil #1 \rceil}

\newcommand{\rbr}[1]{\left(#1\right)}

\newcommand{\cbr}[1]{\left\{#1\right\}}

\newcommand{\inv}[1]{#1^{-1}}

\newcommand{\N}{\mathbb{N}}

\newcommand{\Q}{\mathbb{Q}}
\newcommand{\R}{\mathbb{R}}

\newcommand{\E}{\mathbb{E}}
\renewcommand{\Pr}{\mathbb{P}}
\newcommand{\Ind}{\mathbbm{1}}
\newcommand{\ipr}[2]{\langle #1,#2 \rangle}

\newcommand{\restr}[2]{\left.#1\right\vert_{#2}}

\DeclareMathOperator\sgn{sgn}
\DeclareMathOperator\Var{Var}

\DeclareMathOperator\Img{Im}

\DeclareMathOperator\Span{span}

\DeclareMathOperator*\argmin{argmin}
\DeclareMathOperator*\sargmin{sargmin}

\DeclareMathOperator\Int{Int}
\DeclareMathOperator\relint{relint}
\DeclareMathOperator\Cl{Cl}
\DeclareMathOperator\conv{conv}
\DeclareMathOperator\aff{aff}
\DeclareMathOperator\supp{supp}
\DeclareMathOperator\csupp{csupp}

\newcommand{\iid}{\overset{\mathrm{iid}}{\sim}}

\newcommand{\cvp}{\overset{p}{\to}}
\newcommand{\cvd}{\overset{d}{\to}}

\newcommand{\eqd}{\overset{d}{=}}

\DeclareMathOperator\Bin{Bin}

\DeclareMathOperator\KL{KL}

\newcommand{\dtv}{d_{\mathrm{TV}}}

\theoremstyle{plain} 
\newtheorem{theorem}{Theorem}
\newtheorem{proposition}[theorem]{Proposition}
\newtheorem*{proposition*}{Proposition}
\newtheorem{lemma}[theorem]{Lemma}
\newtheorem*{lemma*}{Lemma}
\newtheorem{corollary}[theorem]{Corollary}
\newtheorem{claim}{Claim}
\newtheorem*{claim*}{Claim}

\newtheorem*{definition*}{Definition}
\newtheorem{assumption}{Assumption}

\newtheorem*{question*}{Question}

\newtheorem*{conjecture*}{Conjecture}

\theoremstyle{definition} 
\newtheorem{remark}[theorem]{Remark}
\newtheorem*{remark*}{Remark}

\newtheorem*{aside*}{Aside}
\newtheorem{example}[theorem]{Example}
\newtheorem*{example*}{Example}
\newtheorem*{aim*}{Aim}

\newtheorem*{observation*}{Observation}

\newcounter{algorithm}
\newenvironment{algorithm}[1][]{\refstepcounter{algorithm}\begin{framed}\textbf{Algorithm~\thealgorithm.} #1}{\end{framed}}

\begin{document}

\title{Nonparametric, tuning-free estimation of S-shaped functions}

\author{Oliver Y. Feng$^\ast$, Yining Chen$^\dagger$, Qiyang Han$^\ddagger$, Raymond J. Carroll$^{\sharp\P}$ \\
and Richard J. Samworth$^\ast$ \\ \\
$^*$Statistical Laboratory, University of Cambridge \\
$^\dagger$Department of Statistics, London School of Economics and Political Science \\
$^\ddagger$Department of Statistics, Rutgers University \\
$^\sharp$Department of Statistics, Texas A\&M University \\
$^\P$School of Mathematical and Physical Sciences, University of Technology Sydney}

\maketitle

\begin{abstract}
We consider the nonparametric estimation of an S-shaped regression function.  The least squares estimator provides a very natural, tuning-free approach, but results in a non-convex optimisation problem, since the inflection point is unknown.  We show that the estimator may nevertheless be regarded as a projection onto a finite union of convex cones, which allows us to propose a mixed primal-dual bases algorithm for its efficient, sequential computation.  After developing a projection framework that demonstrates the consistency and robustness to misspecification of the estimator, our main theoretical results provide sharp oracle inequalities that yield worst-case and adaptive risk bounds for the estimation of the regression function, as well as a rate of convergence for the estimation of the inflection point.  These results reveal not only that the estimator achieves the minimax optimal rate of convergence for both the estimation of the regression function and its inflection point (up to a logarithmic factor in the latter case), but also that it is able to achieve an almost-parametric rate when the true regression function is piecewise affine with not too many affine pieces.  Simulations and a real data application to air pollution modelling also confirm the desirable finite-sample properties of the estimator, and our algorithm is implemented in the \texttt{R} package \texttt{Sshaped}. 
\end{abstract}

\section{Introduction}

We define a function $f\colon [0,1] \rightarrow \mathbb{R}$ to be \emph{S-shaped} if it is increasing, and if there exists $m_0 \in [0,1]$ such that $f$ is convex on $[0,m_0]$ and concave on $[m_0,1]$.  The point $m_0$ is called an \emph{inflection point}, and we do not insist that $f$ is continuous at $m_0$; the cases $m_0 = 0$ and $m_0=1$ correspond to increasing concave and increasing convex functions respectively.  Various examples of S-shaped functions are shown in Figure~\ref{Fig:Examples}.  In many areas of applied science, there are domain-specific reasons to model the regression of a response variable on a covariate as an S-shaped function. For instance, development curves for individuals or populations often exhibit S-shaped behaviour in the context of biological growth \citep{Zeidi1993,ArchontoulisMiguez2015,CSLC2019} or skill proficiency \citep{Gibbs2000}.  Further examples where time is the covariate can be found in audio signal processing \citep{Smith10} and sociology \citep{Tarde1903}.  In agronomy, the van Genuchten--Gupta model \citep{vanGenuchtenGupta1993} postulates an inverted S-shaped relationship between crop yield and soil salinity, and S-shaped trends are also observed for the production levels of commercial goods as labour or other resources are scaled up \citep{Ginsberg1974}. For the latter, economic principles such as the Regular Ultra Passum law \citep{Frisch1964} have been formulated to describe scenarios where marginal gains (i.e.~returns to scale) increase up to a point of maximal productivity and then taper off.

\begin{figure}[htbp!]
\begin{center}
\includegraphics[width=0.5\textwidth]{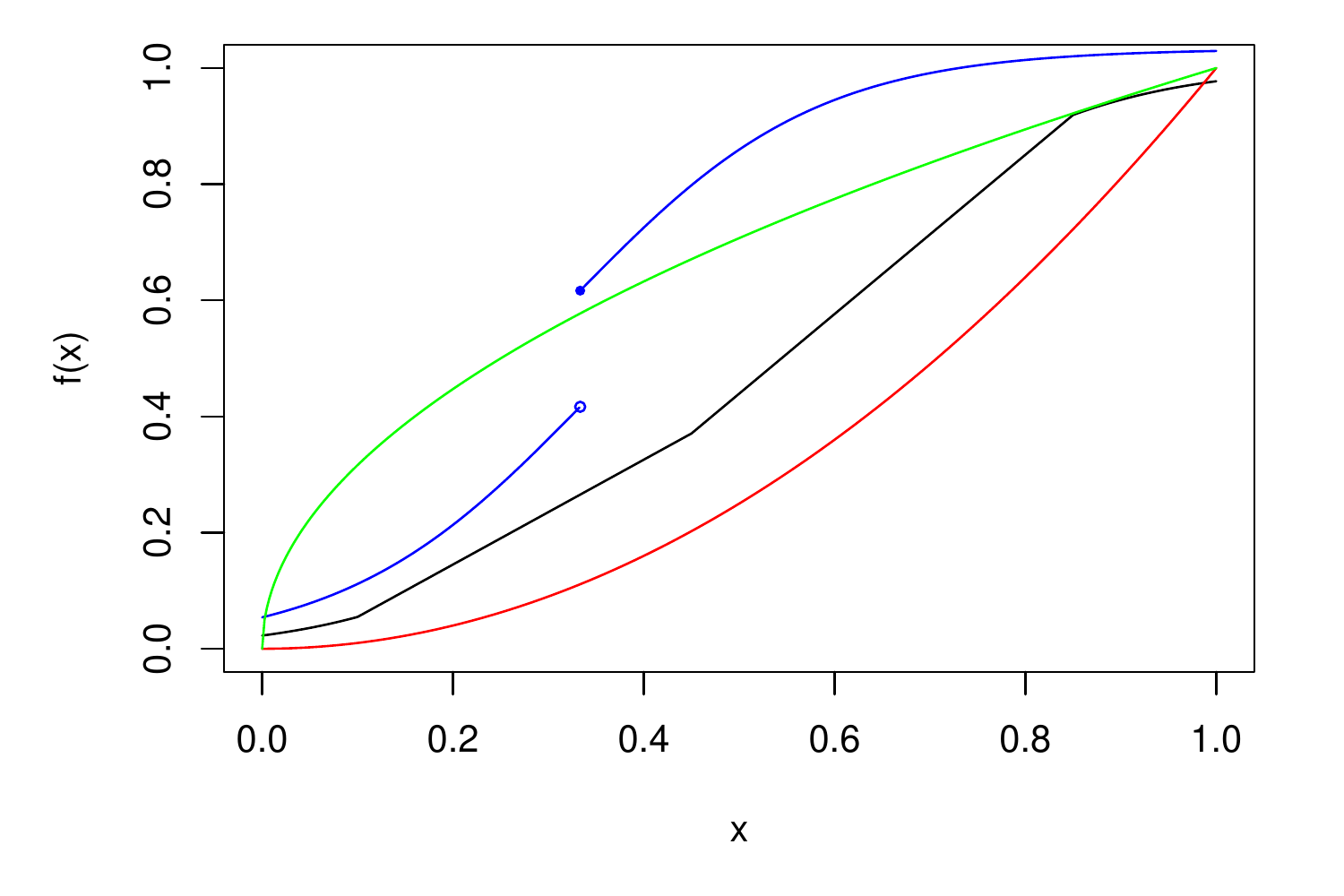}
\end{center}
\caption{\label{Fig:Examples}Some examples of S-shaped functions on $[0,1]$.}
\end{figure}

In some of the examples above, for instance when population or disease dynamics can be modelled by some governing differential equation, it may be natural to confine attention to certain parametric subclasses of S-shaped functions, such as those consisting of sigmoidal (i.e.~logistic) functions of the form
\begin{equation}
\label{eq:logistic}
f(x;A,a,b) = \frac{A}{1+e^{-ax+b}},
\end{equation}
with $A,a > 0$ and $b \in \mathbb{R}$; see also~\citet{JSF2007}. However, in many other settings, such domain-specific knowledge is often lacking, and parametric assumptions may be excessively restrictive.  To illustrate this effect, see Figure~\ref{Fig:Parametric}, where we compare two popular parametric fits of an S-shaped regression function with the estimator we propose in this paper.  The first parametric method fits a logistic curve of the form~\eqref{eq:logistic} using nonlinear least squares.  The second uses segmented linear regression with two kinks, fitted using least squares and a search over the locations of the kinks.  Although these parametric fits appear to the naked eye to be satisfactory, it turns out that their estimation performance, as measured by the squared error loss on the training data, is roughly six times worse than that of our proposal (on average 0.38 and 0.43 compared with 0.067, over 100 repetitions).  If the noise standard deviation is halved, then these parametric methods become 17 times and 19 times worse than our proposal respectively. Notice also that our S-shaped estimator is sufficiently flexible to be able to capture the discontinuity of the regression function, whereas the parametric methods struggle in this respect. The benefits of our nonparametric approach are also apparent in the analysis of real data: see Section~\ref{subsec:realdata}, where we study the way that a quantity related to atmospheric mercury concentration varies with distance from an experimental device close to a geothermal power station.

\begin{figure}[htb!]
\begin{center}
\includegraphics[width=0.475\textwidth]{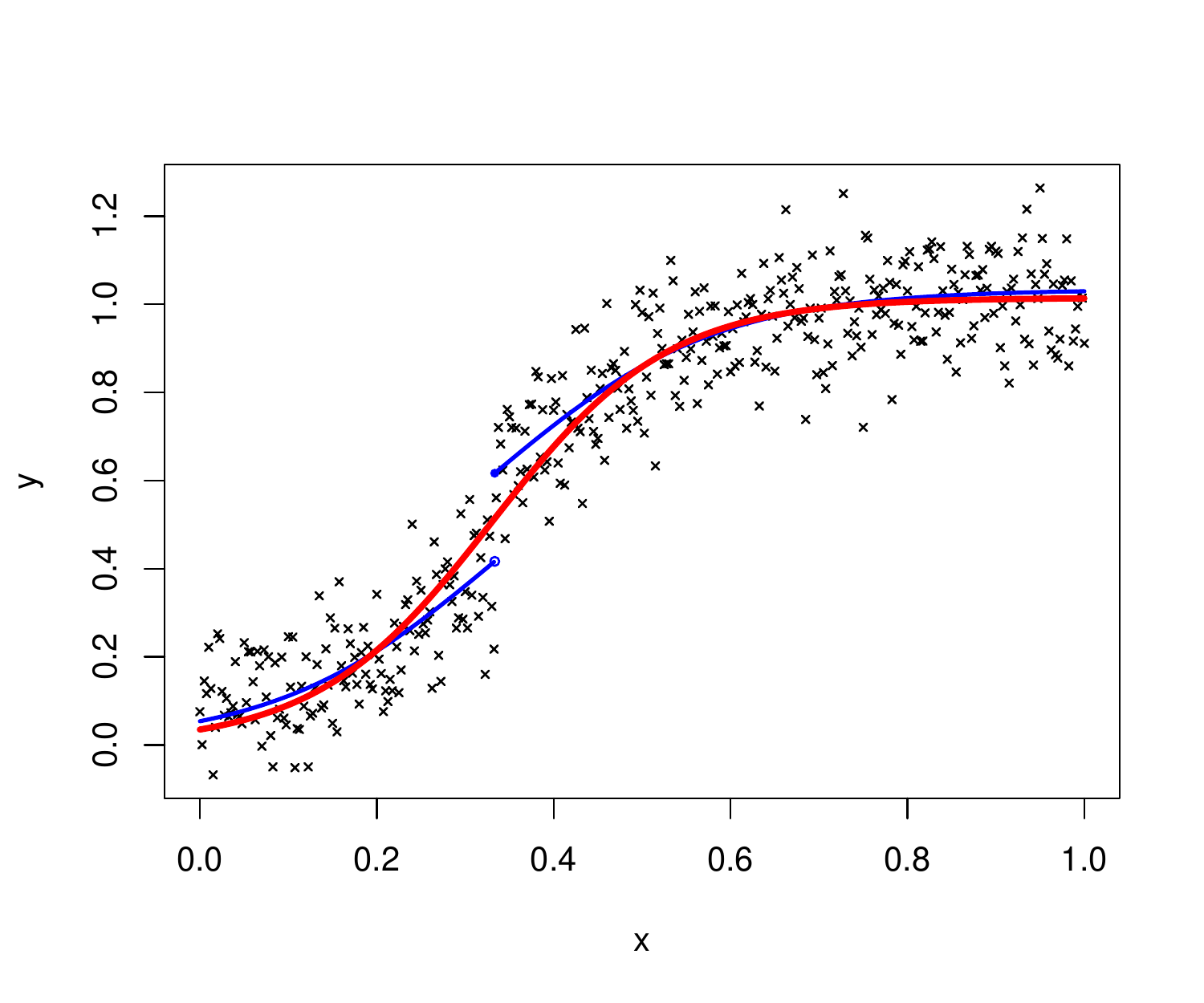}
\includegraphics[width=0.475\textwidth]{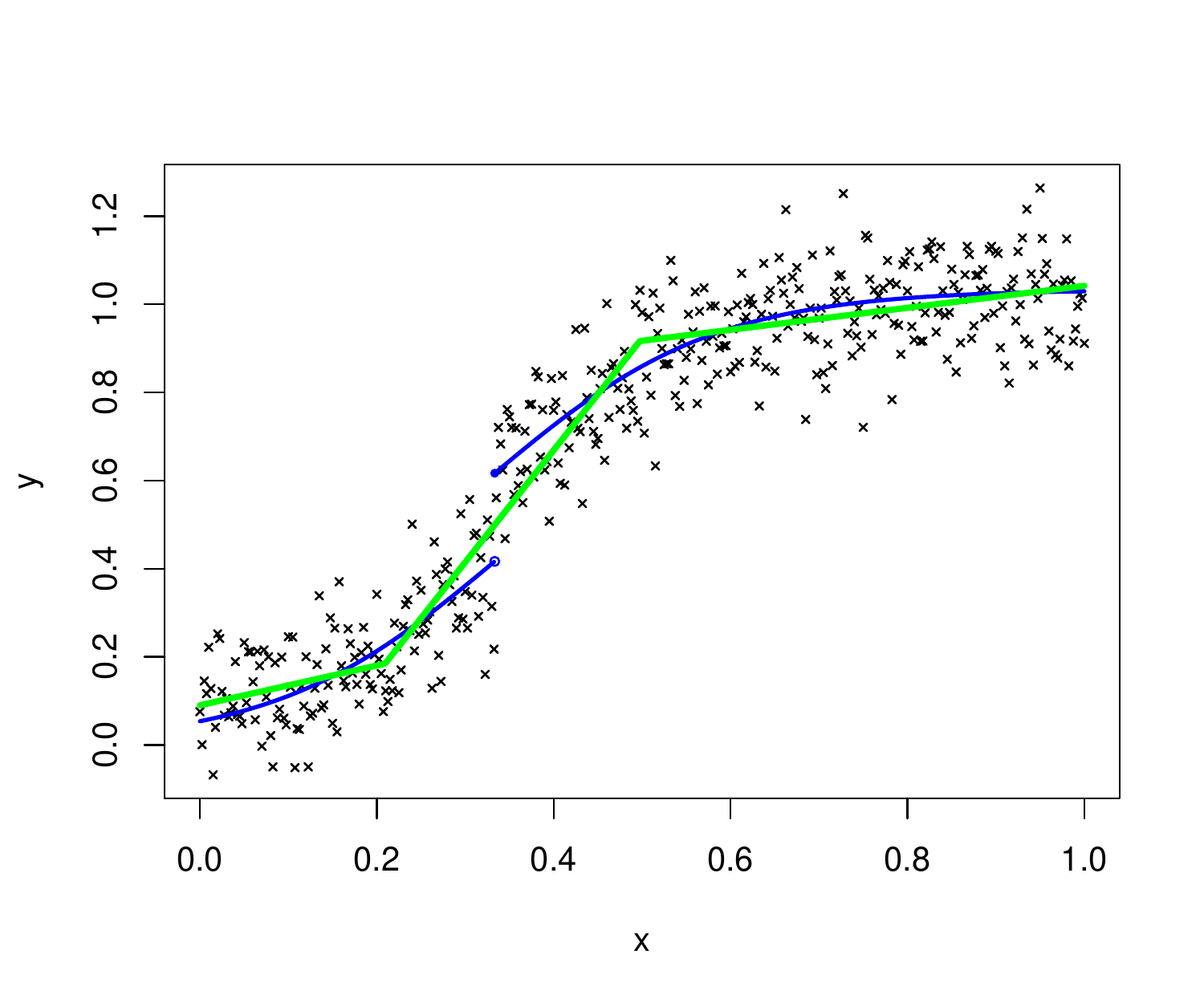}

\vspace{-1cm}
\includegraphics[width=0.475\textwidth]{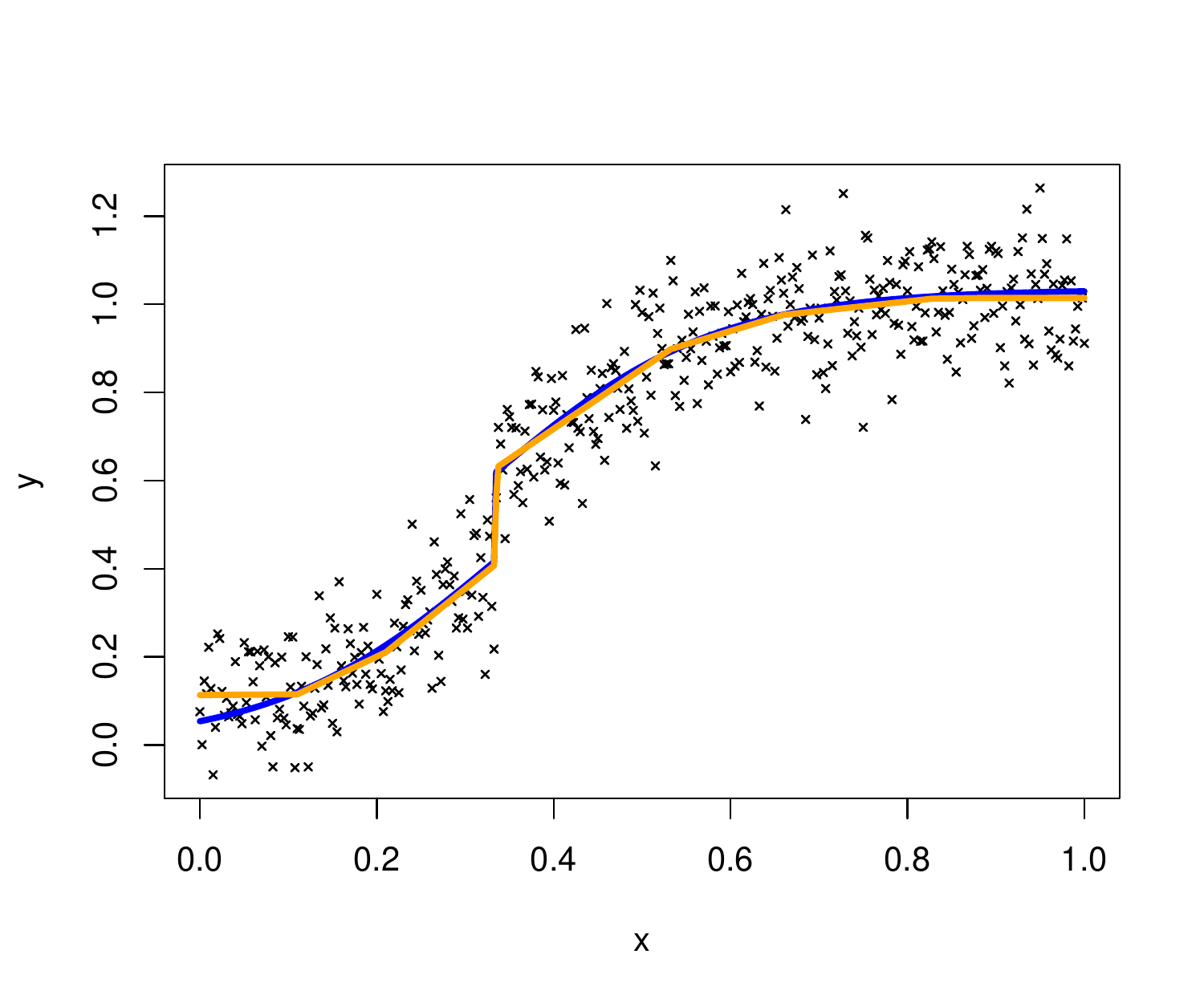}
\includegraphics[width=0.475\textwidth]{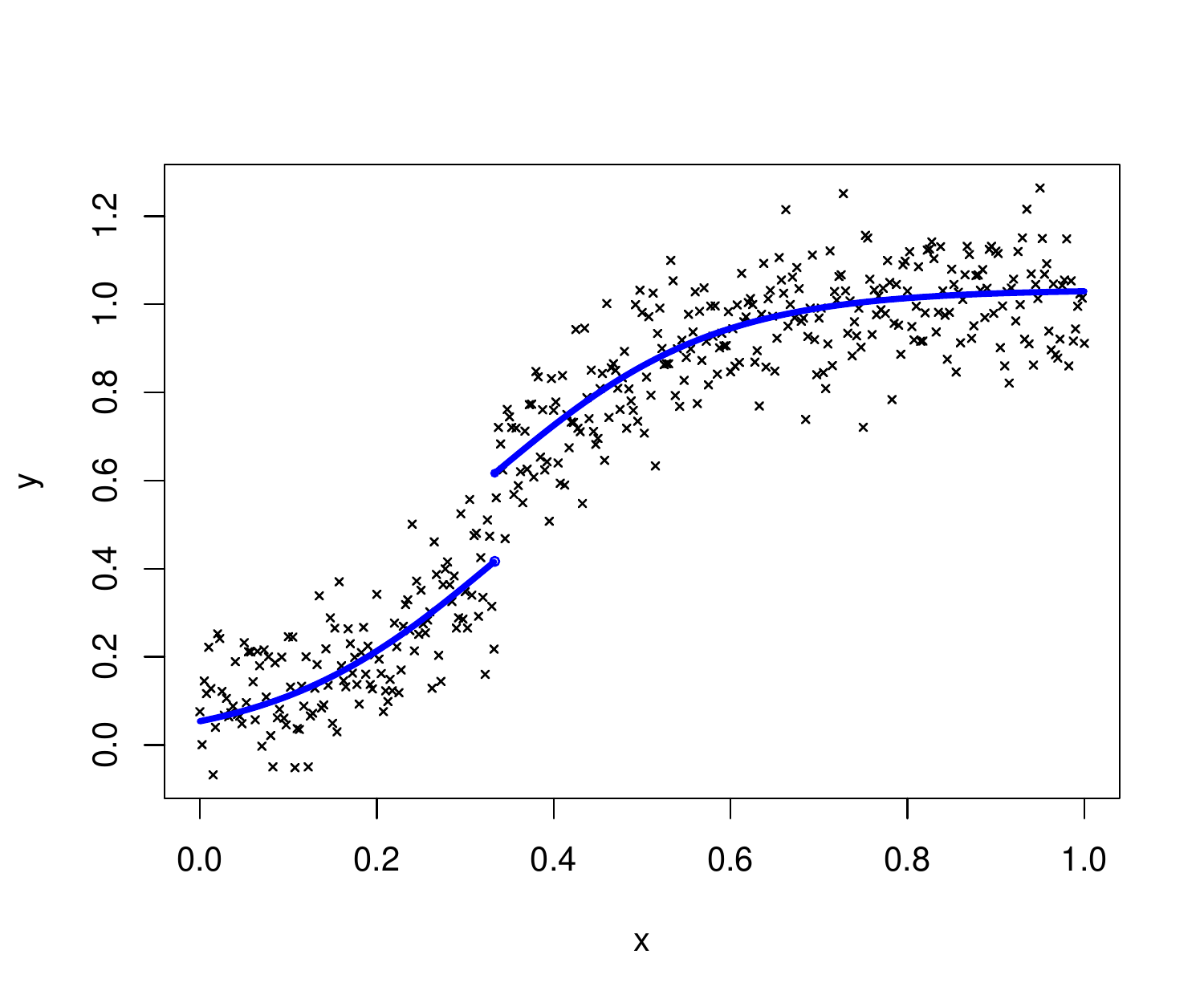}
\end{center}
\caption{\label{Fig:Parametric}Logistic (red, top left), segmented linear regression (green, top right) and our S-shaped estimator (orange, bottom left) of the true regression function (blue, all plots).}
\end{figure}

Motivated by the limitations described in the previous paragraph, the goal of this paper is to introduce a flexible framework for nonparametric estimation of S-shaped functions.  The main challenges in removing the parametric restrictions are two-fold: first, the class $\mathcal{F}$ of S-shaped functions on $[0,1]$ is infinite-dimensional; and second, since the inflection point is unknown, the family $\mathcal{F}$ is non-convex.  Despite this non-convexity, we are able to develop methodology based on suitably defined $L^2$-`projections' of general distributions onto $\mathcal{F}$.  The significant advantage of working in this additional generality is that, having established continuity properties of the projection, results on the consistency and robustness under misspecification of the estimator follow as simple corollaries of basic facts about convergence of empirical distributions.  Nevertheless, since the fully general statements are fairly involved, we defer this formal presentation to Section~\ref{sec:proj2} of the supplementary material
, and focus in Section~\ref{sec:proj} on the special case of projections of the empirical distribution of data of the form $(x_1,Y_1),\ldots,(x_n,Y_n) \in [0,1] \times \mathbb{R}$ with $x_1 < \cdots < x_n$.  This allows us to prove that an S-shaped least squares estimator always exists, and to study its uniqueness properties.  Moreover, when the design is fixed and the errors are independent and identically distributed with mean zero and finite variance, we present a basic consistency result that follows from the general theory in Section~\ref{sec:proj2}.

In Section~\ref{sec:computation}, we take up the challenge of computing the S-shaped least squares estimator.  Since its inflection point occurs at one of the design points, a naive strategy would be to fit, for each choice of $m \in \{x_1,\ldots,x_n\}$, the least squares estimate over the class of S-shaped functions with inflection point $m$, before selecting a solution that minimises the residual sum of squares.  The individual constrained estimates are straightforward to compute using, e.g., active set methods \citep[][Chapters~12 and~16.5]{DHR07,NW06}, but it can be time-consuming to run the active set method $n$ times.  We show how a simple refinement of the search strategy can improve the running time by a factor of around 4, but our major contribution here begins with the observation that the global S-shaped least squares estimate can be obtained as a concatenation of a convex increasing least squares estimate to the left of an estimated inflection point, with a concave increasing least squares estimate to the right.  This enables us to pursue a sequential approach, where we reveal new observations one by one, and update the least squares fits using a mixed primal-dual bases algorithm \citep{FM89,Mey99}.  Our algorithm, which is available in the \texttt{R} package \texttt{Sshaped} \citep{SShaped}, is shown to be around 40 times faster than the naive strategy in examples; see Figure~\ref{Fig:Timing}.  

Our main theoretical contributions are presented in Section~\ref{sec:lsreg}, under an independent and sub-Gaussian error assumption. Here, we derive worst-case and adaptive sharp oracle inequalities for the S-shaped least squares estimator. When combined with our corresponding minimax lower bounds, this theory reveals in particular that the S-shaped least squares estimator attains the optimal worst-case risk of order $n^{-2/5}$ with respect to $L^2$-loss, in the case where the design points are not too irregularly spaced. These results apply both when the S-shaped regression function hypothesis is correctly specified, and where it is misspecified, provided in the latter case that we interpret the loss as the distance to the projection of the signal onto $\mathcal{F}$. For adversarially-chosen design configurations, we show that the risk bound can deteriorate to $n^{-1/3}$ in the worst case. Moreover, the S-shaped least squares estimator adaptively attains the parametric rate of order $n^{-1/2}$ (up to a logarithmic factor), when the projection of the signal is piecewise affine with a relatively small number of affine pieces. Finally, we study the delicate problem of estimating the true inflection point $m_0$, which represents the boundary between the convex and concave parts of the signal. Under an appropriate local smoothness assumption indexed by a parameter $\alpha > 0$, we show that the inflection point $\hat{m}_n$ of the least squares estimator converges to $m_0$ at rate $O_p\bigl((n^{-1}\log n)^{1/(2\alpha+1)}\bigr)$, which matches our local asymptotic minimax lower bound, up to the logarithmic factor. Interestingly, the combination of the monotonicity with the convexity/concavity means that our S-shaped estimator is sufficiently regularised to avoid boundary problems at the endpoints $\{0,1\}$ of the covariate domain; other common shape-constrained methods are known to lead to boundary estimation inconsistency \citep{KL06,CSS2010,BJPSW11,BGS15,Samworth2018,HK21}.

In Section~\ref{sec:simulations}, we study the empirical properties of our S-shaped least squares estimator, comparing both its running time and statistical performance with those of alternative approaches on simulated data. We also present a real data application of these techniques in air pollution modelling, which highlights the convenience and efficacy of our proposal. We conclude by discussing some possible directions for future research in Section~\ref{sec:discussion}. The appendix (Section~\ref{sec:appendix})
provides further details of the mixed primal-dual bases algorithm that we use to compute our estimator. The proofs of our main results are deferred to the supplementary material
, in which the results and sections appear with an `S' before the relevant label number.

Previous work on nonparametric estimation of S-shaped functions includes \citet{YCJM19,YCJK20}, who, in the context of production theory in economics, apply a method known as shape-constrained kernel least squares to estimate multivariate production functions that are S-shaped along one-dimensional rays. \citet{KS13} use local polynomial regression techniques to identify an inflection point of a smooth signal from corrupted observations. In both of these works, kernel bandwidths must be chosen carefully to control the bias-variance tradeoff and (for the approach of \citet{KS13} in particular) to ensure that the fitted curve does not have multiple inflection points. \citet{LM17} instead estimate univariate convex-concave functions using cubic splines defined with respect to a number of user-specified knots, and establish rates of convergence for the inflection points of the resulting estimators. Their method is implemented in the \texttt{R} package \texttt{ShapeChange}~\citep{LM16}, which~\citet{LCMWG20} subsequently used in combination with the \texttt{scam} (Shape Constrained Additive Models) package of~\citet{PW15} to estimate S-shaped disease trajectories of patients with Huntington's disease. We also mention the extremum distance estimator and extremum surface estimator proposed by~\citet{Christopoulos2016}, with the aim of locating the inflection point of a smooth function based on its geometric properties. We provide a numerical comparison of our 
procedure with those of \citet{LM17}, \citet{YCJM19,YCJK20} and \citet{Christopoulos2016} in Section~\ref{sec:StatPerf}.

\hfparskip
\subsection{Notation}
\label{subsec:notation}
For $n\in\N$, we write $[n]:=\{1,\dotsc,n\}$, and given $0\leq x_1<\cdots<x_n\leq 1$, define $\mathcal{G}\equiv\mathcal{G}[x_1,\dotsc,x_n]$ to be the set of continuous, piecewise affine $f\colon [0,1]\to\R$ with kinks in $\{x_2,\dotsc,x_{n-1}\}$. If $\tilde{f}_n\colon [0,1]\to\R$ minimises\footnote{Since there may be multiple minimisers, we will also assume throughout and without further comment that $\tilde{f}_n$ is chosen to depend measurably on $(x_1,Y_1),\dotsc,(x_n,Y_n)$.  Likewise, we will assume the same property for estimated inflection points.}  $f\mapsto \sum_{i=1}^n\bigl(Y_i-f(x_i)\bigr)^2 =: S_n(f)$ over some class $\tilde{\mathcal{F}}$ of functions on $[0,1]$, we say that $\tilde{f}_n$ is a \emph{least squares estimator (LSE) over $\tilde{\mathcal{F}}$} based on $\{(x_i,Y_i):1\leq i\leq n\}$. We write $a_n \lesssim b_n$ to mean that there exists a universal constant $C > 0$ such that $a_n \leq Cb_n$ for all~$n$.

\section{Existence and consistency of S-shaped least squares estimators}
\label{sec:proj}

The purpose of this section is to study the existence, uniqueness and consistency of S-shaped least squares estimators. We will see later that in a suitable sense, these estimators can be regarded as $L^2$-projections onto $\mathcal{F}$ of the empirical distribution of the data.  As such, the results in this section turn out to be special cases of a much more general theory, presented in Section~\ref{sec:proj2}, concerning the existence and continuity of $L^2$-projections of arbitrary distributions on $[0,1] \times \mathbb{R}$ having finite variance.  The generality of this projection framework remains of importance to statisticians, particularly in terms of providing results on the robustness of S-shaped least squares estimators to model misspecification; however, the results are of a more technical nature, so to facilitate understanding of the main ideas, we focus on the well-specified case here. 

Suppose we have observations $(x_1,Y_1),\dotsc,(x_n,Y_n) \in [0,1] \times \mathbb{R}$ with $x_1 < \cdots < x_n$. For each $m\in [0,1]$, we denote by $\mathcal{F}^m$ the class of S-shaped functions with an inflection point at $m$, i.e.\ the set of all $f\colon [0,1]\to\R$ that are convex on $[0,m]$, concave on $[m,1]$ and increasing (i.e.\ non-decreasing) on $[0,1]$.  Thus $\mathcal{F}:=\bigcup_{m\in [0,1]}\mathcal{F}^m$ is the set of all S-shaped functions on $[0,1]$, but this union of convex sets is not itself convex. 
\begin{proposition}
\label{prop:existence}
For each $m \in [0,1]$, there exists an LSE $\tilde{f}_n^m$ over $\mathcal{F}^m$ that is uniquely determined at $x_1,\ldots,x_n$.  Moreover, there exists an LSE $\tilde{f}_n$ over $\mathcal{F}$ with an inflection point in $\{x_1,\ldots,x_n\}$.
\end{proposition}

\unparskip
A straightforward and direct proof of this result is given in Section~\ref{sec:bdadj}. As part of the projection framework in Section~\ref{sec:proj2}, we obtain generalisations of Proposition~\ref{prop:existence} in Corollaries~\ref{cor:proj}(d) and~\ref{cor:cont}(a). Since our objective criterion only measures the error incurred at the design points, it is no surprise that any LSE $\tilde{f}_n^m$ over $\mathcal{F}^m$ can only be unique at $x_1,\ldots,x_n$.  There is a canonical way to define $\tilde{f}_n^m$ on the whole of $[0,1]$, namely by linear interpolation between its kinks.  Thus, the slope remains constant on $[0,x_2],[x_2,x_3],\ldots,[x_{n-2},x_{n-1}], [x_{n-1},1]$, and we denote this interpolating function by~$\hat{f}_n^m\in\mathcal{G}\equiv\mathcal{G}[x_1,\dotsc,x_n]$.  A subtle issue, however, is that when $m$ is not a design point, $\hat{f}_n^m$ need not belong to $\mathcal{F}^m$; see the left panel of Figure~\ref{Fig:NonUniqueness}.  To finesse this point, for $m \in [0,1]$, denote by $\mathcal{H}^m\equiv\mathcal{H}^m[x_1,\dotsc,x_n]$ the class of all $f \in \mathcal{G}$ for which there exists $g\in\mathcal{F}^m$ with $f=g$ on $\{x_1,\dotsc,x_n\}$. Then $\mathcal{H}^m$ is a closed, convex cone, and the LSE over $\mathcal{H}^m$ based on $\{(x_i,Y_i):1\leq i\leq n\}$ is precisely the function $\hat{f}_n^m$. We refer to $\hat{f}_n^0$ and $\hat{f}_n^1$ as the \emph{increasing concave} LSE and \emph{increasing convex} LSE (based on $\{(x_i,Y_i):1\leq i\leq n\}$) respectively.
\begin{figure}[htbp!]
\begin{center}
\includegraphics[width=0.45\textwidth]{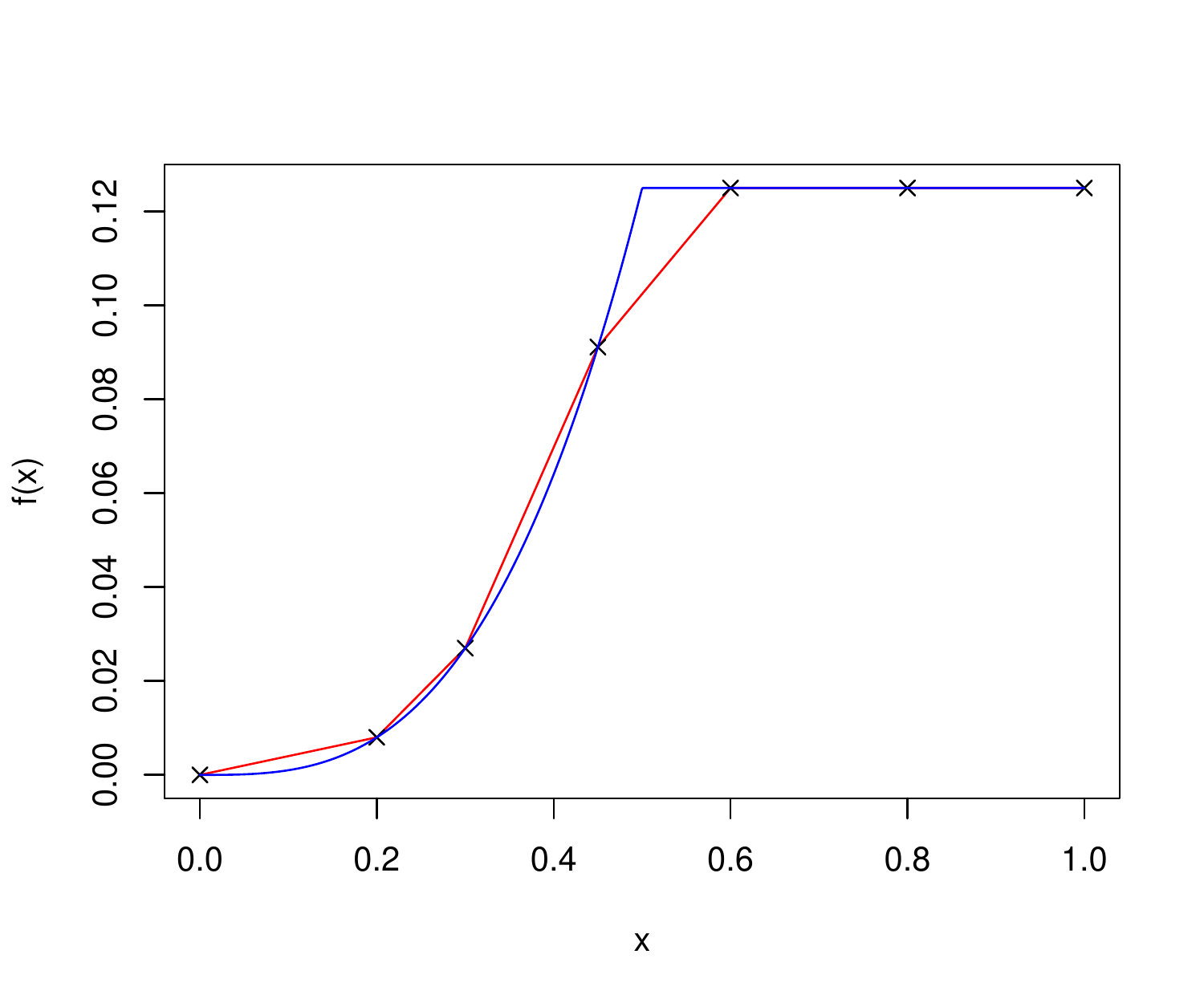}
\includegraphics[width=0.45\textwidth]{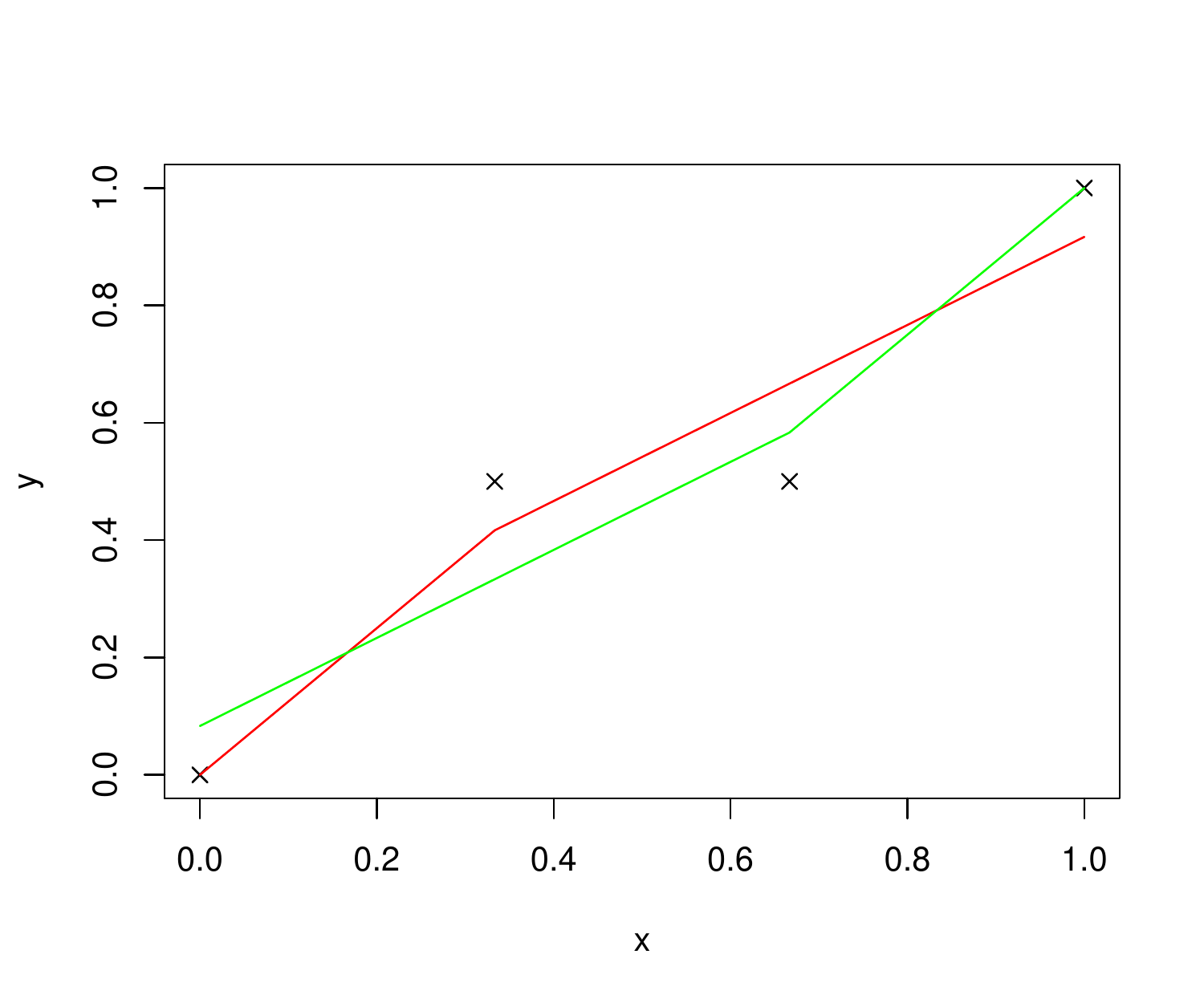}
\caption{\label{Fig:NonUniqueness}Left: For noiseless observations of the blue regression function at the black crosses, the red curve illustrates the linear interpolation $\hat{f}_n^m$ of the LSE, with $m=0.5$; here, the segment of steepest slope does not contain $x=0.5$, so $\hat{f}_n^m$ does not belong to $\mathcal{F}^m$ with $m=0.5$.  Right: For the data given by the black crosses, both the red curve and the green curve are LSEs over $\mathcal{F}$.}
\end{center}
\end{figure}

\unparskip
It turns out, however, that in general an LSE $\tilde{f}_n$ over $\mathcal{F}$ is not even uniquely defined at the design points.  For instance, if our data are $(0,0),(1/3,1/2),(2/3,1/2),(1,1)$, then the linear interpolations of both $(0,0),(1/3,5/12),(2/3,2/3),(1,11/12)$ and $(0,1/12),(1/3,1/3),(2/3,7/12),(1,1)$ are LSEs over~$\mathcal{F}$; see the right panel of Figure~\ref{Fig:NonUniqueness}.  We remark that this non-uniqueness is not related to the small number of data points, but rather to the symmetry of the data configuration.

In order to present a basic consistency result, we introduce a model where we regard our data $\{(x_1,Y_1),\ldots,(x_n,Y_n)\} \equiv \{(x_{n1},Y_{n1}),\ldots,(x_{nn},Y_{nn})\}$ as being realised from a triangular array sampling scheme
\begin{equation}
\label{eq:Model}
Y_{ni}=f_0(x_{ni})+\xi_{ni},\quad i=1,\dotsc,n,
\end{equation}
where $f_0 \colon [0,1]\to\R$ is a Borel measurable regression function, where $\xi_{n1},\dotsc,\xi_{nn}$ are independent noise variables with mean zero and finite variance for each $n$, and where $0\leq x_{n1}<\cdots<x_{nn}\leq 1$ are fixed design points.  We write $\Pr_n:=n^{-1}\sum_{i=1}^n\delta_{(x_{ni},Y_{ni})}$ and $\Pr_n^X:=n^{-1}\sum_{i=1}^n\delta_{x_{ni}}$ for the joint and $X$-marginal empirical distributions respectively.

For a finite Borel measure $\nu$ on $[0,1]$, we denote by $\supp\nu$ the \emph{support} of $\nu$, which is defined as the smallest closed set $A$ such that $\nu(A^c)=0$, or equivalently the set of all $x\in [0,1]$ with the property that $\nu(U)>0$ for any open neighbourhood $U$ of $x$ in $[0,1]$.

\begin{proposition}
\label{cor:consistency}
In model~\eqref{eq:Model}, assume that $f_0 \in \mathcal{F}$ has unique inflection point $m_0 \in [0,1]$ and that $\xi_{n1},\dotsc,\xi_{nn}$ are independent and identically distributed for each $n$.  For each $n \in \mathbb{N}$, let $\hat{f}_n^{m_0}$ and $\tilde{f}_n$ denote LSEs over $\mathcal{F}^{m_0}$ and $\mathcal{F}$ respectively.  Suppose further that $(\Pr_n^X)$ converges weakly to a distribution $P_0^X$ on $[0,1]$ satisfying $\supp P_0^X=[0,1]$ and $P_0^X(\{m\})=0$ for all $m\in [0,1]$.  Then, for $\tilde{g}_n \in \{\hat{f}_n^{m_0},\tilde{f}_n\}$ and with $\tilde{m}_n$ denoting any inflection point of $\tilde{g}_n$, we have
\unparskip
\begin{enumerate}[label=(\alph*)]
\item $\tilde{m}_n\cvp m_0$;
\item $\sup_{x\in A}\,\abs{(\tilde{g}_n-f_0)(x)}\cvp 0$ for any closed set $A\subseteq [0,1]\setminus\{m_0\}$;
\item If $m_0\in (0,1)$, then $\int_0^1 |\tilde{g}_n-f_0|^q \, dP_0^X \cvp 0$ for all $q\in [1,\infty)$;
\item If $m_0\in (0,1)$ and in addition $f_0$ is continuous at $m_0$, then $\sup_{x\in [0,1]}\,\abs{(\tilde{g}_n-f_0)(x)}\cvp 0$.
\end{enumerate}

\unparskip
\end{proposition}

\unparskip
Proposition~\ref{cor:consistency} follows from Proposition~\ref{prop:consistency} in Section~\ref{sec:proj2}, which handles the more general case where $f_0$ need not belong to $\mathcal{F}$, and where it may have multiple inflection points. A proof of the latter result is given in Section~\ref{sec:projproofs}.

\section{Computation of S-shaped least squares estimators}
\label{sec:computation}
Returning to the setting of data $(x_1,Y_1),\dotsc,(x_n,Y_n) \in [0,1] \times \mathbb{R}$ with $x_1 < \cdots < x_n$, we now consider the problem of computing an S-shaped LSE over $\mathcal{F}$.  In light of the non-uniqueness discussion in Section~\ref{sec:proj}, we will take as our target the LSE $\hat{f}_n:=\hat{f}_n^{\hat{m}_n}$, where $\hat{m}_n:=x_{\hat{\jmath}_n}$ and $\hat{\jmath}_n:=\sargmin_{1\leq j\leq n}S_n(\hat{f}_n^{x_j})$; here and below, $\sargmin$ denotes the smallest element of the $\argmin$.  One of the main challenges here is that in general the function $j\mapsto S_n(\hat{f}_n^{x_j})$ has multiple local minima; see Figure~\ref{Fig:MultipleMinima}.  A `brute-force' method that we call \texttt{ScanAll}, then, is to compute each of the LSEs $\hat{f}_n^{x_1},\dotsc,\hat{f}_n^{x_n}$ directly by solving~$n$ separate constrained least squares problems. 
In each instance, we can run the support reduction algorithm~\citep{GJW08} or a generic active set algorithm~\citep[][Chapters~12 and~16.5]{DHR07,NW06} on the whole dataset $\{(x_i,Y_i):1\leq i\leq n\}$, but it is computationally expensive to repeat this $n$ times, even when $n$ is only moderately large; see Section~\ref{sec:comptime}.
\begin{figure}
\centering
\includegraphics[width=0.45\textwidth]{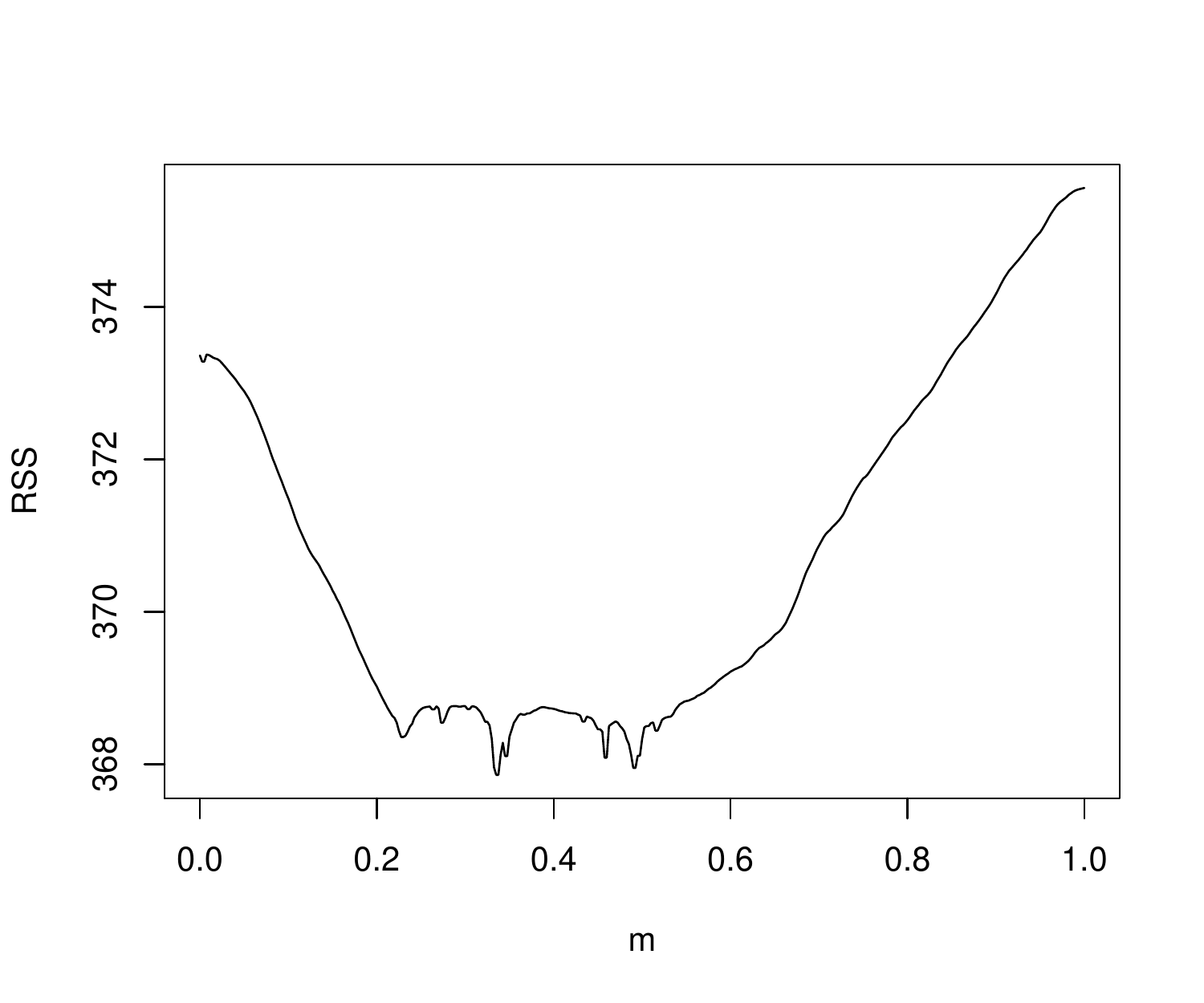}
\includegraphics[width=0.45\textwidth]{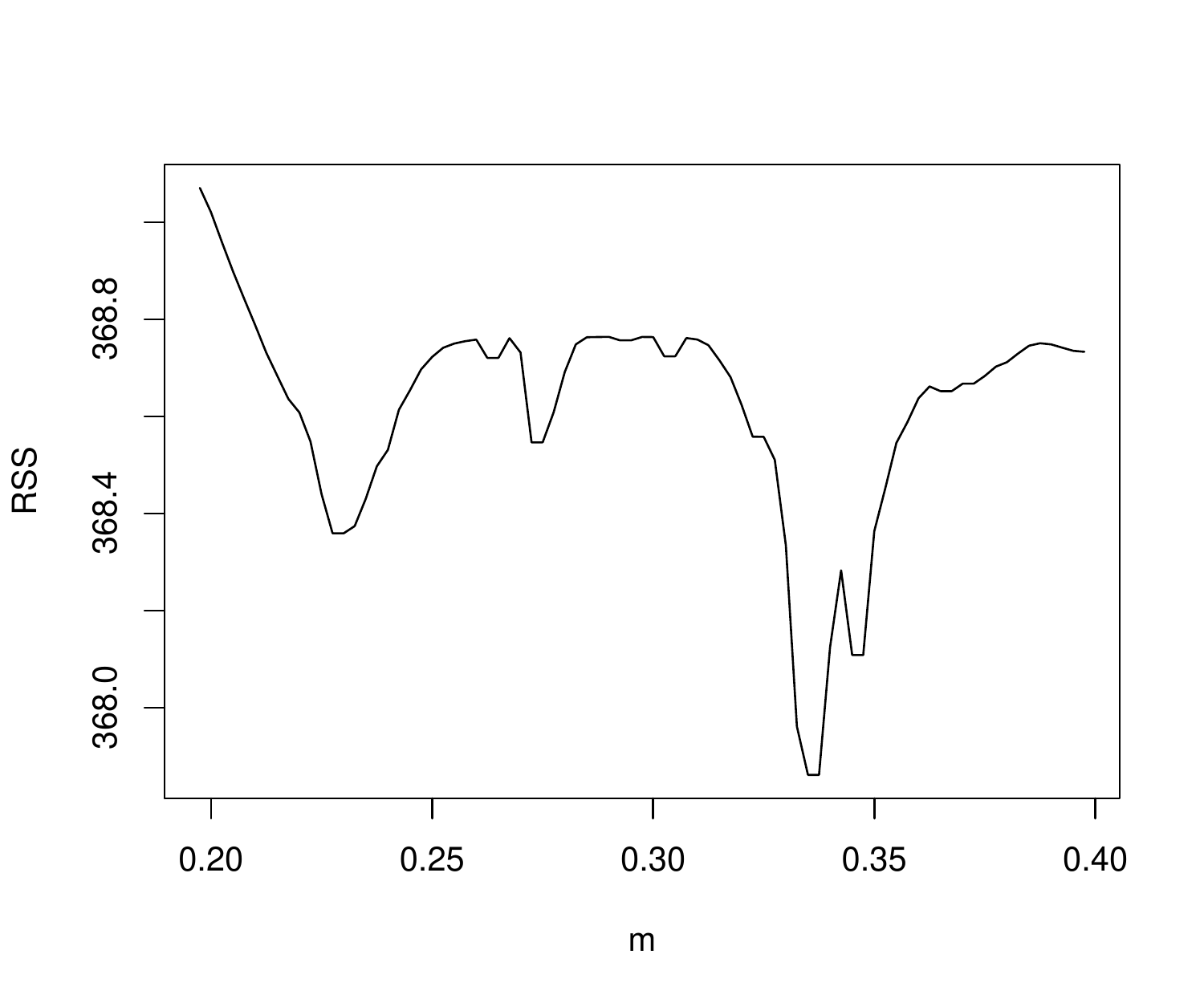}
\caption{\label{Fig:MultipleMinima}Plots of the residual sum of squares $S_n(\hat{f}_n^m)$ of the least squares estimator with inflection point at~$m$ over $m \in [0,1]$ (left) and $m \in [0.2,0.4]$ (right), illustrating the multiple local minima of this function.  Here, with $n=400$, the data were generated according to $Y_i = f(x_i) + \xi_i$ for $i=1,\ldots,n$, with $f$ taken to be the blue regression function from Figure~\ref{Fig:Examples}, $x_i = i/n$ for $i=1,\ldots,n$ and $\xi_1,\ldots,\xi_n$ independent $N(0,1)$ random errors.}
\end{figure}

\unparskip
To improve the overall efficiency of this procedure, it would therefore be desirable to both refine the initial search strategy as well as exploit any common structure underlying the individual minimisation problems. For instance, we might hope to be able to obtain $\hat{f}_n^{x_j}$ via a faster update step that takes as input the previous LSE $\hat{f}_n^{x_{j-1}}$, but it is not immediately clear how this can be done. 

We now describe and justify an alternative approach that achieves both of the above objectives. For $j\in [n]$, we write $\hat{f}_{1,j}\in\mathcal{G}[x_1,\dotsc,x_j]$ for the increasing convex LSE based on $\{(x_i,Y_i):1\leq i\leq j\}$ and $\hat{f}_{n,j}\in\mathcal{G}[x_j,\dotsc,x_n]$ for the increasing concave LSE based on $\{(x_i,Y_i):j\leq i\leq n\}$, recalling from, e.g.,~\citet[Lemma~2.2]{GS17} that
\begin{equation}
\label{eq:GS17}
\hat{f}_{1,j}(x_j)\geq Y_j\geq\hat{f}_{n,j}(x_j) \quad \text{for all }j \in [n].
\end{equation}
We then define $\hat{h}_n^j\in\mathcal{G}[x_1,\dotsc,x_n]$ for $j\in [n]$ by
\begin{equation}
\label{eq:hnj}
\hat{h}_n^j(x_i):=
\begin{cases}
\hat{f}_{1,j}(x_i)\quad&\text{for }i\in\{1,\dotsc,j\}\\
\hat{f}_{n,j+1}(x_i)\quad&\text{for }i\in\{j+1,\dotsc,n\}.
\end{cases}
\end{equation}
In other words, $\hat{h}_n^j$ is obtained by partitioning the data into two disjoint subsets, namely $\{(x_1,Y_1),\dotsc,(x_j,Y_j)\}$ and $\{(x_{j+1},Y_{j+1}),\dotsc,(x_n,Y_n)\}$, and then fitting separate increasing convex and increasing concave LSEs on the left and right pieces respectively. In general, $\hat{h}_n^j$ is not guaranteed to be S-shaped or even increasing on $[0,1]$, in which case $\hat{h}_n^j$ does not coincide with the LSE $\hat{f}_n^{x_j}$ over $\mathcal{H}^{x_j}\equiv\mathcal{H}^{x_j}[x_1,\dotsc,x_n]=\mathcal{F}^{x_j}\cap\mathcal{G}[x_1,\dotsc,x_n]$. Nevertheless, observe that $\hat{h}_n^j$ is the LSE over a larger subclass of $\mathcal{G}[x_1,\dotsc,x_n]$ that contains $\mathcal{H}^{x_j}$. Together with~\eqref{eq:GS17}, this immediately implies Proposition~\ref{prop:factA} below, a key fact that we will exploit in our algorithm.
\begin{proposition}
\label{prop:factA}
For $j\in [n]$, we have $\hat{h}_n^j=\hat{f}_n^{x_j}$ if and only if $\hat{h}_n^j\in\mathcal{H}^{x_j}$, i.e.\ if and only if either $j=n$ or 
\begin{equation}
\label{eq:hnjslope}
\frac{\hat{f}_{n,j+1}(x_{j+2})-\hat{f}_{n,j+1}(x_{j+1})}{x_{j+2}-x_{j+1}}\leq\frac{\hat{f}_{n,j+1}(x_{j+1})-\hat{f}_{1,j}(x_j)}{x_{j+1}-x_j}.
\end{equation}
If~\eqref{eq:hnjslope} holds, then $Y_j\leq\hat{h}_n^j(x_j)\leq\hat{h}_n^j(x_{j+1})\leq Y_{j+1}$.
\end{proposition}

\unparskip
In addition, we have the following crucial result for all global S-shaped LSEs over the class $\mathcal{H}\equiv\mathcal{H}[x_1,\dotsc,x_n]:=\mathcal{F}\cap\mathcal{G}$, namely those $\hat{f}_n^{x_{j'}}$ for which $j'\in\argmin_{1\leq j\leq n}S_n(\hat{f}_n^{x_j})$.
\begin{proposition}
\label{prop:srestrexact}
Given any S-shaped LSE $\tilde{f}_n$ over $\mathcal{H}$, if $j\in [n-1]$ is such that either $x_j$ is the smallest inflection point of $\tilde{f}_n$ or $x_{j+1}$ is the largest inflection point of $\tilde{f}_n$, then $\hat{h}_n^j=\tilde{f}_n$ and hence $Y_j\leq\tilde{f}_n(x_j)\leq\tilde{f}_n(x_{j+1})\leq Y_{j+1}$.
\end{proposition}

\unparskip
We explain in the final example of Section~\ref{sec:bdadj} that Proposition~\ref{prop:srestrexact} is a consequence of Proposition~\ref{prop:srestr}(c,\,d,\,e), whose proof also reveals why $\hat{h}_n^j=\hat{f}_n^{x_j}$ is not guaranteed to hold for a pre-specified $j\in [n]$. A further remark is that the localisation property for $\tilde{f}_n$ in Proposition~\ref{prop:srestrexact} is only valid for particular choices of partition of our data into subintervals, namely where the split occurs at the smallest or largest inflection points of $\tilde{f}_n$. In other words, if for example $x_j$ is chosen to be a kink of $\tilde{f}_n$ that is strictly to the left of the smallest inflection point, then $\tilde{f}_n$ is not guaranteed to agree with the increasing convex LSE $\hat{f}_{1,j}$ on $[x_1,x_j]$. This presents a substantial additional difficulty for both computation and theory in comparison with the problem of unimodal regression \citep{SZ01,Sto08}, where, for every jump $x_j$ of the unimodal LSE $\tilde{g}_n$ to the left of its mode, it is the case that $\tilde{g}_n$ agrees on $[x_1,x_j]$ with the increasing LSE based on $\{(x_i,Y_i):1 \leq i \leq j\}$. These issues are discussed in greater depth in Section~\ref{sec:bdadj}.

Propositions~\ref{prop:factA} and~\ref{prop:srestrexact} motivate the following generic procedure as an improvement on \texttt{ScanAll}:
\begin{algorithm}
\label{alg:sshape}
Generic algorithm for computing $(\hat{m}_n,\hat{f}_n)$.

\unparskip
\begin{enumerate}[label=(\Roman*)]
\item Discard all $j\in [n-1]$ for which $Y_j>Y_{j+1}$.
\item For each of the remaining indices $j\in [n]$, compute $\hat{f}_{1,j}$ based on $\{(x_i,Y_i):1\leq i\leq j\}$ and $\hat{f}_{n,j+1}$ based on $\{(x_i,Y_i):j+1\leq i\leq n\}$, and concatenate these to obtain $\hat{h}_n^j$ via~\eqref{eq:hnj}. Discard $j$ if $\hat{h}_n^j\notin\mathcal{H}^{x_j}$, i.e.\ if $j\leq n-1$ and
\[\frac{\hat{f}_{n,j+1}(x_{j+2})-\hat{f}_{n,j+1}(x_{j+1})}{x_{j+2}-x_{j+1}}>\frac{\hat{f}_{n,j+1}(x_{j+1})-\hat{f}_{1,j}(x_j)}{x_{j+1}-x_j}.\]
\item Let $\mathcal{J}$ be the set of indices $j\in [n]$ that are retained after Step II. Find $\tilde{\jmath}:=\sargmin_{j\in\mathcal{J}}S_n(\hat{h}_n^j)$ by computing $S_n(\hat{h}_n^j)=\inv{n}\sum_{i=1}^n\bigl(Y_i-\hat{h}_n^j(x_i)\bigr)^2$ for each $j\in\mathcal{J}$, and return $(x_{\tilde{\jmath}},\hat{h}_n^{\tilde{\jmath}})$.
\end{enumerate}

\unparskip
\end{algorithm}

\unparskip
To see that the output $(x_{\tilde{\jmath}},\hat{h}_n^{\tilde{\jmath}})$ of Algorithm~\ref{alg:sshape} is indeed $(\hat{m}_n,\hat{f}_n)$, note first that by Proposition~\ref{prop:factA}, the set $\mathcal{J}$ in Step III consists precisely of those $j\in [n]$ for which $\hat{h}_n^j=\hat{f}_n^{x_j}$. In addition, by Proposition~\ref{prop:srestrexact}, $\hat{\jmath}_n=\sargmin_{1\leq j\leq n}S_n(\hat{f}_n^{x_j})\in\mathcal{J}$ since $\hat{m}_n=x_{\hat{\jmath}_n}$ is the smallest inflection point of $\hat{f}_n=\hat{f}_n^{\hat{m}_n}$. Thus, $\tilde{\jmath}=\sargmin_{j\in\mathcal{J}}S_n(\hat{f}_n^{x_j})=\hat{\jmath}_n$, and hence $x_{\tilde{\jmath}}=\hat{m}_n$ and $\hat{h}_n^{\tilde{\jmath}}=\hat{f}_n$, as desired.

\unparskip
\begin{enumerate}[label=(\roman*)]
\item In advance of carrying out any least squares minimisation, we can restrict the set of candidates for $\hat{\jmath}_n$ based on just $n-1$ pairwise comparisons. If $(x_1,Y_1),\dotsc,(x_n,Y_n)$ are drawn according to a regression model~\eqref{eq:Model} featuring a continuous $f_0$ and independent and identically distributed errors with zero mean, then Step~I typically screens out about half of the indices in $[n]$ when $n$ is reasonably large.
\item For the remaining indices $j$ in Step~II, we do not attempt to compute the S-shaped function $\hat{f}_n^{x_j}$ based on all $n$ data points, but instead fit the increasing convex LSE $\hat{f}_{1,j}$ and the increasing concave LSE $\hat{f}_{n,j+1}$ using $j$ and $n-j$ observations respectively. 
\end{enumerate}

\unparskip
The main drawback of the \texttt{ScanSelected} algorithm, however, is that it fails to exploit the commonalities in the computation of $\hat{f}_{1,j}$ for different $j$ (and similarly of $\hat{f}_{n,j+1}$ for different $j$).  Our main computational contribution, then, is to show that for $k\in [j-1]$, it is possible to obtain $\hat{f}_{1,j}$ by modifying $\hat{f}_{1,k}$ appropriately when the observations $\{(x_i,Y_i):k<i\leq j\}$ are introduced.  We can therefore proceed in a sequential manner and hence make significant computational gains.

Recall that for $j \in [n]$ and a closed, convex cone $\Lambda \subseteq \R^j$, there exists a unique $L^2$-projection $\Pi_\Lambda\colon\mathbb{R}^j \rightarrow \Lambda$, given by
\[
\Pi_\Lambda(y) := \argmin_{u \in \Lambda} \|u-y\|.
\]
The key to our approach is to develop a mixed primal-dual bases algorithm \citep{FM89,Mey99} that allows us to compute $\Pi_{\Lambda}(L)$ when $L \subseteq \R^j$ is a line segment and $\Lambda$ is a polyhedral convex cone.  An important observation is that, given $v(0),v(1)\in\R^j$, the map $t \mapsto \Pi_\Lambda\bigl((1-t)v(0)+tv(1)\bigr)$ is continuous and piecewise linear on $[0,1]$, where the individual linear pieces correspond to projections onto different faces of $\Lambda$; see Remark~\ref{rem:conefaces}.  This enables us to compute $\Pi_\Lambda\bigl(v(1)\bigr)$ when $\Pi_\Lambda\bigl(v(0)\bigr)$ is known.  Indeed, we give a detailed description of a general procedure for this task in Algorithm~\ref{alg:coneproj} in Section~\ref{sec:appendix}, and we focus here on its application to increasing convex regression (increasing concave regression for the right-hand end can be handled very similarly). 
In this case, the cones of particular interest to us are those of increasing convex sequences based on $x_1,\ldots,x_j$ for some $j \in [n]$, which we denote by
\begin{equation}
\label{eq:Thetak}
\Lambda^j:=\bigl\{\bigl(g(x_1),\dotsc,g(x_j)\bigr):g\in\mathcal{F}^1\bigr\}=\biggl\{(z_1,\dotsc,z_j) \in \R^j:0\leq\frac{z_2-z_1}{x_2-x_1}\leq\cdots\leq\frac{z_j-z_{j-1}}{x_j-x_{j-1}}\biggr\}.
\end{equation}
Given $k \in [j-1]$ and supposing that we have already fitted the increasing convex LSE $\hat{f}_{1,k}$ (which is linear on $[x_{k-1},1]$), an appropriate choice of $v(0),v(1)$ is
\begin{equation}
\label{eq:seqconv}
v(0)=\bigl(Y_1,\dotsc,Y_{k},\hat{f}_{1,k}(x_{k+1}),\dotsc,\hat{f}_{1,k}(x_j)\bigr)\quad\text{and}\quad v(1)=(Y_1,\dotsc,Y_j);
\end{equation}
indeed, $\Pi_{\Lambda^j}\bigl(v(1)\bigr)=\bigl(\hat{f}_{1,j}(x_1),\dotsc,\hat{f}_{1,j}(x_j)\bigr)$ is what we seek to compute, and moreover we claim that $\Pi_{\Lambda^j}\bigl(v(0)\bigr)=\bigl(\hat{f}_{1,k}(x_1),\dotsc,\hat{f}_{1,k}(x_j)\bigr)$ (which is known).  To establish this claim, observe that for any $u \equiv (u_1,\ldots,u_j) \in \Lambda^j$, we have
\begin{equation}
\label{Eq:AddPoints}
\|v(0) - u\|^2 \geq \sum_{i=1}^k (Y_i - u_i)^2 \geq \sum_{i=1}^k \bigl(Y_i - \hat{f}_{1,k}(x_i)\bigr)^2 = \bigl\|v(0) - \bigl(\hat{f}_{1,k}(x_1),\dotsc,\hat{f}_{1,k}(x_j)\bigr)\bigr\|^2,
\end{equation}
and $\bigl(\hat{f}_{1,k}(x_1),\dotsc,\hat{f}_{1,k}(x_j)\bigr) \in \Lambda^j$.  In fact, we will apply this version of the mixed primal-dual bases algorithm with $k=j-1$, so that the observations $Y_1,\dotsc,Y_n$ are introduced sequentially.  Note that when $Y_j\geq\hat{f}_{1,j-1}(x_j)$, we have by the same argument as in~\eqref{Eq:AddPoints} that $\bigl(\hat{f}_{1,j}(x_1),\dotsc,\hat{f}_{1,j}(x_j)\bigr)=\bigl(\hat{f}_{1,j-1}(x_1),\dotsc,\hat{f}_{1,j-1}(x_{j-1}),Y_j\bigr)$, so no calculations are required.  We refer to this sequential implementation of Algorithm~\ref{alg:sshape} as \texttt{SeqConReg}.

\hfparskip
\section{Theoretical properties of S-shaped least squares estimators}
\label{sec:lsreg}

\subsection{Worst-case and adaptive sharp oracle inequalities}
\label{sec:lsregoracle}

Our first main results of this section consist of worst-case and adaptive sharp oracle inequalities for S-shaped least squares estimators.  These reveal not only risk bounds when our S-shaped regression function hypothesis is correctly specified, but also control the way in which the performance of the estimators deteriorate as the model becomes increasingly misspecified.  

We will work in the setting of model~\eqref{eq:Model}, and now make the following assumption on the errors:
\begin{assumption}
\label{ass:sG}
$\{\xi_i\equiv\xi_{ni}:1\leq i\leq n\}$ is a collection of independent sub-Gaussian random variables with parameter 1, so that $\E(e^{t\xi_{ni}})\leq e^{t^2/2}$ for all $t\in\R$ and $i\in [n]$. 
\end{assumption}

\unparskip
For fixed $n\in\N$ and $f\colon [0,1]\to\R$, we write $x_i\equiv x_{ni}$ for $i\in [n]$ and let $\norm{f}_{n}:=\norm{f}_{L^2(\Pr_n^X)}=\bigl(\sum_{i=1}^n f^2(x_i)/n\bigr)^{1/2}$. Also, for $f\in\mathcal{H}\equiv\mathcal{H}[x_1,\dotsc,x_n]$, define $V(f):=f(x_n)-f(x_1)=\max_{1\leq i\leq n}f(x_i)-\min_{1\leq i\leq n}f(x_i)$ and denote by $k(f)$ the number of affine pieces of $f$, so that $k(f)$ is the smallest $k\in [n]$ with the property that $f$ is affine on each of $k$ subintervals $I_1,\dotsc,I_k$ that partition $[0,1]$.
\begin{theorem}
\label{thm:worstcase}
For fixed $n\geq 2$, suppose that Assumption~\ref{ass:sG} holds and let $\tilde{f}_n$ be any LSE over $\mathcal{F}$. Let $R:=n^{-1}(x_n-x_1)/\min_{2\leq i\leq n}(x_i-x_{i-1})$. Then there exists a universal constant $C>0$ such that for every $f_0\colon [0,1]\to\R$ and $t>0$, we have
\begin{equation}
\label{eq:worstpr}
\norm{\tilde{f}_n-f_0}_n\leq\inf_{f\in\mathcal{H}}\:\biggl\{\norm{f-f_0}_n+\frac{C\bigl(1+V(f)\bigr)^{1/3}}{n^{1/3}}\wedge\frac{CR^{1/10}\bigl(1+V(f)\bigr)^{1/5}}{n^{2/5}}\biggr\}+\sqrt{\frac{8t}{n}}
\end{equation}
with probability at least $1-e^{-t}$. 
\end{theorem}

\unparskip
By integrating this tail bound, we obtain the worst-case risk bound
\begin{align}
\label{eq:worstex1}
\E_{f_0}(\norm{\tilde{f}_n-f_0}_n) &\leq\inf_{f\in\mathcal{H}}\:\biggl\{\norm{f-f_0}_n+\frac{C\bigl(1+V(f)\bigr)^{1/3}}{n^{1/3}}\wedge\frac{CR^{1/10}\bigl(1+V(f)\bigr)^{1/5}}{n^{2/5}}\biggr\}+ \sqrt{\frac{2\pi}{n}}.
\end{align}
In the special case where $f_0 \in \mathcal{F}$, we may take $f=f_0$ in Theorem~\ref{thm:worstcase} to conclude that
\[
\E_{f_0}(\norm{\tilde{f}_n-f_0}_n) \lesssim \frac{\bigl(1+V(f_0)\bigr)^{1/3}}{n^{1/3}}\wedge\frac{R^{1/10}\bigl(1+V(f_0)\bigr)^{1/5}}{n^{2/5}};
\]
thus, when $R$ and $V(f_0)$ are of constant order, we obtain a worst-case risk bound of order $n^{-2/5}$.  More generally,~\eqref{eq:worstpr} and~\eqref{eq:worstex1} reveal the impact of both non-equispaced design and the range of the signal.  In fact, an alternative, more complicated definition of $R$ is possible, and this further refines our bounds for certain designs; see the discussion following the proof of Theorem~\ref{thm:worstcase} in Section~\ref{subsec:oracle}.  To see that the rate of order $n^{-2/5}$ cannot in general be attained for arbitrary configurations of design points, we appeal to~\citet[Theorem~4.5]{Bel18} for a suitable minimax lower bound: for any $V\geq n^{-1/2}$, there exist design points $x_1<\cdots<x_n$ that depend on $V$ such that if $\xi_1,\dotsc,\xi_n\iid N(0,1)$ in~\eqref{eq:Model}, then
\[\inf_{\breve{g}_n}\sup_{f_0\in\mathcal{F}^1:V(f_0)\leq 2V}\Pr_{f_0}\bigl(\norm{\breve{g}_n-f_0}_n\geq C(V/n)^{1/3}\bigr)\geq c,\]
where the infimum is taken over all estimators $\breve{g}_n\equiv\breve{g}_n(x_1,Y_1,\dotsc,x_n,Y_n)$, and $c,C>0$ are universal constants.

Another very attractive aspect of Theorem~\ref{thm:worstcase} is that, in cases where $f_0 \notin \mathcal{F}$, we can control the performance of an LSE $\tilde{f}_n$ over $\mathcal{F}$ via approximation error and estimation error terms.  The fact that the approximation error term $\norm{f-f_0}_n$ has leading constant 1 (which is the best possible) is the reason that~\eqref{eq:worstpr} and~\eqref{eq:worstex1} are referred to as sharp oracle inequalities. 

To complement the worst-case sharp oracle inequality~\eqref{eq:worstex1} above, we now consider the more favourable situation where $f_0$ is well approximated by a piecewise affine function with not too many affine pieces.  The fact that an LSE $\tilde{f}_n$ over $\mathcal{F}$ can approximate such a signal with a relatively small number of kinks suggests that we may be able to obtain improved sharp oracle inequalities in such cases.
\begin{theorem}
\label{thm:adaptive}
For fixed $n \geq 2$, suppose that Assumption~\ref{ass:sG} holds, and let $\tilde{f}_n$ be any LSE over $\mathcal{F}$. Then for every $f_0\colon [0,1]\to\R$ and $t>0$, we have
\begin{equation}
\label{eq:adaptpr}
\norm{\tilde{f}_n-f_0}_n\leq\inf_{f\in\mathcal{H}}\:\biggl\{\norm{f-f_0}_n+\sqrt{\frac{32\bigl(k(f)+1\bigr)}{n}\log\rbr{\frac{en}{k(f)+1}}}\biggr\}+\sqrt{\frac{2(t+\log n)}{n}}
\end{equation}
with probability at least $1-e^{-t}$.
\end{theorem}

\unparskip
As with Theorem~\ref{thm:worstcase}, we can integrate the tail bound from~\eqref{eq:adaptpr} to obtain
\begin{align}
\label{eq:adaptex1}
\E_{f_0}(\norm{\tilde{f}_n-f_0}_n)&\leq\inf_{f\in\mathcal{H}}\:\biggl\{\norm{f-f_0}_n+\sqrt{\frac{32\bigl(k(f)+1\bigr)}{n}\log\rbr{\frac{en}{k(f)+1}}}\biggr\}+\sqrt{\frac{2\log n}{n}}+\sqrt{\frac{\pi}{2n}}\notag\\
&\leq\inf_{f\in\mathcal{H}}\:\biggl\{\norm{f-f_0}_n+8\,\sqrt{\frac{k(f)+1}{n}\log\rbr{\frac{en}{k(f)+1}}}\biggr\}.
\end{align}
In particular, we see from~\eqref{eq:adaptex1} that if $f_0 \in \mathcal{F}$ has $k$ affine pieces, then any LSE $\tilde{f}_n$ over $\mathcal{F}$ attains the parametric rate $k^{1/2}/n^{1/2}$, up to a logarithmic factor. 

Adaptation to signals of low complexity is one of the particularly intriguing aspects of shape-constrained estimators~\citep{GS18,Samworth2018}.  For instance,~\citet{GS13},~\citet{CGS15} and~\cite{CL19} investigated the adaptive behaviour of univariate convex, isotonic and unimodal LSEs respectively when the truth is well approximated by a function with a small number of affine or constant pieces. For multivariate extensions of these results, see for example~\citet{HW16},~\citet{KGGS20} and~\citet{Han21} among others. Sharp oracle inequalities of a similar flavour to Theorem~\ref{thm:adaptive} have been obtained for a variety of LSEs~\citep{Bel18}, including multivariate isotonic LSEs~\citep{HWCS19,PS21}. In log-concave density estimation, adaptation results of this type were established for the log-concave maximum likelihood estimator by~\cite{KGS18} and~\cite{FGKS21} in univariate and multivariate settings respectively.  Finally, \citet{BB16} introduced a $\rho$-estimation framework for univariate shape-constrained estimation and studied its adaptation properties.

\hfparskip
\subsection{Inflection point estimation}
\label{subsec:inflectpt}
A particular feature of S-shaped function estimation that differentiates it from other shape-constrained estimation problems is the existence of an inflection point $m_0$.  In some respects, this is like a boundary point, because it represents the point of transition from convex to concave parts of the function, and the behaviour of the function is therefore less regulated there (in particular, the derivative of an S-shaped function may diverge to infinity as we approach the inflection point).  On the other hand, when $m_0 \in (0,1)$, we may well have design points on either side of $m_0$, and in that sense the inflection point may be regarded as an interior point. The distinguished nature of the inflection point means that its location is often of interest in applications such as the modelling of economic growth \citep[e.g.][]{JSF2007} and disease progression in longitudinal studies~\citep[e.g.][]{LCMWG20}. For instance, in the latter work, S-shaped functions were used to model the deterioration in motor function associated with Huntington's disease, and the estimated inflection points from a nonparametric procedure were seen to be clinically useful indicators of the onset of severe motor dysfunction, in the sense of having the potential to facilitate timely diagnosis and intervention.

In studying the inflection point estimation problem, we will assume that $f_0 \in \mathcal{F}$ and the following additional conditions hold:
\begin{assumption}
\label{ass:inflection}
Suppose that $f_0\in\mathcal{F}$ has a unique inflection point $m_0\in (0,1)$, and that there exist $B>0$ and $\alpha\in (0,1)\cup (1,\infty)$ such that as $x\to m_0$, we have
\begin{equation}
f_0(x)=
\label{eq:smoothness}
\begin{cases}
f_0(m_0)-B\bigl(1+o(1)\bigr)\sgn(x-m_0)\abs{x-m_0}^\alpha&\text{when }\alpha\in (0,1)\\
f_0(m_0)+f_0'(m_0)(x-m_0)+B\bigl(1+o(1)\bigr)\sgn(x-m_0)\abs{x-m_0}^\alpha&\text{when }\alpha>1.
\end{cases}
\end{equation}
In the regression model~\eqref{eq:Model}, suppose also that $x_{ni}=i/n$ and $\xi_{ni}\eqd\xi$ for all $n \in \mathbb{N}$ and $i\in [n]$, where $\xi$ is a sub-Gaussian random variable with parameter 1.
\end{assumption}

\unparskip
When $\alpha\geq 3$ is an integer,~\eqref{eq:smoothness} holds if (a) $f_0$ is $\alpha$-times continuously differentiable in a neighbourhood of $m_0$, and (b) $f_0^{(k)}(m_0)=0\neq f_0^{(\alpha)}(m_0)$ for $2\leq k\leq \alpha-1$. Under this stronger assumption, $\alpha$ must in fact be odd, and $f_0^{(\alpha)}(m_0)<0$. Indeed, for all $x\in [0,1]$ sufficiently close to the inflection point $m_0$, we have $f_0''(x)\geq 0$ if $x\leq m_0$ and $f_0''(x)\leq 0$ if $x\geq m_0$, and since $f_0^{(\alpha)}$ is continuous at $m_0$, a Taylor expansion reveals that $f_0''(x)=f_0^{(\alpha)}(m_0)\bigl(1+o(1)\bigr)(x-m_0)^{\alpha-2}/(\alpha-2)!$ as $x\to m_0$. 
\begin{theorem}
\label{thm:inflection}
Let $(\tilde{f}_n)$ be any sequence of LSEs over $\mathcal{F}$, and for each $n$, let $\tilde{m}_n$ be an inflection point of $\tilde{f}_n$. Under Assumption~\ref{ass:inflection}, we have $\tilde{m}_n-m_0=O_p\bigl((n/\log n)^{-1/(2\alpha+1)}\bigr)$.
\end{theorem}

\unparskip
We mention that~\citet{LM17} study a least squares estimator over a subclass of $\mathcal{F}$ consisting of cubic splines (where the number of knots is of order $n^{1/9}$); they show that its inflection point converges to the true $m_0$ at rate $O_p(n^{-8/63})$ in a random design setting 
where $f_0$ satisfies (a stronger version of)~\eqref{eq:smoothness} with $\alpha=3$. The proof of their Theorem~2 relies on a quantitative result on the quality of local approximations to $f_0$ near $m_0$ by convex or concave functions~\citep[Lemma~2]{LM17}, as well as a global rate of convergence for their spline-based estimator. 

In our setting, Theorem~\ref{thm:inflection} shows that the inflection point estimator $\tilde{m}_n$ (based on an LSE $\tilde{f}_n$ over the entire class $\mathcal{F}$) converges to $m_0$ at rate $O_p\bigl((n/\log n)^{-1/7}\bigr)$ when $\alpha=3$.  The proof of Theorem~\ref{thm:inflection}, which is given in Section~\ref{sec:lsregmainproofs}, is lengthy and broken up into several steps, each of which requires some delicate technical arguments; see Figure~\ref{fig:inflection} for an illustration. The crucial Step 2a exploits the observation that if $\tilde{m}_n$ is a long way from $m_0$, then there is a long interval between the two on which one of $f_0,\tilde{f}_n$ is convex and the other is concave. On such an interval, we show that $\tilde{f}_n$ has a long affine piece, as would be intuitively expected, and thereby quantify the approximation error due to misspecification; see Lemma~\ref{lem:step2lbd}.  Another important aspect of our proof strategy is that we find a suitable way to localise the analysis of $\tilde{f}_n$ to a neighbourhood of~$m_0$, rather than rely on global considerations that would lead to a suboptimal bound. As we explain in Section~\ref{sec:bdadj}, our localisation technique for convex or S-shaped LSEs relies on non-trivial `boundary adjustments' that are not needed for isotonic or unimodal LSEs. 
Nevertheless, a simpler version of the proof of Theorem~\ref{thm:inflection} allows us to recover the result of~\citet{SZ01} on the rate of convergence of the mode of the LSE of a unimodal regression function, at least under our sub-Gaussian assumption on the errors $\xi_{ni}$ and their local smoothness condition (1.3).

The rate of convergence of $\tilde{m}_n$ to $m_0$ in Theorem~\ref{thm:inflection} matches that in the following complementary local asymptotic minimax lower bound, up to a logarithmic factor. For $r>0$, let $\mathcal{F}(f_0,r):=\{f\in\mathcal{F}:\int_0^1\,(f-f_0)^2<r^2\}$. Although $f_0$ has a unique inflection point $m_0$ under Assumption~\ref{ass:inflection}, not every function in $\mathcal{F}(f_0,r)$ has a unique inflection point, so for $f\in\mathcal{F}$, we denote by $\mathcal{I}_f$ the subinterval of inflection points of $f$ and define $d(x,\mathcal{I}_f):=\inf_{z\in\mathcal{I}_f}\,\abs{x-z}$ for $x\in [0,1]$.
\begin{proposition}
\label{prop:lamlb}
Under Assumption~\ref{ass:inflection}, and with $\xi_{n1},\dotsc,\xi_{nn}\iid N(0,1)$ for all $n$, we have
\begin{equation}
\label{eq:lamlb}
\sup_{\tau>0}\,\liminf_{n\to\infty}\,\inf_{\breve{m}_n}\,\sup_{f\in\mathcal{F}(f_0,\tau/\sqrt{n})}n^{1/(2\alpha+1)}\,\E_f\bigl(d(\tilde{m}_n,\mathcal{I}_f)\bigr)>0,
\end{equation}
where the infimum is taken over all estimators $\breve{m}_n\equiv\breve{m}_n(x_1,Y_1,\dotsc,x_n,Y_n)$ taking values in $[0,1]$, and $\E_f$ is the expectation operator under the model~\eqref{eq:Model} with $f$ in place of $f_0$.
\end{proposition}

\section{Simulations and real data example}
\label{sec:simulations}

In this section, we first investigate the computation time and empirical performance of our S-shaped estimator in some numerical experiments. We then demonstrate the use of our estimator in a real data application to air pollution modelling.

\hfparskip
\subsection{Computation time}
\label{sec:comptime}

We compare the running time of our sequential cone projection Algorithm~\ref{alg:coneproj}, which we henceforth refer to as \texttt{SeqConReg}, with two other possible approaches.  The first, which we call \texttt{ScanAll}, relies on a brute-force search that scans through all possible inflection points $m \in \{x_1,\ldots,x_n\}$ as described in the introduction, performing least squares over~$\mathcal{F}^m$, and determining the candidate that minimises the residual sum of squares.  Here the active set least squares procedure used for each $m$ is based on a simple modification of the \texttt{R} package \texttt{scar} \citep{scar14,ChenSamworth2016}.  The second approach, which we call \texttt{ScanSelected}, is based on the observation in Step I of Algorithm~\ref{alg:sshape} that there is no need to scan through all design points. Instead, we restrict attention to those indices $j$ for which $Y_j\leq Y_{j+1}$, fitting a convex increasing function to $\{(x_i,Y_i):1\leq i\leq j\}$, a concave and increasing function to $\{(x_i,Y_i):j+1\leq i\leq n\}$ (both using \texttt{scar}), before finding the smallest $j$ that minimises the residual sum of squares.

For $n \in \{100,200,500,1000,2000\}$, we set $x_i = i/(n+1)$ and $Y_i = \sin\big(\pi(x_i-0.5)\big)+\sigma \epsilon_i$ for $i=1,\ldots,n$, where $\epsilon_1,\ldots,\epsilon_n$ are independent normal random variables with zero mean and unit variance. Here, to examine the impact of the signal-to-noise ratio on the running time, we also vary the value of~$\sigma \in \{1,0.1,0.01\}$, and plot the average running time of the different approaches in Figure~\ref{Fig:Timing}. We see that \texttt{SeqConReg} is the fastest among all three approaches, being approximately 10 times more efficient than \texttt{ScanSelected} and 40 times faster than \texttt{ScanAll}.  The ratio of the timings becomes larger as the signal-to-noise ratio increases, because the resulting fitted function has more knots, which makes it more appealing to use algorithms of a sequential nature, such as \texttt{SeqConReg}.  

\begin{figure}[!ht]
\centering
\includegraphics[width=\textwidth]{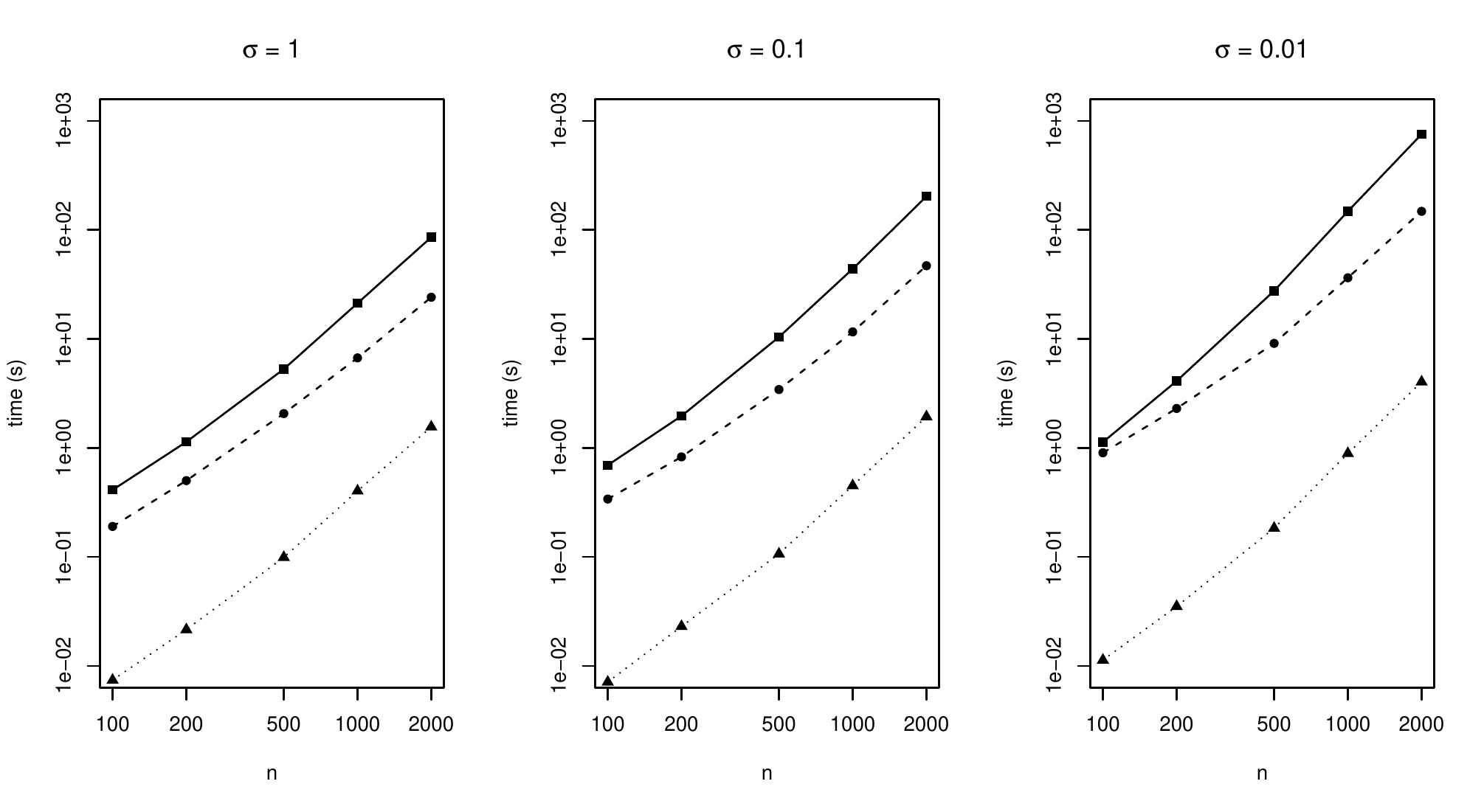}
\caption{\label{Fig:Timing}Log-log plots of the running time (in seconds) of the \texttt{SeqConReg} (\raisebox{1pt}{\scalebox{0.8}{$\blacktriangle$}}), \texttt{ScanSelected} ($\bullet$) and \texttt{ScanAll}~(\raisebox{0.2pt}{\scalebox{0.6}{$\blacksquare$}}) algorithms for least squares estimation of an S-shaped function, for sample sizes $n \in \{100,200,500,1000, 2000\}$ and noise levels $\sigma \in \{1, 0.1, 0.01\}$.}
\end{figure}

\hfparskip
\subsection{Statistical performance}
\label{sec:StatPerf}

We compare our estimator (denoted by LSE below) with the following alternatives: 

\unparskip
\begin{itemize}
\item Spline: The method of \citet{LM17}, based on cubic B-splines with shape constraints, which is implemented in the \texttt{R} package \texttt{ShapeChange} \citep{LM16};
\item SCKLS: The shape-constrained kernel least squares method of \citet{YCJM19, YCJK20} based on local linear kernels;\footnote{To give more implementation details, we run SCKLS with $M=50$ evaluation points and select the kernel bandwidth according to the method of \citet{RSW1995}.} 
\item BEDE and BESE: The bisection extremum distance estimator and bisection extremum surface estimators of \citet{Christopoulos2016}, both developed based on the geometric properties of the inflection point for a smooth function and implemented in the \texttt{R} package \texttt{inflection} \citep{Christopoulos2019}.  
\end{itemize}

\unparskip
For LSE, Spline and SCKLS, we assess their performance based on both the average $L^2(\mathbb{P}_n)$ loss and the mean absolute error of the estimated inflection point location, while for BEDE and BESE we compute only the mean absolute error of the estimated inflection point location. All results are based on numerical experiments over 1000 repetitions.

For $n \in \{100,200,500,1000\}$, and design points $x_1,\ldots,x_n$, we set $Y_i = f_j(x_i)+ 0.1 \epsilon_i$ for $i=1,\ldots,n$, where $\epsilon_1,\ldots,\epsilon_n \stackrel{\mathrm{iid}}{\sim} N(0,1)$, for four different choices of signal function $f_j$:
\begin{alignat}{2}
&f_1(x)=
\begin{cases}
2(0.3 - \sqrt{0.09-x^2}) & \mbox{ for } x \in [0,0.3)\\
2\{0.3+\sqrt{0.49-(1-x)^2}\} & \mbox{ for } x \in [0.3,1] 
\end{cases};
\qquad\qquad&&f_3(x) = x + \mathbbm{1}_{\{x \geq 0.3\}};\notag\\
\label{eq:signals}
&f_2(x)= \sin\bigl((x-0.3)\pi/1.4\bigr)\mathbbm{1}_{\{x \geq 0.3\}};
\qquad\qquad&&f_4(x) = 4/\bigl(1+e^{-2(x-0.3)}\bigr).
\end{alignat}
These signals are plotted in Figure~\ref{Fig:Signals}. The signals are designed in such a way that their ranges over $[0,1]$ are roughly the same. Furthermore, they all belong to $\mathcal{F}$ and have a unique inflection point at $m_0=0.3$.  Note that $f_1$ satisfies Assumption~\ref{ass:inflection} with $\alpha=1/2$, and $f_2$ and $f_3$ do not satisfy Assumption~\ref{ass:inflection} for any $\alpha>0$, while $f_4$ satisfies the assumption with $\alpha=3$.

\begin{figure}[htb!]
\centering
\includegraphics[scale=0.85]{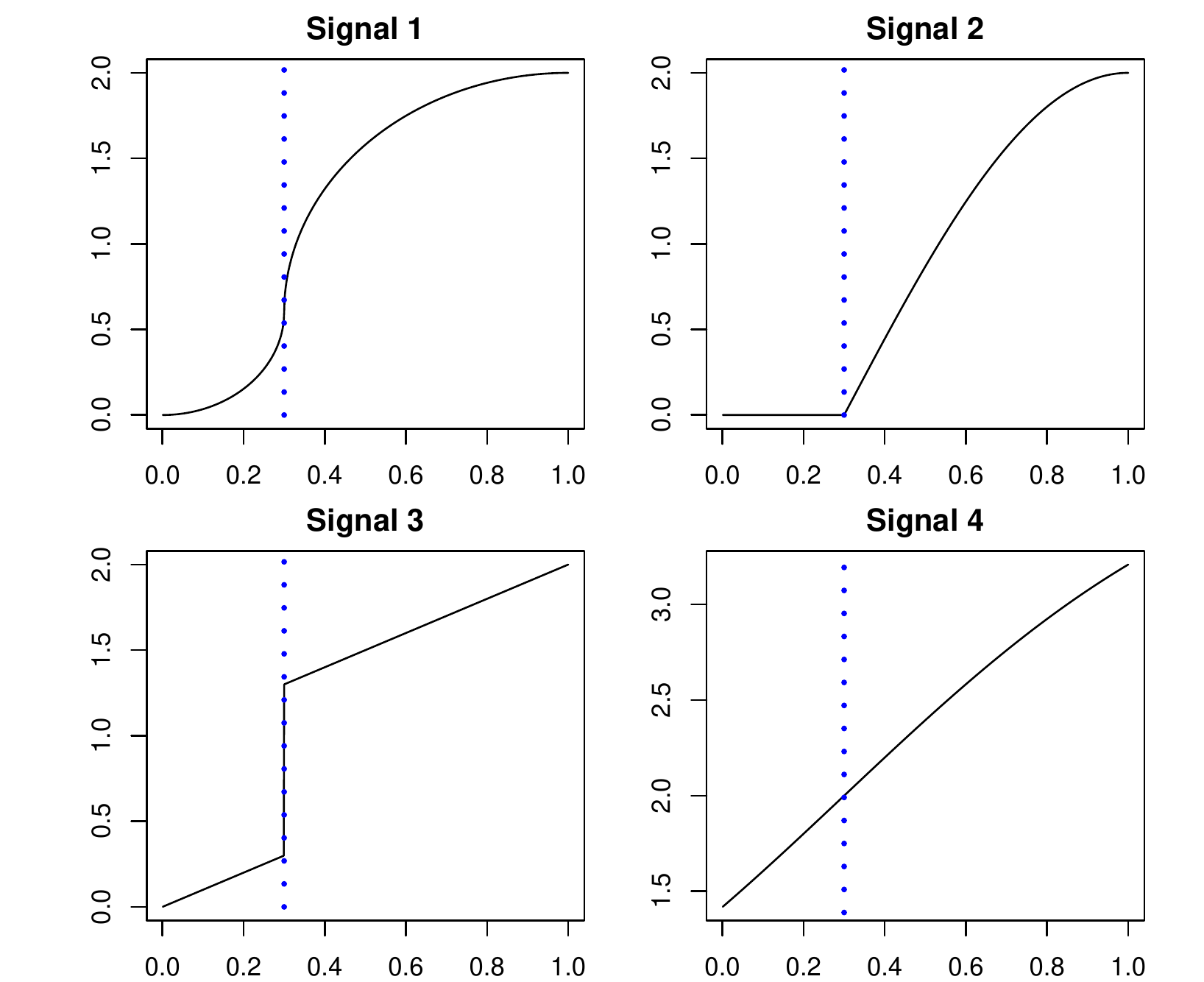}
\medskip
\caption{\label{Fig:Signals}Plots of the signals $f_1,f_2,f_3,f_4$ defined in~\eqref{eq:signals}, with the inflection points highlighted by dashed blue lines.}
\end{figure}

We consider two different designs by setting $x_i = F^{-1}\bigl(i/(n+1)\bigr)$ for $i = 1,\ldots,n$, where $F$ is the distribution function of either the $U[0,1]$ or $\mathrm{Beta}(4,8)$ distributions. In the second setting, the design points are not equally spaced, and $m_0=0.3$ is the mode of the $\mathrm{Beta}(4,8)$ distribution. The results are shown in Figures~\ref{Fig:sim_res_1} and~\ref{Fig:sim_res_2}.

\begin{figure}[htbp]
\centering
\includegraphics[width=\textwidth,height=0.88\textheight]{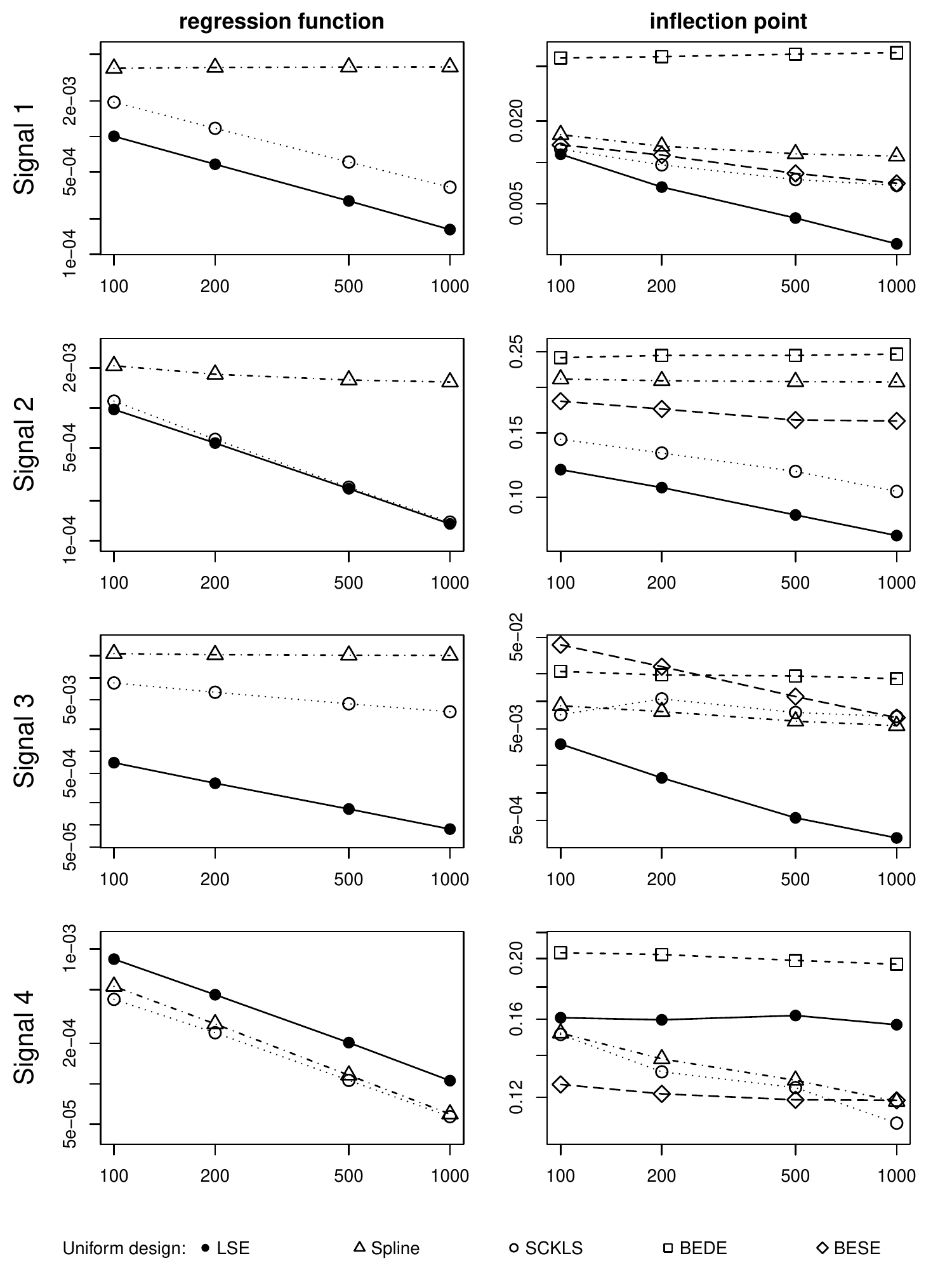}
\medskip
\caption{\label{Fig:sim_res_1} Log-log plots of the mean squared error of the fitted function on the design points, as well as the mean absolute distance between the estimated and true inflection points, based on $n=100,200,500, 1000$ observations when the design points are equispaced and the signals are as in Figure~\ref{Fig:Signals}.}
\end{figure}

\begin{figure}[htbp]
\centering
\includegraphics[width=\textwidth,height=0.88\textheight]{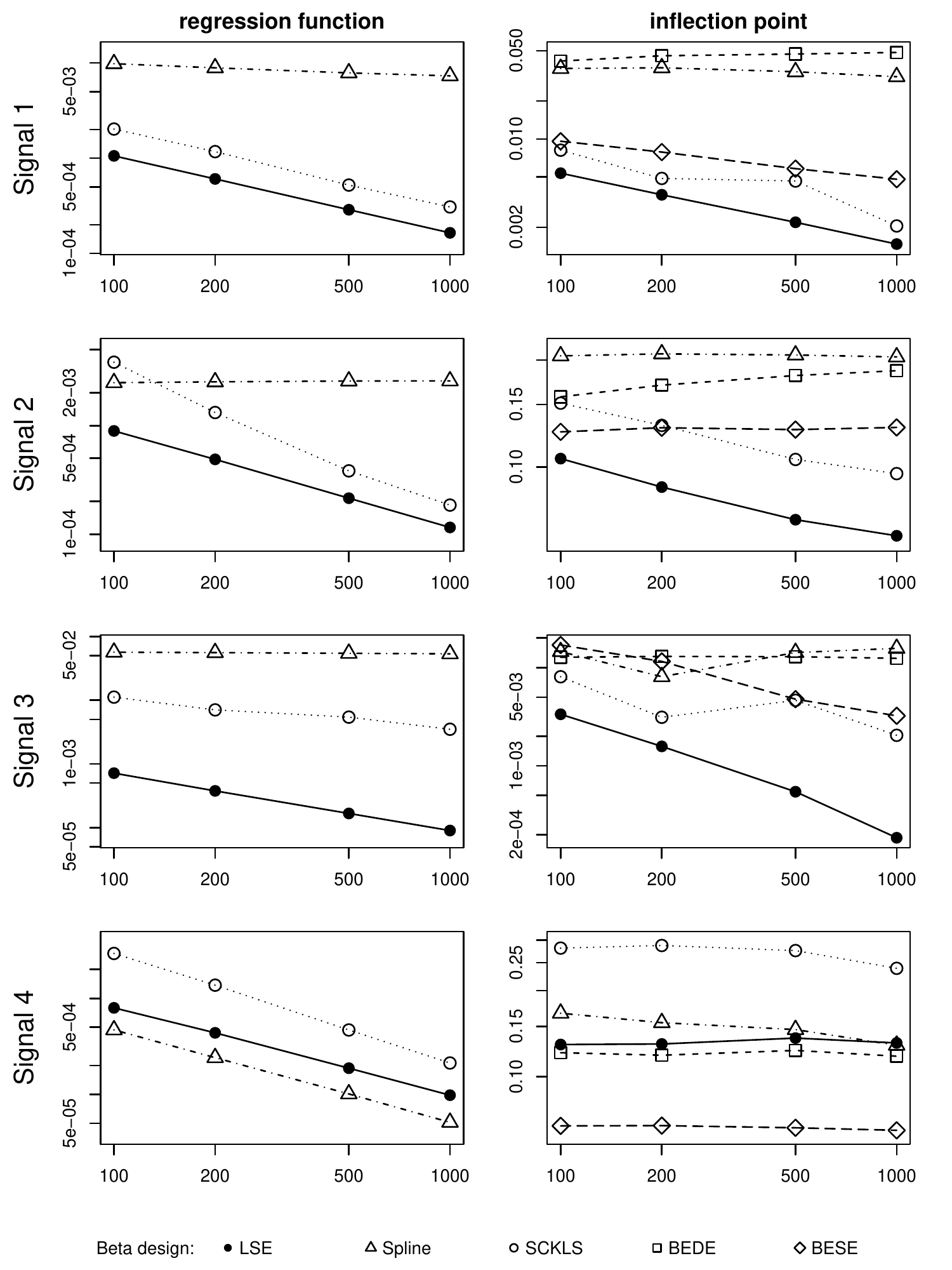}
\medskip
\caption{\label{Fig:sim_res_2}  Log-log plots of the mean squared error of the fitted function on the design points, as well as the mean absolute distance between the estimated and true inflection points, based on $n=100,200,500, 1000$ observations when the design points are quantiles of a $\mathrm{Beta}(4,8)$ distribution and the signals are as in Figure~\ref{Fig:Signals}.}
\end{figure}

For the estimation of the regression function, the LSE performs well in all cases; in particular, it is able to adapt to inhomogeneous smoothness levels and asymmetric designs.  On the other hand, the spline- and kernel-based approaches struggle in this regard, and perform much worse for signals $f_1$ and $f_3$ especially.  In fact, the spline-based method appears to be inconsistent for signals $f_1$ and $f_3$, and the kernel-based approach seems to suffer the same problem for signal $f_3$ too. 
For the estimation of the inflection point, the story has some similarities, but also some differences: for signals $f_1$, $f_2$ and $f_3$, the least squares approach provides more reliable estimates, for two main reasons.  First, it is able to adapt to a much wider range of local smoothnesses around $m_0$.  Second, by carefully comparing Figure~\ref{Fig:sim_res_2} to Figure~\ref{Fig:sim_res_1}, we see that the least squares approach is also able to take advantage of the additional design points near $m_0$ under the beta design to obtain improved estimation performance (relative to the uniform design).  For signal $f_4$, the other methods are able to exploit the homogeneity of the signal across the entire domain (and the symmetry of the signal around the inflection point) and tend to have smaller mean absolute error than the least squares approach.  We recall Figure~\ref{Fig:Parametric}, which further illustrates the dangers of assuming smoothness of an S-shaped signal when it is not present.

\hfparskip
\subsection{Real data example}
\label{subsec:realdata}
In this subsection, we apply our nonparametric S-shaped procedure to $n=221$ LIDAR (light detection and ranging) 
measurements for determining atmospheric concentrations of mercury emissions from the Bella Vista geothermal power station in Italy. This dataset, which is of interest from an air pollution modelling perspective, is discussed at length by~\citet{RWC03} and included in the \texttt{R} package \texttt{SemiPar}~\citep{Wan18}. 

To explain the rationale behind the use of the S-shaped regression model~\eqref{eq:Model} in this context, we begin by briefly outlining the physical background and experimental setup; see~\citet{EFSS89,ERSWDFM92} and~\citet[Section~2]{HHBRE96} for further details.\footnote{For additional physical explanations and graphical illustrations, see for example \url{http://www.nist.gov/programs-projects/differential-absorption-lidar-detection-and-quantification-greenhouse-gases} as well as \url{http://dialtechnology.info/history.html}.} 
In this instance, the LIDAR equipment was set up at a fixed location downwind of the power station, at a distance of 390--720m from the bulk of the mercury plume. The DIAL (differential absorption LIDAR) technique involves firing two laser beams in quick succession in the same direction
towards the plume, where the first beam contains light at the resonant wavelength $\lambda_{\mathrm{on}}=253.6\mathrm{nm}$ of mercury while the second `reference' beam is set to a slightly different `off-resonant' wavelength $\lambda_{\mathrm{off}}$. The light in both beams is scattered (or reflected back) to roughly the same extent by particles and aerosols in the atmosphere, but the light at wavelength $\lambda_{\mathrm{on}}$ is absorbed much more strongly by atoms of mercury, the pollutant of interest. The LIDAR apparatus records the intensity (i.e.\ power) of the reflected signals from both incident beams as a function of time elapsed, which is proportional to the distance travelled by the light before it is reflected back towards the source. The latter is the independent variable \texttt{range} in the dataset. 
The intensity curves from 100 pairs of laser shots in the same direction were then averaged to produce power estimates $P(r_i;\lambda_{\mathrm{on}})$ and $P(r_i;\lambda_{\mathrm{off}})$ for $n=221$ equispaced values $r_i$ of \texttt{range} between 390m and 720m (at intervals of 1.5m). In view of the physical reasons outlined above, the relative sizes of these two quantities for different $r_i$ can be used to estimate how the atmospheric concentration $g_0(r)$ of mercury (in $\mathrm{ng}/\mathrm{m}^3$) varies with distance $r$ (in metres) along the path of the laser beams.

\unparskip
\begin{figure}[htb!]
\centering
\includegraphics[width=0.475\textwidth]{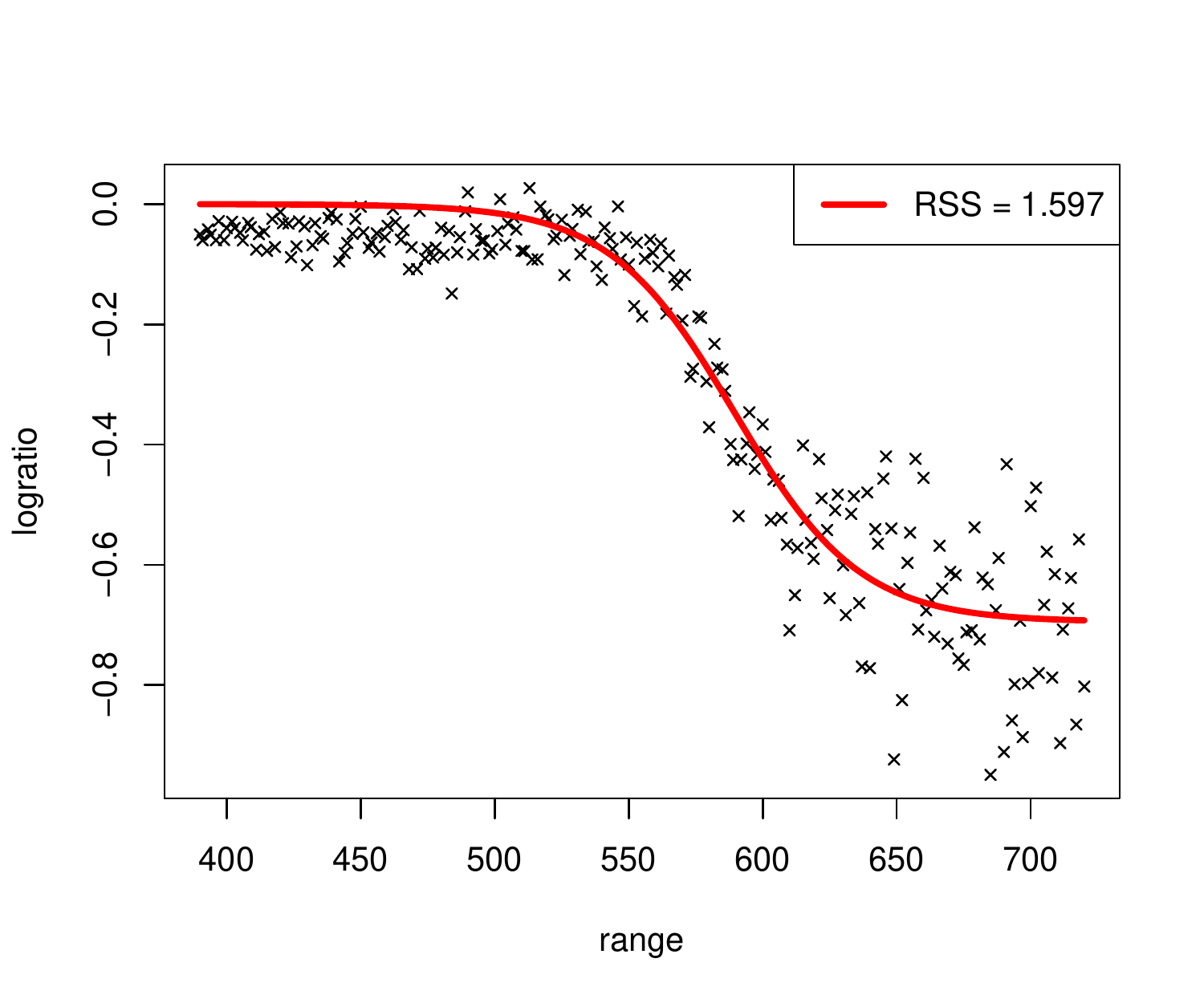}
\includegraphics[width=0.475\textwidth]{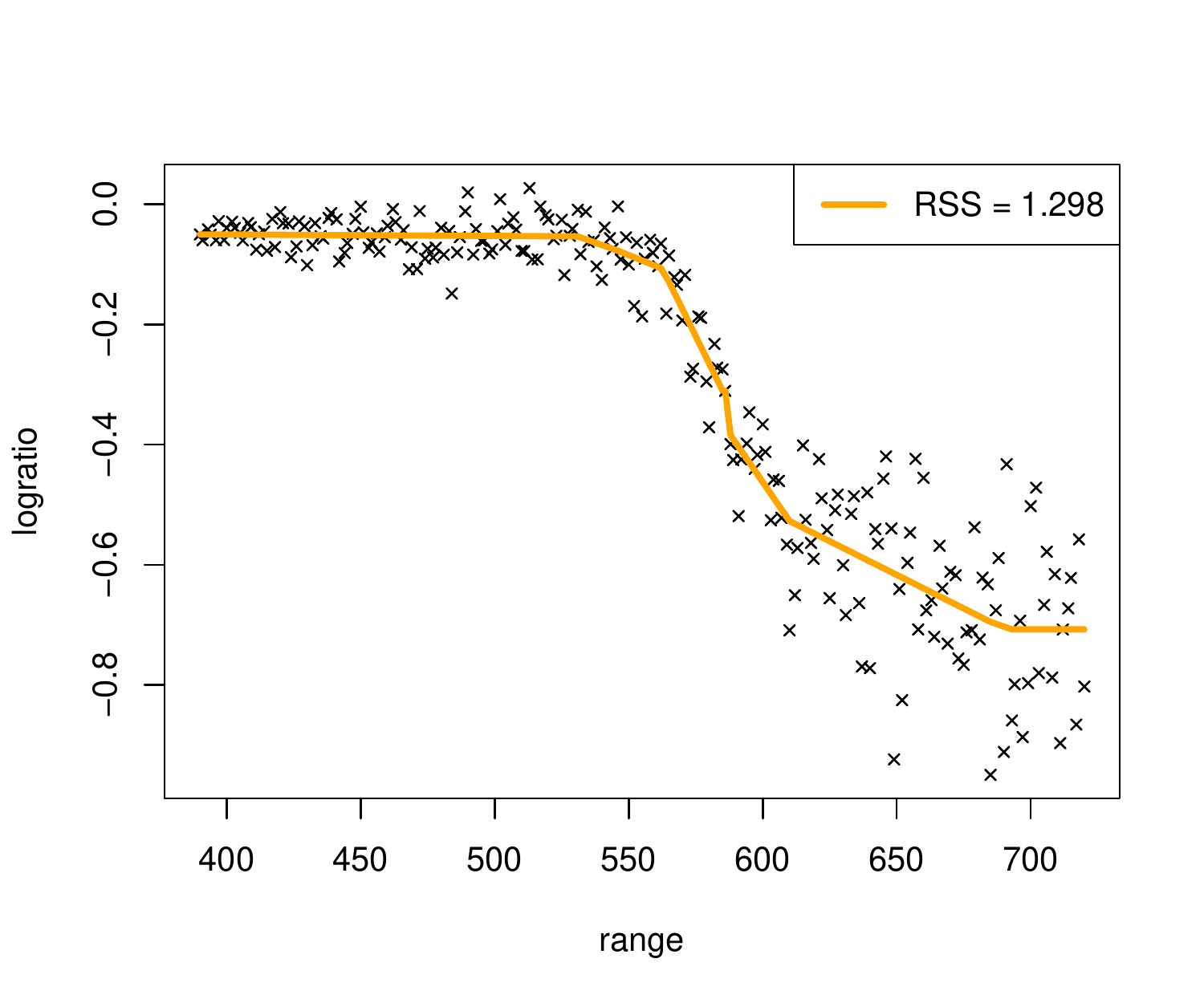}

\vspace{-1cm}
\includegraphics[width=0.475\textwidth]{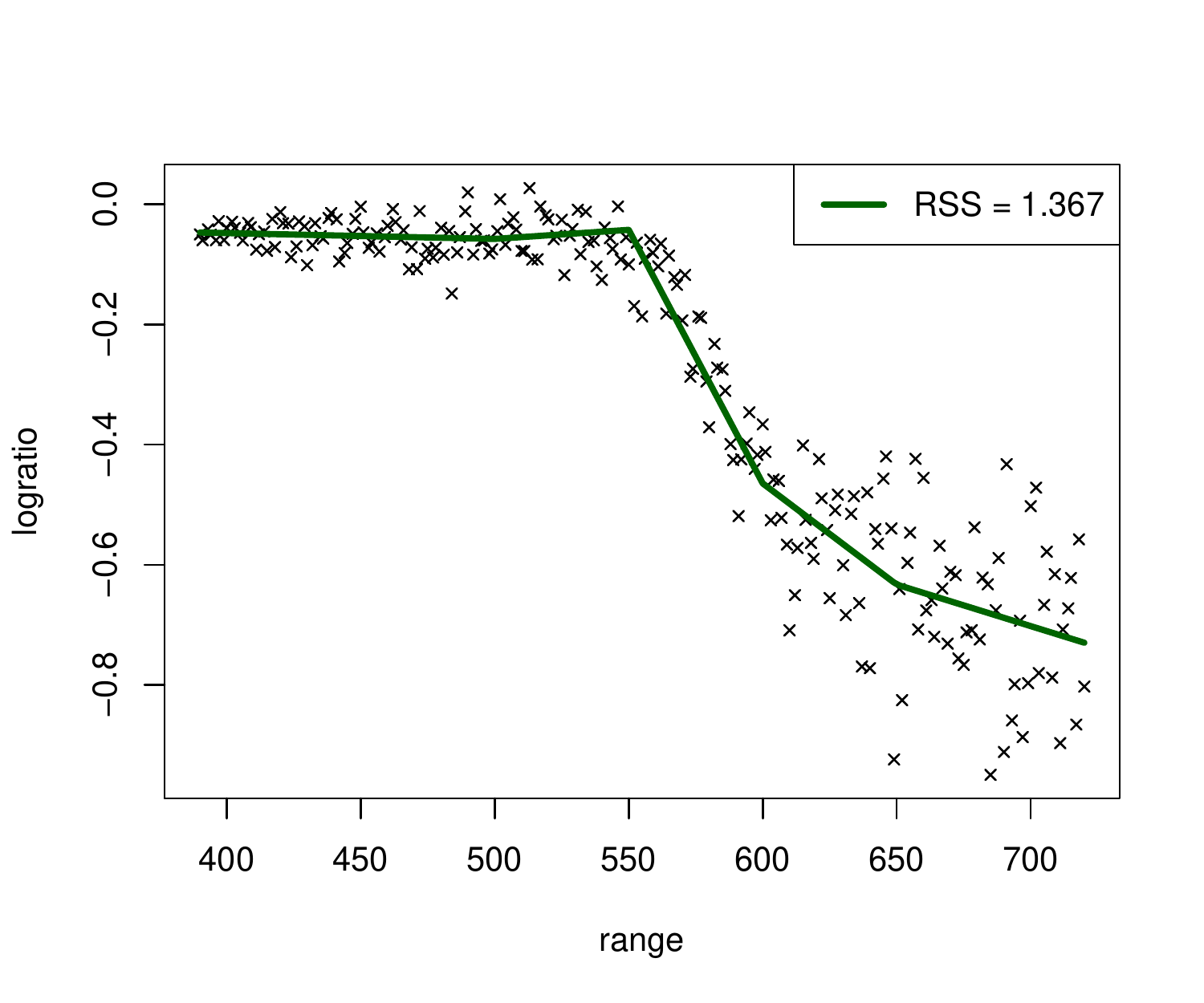}
\includegraphics[width=0.475\textwidth]{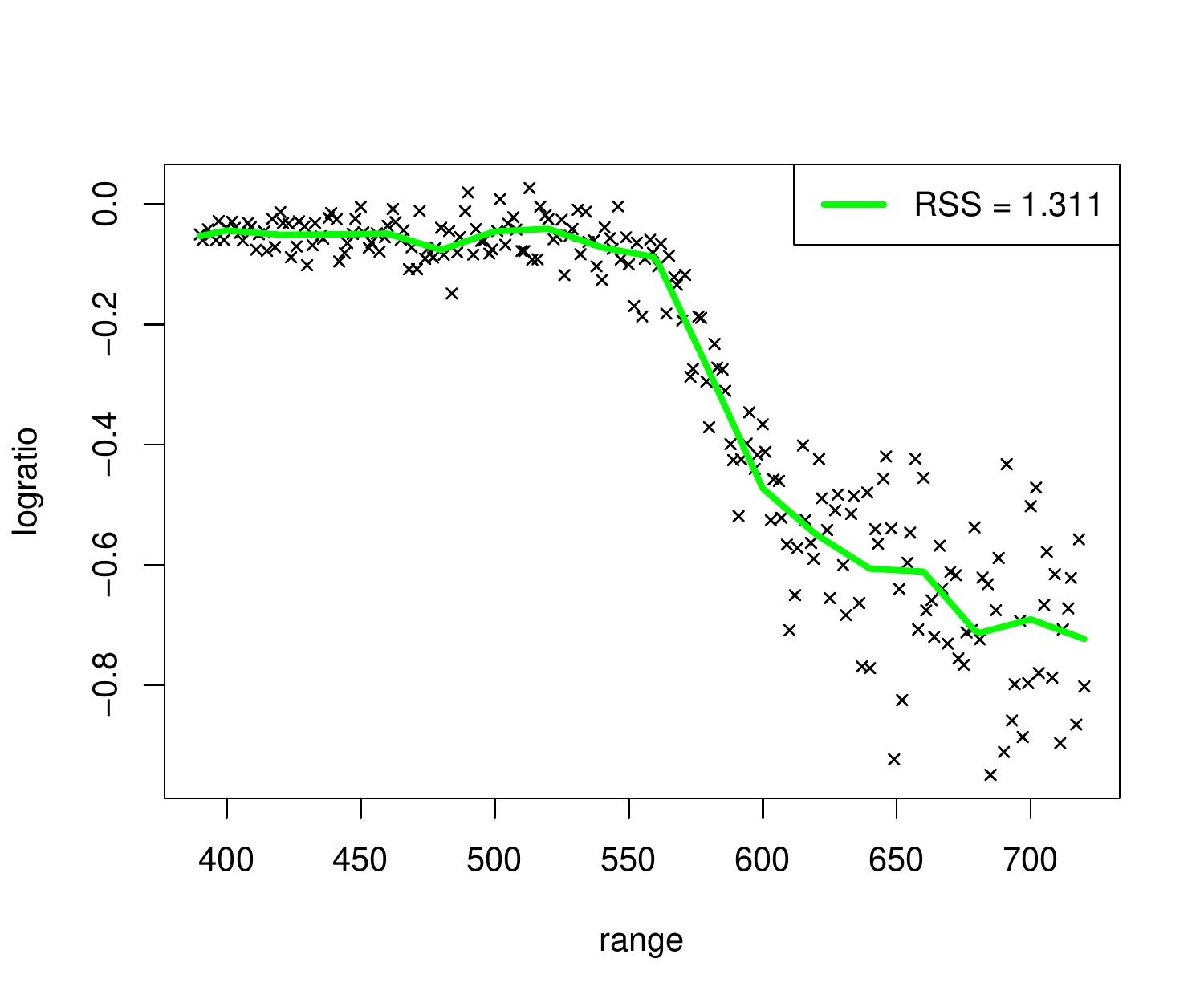}

\medskip
\caption{\label{Fig:LIDAR} Least squares fits to the LIDAR dataset ($n=221$) from~\citet{HHBRE96}: logistic (\textit{top left}), segmented linear with knots at $\mathtt{range}=500,550,600,650$ (\textit{bottom left}), segmented linear with knots at $\mathtt{range}=400,420,\dotsc,680,700$ (\textit{bottom right}) and S-shaped (\textit{top right}), along with their respective residual sums of squares (RSS).} 
\end{figure}

More precisely, based on an approximation of the governing equation for LIDAR scattering,~\citet[Section~3]{HHBRE96} consider a regression model for the \texttt{logratio} values
\[
\log\frac{P(r_i;\lambda_{\mathrm{on}})}{P(r_i;\lambda_{\mathrm{off}})}=f_0(r_i)+\xi_i,\quad i=1,\dotsc,n,
\]
where on physical grounds, $f_0(r)=-C\int_0^r g_0(s)\,ds$ is defined for $r\geq 0$ as the integral of the concentration function $g_0$ over $[0,r]$ multiplied by $-C\equiv -C(\lambda_{\mathrm{on}},\lambda_{\mathrm{off}})=-1.6\times 10^{-5}\,\mathrm{ng}^{-1}\,\mathrm{m}^2$. Since mercury concentration is always non-negative and would generally be expected to decrease with distance from the interior of the plume, $g_0$ can reasonably be modelled as a non-negative unimodal function,
in which case its antiderivative satisfies our definition of an S-shaped function. The data, shown in Figure~\ref{Fig:LIDAR}, do indeed appear to support $f_0$ as an inverted S-shaped regression function. Moreover,~\citet[Figure~4]{HHBRE96} present plots of suitably normalised residuals against \texttt{range} as well as the sample autocorrelations at different lags, which provide some empirical justification for the assumption that the errors~$\xi_1,\ldots,\xi_n$ are independent.

The different panels of Figure~\ref{Fig:LIDAR} illustrate least squares fits over different classes of regression functions.  In the top-left panel, we plot a fit of a logistic function
\[
x\mapsto -\frac{A}{1+e^{-ax+b}};
\]
here we see the limitations of the parametric model in terms of its inability to capture the behaviour of the regression function in the range 390--550m.
The segmented linear regression fits shown in the two bottom panels require the choice of a set of knots, and the left and right panels use 4 and 16 knots respectively.  We see that the selection of the set of knots can have quite a significant influence, and moreover, the fits are not guaranteed to be S-shaped or even monotone. Interestingly, despite the overfitting that is apparent in the bottom-right plot of the figure, the residual sum of squares remains higher than that of the S-shaped LSE\footnote{Note that all the algorithms in Section~\ref{sec:computation} can be used without further modifications to compute S-shaped LSEs on any other interval $[a,b]$ besides $[0,1]$.} illustrated in the top-right panel.  Moreover, the S-shaped LSE selects the number and location of its knots adaptively, with no input required from the practitioner. Another attractive feature of the S-shaped LSE is that its theoretical guarantees presented in Theorems~\ref{thm:worstcase} and~\ref{thm:adaptive} allow for heteroscedasticity, which is clearly present in this dataset. Finally, we note that the inflection point of this LSE at $\mathtt{range}=586\mathrm{m}$ yields an estimate of the distance from the LIDAR equipment to the central part of the plume, where the mercury concentration is highest.

\section{Discussion}
\label{sec:discussion}
In this paper, we have developed a framework for the estimation of S-shaped regression functions and their inflection points via nonparametric least squares. In spite of the challenges of working with a non-convex shape-constrained function class, we have proposed and implemented an efficient sequential algorithm for the computation of S-shaped least squares estimators, 
and also established theoretical guarantees on the consistency, robustness and rates of convergence of our estimators.
We will conclude by discussing some variations and possible extensions of our S-shaped regression problem that may prove to be interesting avenues for future research.

First, while our monotonicity requirement for S-shaped functions is natural in many practical applications, and useful for regulating the boundary behaviour of the least squares estimator at the endpoints of the covariate domain, much of our methodology and theory can be adapted straightforwardly to handle functions that are convex on $[0,m_0]$ and concave on $[m_0,1]$, but not necessarily increasing on $[0,1]$. On the computational side, our sequential strategy \texttt{SeqConReg} would still be applicable after the obvious small modifications to Step II of Algorithm~\ref{alg:sshape}.  This modified algorithm would be justified by analogues of Propositions~\ref{prop:factA} and~\ref{prop:srestrexact}, and we could still use the mixed primal-dual bases algorithm (Algorithm~\ref{alg:coneproj}) to sequentially compute convex LSEs on $\{(x_i,Y_i):1\leq i\leq j\}$ and concave LSEs on $\{(x_i,Y_i):j\leq i\leq n\}$ for $j\in [n]$.  The theoretical results in Section~\ref{sec:lsreg} would also go through with some minor alterations (e.g.~to the smoothness condition~\eqref{eq:smoothness} in Assumption~\ref{ass:inflection}). The proofs of the oracle inequalities would be broadly the same, and the current localisation argument for the inflection point result does not rely in any essential way on monotonicity near $m_0$.  Some properties of our projection framework may need more significant adjustment, however, in order to handle potential boundary issues.

In another direction, one could consider the estimation of `symmetric' S-shaped regression functions, by which we mean S-shaped functions $f_0$ with inflection point $m_0 \in (0,1)$ such that $f_0(x) = 2f_0(m_0) - f_0(2m_0-x)$ for $x \in [0\vee(2m_0-1),(2m_0)\wedge 1]$. We believe that this additional symmetry constraint is likely to bring about considerable challenges when it comes to developing theory and algorithms for the LSE that minimises the residual sum of squares over all symmetric S-shaped functions.  In particular, unlike in our Proposition~\ref{prop:existence}, it is not clear if the global minimiser in the least squares procedure can be attained at some symmetric S-shaped function with inflection point in $\{x_1,\dotsc,x_n\}$.  Moreover, the sequential strategy that underpins our current algorithm may no longer be valid, because in contrast to the conclusion of Proposition~\ref{prop:srestrexact}, the symmetric S-shaped LSE may not coincide with increasing convex or increasing concave LSEs on any subinterval.  Theoretically, although the global risk bounds in Section~\ref{sec:lsregoracle} are likely to carry over even with the additional symmetry constraint, the rate of convergence of the inflection point estimator $\tilde{m}_n$ may be very different to that in Theorem~\ref{thm:inflection}, and may even be (nearly) parametric. 

A further topic for future research could be to seek quantitative versions of the continuity result (Proposition~\ref{prop:cont}) for our $L^2$-projection onto the class of S-shaped functions, in the spirit of the recent work of~\citet{BS21} on the log-concave projection.  Such a result could, for instance, provide insight into the rate at which the estimated inflection point converges to the inflection point of the projected regression function under model misspecification. 

Finally, under local curvature conditions on an S-shaped function $f_0$ similar to those in Assumption~\ref{ass:inflection}, it would be of methodological and theoretical interest to be able to carry out (uniformly) asymptotically valid inference for $f_0(x)$ at fixed $x\in [0,1]$, as well as for the inflection point $m_0$. For $x\neq m_0$, defining $[\tilde{u}_n(x),\tilde{v}_n(x)]$ to be the largest interval containing $x$ on which the LSE $\tilde{f}_n$ is linear, we anticipate that the techniques of~\citet{DHS20} can be applied to obtain a limiting distribution for \[\sqrt{n\bigl(\tilde{v}_n(x)-\tilde{u}_n(x)\bigr)}\,\bigl(\tilde{f}_n(x)-f_0(x)\bigr)\]
that does not depend on $f_0$, and hence construct asymptotically valid confidence intervals for $f_0(x)$. On the other hand, since $m_0$ marks the boundary between the convex and concave parts of $f_0$, we expect the problem of uncertainty quantification for $m_0$ and $f_0(m_0)$ to be more challenging and of a qualitatively different character. With this end in view, it is natural to seek tractable asymptotic distributions for $\tilde{m}_n$ and $\tilde{f}_n(\tilde{m}_n)$. As an initial step, one would need to refine the results in Section~\ref{subsec:inflectpt} by closing the logarithmic gap between the upper and lower bounds therein on the rate of convergence of $\tilde{m}_n$ to $m_0$. A satisfactory solution to this problem would ideally also settle the analogous problem for the plug-in mode estimator based on the unimodal LSE~\citep{SZ01}, and is likely to require significant further technical developments.

\section{Appendix: A mixed primal-dual bases algorithm}
\label{sec:appendix}
In this section, we describe a mixed primal-dual bases algorithm the $L^2$-projection of a line segment onto the polyhedral convex cone of increasing convex sequences.  This underpins our \texttt{SeqConReg} algorithm in Section~\ref{sec:computation}.  Our starting point is the following standard characterisation of projections onto general closed, convex cones~\citep[e.g.][Corollary~2.1]{Mor62,Gro96}. Here and below, we write $\norm{{\cdot}}$ and $\ipr{\cdot\,}{\cdot}$ for the standard Euclidean norm and inner product on $\R^n$ for some $n\in\N$.
\begin{lemma}
\label{lem:cone}
Let $\Lambda\subseteq\R^n$ be a closed, convex cone. For each $y\in\R^n$, there exists a unique projection of $y$ onto $\Lambda$, given by $\Pi_\Lambda(y)=\argmin_{u\in\Lambda}\norm{u-y}$, and we have the following:

\unparskip
\begin{enumerate}[label=(\alph*)]
\item $\Pi_\Lambda(y)$ is the unique $\hat{y}\in\Lambda$ for which $\ipr{v}{y-\hat{y}}\leq 0$ for all $v\in\Lambda$ and $\ipr{\hat{y}}{y-\hat{y}}=0$.
\item Suppose in addition that $\Lambda$ is \emph{finitely generated}, i.e.\ that $\Lambda=\bigl\{\sum_{\ell=1}^r\lambda_\ell v^\ell:\lambda_1,\dotsc,\lambda_r\geq 0\bigr\}$ for some generators $v^1,\dotsc,v^r\in\Lambda$. Then $\hat{y}=\Pi_\Lambda(y)$ if and only if $\hat{y}=\sum_{\ell=1}^r\hat{\lambda}_\ell v^\ell$ for some $\hat{\lambda}_1,\dotsc,\hat{\lambda}_r\geq 0$, and $\ipr{v^\ell}{y-\hat{y}}\leq 0$ for all $\ell$, with $\ipr{v^\ell}{y-\hat{y}}=0$ for any $\ell$ such that $\hat{\lambda}_\ell>0$.
\end{enumerate}

\unparskip
\end{lemma}

\unparskip
In Lemma~\ref{lem:cone}(b), the vector $(\hat{\lambda}_1,\dotsc,\hat{\lambda}_r)$ is the minimiser of the quadratic function $(\lambda_1,\dotsc,\lambda_r)\mapsto\norm{y-\sum_{\ell=1}^r\lambda_\ell v^\ell}^2$ over the convex set $[0,\infty)^r$. When this constrained minimisation problem is written in Lagrangian form, the associated KKT optimality conditions~\citep[e.g.][Theorem~28.3]{Rock97} correspond precisely to the three conditions in (a) that uniquely define $\Pi_\Lambda(y)$, namely (i) $\hat{y}\in\Lambda$ (\emph{primal feasibility}); (ii) $y-\hat{y}\in\{u\in\R^n:\ipr{u}{v}\leq 0\text{ for all }v\in\Lambda\}$, the \emph{polar cone} of $\Lambda$ (\emph{dual feasibility}); and (iii) $\ipr{\hat{y}}{y-\hat{y}}=0$ (\emph{complementary slackness}).

Given $(x_1,Y_1),\dotsc,(x_n,Y_n) \in [0,1]\times\R$ with $x_1<\cdots<x_n$, we now fix $j\in [n]$ and work with the cone $\Lambda^j$ of increasing convex sequences based on $x_1,\dotsc,x_j$, as defined in~\eqref{eq:Thetak}. The projection of $(Y_1,\dotsc,Y_j)$ onto~$\Lambda^j$ is $\bigl(\hat{f}_{1,j}(x_1),\dotsc,\hat{f}_{1,j}(x_j)\bigr)$, where $\hat{f}_{1,j}$ is the increasing convex LSE based on $\{(x_i,Y_i):i \in [j]\}$. The generators of $\Lambda^j$ are $\pm u^0,u^1,\dotsc,u^{j-1}\in\R^j$, where $u^0=\mathbf{1}$ and $u_i^\ell=(x_i-x_\ell)^+$ for all $i \in [j]$ and $\ell \in [j-1]$. Since $u^0,u^1,\dotsc,u^{j-1}$ are linearly independent, every $v\equiv (v_1,\dotsc,v_j)\in\R^j$ can be represented uniquely in the form $v=\sum_{\ell=0}^{j-1}\lambda_\ell u^\ell$, where
\begin{equation}
\label{eq:primal}
\lambda_0\equiv\lambda_0(v)=v_1;\quad\lambda_1\equiv\lambda_1(v)=\frac{v_2-v_1}{x_2-x_1};\quad\lambda_\ell\equiv\lambda_\ell(v)=\frac{v_{\ell+1}-v_\ell}{x_{\ell+1}-x_\ell}-\frac{v_\ell-v_{\ell-1}}{x_\ell-x_{\ell-1}},\;\;2\leq\ell\leq j-1,
\end{equation}
so that $v\in\Lambda^j$ if and only if $\lambda_\ell(v)\geq 0$ for all $\ell\in [j-1]$; this is the \emph{primal feasibility} condition from~Lemma~\ref{lem:cone}. For each $v=\sum_{\ell=0}^{j-1}\lambda_\ell u^\ell\in\R^j$, the unique $g_v\in\mathcal{G}[x_1,\dotsc,x_j]$ satisfying $v=\bigl(g_v(x_1),\dotsc,g_v(x_j)\bigr)$ has a knot at $x_\ell$ if and only if $\lambda_\ell\neq 0$, so we refer to $A(v):=\{1\leq\ell\leq j-1:\lambda_\ell\neq 0\}$ as the set of \emph{knots} of $v$ (or `active indices'). 

The following useful property of the projection map $\Pi_{\Lambda^j}\colon\R^j\to\Lambda^j$ can be derived easily from Lemma~\ref{lem:cone}. A general version of this result for arbitrary closed, convex sets is stated as Lemma~\ref{lem:projconv}.
\begin{lemma}
\label{lem:projactive}
Let $A\subseteq [j-1]$ and $v',v''\in\R^j$ be such that $A\bigl(\Pi_{\Lambda^j}(v)\bigr)=A$ for each $v\in\{v',v''\}$. Then for all $v\in [v',v'']:=\{(1-t)v'+tv'':t\in [0,1]\}$, we have $A\bigl(\Pi_{\Lambda^j}(v)\bigr)=A$ and, defining the linear subspace $\mathcal{L}_A:=\Span\{u^\ell:\ell\in A\cup\{0\}\}=\{v\in\R^j:A(v)\subseteq A\}$, we have $\Pi_{\Lambda^j}(v)=\Pi_{\mathcal{L}_A}(v)$.
\end{lemma}

\unparskip
\begin{remark}
\label{rem:conefaces}
For $A\subseteq [j-1]$, the orthogonal projection onto the linear subspace $\mathcal{L}_A$ is represented by $P_A:=U_A(U_A^\top U_A)^{-1}U_A^\top\in\R^{j\times j}$, where $U_A\in\R^{j\times (\abs{A}+1)}$ is the matrix obtained by extracting the columns of $U:=(u^0\:u^1\cdots\:u^{j-1})\in\R^{j\times j}$ indexed by $A\cup\{0\}$. By taking $v'=v''$ in Lemma~\ref{lem:projactive}, we recover a version of~\citet[Proposition~2.1]{GS17}: suppose that we are given $v\in\R^j$ and have oracle knowledge of $A\equiv A\bigl(\Pi_{\Lambda^j}(v)\bigr)$, i.e.\ the locations of the knots of $\Pi_{\Lambda^j}(v)$. Then to compute $\Pi_{\Lambda^j}(v)$, we can note that $\Pi_{\Lambda^j}(v) = P_Av=\sum_{\ell=0}^{j-1}\hat{\lambda}_\ell u^\ell$, where $\hat{\lambda}_\ell\equiv\hat{\lambda}_\ell^A(v):=\lambda_\ell(P_Av)$ for $0\leq\ell\leq j-1$, so that $\hat{\lambda}_\ell=0$ for all $\ell\notin A$ and
\begin{equation}
\label{eq:lambdahat}
\bigl(\hat{\lambda}_\ell:\ell\in A\cup\{0\}\bigr)=(U_A^\top U_A)^{-1}U_A^\top v=\argmin_{(\lambda_\ell\,:\,\ell\in A\cup\{0\})}\,\sum_{i=1}^n\,\biggl(v_i-\lambda_0-\sum_{\ell\in A}\lambda_\ell(x_i-x_\ell)^+\biggr)^2
\end{equation}
solves an ordinary (\emph{unconstrained}) least squares problem.
\end{remark}

\unparskip
Observe now that if $v(0),v(1)\in\R^j$ are arbitrary and $v(t):=(1-t)v(0)+tv(1)$ for all $t\in (0,1)$, then $t\mapsto\Pi_{\Lambda^j}\bigl(v(t)\bigr)$ is a continuous, piecewise affine function from $[0,1]$ to $\Lambda^j$. Indeed, by Lemma~\ref{lem:projactive} (and the continuity of projections onto closed, convex cones),
there exist $0=t_0'<t_1'<\cdots<t_{s+1}'=1$ and distinct subsets $A_0',A_1',\dotsc,A_s'\subseteq [j-1]$ such that for each $0\leq r\leq s$, we have $\Pi_{\Lambda^j}\bigl(v(t)\bigr)=\Pi_{\mathcal{L}_{A_r'}}\bigl(v(t)\bigr)=P_{A_r'}v(t)$ for all $t\in [t_r',t_{r+1}']$.

Suppose that we are given $v(0),v(1)\in\R^j$ and the projection $\Pi_{\Lambda^j}\bigl(v(0)\bigr)\in\Lambda^j$, and now seek to compute  $\Pi_{\Lambda^j}\bigl(v(1)\bigr)$. The reasoning in the previous paragraph suggests that we can 
proceed as in Algorithm~\ref{alg:coneproj} below. 

\begin{algorithm}
\label{alg:coneproj}
Mixed primal-dual bases algorithm to compute projections onto the cone $\Lambda^j$.

\unparskip
\begin{enumerate}[label=(\Roman*)]
\item Starting at $t=t_0:=0$, define $\hat{v}_0(t_0):=\Pi_{\Lambda^j}\bigl(v(0)\bigr)$ and let the initial active set be $A_0:=A\bigl(\hat{v}_0(t_0)\bigr)$, so that $\hat{v}_0(t_0)=\Pi_{\Lambda^j}\bigl(v(0)\bigr)=P_{A_0}v(0)$. 
\item For $r\in\N_0$, suppose inductively that at $t=t_r$, we are given that $\hat{v}_r(t_r):=\Pi_{\Lambda^j}\bigl(v(t_r)\bigr)=P_{A_r}v(t_r)$ for some $A_r\subseteq [j-1]$. 
Let $\hat{v}_r(t):=P_{A_r}v(t)=\hat{v}_r(t_r)-(t-t_r)P_{A_r}u$ for $t \in [t_r,1]$, where $u:=v(0)-v(1)$, and
\begin{align}
t_{r+1}:=\sup\,\bigl\{t\geq t_r:{}&\lambda_\ell\bigl(\hat{v}_r(s)\bigr)\geq 0,\notag\\
\label{eq:tr+1}
&\ipr{u^\ell}{v(s)-\hat{v}_r(s)}\leq 0\text{ for all }s\in [t_r,t]\text{ and } \ell \in [j-1]\bigr\}.
\end{align}
By Lemma~\ref{lem:cone} and the fact that $\hat{v}_r(t_r)=\Pi_{\Lambda^j}\bigl(v(t_r)\bigr)$, the set on the right-hand side always contains $t_r$. In order to compute $t_{r+1}$ explicitly, observe that for all $t\in [t_r,t_{r+1}]$, we have

\unparskip	
\begin{enumerate}[label=(\roman*)]
\item \emph{Primal feasibility}: $\beta_\ell(t):=\lambda_\ell\bigl(\hat{v}_r(t)\bigr)=\lambda_\ell\bigl(\hat{v}_r(t_r)-(t-t_r)P_{A_r}u\bigr)=\beta_\ell(t_r)-(t-t_r)\,\lambda_\ell(P_{A_r}u)\geq 0$ for every $\ell \in [j-1]$, where equality holds if $\ell\in A_r^c$;
\item \emph{Dual feasibility}: $\gamma_\ell(t):=\ipr{u^\ell}{v(t)-\hat{v}_r(t)}=\ipr{u^\ell}{(I-P_{A_r})v(t)}=\gamma_\ell(t_r)-(t-t_r)\,\hat{\zeta}_\ell^{A_r}(u)\leq 0$ for every $0\leq\ell\leq j-1$, where equality holds if $\ell\in A_r\cup\{0\}$, and $\hat{\zeta}_\ell^{A_r}(u):=\ipr{u^\ell}{(I-P_{A_r})u}$.
\end{enumerate}

\unparskip
In particular, $\beta_\ell(t),\gamma_\ell(t)$ depend linearly on $t\in [t_r,t_{r+1}]$, so 
\begin{align}
t_{r+1}=t_r+\biggl(&\min\,\biggl\{\frac{\beta_\ell(t_r)}{\hat{\lambda}_\ell^{A_r}(u)}:\ell\in [j-1],\,\hat{\lambda}_\ell^{A_r}(u)>0\biggr\}\notag\\
\label{eq:algconeproj}
&\hspace{2.5cm}\wedge\min\,\biggl\{\frac{\gamma_\ell(t_r)}{\hat{\zeta}_\ell^{A_r}(u)}:\ell\in A_r^c,\,\hat{\zeta}_\ell^{A_r}(u)<0\biggr\}\biggr),
\end{align}
observing that $\hat{\zeta}_\ell^{A_r}(u)=\ipr{(I-P_{A_r})u^\ell}{u}=0$ for $\ell\in A_r\cup\{0\}$. 

\unparskip
\begin{enumerate}[label=(\roman*),resume]
\item \emph{Complementary slackness} is maintained throughout this step: $\ipr{\hat{v}_r(t)}{v(t)-\hat{v}_r(t)}=\ipr{P_{A_r}v(t)}{(I-P_{A_r})v(t)}=0$, so $\Pi_{\Lambda^j}\bigl(v(t)\bigr)=\hat{v}_r(t)=P_{A_r}v(t)$ for all $t\in [t_r,t_{r+1}]$ by Lemma~\ref{lem:cone}.
\end{enumerate}

\unparskip
\item If $t_{r+1}\geq 1$, then return $\hat{v}_r(1)=\hat{v}_r(t_r)-(1-t_r)P_{A_r}u$ and terminate the algorithm. Otherwise, go to (IV), noting that when $t$ approaches $t_{r+1}$ from below, either 

\unparskip		
\begin{itemize}[leftmargin=0.4cm]
\item A \emph{primal variable} $\beta_\ell(t)$ with $\ell\in A_r$ is about to hit 0 and turn negative, or
\item A \emph{dual variable} $\gamma_\ell(t)$ with $\ell\in A_r^c$ is about to hit 0 and turn positive.
\end{itemize}

\unparskip		
\item \emph{Changing the `active set'}: Define $A_r^-:=\{\ell\in A_r:\beta_\ell(t_{r+1})=0\}$ and $A_r^+:=\{\ell\in A_r^c:\gamma_\ell(t_{r+1})=0\}$. 

\unparskip
\begin{enumerate}
\item If $\abs{A_r^-\cup A_r^+}=1$, then repeat (II) and (III) with $r+1$ in place of $r$ and $A_{r+1}:=(A_r\setminus A_r^-)\cup A_r^+$, observing that $\Pi_{\Lambda^j}\bigl(v(t_{r+1})\bigr)=P_{A_r}v(t_{r+1})=P_{A_{r+1}}v(t_{r+1})$. 
\item If $\abs{A_r^-\cup A_r^+}>1$, i.e.\ there is a \emph{degeneracy} at $t_{r+1}$, then 
choose $A^\pm\subseteq A_r^\pm$ and carry out (II) with $r+1$ in place of $r$ and $A_{r+1}=(A_r\setminus A^-)\cup A^+$. In doing so, if~\eqref{eq:algconeproj} yields a strict increase in $t$, then let the algorithm continue from there and pass to (III). Otherwise, retry this for different pairs of subsets $A^\pm\subseteq A_r^\pm$ until we can move a strictly positive distance in the next iteration of (II).
\end{enumerate}

\unparskip
\end{enumerate}

\unparskip
\end{algorithm}

\unparskip
When defining the primal variables $\beta_\ell(t)$ in (II), it is convenient here that every $v\in\Lambda^j$ has a unique primal representation, which in this case is given by~\eqref{eq:primal}. The same is true of any cone in~$\R^j$ generated by $\pm\tilde{u}^0,\dots,\pm\tilde{u}^{q-1},\tilde{u}^q,\dotsc,\tilde{u}^{j-1}$, for some linearly independent $\tilde{u}^0,\tilde{u}^1,\dotsc,\tilde{u}^{j-1}$. Thus, Algorithm~\ref{alg:coneproj} is applicable to all such cones, provided that the `active sets' are taken to be subsets of $\{q,q+1,\dotsc,j-1\}$~\citep{FM89}, so in particular, it can also be used to compute isotonic and convex LSEs (in a sequential manner, as described in Section~\ref{sec:computation}).  Indeed, the sequential application of this mixed primal-dual bases algorithm to the monotone cone $\Theta^\uparrow$ from the proof of Corollary~\ref{cor:bdiso} yields the widely-used, linear time `pool adjacent violators' algorithm (PAVA) \citep{BBBB1972}.  Moreover, with appropriate modifications, Algorithm~\ref{alg:coneproj} can be extended to general polyhedral cones~\citep{Mey99} and polyhedral convex sets. 

\begin{lemma}
\label{lem:terminate}
Algorithm~\ref{alg:coneproj} always terminates after finitely many steps with the correct solution $\Pi_{\Lambda^j}\bigl(v(1)\bigr)$.
\end{lemma}

\unparskip
This follows from (i)--(iii) in Stage (II) and the following two observations:

\unparskip
\begin{enumerate}[label=(\roman*)]
\setcounter{enumi}{3}
\item The algorithm does not get stuck at any of the thresholds $t_r$; i.e.\ when $t=t_r$ for some $r$, there is always a subsequent iteration of (II) that strictly increases $t$;
\item At distinct thresholds $t_r$, the corresponding `active sets' $A_r$ are distinct subsets of $[j-1]$.
\end{enumerate}

\unparskip
We will justify (iv) and (v) in Section~\ref{sec:compproofs} in the supplementary material, where we also exploit the specific structure of $\Lambda^j$ to handle the degeneracies mentioned in Stage~(IV)(b); see in particular modification (IV') and Proposition~\ref{prop:ejactive}.

\section*{Acknowledgements}
We thank the editor and three anonymous reviewers for their constructive comments and suggestions. QH was supported by the NSF grant DMS-1916221. RJC was supported by the National Cancer Institute grant R01-CA057030 and the TRIPODS grant NSF CCF-1934904 from the U.S. National Science Foundation. RJS was supported by EPSRC grants EP/P031447/1 and EP/N031938.

\clearpage
\setcounter{section}{0}
\setcounter{equation}{0}
\setcounter{theorem}{0}
\def\theequation{S\arabic{equation}}
\def\thesection{S\arabic{section}}
\def\thetheorem{S\arabic{theorem}}

\begin{center}
\LARGE{Supplementary material for `Nonparametric, tuning-free estimation of S-shaped functions'} \\ \vspace{0.3in}
\end{center}

This is the supplementary material for the main paper `Nonparametric, tuning-free estimation of S-shaped functions', which is hereafter referred to as the main text. We present the proofs of our main theorems and various auxiliary results.

\section{Subinterval localisation and boundary adjustment results}
\label{sec:bdadj}
As in Section~\ref{sec:computation}, we consider pairs $(x_1,Y_1),\dotsc,(x_n,Y_n)$ taking values in $[0,1]\times\R$, where $0\leq x_1<\cdots<x_n\leq 1$ are fixed. The purpose of this section is to generalise the following, known `subinterval localisation' property of the univariate isotonic LSE $\bar{f}_n$ based on $\{(x_i,Y_i):1\leq i\leq n\}$: if $\bar{f}_n$ has a jump after $x_k$, so that $\bar{f}_n(x_k)<\bar{f}_n(x_{k+1})$, then the isotonic LSE based on $\{(x_i,Y_i):1\leq i\leq k\}$ agrees with $\bar{f}_n$ on $[x_1,x_k]$, and the isotonic LSE based on $\{(x_i,Y_i):k+1\leq i\leq n\}$ agrees with $\bar{f}_n$ on $[x_{k+1},x_n]$. One way to see this is to invoke the explicit representation of $\bar{f}_n$ as the left derivative of the greatest convex minorant of the cumulative sum diagram associated with $(x_1,Y_1),\dotsc,(x_n,Y_n)$~\citep[e.g.][Lemma~2.1]{GJ14}. 

By comparison with the isotonic LSE, the lack of an explicit representation makes the situation much more complicated for convex and concave LSEs, as well as the S-shaped LSEs (with known or unknown inflection point $m\in [0,1]$) defined in Section~\ref{sec:proj}.  
Writing $\tilde{g}_n$ for any one of these LSEs and $x_k$ for one of its kinks, we will see below that $\tilde{g}_n$ cannot in general be localised exactly to either $[x_1,x_k]$ or $[x_k,x_n]$. In fact, one of our significant technical contributions is to show that, on each of these subintervals, the restriction of $\tilde{g}_n$ minimises a \emph{weighted} sum of squares, in which the observation $(x_k,Y_k)$ is assigned a fraction of the weight placed on all the other points in the subinterval; see~\eqref{eq:bdweights}. Although the adjusted `boundary weight' usually depends on the LSE $\tilde{g}_n$ and is not an accessible quantity in its own right, the merit of this boundary reweighting idea is seen in the proof of Theorem~\ref{thm:inflection} on inflection point estimation. A special case where no boundary adjustment is needed (Proposition~\ref{prop:srestrexact}) is the basis of Algorithm~\ref{alg:sshape} for computing S-shaped LSEs.

The subinterval localisation properties of all the LSEs mentioned above will be derived as consequences of the general Lemma~\ref{lem:bdadj} below, for which we require the following additional notation. Let $\mathbf{1}:=(1,1,\dotsc,1)=\sum_{i=1}^n e_i\in\R^n$, where $e_1,\dotsc,e_n$ are the standard basis vectors in~$\R^n$. For $1\leq a\leq b\leq n$, let $\mathbf{1}^{[a:b]}:=\sum_{i=a}^b e_i$, and for $\theta\in\R^n$, define $\theta^{(a:b)}\in\R^n$ by $\theta_i^{(a:b)}:=\theta_{a\vee i\wedge b}$ for $i\in [n]$. In addition, for $w\equiv (w_1,\dotsc,w_n)\in [0,\infty)^n$ and $u,v\in\R^n$, define $\ipr{u}{v}_w:=\sum_{i=1}^n w_iu_iv_i$ and $\norm{u}_w:=\ipr{u}{u}_w^{1/2}$, so that $\ipr{\cdot\,}{\cdot}_w$ is a non-negative definite symmetric bilinear form. It is convenient to study weighted LSEs defined with respect to arbitrary weight vectors $w\in [0,\infty)^n$, even though we are primarily interested in the case $w=\mathbf{1}$ in subsequent applications of the result below. 
\begin{lemma}
\label{lem:bdadj}
Let $\Theta\subseteq\R^n$ be a closed, convex set, let $Y:=(Y_1,\dotsc,Y_n) \in \R^n$ and, for some weight vector $w\equiv (w_1,\dotsc,w_n)\in [0,\infty)^n$, let $\hat{\theta}\equiv\hat{\theta}_n(w)\in\argmin_{\theta\in\Theta}\norm{Y-\theta}_w$. Suppose that $\hat{\theta}\pm\eta\mathbf{1}\in\Theta$ for some $\eta>0$. 

\unparskip
\begin{enumerate}[label=(\alph*)]
\item Assume that at least one of the following conditions is satisfied for some $k \in [n]$: 

\unparskip
\begin{enumerate}[label=(\roman*)]
\item $\hat{\theta}+\varepsilon\eta\mathbf{1}^{[1:k]}\in\Theta$ and $\hat{\theta}+\varepsilon\eta\mathbf{1}^{[k:n]}\in\Theta$ for some $\varepsilon\in\{-1,1\}$ and $\eta>0$;
\item $\hat{\theta}\pm\eta u\in\Theta$ for some $u\in\bigl\{\mathbf{1}^{[1:k]},\mathbf{1}^{[k:n]}\bigr\}$ and $\eta>0$. 
\end{enumerate}

\unparskip		
Then defining
\begin{equation}
\label{eq:bdweights}
\underline{w}_k:=\frac{\sum_{i=1}^{k-1}w_i(Y_i-\hat{\theta}_i)}{\hat{\theta}_k-Y_k}\,\Ind_{\{\hat{\theta}_k\neq Y_k\}}\quad\text{and}\quad\overline{w}_k:=\frac{\sum_{i=k+1}^{n}w_i(Y_i-\hat{\theta}_i)}{\hat{\theta}_k-Y_k}\,\Ind_{\{\hat{\theta}_k\neq Y_k\}},
\end{equation}
we have $\underline{w}_k,\overline{w}_k\in [0,w_k]$ and $\underline{w}_k+\overline{w}_k\leq w_k$, with equality when $Y_k\neq\hat{\theta}_k$.

If (ii) holds with $u=\mathbf{1}^{[1:k]}$, then $\overline{w}_k=0$, and if (ii) holds with $u=\mathbf{1}^{[k:n]}$, then $\underline{w}_k=0$.
\item Let $1\leq a\leq b\leq n$ be such that for each $k\in\{a,b\}$, either (i) or (ii) holds, and suppose that for each $\theta\in\Theta\cup\{-\hat{\theta}\}$, there exists $\eta>0$ such that $\hat{\theta}+\eta\theta^{(a:b)}\in\Theta$. Then defining
\begin{equation}
\label{Eq:wtilde}
\tilde{w}^{a;b}:=(0,\dotsc,0,\overline{w}_a,w_{a+1}\dotsc,w_{b-1},\underline{w}_b,0,\dotsc,0)\in\R^n,
\end{equation}
so that $\tilde{w}_i^{a;b}=0$ for $1\leq i<a$ and $b<i\leq n$, we have $\hat{\theta}\in\argmin_{\theta\in\Theta}\norm{Y-\theta}_{\tilde{w}^{a;b}}$.
\end{enumerate}

\unparskip
\end{lemma}

\unparskip
The main conclusion of Lemma~\ref{lem:bdadj} comes at the end of part~(b): under certain conditions on $a,b$, there exists a non-negative weight vector $\tilde{w}^{a;b}$, whose only non-zero weights occur for indices $i$ with $a \leq i \leq b$, for which the sub-vector $(\hat{\theta}_a,\hat{\theta}_{a+1},\ldots,\hat{\theta}_b)$ of the overall LSE $\hat{\theta}$ can be computed as the $\tilde{w}^{a;b}$-weighted LSE of our data vector $Y$.  Note that $\tilde{w}^{a;b}$ differs from $w^{a;b}$ only at the endpoints $a$ and~$b$; we therefore refer to $\tilde{w}^{a;b}$ as the \emph{boundary-adjusted weight vector}.  Condition~(ii) in~(a) yields sufficient conditions for exact subinterval localisation (without non-trivial boundary adjustments $\underline{w}_k,\overline{w}_k$). Since it is assumed that $\hat{\theta}\pm\eta\mathbf{1}\in\Theta$ for some $\eta>0$, condition (ii) in (a) holds for $k \in \{1,n\}$; although this is vacuous as far as the conclusion of~(a) is concerned, it means that we can take $a=1$ or $b=n$ in~(b).

\unparskip
\begin{proof}
For a closed, convex set $\Theta$ and a weight vector $w\in [0,\infty)^n$, the existence of $\hat{\theta}\in\argmin_{\theta\in\Theta}\norm{Y-\theta}_w$ is guaranteed (and uniqueness holds if $w_i>0$ for all $i\in [n]$). In all cases, we have $\hat{\theta}\in\argmin_{\theta\in\Theta}\norm{Y-\theta}_w$ if and only if
\begin{equation}
\label{eq:projconv}
\ipr{Y-\hat{\theta}}{\theta-\hat{\theta}}_w=\sum_{i=1}^n w_i(Y_i-\hat{\theta}_i)(\theta_i-\hat{\theta}_i)\leq 0
\end{equation}
for all $\theta\in\Theta$; see Lemma~\ref{lem:projconv}(a). 

(a) By assumption, we can take $\theta=\hat{\theta}\pm\eta\mathbf{1}\in\Theta$ in~\eqref{eq:projconv} for a suitable $\eta>0$, so $\sum_{i=1}^n w_i(Y_i-\hat{\theta}_i)=0$ and therefore $\underline{w}_k+\overline{w}_k=w_k$ when $Y_k\neq\hat{\theta}_k$.

\unparskip
\begin{itemize}[leftmargin=0.4cm]
\item If (i) holds, then we can take $\theta=\hat{\theta}+\varepsilon\eta\mathbf{1}^{[1:k]}$ and $\theta=\hat{\theta}+\varepsilon\eta\mathbf{1}^{[k:n]}$ in~\eqref{eq:projconv} for some $\eta>0$ and $\varepsilon\in\{-1,1\}$, whence 
\begin{equation}
\label{eq:bdadj1}
-\varepsilon\sum_{i=k}^n w_i(Y_i-\hat{\theta}_i)=\varepsilon\sum_{i=1}^{k-1}w_i(Y_i-\hat{\theta}_i)\geq 0\geq\varepsilon\sum_{i=1}^k w_i(Y_i-\hat{\theta}_i)=-\varepsilon\sum_{i=k+1}^n w_i(Y_i-\hat{\theta}_i).
\end{equation}
Thus, $\varepsilon w_k(\hat{\theta}_k-Y_k)\geq\varepsilon\sum_{i=1}^{k-1}w_i(Y_i-\hat{\theta}_i)\geq 0$, so $\underline{w}_k\in [0,w_k]$, and similarly $\overline{w}_k\in [0,w_k]$. 
\item Note that if $Y_k=\hat{\theta}_k$, then it follows from~\eqref{eq:bdadj1} that $\sum_{i=1}^{k-1} w_i(Y_i-\hat{\theta}_i)=0=\sum_{i=k+1}^n w_i(Y_i-\hat{\theta}_i)$.
\item Under (ii), if $\theta=\hat{\theta}\pm\eta\mathbf{1}^{[1:k]}\in\Theta$ for some $\eta>0$, then~\eqref{eq:projconv} implies that $\sum_{i=1}^k w_i(Y_i-\hat{\theta}_i)=0=\sum_{i=k+1}^n w_i(Y_i-\hat{\theta}_i)$, in which case $\overline{w}_k=0$. The other case where $u=\mathbf{1}^{[k:n]}$ is similar.
\end{itemize}

\unparskip
(b) For each $k\in\{a,b\}$, either (i) or (ii) holds by hypothesis, so it follows from part (a) that $\underline{w}_k,\overline{w}_k\in [0,w_k]$ and $(\underline{w}_k+\overline{w}_k)(Y_k-\hat{\theta}_k)=w_k(Y_k-\hat{\theta}_k)$. Thus, defining the weight vectors $\tilde{w}^{1;a}:=(w_1,\dotsc,w_{a-1},\underline{w}_a,0,\dotsc,0)\in\R^n$ and $\tilde{w}^{b;n}:=(0,\dotsc,0,\overline{w}_b,w_{b+1},\dotsc,w_n)\in\R^n$, we have $\tilde{w}^{1;a},\tilde{w}^{a;b},\tilde{w}^{b;n}\in [0,\infty)^n$ and  
\begin{equation}
\label{eq:bdadj2}
w_i(Y_i-\hat{\theta}_i)=(\tilde{w}_i^{1;a}+\tilde{w}_i^{a;b}+\tilde{w}_i^{b;n})(Y_i-\hat{\theta}_i)
\end{equation}
for $i\in [n]$. Moreover,
\begin{equation}
\label{eq:bdadj3}
\sum_{i=1}^n\tilde{w}_i^{1;a}(Y_i-\hat{\theta}_i)=0=\sum_{i=1}^n\tilde{w}_i^{b;n}(Y_i-\hat{\theta}_i)
\end{equation}
by the definitions of $\underline{w}_a,\overline{w}_b$ and the second bullet point above. Now for each $\theta\in\Theta\cup\{-\hat{\theta}\}$, we have $\hat{\theta}+\eta\theta^{(a:b)}\in\Theta$ for some $\eta>0$ by assumption, so it follows from~\eqref{eq:projconv},~\eqref{eq:bdadj2} and~\eqref{eq:bdadj3} that
\begin{align*}
0\geq\sum_{i=1}^n w_i(Y_i-\hat{\theta}_i)\,\theta_i^{(a:b)}&=\sum_{i=1}^n\tilde{w}_i^{1;a}(Y_i-\hat{\theta}_i)\,\theta_a+\sum_{i=1}^n\tilde{w}_i^{a;b}(Y_i-\hat{\theta}_i)\,\theta_i+\sum_{i=1}^n\tilde{w}_i^{b;n}(Y_i-\hat{\theta}_i)\,\theta_b\\
&=\sum_{i=1}^n\tilde{w}_i^{a;b}(Y_i-\hat{\theta}_i)\,\theta_i
\end{align*}
for all $\theta\in\Theta\cup\{-\hat{\theta}\}$. In particular, this holds for $\theta=\pm\hat{\theta}$, so $\sum_{i=1}^n\tilde{w}_i^{a;b}(Y_i-\hat{\theta}_i)\,\hat{\theta}_i=0$. We conclude that $\sum_{i=1}^n\tilde{w}_i^{a;b}(Y_i-\hat{\theta}_i)(\theta_i-\hat{\theta}_i)\leq 0$ for all $\theta\in\Theta$, and hence that $\hat{\theta}\in\argmin_{\theta\in\Theta}\norm{Y-\theta}_{\tilde{w}^{a;b}}$ by~\eqref{eq:projconv}, as required.
\end{proof}

\unparskip
For LSEs $\tilde{f}_n$ over classes $\tilde{\mathcal{F}}$ of shape-constrained functions on $[0,1]$, we will now apply Lemma~\ref{lem:bdadj} to $\Theta\equiv\Theta(\tilde{\mathcal{F}}):=\bigl\{\bigl(f(x_1),\dotsc,f(x_n)\bigr):1\leq i\leq n\bigr\}$ and the LSE $\hat{\theta}=\argmin_{\theta\in\Theta}\norm{Y-\theta}=\bigl(\tilde{f}_n(x_1),\dotsc,\tilde{f}_n(x_n)\bigr)$, where $\norm{{\cdot}}$ denotes the standard Euclidean norm on $\R^n$ corresponding to $w=\mathbf{1}$. The key observation is that the conditions in parts (a) and (b) of the lemma are satisfied when $k$ (or $a,b$) is the index of a \emph{jump} or \emph{knot} of $\tilde{f}_n$.  As mentioned previously, in our first setting of isotonic regression, Corollary~\ref{cor:bdiso} provides an alternative proof of a known result \citep[e.g.][Lemma~2.1]{GJ14}; however, for the convex and S-shaped LSEs, treated in Corollary~\ref{cor:bdconv} and Proposition~\ref{prop:srestr} respectively, the results are new to the best of our knowledge.  Henceforth, for $f\colon [0,1]\to\R$ and a weight vector $w\equiv (w_1,\dotsc,w_n)\in [0,\infty)^n$, we write $S_n(f,w):=\sum_{i=1}^n w_i\bigl(Y_i-f(x_i)\bigr)^2$.  
\begin{corollary}
\label{cor:bdiso}
Let $\mathcal{F}^\uparrow$ denote the class of all non-decreasing functions $f\colon [0,1]\to\R$. Denote by $\bar{f}_n$ the (isotonic) LSE over $\mathcal{F}^\uparrow$ based on $\{(x_i,Y_i):1\leq i\leq n\}$, which for definiteness is taken to be a left continuous, piecewise constant function with jumps only at the design points $x_i$. Let $1\leq a\leq b\leq n$ be such that either $a=1$ or $\bar{f}_n(x_{a-1})<\bar{f}_n(x_a)$, and either $b=n$ or $\bar{f}_n(x_b)<\bar{f}_n(x_{b+1})$. Then $\bar{f}_n$ minimises $f\mapsto S_n(f,\mathbf{1}^{[a:b]})=\sum_{i=a}^b\bigl(Y_i-f(x_i)\bigr)^2$ over $\mathcal{F}^\uparrow$, so that its restriction to $[x_a,x_b]$ coincides with the isotonic LSE based on $\{(x_i,Y_i):a\leq i\leq b\}$.
\end{corollary}

\deparskip
\begin{proof}
Here, $\Theta^\uparrow\equiv\Theta(\mathcal{F}^\uparrow)=\{\theta = (\theta_1,\ldots,\theta_n) \in\R^n:\theta_1\leq\cdots\leq\theta_n\}$ is the monotone cone, $w=\mathbf{1}$ is the weight vector and $\hat{\theta}=\argmin_{\theta\in\Theta^\uparrow}\norm{Y-\theta}=\bigl(\bar{f}_n(x_1),\dotsc,\bar{f}_n(x_n)\bigr)$.  Since $\hat{\theta}\pm\eta\mathbf{1}^{[a:n]}\in\Theta^\uparrow$ and $\hat{\theta}\pm\eta\mathbf{1}^{[1:b]}\in\Theta^\uparrow$ for all sufficiently small $\eta>0$, condition (ii) of Lemma~\ref{lem:bdadj}(a) holds for $a,b$.  By Lemma~\ref{lem:bdadj}(a), $\underline{w}_a=0=\overline{w}_b$, so $\tilde{w}^{a;b} = \mathbf{1}^{[a:b]}$.  Since $\theta^{(a:b)} \in \Theta^\uparrow$ whenever $\theta \in \Theta$, we have $\hat{\theta} + \eta \theta^{(a:b)} \in \Theta^\uparrow$ for every $\eta > 0$, and moreover
\[
\hat{\theta} + \eta (-\hat{\theta}^{(a:b)}) = \bigl(\hat{\theta}_1-\eta \hat{\theta}_a,\ldots,\hat{\theta}_{a-1}-\eta \hat{\theta}_a,(1-\eta)\hat{\theta}_a,\ldots,(1-\eta)\hat{\theta}_b,\hat{\theta}_{b+1}-\eta \hat{\theta}_b,\ldots,\hat{\theta}_n-\eta \hat{\theta}_b\bigr) \in \Theta^\uparrow
\]
for every $\eta \in (0,1]$.  We may therefore apply Lemma~\ref{lem:bdadj}(b) to deduce the result.
\end{proof}

\unparskip

\begin{corollary}
\label{cor:bdconv}
Let $\mathcal{C}$ denote the class of all convex functions $f\colon [0,1]\to\R$. Denote by $\breve{f}_n$ the (convex) LSE over $\mathcal{C}$ based on $\{(x_i,Y_i):1\leq i\leq n\}$, which for definiteness is taken to an element of $\mathcal{G}\equiv\mathcal{G}[x_1,\ldots,x_n]$.  Let $1\leq a\leq b\leq n$ be such that for each $k\in\{a,b\}$, either $k\in\{1,n\}$ or $x_k$ is a kink of $\breve{f}_n$. Let $\hat{\theta}_i:=\breve{f}_n(x_i)$ for $i\in [n]$ and define $\tilde{w}^{a;b}$ in accordance with~\eqref{eq:bdweights} and~\eqref{Eq:wtilde}. Then $\breve{f}_n$ minimises $f\mapsto S_n(f,\tilde{w}^{a;b})=\sum_{i=a}^b\tilde{w}_i^{a;b}\bigl(Y_i-f(x_i)\bigr)^2$ over $\mathcal{C}$.
\end{corollary}

\deparskip
\begin{proof}
Here, we take $\Theta\subseteq\R^n$ to be the closed, convex cone of convex sequences based on $x_1,\dotsc,x_n$, i.e.~
\begin{equation}
\label{eq:Slopes}
\Theta(\mathcal{C})=\biggl\{(\theta_1,\dotsc,\theta_n) \in \R^n:\frac{\theta_2-\theta_1}{x_2-x_1}\leq\cdots\leq\frac{\theta_n-\theta_{n-1}}{x_n-x_{n-1}}\biggr\},
\end{equation}
and $\hat{\theta}\equiv (\hat{\theta}_1,\dotsc,\hat{\theta}_n)=\argmin_{\theta\in\Theta}\norm{Y-\theta}$.  Observe that $\hat{\theta}+\eta\mathbf{1}^{[1:k]}\in\Theta$ and $\hat{\theta}+\eta\mathbf{1}^{[k:n]}\in\Theta$ for each $k\in\{a,b\}$ and sufficiently small $\eta>0$, so condition (i) of Lemma~\ref{lem:bdadj}(a) holds with $\varepsilon=1$ for both $a$ and~$b$.  Similar considerations to those in the proof of Corollary~\ref{cor:bdiso}, but now with reference to the slopes $(\theta_i - \theta_{i-1})/(x_i - x_{i-1})$ for $i \in \{2,\ldots,n\}$, reveal that $\hat{\theta} + \eta \theta^{(a:b)} \in \Theta$ for sufficiently small $\eta > 0$ and for every $\theta \in \Theta \cup \{-\hat{\theta}\}$.  The result therefore follows again from Lemma~\ref{lem:bdadj}(b).
\end{proof}

\unparskip
When localising convex LSEs $\breve{f}_n$ to subintervals $[x_a,x_b]$ where $a,b$ are kinks of $\breve{f}_n$, we usually require non-trivial boundary weights $\overline{w}_a,\underline{w}_b\in (0,1)$, as defined in~\eqref{eq:bdweights}. 
We also mention that the conclusion of Corollary~\ref{cor:bdconv} remains valid if $\mathcal{C}$ is replaced throughout with $-\mathcal{C}$, the set of all concave functions $f\colon [0,1]\to\R$. Indeed, this result for concave LSEs follows from essentially the same proof (taking $\varepsilon=-1$ instead in condition (i) of Lemma~\ref{lem:bdadj}), or alternatively by a symmetry argument: if $\breve{f}_n$ is the LSE over $\mathcal{C}$ based on $\{(x_i,Y_i):1\leq i\leq n\}$, then $-\breve{f}_n$ is the LSE over $-\mathcal{C}$ based on $\{(x_i,-Y_i):1\leq i\leq n\}$.

Finally, we turn to the S-shaped LSEs that we study in this paper. For completeness, we first give the proof of the existence result in Section~\ref{sec:proj}. Throughout, we suppress the dependence on $x_1,\dotsc,x_n$ of sets such as $\mathcal{G}$ and $\mathcal{H}$, which are defined in Sections~\ref{subsec:notation}--\ref{sec:computation}.

\unparskip
\begin{proof}[Proof of Proposition~\ref{prop:existence}]
For $m\in [0,1]$, note that $\Gamma^m:=\Theta(\mathcal{F}^m)\subseteq\R^n$ is a closed convex cone. Thus, $f\colon [0,1]\to\R$ is an LSE over $\mathcal{F}^m$ if and only if $\bigl(f(x_1),\dotsc,f(x_n)\bigr)=\argmin_{\theta\in\Gamma^m}\norm{Y-\theta}=:\hat{\theta}^m$, the unique projection of $Y$ onto $\Theta^m$, so an LSE over $\mathcal{F}^m$ exists and is unique on $\{x_1,\dotsc,x_n\}$. 

Moreover, every $f\in\mathcal{F}$ agrees on $\{x_1,\dotsc,x_n\}$ with some $h\in\mathcal{H}=\mathcal{F}\cap\mathcal{G}$, which has an inflection point in $\{x_1,\dotsc,x_n\}$. Thus, $\mathcal{H}=\bigcup_{j=1}^n\mathcal{H}^{x_j}$ and hence $\Gamma:=\Theta(\mathcal{F})=\Theta(\mathcal{H})=\bigcup_{j=1}^n\Gamma^{x_j}$ is a finite union of convex cones. It follows that $f\colon [0,1]\to\R$ is an LSE over $\mathcal{F}$ if and only if $\bigl(f(x_1),\dotsc,f(x_n)\bigr)\in\argmin_{\theta\in\Gamma}\norm{Y-\theta}=\argmin_{\theta\in\{\hat{\theta}^{x_1},\dotsc,\hat{\theta}^{x_n}\}}\norm{Y-\theta}$, which is non-empty. Thus, an LSE over $\mathcal{F}$ exists and belongs to $\mathcal{H}=\bigcup_{j=1}^n\mathcal{H}^{x_j}$.
\end{proof}

\unparskip
The main result of this section of direct relevance for the rest of our work in the main text is Proposition~\ref{prop:srestr} below, which reveals that the situation for localisation of S-shaped LSEs is more similar to that for convex LSEs than for isotonic LSEs, in that non-trivial boundary weights are generally required for localisation. Nevertheless, the examples following the proof show that exact localisation holds in some special cases, most notably in the setting of Proposition~\ref{prop:srestrexact}. Recall the definition of $\mathcal{H}^m$ from Section~\ref{sec:proj}.
\begin{proposition}
\label{prop:srestr}
For $m\in [0,1]$, let $\hat{f}_n^m$ be the LSE over $\mathcal{H}^m$ based on $\{(x_i,Y_i):i \in [n]\}$. For $j \in [n]$, let
\begin{equation}
\label{eq:wkkll}
\underline{w}_j:=\frac{\sum_{i=1}^{j-1}\bigl(Y_i-\hat{f}_n^m(x_i)\bigr)}{\hat{f}_n^m(x_j)-Y_j}\,\Ind_{\{\hat{f}_n^m(x_j)\neq Y_j\}}\quad\text{and}\quad\overline{w}_j:=\frac{\sum_{i=j+1}^{n}\bigl(Y_i-\hat{f}_n^m(x_i)\bigr)}{\hat{f}_n^m(x_j)-Y_j}\,\Ind_{\{\hat{f}_n^m(x_j)\neq Y_j\}},
\end{equation}
similarly to~\eqref{eq:bdweights}. For $1\leq a\leq b\leq n$, define $\tilde{w}^{a;b}:=(0,\dotsc,0,\overline{w}_a,1,\dotsc,1,\underline{w}_b,0,\dotsc,0)\in\R^n$ similarly to~\eqref{Eq:wtilde}, so that $\tilde{w}_i^{a;b}=0$ for $1\leq i<a$ and $b<i\leq n$ and $\tilde{w}_i^{a;b}=1$ for $a<i<b$. If $x_k,x_\ell$ are knots of~$\hat{f}_n^m$ with $x_{k+1}\leq m\leq x_{\ell-1}$, then $\tilde{w}^{1;k},\tilde{w}^{k;\ell},\tilde{w}^{\ell;n}\in [0,1]^n$ and the following hold:

\unparskip
\begin{enumerate}[label=(\alph*)] 
\item $\hat{f}_n^m$ minimises $f\mapsto S_n(f,\tilde{w}^{1;k})$ over all $f\colon [0,1]\to\R$ that are increasing and convex on $[x_1,x_k]$;
\item $\hat{f}_n^m$ minimises $f\mapsto S_n(f,\tilde{w}^{\ell;n})$ over all $f\colon [0,1]\to\R$ that are increasing and concave on $[x_\ell,x_n]$;
\item $\hat{f}_n^m$ minimises $f\mapsto S_n(f,\tilde{w}^{k;\ell})$ over $\mathcal{H}^m$, and hence $S_n(\hat{f}_n^m,\tilde{w}^{k;\ell})\leq S_n(f,\tilde{w}^{k;\ell})$ for all $f\in\mathcal{F}^m$.
\end{enumerate}

\unparskip
In addition, let $x_K,x_L$ be the smallest and largest inflection points of $\hat{f}_n^m$ respectively.

\unparskip	
\begin{enumerate}[label=(\alph*),resume] 
\item If $m\in (x_K,x_{K+1}]$, then $\underline{w}_K=1$ and $\overline{w}_K=0$. In this case, the increasing convex LSE $\hat{f}_{1,K}$ based on $\{(x_i,Y_i):1\leq i\leq K\}$ agrees with $\hat{f}_n^m$ on $[x_1,x_K]$, and the increasing concave LSE $\hat{f}_{n,K+1}$ based on $\{(x_i,Y_i):K+1\leq i\leq n\}$ agrees with $\hat{f}_n^m$ on $[x_{K+1},x_n]$.
\item If $m\in [x_{L-1},x_L)$, then $\underline{w}_L=0$ and $\overline{w}_L=1$. In this case, the increasing convex LSE $\hat{f}_{1,L-1}$ based on $\{(x_i,Y_i):1\leq i\leq L-1\}$ agrees with $\hat{f}_n^m$ on $[x_1,x_{L-1}]$, and the increasing concave LSE $\hat{f}_{n,L}$ based on $\{(x_i,Y_i):L\leq i\leq n\}$ agrees with $\hat{f}_n^m$ on $[x_L,x_n]$.
\end{enumerate}
\unparskip
\end{proposition}

\deparskip
\begin{proof}
Here, $\hat{\theta}\equiv\hat{\theta}^m=\argmin_{\theta\in\Gamma^m}\norm{Y-\theta}=\bigl(\hat{f}_n^m(x_1),\dotsc,\hat{f}_n^m(x_n)\bigr)$ corresponds to the weight vector $w=\mathbf{1}$ and closed, convex cone $\Gamma^m=\Theta(\mathcal{F}^m)\subseteq\R^n$. For $k,\ell$ as in (a,\,b,\,c), the facts $\tilde{w}^{1;k},\tilde{w}^{k;\ell},\tilde{w}^{\ell;n}\in [0,1]^n$ follow from Lemma~\ref{lem:bdadj}(a), where it can be verified that condition (i) holds for $k$ with $\varepsilon=1$ and for $\ell$ with $\varepsilon=-1$. For (d,\,e), it can be seen that $\hat{\theta}\pm\eta\mathbf{1}^{[1:K]}\in\Gamma^m$ and $\hat{\theta}\pm\eta\mathbf{1}^{[L:n]}\in\Gamma^m$ for all sufficiently small $\eta>0$, so condition (ii) in Lemma~\ref{lem:bdadj}(a) holds and therefore $\underline{w}_K=1$ and $\overline{w}_L=1$. The remaining assertions in (a)--(e) then follow by checking the hypotheses of Lemma~\ref{lem:bdadj}(b).
\end{proof}

\unparskip

\textbf{Exact subinterval localisation}: We now give some examples of situations where (d) and (e) hold, in which case the LSE~$\hat{f}_n^m$ over $\mathcal{H}^m$ can be localised exactly to subintervals (without a non-trivial boundary adjustment) in the same way as for the isotonic LSE in Corollary~\ref{cor:bdiso}. Let $s_n(j):=S_n(\hat{f}_n^{x_j})$ for each $j\in [n]$.

\unparskip
\begin{enumerate}[label=(\roman*)]
\item If $K\in [n]$ is a \emph{local} minimum of $j\mapsto s_n(j)$ satisfying $s_n(K-1)>s_n(K)=s_n(K+1)$, then~(d) holds for $m=x_{K+1}$. 
Indeed, since $s_n(K-1)>s_n(K)$, we have $\hat{f}_n^{x_K}\notin\mathcal{H}^{x_{K-1}}$, so $x_K$ must be the smallest inflection point of $\hat{f}_n^{x_K}$. Since $s_n(K)=s_n(K+1)$, it follows that $\hat{f}_n^{x_K}=\hat{f}_n^{x_{K+1}}$ minimises $f\mapsto S_n(f)$ over $\mathcal{H}^{x_K}\cup\mathcal{H}^{x_{K+1}}$, so the hypotheses of (d) are satisfied. 
\item Similarly, if $L\in [n]$ is such that $s_n(L-1)=s_n(L)<s_n(L+1)$, then (e) holds for $m=x_{L-1}$.
\item If $\tilde{f}_n$ is an S-shaped LSE over $\mathcal{H}=\bigcup_{j=1}^{\,n}\mathcal{H}^{x_j}$ and $x_K,x_L$ are its smallest and largest inflection points respectively, then when $m = x_{K+1}$, we have that~(d) holds, and when $m = x_{L-1}$, we have that (e) holds. This yields the key Proposition~\ref{prop:srestrexact} in Section~\ref{sec:computation}.
\end{enumerate}

\section{Proofs for Section~\ref{sec:lsreg}}
\label{sec:lsregmainproofs}
For $\theta\in\R^n$ and $J=\{a,a+1,\dotsc,b\}$ with $1\leq a\leq b\leq n$, we write $\theta_J:=(\theta_i:i\in J)$ for the subvector indexed by $J$. We say that $u\equiv (u_a,u_{a+1},\dotsc,u_b)$ is a convex sequence (based on $x_a,x_{a+1},\dotsc,x_b$) if $u=\bigl(f(x_a),f(x_{a+1}),\dotsc,f(x_b)\bigr)$ for some convex $f\colon [0,1]\to\R$, and define concave and affine sequences analogously. Denote by $K^J\equiv K^{a,b}$ the set of all convex sequences based on $x_a,\dotsc,x_b$, which is a closed, convex cone; see~\eqref{eq:Slopes}. Recall from Section~\ref{sec:bdadj} the definitions of the monotone cone $\Theta^\uparrow=\{(\theta_1,\dotsc,\theta_n)\in\R^n:\theta_1\leq\cdots\leq \theta_n\}$ and cone $\Gamma^m=\Theta(\mathcal{F}^m)=\bigl\{\bigl(f(x_1),\dotsc,f(x_n)\bigr):f\in\mathcal{F}^m\bigr\}$ for $m\in [0,1]$. Let $\Gamma:=\Theta(\mathcal{F})=\bigcup_{j=1}^{\,n}\Gamma^{x_j}$, so that if $\tilde{f}_n$ is an LSE over $\mathcal{F}$, then $\tilde{\theta}_n:=\bigl(\tilde{f}_n(x_1),\dotsc,\tilde{f}_n(x_n)\bigr)\in\argmin_{\theta\in\Gamma}\norm{Y-\theta}$. Sometimes, we will write, e.g., $\Gamma\equiv\Gamma[\mathcal{D}]$ to emphasise the dependence on the set $\mathcal{D}$ of design points $x_1<\cdots<x_n$. For a general closed, convex cone $\Lambda\subseteq\R^n$ and $\theta\in\R^n$, we write $T_{\Lambda}(\theta):=\{\lambda(v-\theta):v\in \Lambda,\,\lambda\geq 0\}$ for the corresponding \emph{tangent cone} at $\theta$. 

For fixed $n\in\N$, let $Y:=(Y_1,\dotsc,Y_n)$, $\theta_0:=\bigl(f_0(x_1),\dotsc,f_0(x_n)\bigr)$ and $\xi:=(\xi_1,\dotsc,\xi_n)$, so that $Y=\theta_0+\xi$ under the model~\eqref{eq:Model}.

\hfparskip
\subsection{Sharp oracle inequalities}
\label{subsec:oracle}
\begin{proof}[Proof of Theorem~\ref{thm:worstcase}]
For a fixed $\theta\in\Gamma$, define $V(\theta):=\theta_n-\theta_1$, and for $r>0$, let $\Gamma(\theta,r)\equiv\Gamma(\theta,r)[\mathcal{D}]:=\{v\in\Gamma[\mathcal{D}]:\norm{v-\theta}\leq r\}$. To prove~\eqref{eq:worstpr}, we claim that it suffices to find $r_\ast(\theta)>0$ such that 
\begin{equation}
\label{eq:lgwineq}
\E\,\biggl(\sup_{v\in\Gamma(\theta,r_\ast(\theta))}\:\abs{Z^\top (v-\theta)}\biggr)\leq\frac{r_\ast(\theta)^2}{2},
\end{equation}
where $Z\sim N_n(0,I_n)$.  Indeed, by the sub-Gaussianity of the errors in Assumption~\ref{ass:sG}, it then follows from~\citet[Propositions~2.4,~6.3 and~6.4 and their proofs]{Bel18} that for every $t>0$, we have
\[\norm{\tilde{\theta}_n-\theta_0}\leq\norm{\theta-\theta_0}+r_\ast(\theta)+\sqrt{8t}\]
with probability at least $1-e^{-t}$. 

First, we note that $\Gamma(\theta,r)\subseteq\Theta^\uparrow(\theta,r):=\{v\in\Theta^\uparrow:\norm{v-\theta}\leq r\}$ for each $r>0$ and deduce from the proof of~\citet[Theorem~2.2]{Cha14} that if we set $r_{1,\ast}(\theta):=Cn^{1/6}\,\bigl(1+V(\theta)\bigr)^{1/3}$ for a sufficiently large universal constant $C>0$, then
\[\E\,\biggl(\sup_{v\in\Gamma(\theta,r_{1,\ast}(\theta))}\:\abs{Z^\top (v-\theta)}\biggr)\leq\E\,\biggl(\sup_{v\in\Theta^\uparrow(\theta,r_{1,\ast}(\theta))}\:\abs{Z^\top (v-\theta)}\biggr)\leq\frac{r_{1,\ast}(\theta)^2}{2};\]
see also (3.4) in~\citet{Bel18}. Moreover, by taking $\tilde{C}\geq 1$ to be sufficiently large in Lemma~\ref{lem:lgwTheta}, we see from~\eqref{eq:lgwTheta} that~\eqref{eq:lgwineq} is satisfied if we take $r_\ast(\theta)=r_{2,\ast}(\theta):=C'(Rn)^{1/10}\,\bigl(1+V(\theta)\bigr)^{1/5}$ for some suitably large universal constant $C'>0$. The desired conclusion follows upon setting $r_\ast(\theta):=r_{1,\ast}(\theta)\wedge r_{2,\ast}(\theta)$.
\end{proof}

\unparskip
As mentioned in Section~\ref{sec:lsreg}, it is possible to modify the definition of $R$ in Theorem~\ref{thm:worstcase} to yield further refinements for certain designs.  In particular, for a set $\mathcal{D}$ of design points $x_1<\cdots<x_n$, define $\tilde{R}(\mathcal{D}):=1$ if $n=1$,
and otherwise inductively set
\begin{equation}
\label{eq:RD}
\tilde{R}(\mathcal{D}):= \frac{x_n-x_1}{\min_{2\leq i\leq
n}(x_i-x_{i-1})} \wedge \min_{\mathcal{D}_1,\dotsc,\mathcal{D}_k}\rbr{\sum_{\ell=1}^k
\tilde{R}(\mathcal{D}_\ell)^{1/5}}^5,
\end{equation}
where the minimum is taken over all partitions of $\mathcal{D}$ into $k\geq 2$ non-empty sets $\mathcal{D}_1,\dotsc,\mathcal{D}_k$. The proofs of Lemmas~\ref{lem:coverTheta} and~\ref{lem:convcover} reveal that we can replace $R$ in Theorem~\ref{thm:worstcase} with the quantity $n^{-1}\tilde{R}(\{x_1,\ldots,x_n\})$, which, due to the minimum in the definition, is certainly no larger than~$R$. 
This claim follows by partitioning the set $\mathcal{D}$ of design points, then finding, for each subset $\mathcal{D}_\ell$ in the partition, a good approximation to a given S-shaped function at the design points in~$\mathcal{D}_\ell$, and finally constructing an overall approximation by linear interpolation.
To see the advantages of this modified (albeit more complicated) definition of $\tilde{R}(\mathcal{D})$, consider first a perturbation of the equispaced design $x_i=i/n$ for $i \in [n]$, where we set $x_0 := (1-\delta)/n$ for some $\delta \in (0,1)$.  Then our original quantity $R$ is at least $1/(2\delta)$ when $n \geq 2$, whereas
\begin{align*}
\frac{1}{n+1}\tilde{R}(\{x_0,x_1,\ldots,x_n\})&\leq \frac{1}{n+1} \bigl\{\tilde{R}(\{x_1,\ldots,x_n\})^{1/5} + \tilde{R}(\{x_0\})^{1/5}\bigr\}^5\\
&\leq \frac{1}{n+1}\bigl((n-1)^{1/5}+1\bigr)^5 \lesssim 1.
\end{align*}
As another example, fix $k \in \mathbb{N}$, suppose for simplicity that $n/k$ is an integer, and suppose further that
\[
x_{\ell k + j} = \frac{(\ell +\delta_j)k }{n},
\]
for $\ell = 0,1,\ldots,(n/k) -1$ and $j = 1,\ldots,k$, where $0 < \delta_1 < \cdots < \delta_k < 1/2$.  Here, the design points can be partitioned into $k$ groups, within each of which the points are equispaced, so
\[
R \geq \frac{1}{2k\min_j (\delta_{j+1} - \delta_j)},
\]
when $n \geq 2$, while
\[
\frac{1}{n} \tilde{R}(\{x_1,\ldots,x_n\}) \leq \frac{1}{n}\biggl(k \cdot \frac{n^{1/5}}{k^{1/5}}\biggr)^5 \wedge R = k^4 \wedge R.
\]
Thus, in both examples, the modified definition may provide a significant improvement, in the first case when $\delta \ll 1$, and in the second, when $k^5 \ll 1/\min_j (\delta_{j+1} - \delta_j)$.  This enables us to recover a rate of convergence of $n^{-2/5}$ in Theorem~\ref{thm:worstcase} in both cases, provided that $k$ is treated as a constant in the second case.  Overall, this new definition yields additional insight into the effect of the design on the rate of convergence, and provides reassurance about the robustness of the performance of the LSE $\hat{f}_n$ for much wider classes of designs.

\unparskip
\begin{proof}[Proof of Theorem~\ref{thm:adaptive}]
For a closed, convex cone $\Lambda\subseteq\R^n$, recall that the \emph{statistical dimension} of $\Lambda$ is defined as $\delta(\Lambda):=\E\bigl(\norm{\Pi_\Lambda(Z)}^2\bigr)$, where $Z\sim N_n(0,I_n)$ and $\Pi_\Lambda\colon\R^n\to \Lambda$ denotes the projection map onto $\Lambda$~\citep{ALMT14}. Since $\Gamma$ is the union of the closed, convex cones $\Gamma^{x_1},\dotsc,\Gamma^{x_n}$ and by Assumption~\ref{ass:sG}, it follows from (2.7) and Propositions~6.1 and 6.4 of~\citet{Bel18} that for any $\theta\in\Gamma$ and $t>0$, we have
\begin{equation}
\label{eq:2.7}
\norm{\tilde{\theta}_n-\theta_0}\leq\norm{\theta-\theta_0}+2\rbr{\max_{1\leq j\leq n}\delta^{1/2}\bigl(T_{\Gamma^{x_j}}(\theta)\bigr)+\sqrt{2(t+\log n)}}
\end{equation}
with probability at least $1-e^{-t}$. Denoting by $k_\theta$ the smallest $k\in\N$ for which $\theta\equiv(\theta_1,\dotsc,\theta_n)\in\R^n$ is affine on $k$ pieces,
we claim that
\[
\delta\bigl(T_{\Gamma^{x_j}}(\theta)\bigr)\leq 8(k_\theta+1)\log\biggl(\frac{en}{k_\theta+1}\biggr)
\]
for all $j \in [n]$ and $\theta\in\Gamma$. Indeed, for fixed $j \in [n]$ and $\theta\in\Gamma$, we write $k\equiv k_\theta$ and let $0=j_0\leq j_1<\cdots<j_{k'}=j<j_{k'+1}<\cdots<j_k<j_{k+1}=n$ be such that the subvector $\theta_{J_r}=(\theta_i:j_r+1\leq i\leq j_{r+1})$ indexed by $J_r:=\{j_r+1,j_r+2,\dotsc,j_{r+1}\}$ is an affine sequence for every $0\leq r\leq k$. Then for any $v\in\Gamma^{x_j}$, note that $(v-\theta)_{J_r}$ is a convex sequence if $0\leq r\leq k'-1$ and a concave sequence if $k'\leq r\leq k$. Thus, $T_{\Gamma^{x_j}}(\theta)=\{\lambda(v-\theta):v\in\Gamma^{x_j},\,\lambda\geq 0\}\subseteq\prod_{r=0}^{k'-1}K^{J_r}\times\prod_{r=k'}^k\,(-K^{J_r})$. Since $\delta(\pm K^{J_r})\leq 8\log(e\abs{J_r})$ for each $0\leq r\leq k$ by~\citet[Proposition~4.2]{Bel18}, it follows from~\citet[Proposition~3.1(9,\,10)]{ALMT14} that
\begin{equation}
\label{eq:statdimTheta}
\delta\bigl(T_{\Gamma^{x_j}}(\theta)\bigr)\leq\sum_{r=0}^{k'-1}\delta(K^{J_r})+\sum_{r=k'}^k\delta(-K^{J_r})\leq\sum_{r=0}^k 8\log(e\abs{J_r})\leq 8(k+1)\log\rbr{\frac{en}{k+1}},
\end{equation}
as required, where the final inequality follows from Jensen's inequality together with the fact that $\sum_{r=0}^k\,\abs{J_r}=n$. Finally, if $f\in\mathcal{H}$ and $\theta=\bigl(f(x_1),\dotsc,f(x_n)\bigr)$, then $\norm{f-f_0}_n^2=\norm{\theta-\theta_0}^2/n$, so the sharp oracle inequality~\eqref{eq:adaptpr} is a direct consequence of~\eqref{eq:2.7} and~\eqref{eq:statdimTheta}.
\end{proof}

\umparskip
\subsection{Inflection point estimation}
\label{subsec:inflection}
The proofs of some technical lemmas in this subsection are deferred to Section~\ref{subsec:inflectproofs}.

\unparskip
\begin{proof}[Proof of Theorem~\ref{thm:inflection}]
For each $n\in\N$, let $\tilde{m}_-\equiv\tilde{m}_{n-}$ and $\tilde{m}_+\equiv\tilde{m}_{n+}$ be the smallest and largest inflection points of $\tilde{f}_n$ respectively. Letting $(C_n)$ be any deterministic positive sequence with $C_n\to\infty$, and defining the events $E_n^\pm:=\bigl\{\pm(\tilde{m}_\pm-m_0)>C_n(n/\log n)^{-1/(2\alpha+1)}\bigr\}$, we aim to establish that $\Pr(E_n^\pm)\to 0$ as $n\to\infty$. We will consider only the events $E_n^+$; the arguments for $E_n^-$ are analogous.

Our strategy is to show that there exist events $(\Omega_n)$ with $\Pr(\Omega_n^c)\to 0$ such that $\Delta_n:=S_n(\tilde{f}_n)-S_n(\hat{f}_n^{m_0})>0$ on $E_n^+\cap\Omega_n$ for all sufficiently large $n$. Since $\tilde{f}_n$ and $\hat{f}_n^{m_0}$ are LSEs over $\mathcal{F}$ and $\mathcal{H}^{m_0}$ respectively, we have $S_n(\tilde{f}_n)=\min_{f\in\mathcal{F}}S_n(f)\leq S_n(\hat{f}_n^{m_0})$, so $E_n^+\cap\Omega_n=\emptyset$ for all sufficiently large $n$, whence, by the reverse Fatou lemma, $\Pr(E_n^+)\leq\Pr(E_n^+\cap\Omega_n)+\Pr(\Omega_n^c)\to 0$, as desired. 
\begin{figure}[htb]
\centering
\includegraphics[width=\textwidth]{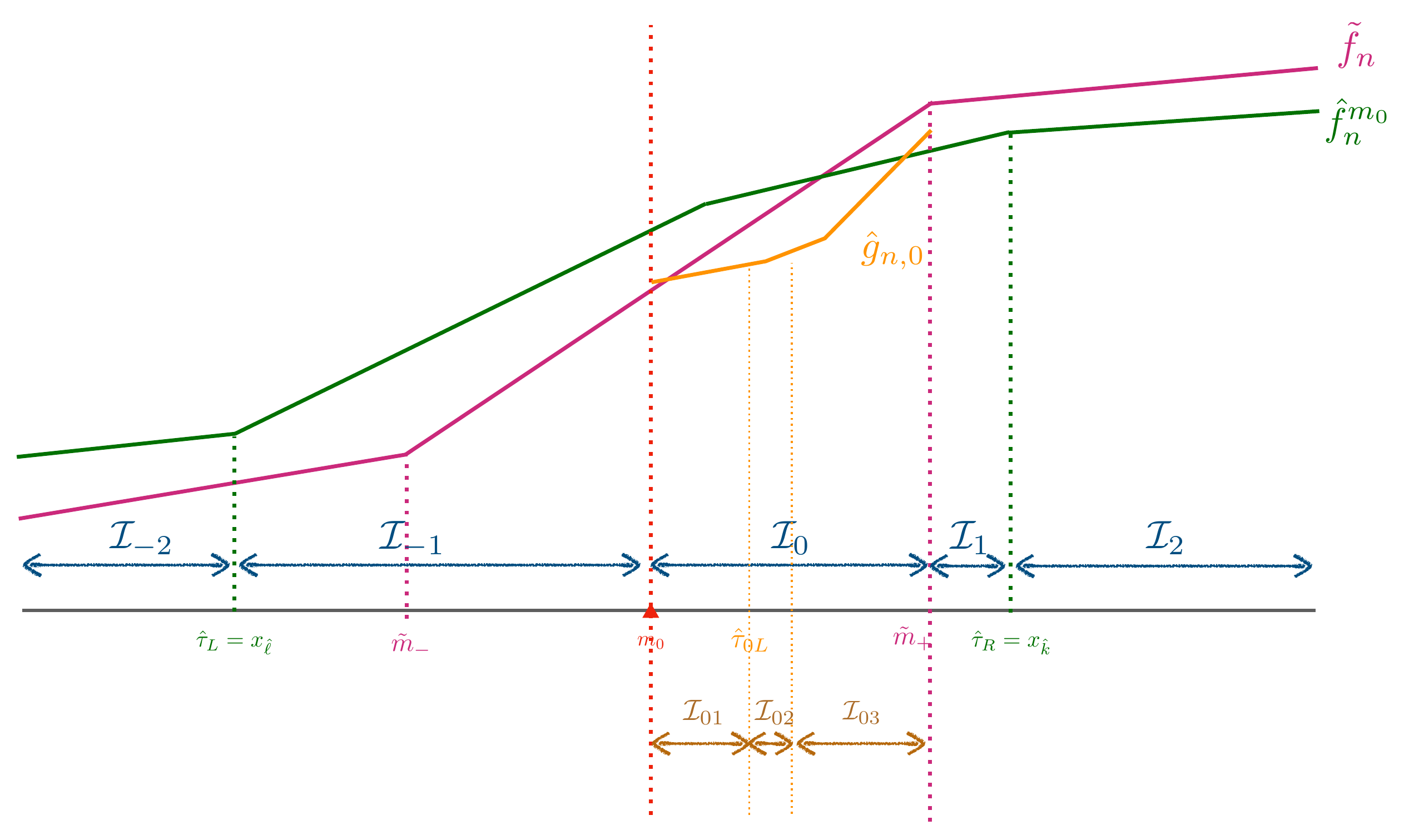}
\caption{Illustration of the proof of Theorem~\ref{thm:inflection}.}
\label{fig:inflection}
\end{figure}

\textbf{Step 1 -- subdividing $[0,1]$ and making `boundary adjustments'}: For each $n$, we make the following definitions, suppressing the dependence on $n$ to ease notation where appropriate. Writing $\hat{\tau}_-=x_\ell$ for the smallest inflection point of $\hat{f}_n^{m_0}$, we set $\hat{\tau}_L:=\hat{\tau}_-$ if $\hat{\tau}_-<m_0$, and otherwise define $\hat{\tau}_L$ to be the largest knot of $\hat{f}_n^{m_0}$ in $[0,x_{\ell-1}]$. Also, let $\hat{\tau}_R$ be the smallest knot of $\hat{f}_n^{m_0}$ in $[\tilde{m}_+,1]$. On the event $E_n^+$, we may decompose $[0,1]$ into the subintervals $\mathcal{I}_{-2}:=[0,\hat{\tau}_L]$, $\mathcal{I}_{-1}:=[\hat{\tau}_L,m_0)$, $\mathcal{I}_0:=[m_0,\tilde{m}_+]$, $\mathcal{I}_1:=(\tilde{m}_+,\hat{\tau}_R]$ and $\mathcal{I}_2:=[\hat{\tau}_R,1]$. For $-2\leq A\leq 2$, we associate $\mathcal{I}_A$ with a weight vector $w^A\in\R^n$ that is defined below in~\eqref{eq:weights}. Let $\hat{k},\hat{\ell}$ be such that $\hat{\tau}_L=x_{\hat{\ell}}$ and $\hat{\tau}_R=x_{\hat{k}}$, and for $s\in\{\hat{k},\hat{\ell}\}$, let
\[\underline{w}_s:=\frac{\sum_{i=1}^{s-1}\bigl(Y_i-\hat{f}_n^{m_0}(x_i)\bigr)}{\hat{f}_n^{m_0}(x_s)-Y_s}\,\Ind_{\{\hat{f}_n^{m_0}(x_s)\neq Y_s\}}\quad\text{and}\quad\overline{w}_s:=\frac{\sum_{i=s+1}^{n}\bigl(Y_i-\hat{f}_n^{m_0}(x_i)\bigr)}{\hat{f}_n^{m_0}(x_s)-Y_s}\,\Ind_{\{\hat{f}_n^{m_0}(x_s)\neq Y_s\}}\]
as in~\eqref{eq:wkkll}, where $x_i\equiv x_{ni}$ and $Y_i\equiv Y_{ni}$ for all $i$. Then by Proposition~\ref{prop:srestr}, $\underline{w}_s,\overline{w}_s\in [0,1]$ and $\underline{w}_s+\overline{w}_s\leq 1$, with equality when $Y_s\neq\hat{f}_n^{m_0}(x_s)$. Now for $i\in [n]$ and $-2\leq A\leq 2$, define
\begin{equation}
\label{eq:weights}
w_i^A:=
\begin{cases}
\,1\quad&\text{if }x_i\in\mathcal{I}_A\setminus\{\hat{k},\hat{\ell}\}\\
\,0\quad&\text{if }x_i\in\mathcal{I}_A^c\\
\,\underline{w}_i&\text{if either }A=-2\text{ and }i=\hat{\ell}\text{, or }A=1\text{ and }i=\hat{k}\\
\,\overline{w}_i&\text{if either }A=-1\text{ and }i=\hat{\ell}\text{, or }A=2\text{ and }i=\hat{k}.
\end{cases}
\end{equation}
Then $w^A\in [0,1]^n$ for all $-2\leq A\leq 2$, and setting $S_n(f,w^A):=\sum_{i=1}^n w_i^A\bigl(Y_i-f(x_i)\bigr)^2$ for each $-2\leq A\leq 2$ and $f\colon [0,1]\to\R$, we have
\[S_n(f)=\sum_{i=1}^n\bigl(Y_i-f(x_i)\bigr)^2\geq\sum_{A=-2}^2\;\sum_{i=1}^n w_i^A\bigl(Y_i-f(x_i)\bigr)^2=\sum_{A=-2}^2 S_n(f,w^A),\] 
with equality when $f=\hat{f}_n^{m_0}$. Thus, defining $\Delta_{n,A}:=S_n(\tilde{f}_n,w^A)-S_n(\hat{f}_n^{m_0},w^A)$ for $-2\leq A\leq 2$, we see that $\Delta_n=S_n(\tilde{f}_n)-S_n(\hat{f}_n^{m_0})\geq\sum_{A=-2}^2\Delta_{n,A}$. 

Since $\tilde{f}_n$ is increasing and convex on $\mathcal{I}_{-2}$ and increasing and concave on $\mathcal{I}_2$, it follows from Proposition~\ref{prop:srestr} that $\Delta_{n,A}=S_n(\tilde{f}_n,w^A)-S_n(\hat{f}_n^{m_0},w^A)\geq 0$ for each $A\in\{-2,2\}$. Moreover, letting $w':=\sum_{A=-1}^{\,1} w^A$, we have $\sum_{A=-1}^{\,1} S_n(\hat{f}_n^{m_0},w^A)=S_n(\hat{f}_n^{m_0},w')\leq S_n(f_0,w')=\sum_{A=-1}^{\,1} S_n(f_0,w^A)$ by Proposition~\ref{prop:srestr}(c) and the fact that $f_0\in\mathcal{F}^{m_0}$. It is for these reasons that we made the `boundary adjustments' at $\hat{k}$ and $\hat{\ell}$ in~\eqref{eq:weights}. We can now write
\begin{equation}
\label{eq:step1}
\Delta_n\geq\sum_{A=-2}^2\Delta_{n,A}\geq\sum_{A=-1}^1\Delta_{n,A}=\sum_{A=-1}^1\bigl\{S_n(\tilde{f}_n,w^A)-S_n(\hat{f}_n^{m_0},w^A)\bigr\}\geq\sum_{A=-1}^1\tilde{\Delta}_{n,A},
\end{equation}
where 
\[\tilde{\Delta}_{n,A}:=S_n(\tilde{f}_n,w^A)-S_n(f_0,w^A)=\sum_{i: x_i\in\mathcal{I}_A}w_i^A\bigl\{\bigl(Y_i-\tilde{f}_n(x_i)\bigr)^2-\xi_i^2\bigr\}\]
for $A\in\{-1,0,1\}$, and seek to bound each of these three terms from below.

\textbf{Step 2 -- bounding $\tilde{\Delta}_{n,0}$}: On the event $E_n^+$, note that $\tilde{f}_n$ is convex on $\mathcal{I}_0=[m_0,\tilde{m}_+]$ and $f_0$ is concave on $\mathcal{I}_0$. We will exploit this mismatch of shape constraints on $\mathcal{I}_0$ to obtain a suitable lower bound on $\tilde{\Delta}_{n,0}$. For each $n$, define $\hat{g}_{n,0}\colon\mathcal{I}_0\to\R$ to be the convex LSE based on $\{(x_i,Y_i):x_i\in\mathcal{I}_0\}$, which for definiteness is taken to be a continuous, piecewise linear function with knots in $\{x_1,\dotsc,x_n\}\cap\mathcal{I}_0$. Then $S_n(\tilde{f}_n,w^0)\geq S_n(\hat{g}_{n,0},w^0)$ and
\begin{equation}
\label{eq:deltatilde0}
\tilde{\Delta}_{n,0}\geq S_n(\hat{g}_{n,0},w^0)-S_n(f_0,w^0)=\sum_{i: x_i\in\mathcal{I}_0}\bigl\{\bigl(\xi_i+f_0(x_i)-\hat{g}_{n,0}(x_i)\bigr)^2-\xi_i^2\bigr\}
\end{equation}
in view of the definition of $w^0$ in~\eqref{eq:weights}. On $E_n^+$, let $\hat{\tau}_{0L}$ be the largest knot of $\hat{g}_{n,0}$ in $[m_0,(m_0+\tilde{m}_+)/2]$, and on $(E_n^+)^c$, set $\hat{\tau}_{0L}=m_0$ for concreteness.
Suppressing the dependence on $n$ for convenience, we define $\mathcal{I}_{01}:=(m_0,\hat{\tau}_{0L}]$, $\mathcal{I}_{02}:=\bigl(\hat{\tau}_{0L},(m_0+\tilde{m}_+)/2\bigr)$ and $\mathcal{I}_{03}:=\bigl[(m_0+\tilde{m}_+)/2,\tilde{m}_+\bigr)$, and decompose the right-hand side of~\eqref{eq:deltatilde0} as $\Lambda_{n1}+\Lambda_{n2}+\Lambda_{n3}$, where
\begin{equation}
\label{eq:Lambdanj}
\Lambda_{nj}:=\sum_{i: x_i\in\mathcal{I}_{0j}}\bigl\{\bigl(\xi_i+f_0(x_i)-\hat{g}_{n,0}(x_i)\bigr)^2-\xi_i^2\bigr\}
\end{equation}
for $j\in\{1,2,3\}$. In the arguments below, a key ingredient is the following fact, whose proof (which we remind the reader is given in Section~\ref{subsec:inflectproofs})
makes use of Assumption~\ref{ass:inflection}.
\begin{lemma}
\label{lem:kinkdist}
$\hat{\tau}_{0L}-m_0=O_p\bigl((n/\log n)^{-1/(2\alpha+1)}\bigr)$.
\end{lemma}

\unparskip
Defining $t_n:=\sqrt{C_n}\,(n/\log n)^{-1/(2\alpha+1)}$ and $u_n:=2^{-1}C_n(n/\log n)^{-1/(2\alpha+1)}$, we deduce that there are events $(E_{n1})$ with $\Pr(E_{n1}^c)\to 0$ such that $m_0\leq\hat{\tau}_{0L}\leq m_0+t_n$ and $(m_0+\tilde{m}_+)/2\geq m_0+u_n$ 
on $E_n^+\cap E_{n1}$, for each $n$. 

\textbf{Step 2a -- bounding $\Lambda_{n2}$}: For each $n$, note that by the definition of $\hat{\tau}_{0L}$, the function $\hat{g}_{n,0}$ is linear on $\mathcal{I}_{02}=\bigl(\hat{\tau}_{0L},(m_0+\tilde{m}_+)/2\bigr)$, whereas $f_0$ is concave on $\mathcal{I}_{02}$. In view of this and the fact that $\mathcal{I}_{02}$ length $(m_0+\tilde{m}_+)/2-\hat{\tau}_{0L}\geq u_n-t_n=u_n\bigl(1+o(1)\bigr)$ on $E_n^+\cap E_{n1}$, we would expect the approximation error $\sum_{i: x_i\in\mathcal{I}_{02}}\bigl(\hat{g}_{n,0}(x_i)-f_0(x_i)\bigr)^2$ to be `large'; see Lemma~\ref{lem:step2lbd} below. Together with the arguments in Steps 2b and 3, this will enable us to prove that the quantity $\Lambda_{n2}$ is positive and dominates (in magnitude) all the other terms $\Lambda_{n1},\Lambda_{n3},\Delta_{n,\pm 1}$ in~\eqref{eq:step1}--\eqref{eq:Lambdanj}. This yields the eventual conclusion~\eqref{eq:Deltan} that $\Delta_n>0$ with high probability on $E_n^+$.

To handle the randomness of $\mathcal{I}_{02}$, let
\[\mathcal{T}_n:=\bigl\{(a,b):1\leq a\leq b\leq n,\,m_0\leq x_a\leq m_0+t_n,\,x_b\geq m_0+u_n\bigr\},\]
and for $(a,b)\in\mathcal{T}_n$, define the vectors $\mathbf{1}^{a,b}:=(1,1,\dotsc,1)\in\R^{b-a+1}$, $x^{a,b}:=(x_a,x_{a+1},\dotsc,x_b)$, $\xi^{a,b}:=(\xi_a,\xi_{a+1},\dotsc,\xi_b)$ and $\theta^{a,b}:=\bigl(f_0(x_a),f_0(x_{a+1}),\dotsc,f_0(x_b)\bigr)$. Then on the event $E_n^+\cap E_{n1}$, we have
\begin{align}
\Lambda_{n2}&\geq\inf_{c_0,c_1 \in \mathbb{R}}\sum_{i: x_i\in\mathcal{I}_{02}}\bigl\{\bigl(\xi_i+f_0(x_i)-c_0-c_1 x_i\bigr)^2-\xi_i^2\bigr\}\notag\\
\label{eq:Lambdan2}
&\geq\inf_{(a,b)\in\mathcal{T}_n}\,\inf_{c_0,c_1 \in \mathbb{R}}\,\bigl\{\norm{\theta^{a,b}-c_0\mathbf{1}^{a,b}-c_1 x^{a,b}}^2-2\,\bigl\langle\xi^{a,b},c_0\mathbf{1}^{a,b}+c_1 x^{a,b}-\theta^{a,b}\bigr\rangle\bigr\}.
\end{align}
For $(a,b)\in\mathcal{T}_n$, denote by $\Pi_{a,b}\,\xi^{a,b}:=\argmin_{v\in L^{a,b}}\norm{\xi^{a,b}-v}$ the projection of $\xi^{a,b}$ onto the subspace $L^{a,b}:=\Span\{\theta^{a,b},\mathbf{1}^{a,b},x^{a,b}\}$, which has dimension $d\equiv d_{a,b}\leq 3$. Then
\begin{equation}
\label{eq:suplin}
\sup_{c_0,c_1 \in \mathbb{R}}\frac{\bigl|\bigl\langle\xi^{a,b},c_0\mathbf{1}^{a,b}+c_1 x^{a,b}-\theta^{a,b}\bigr\rangle\bigr|}{\norm{c_0\mathbf{1}^{a,b}+c_1 x^{a,b}-\theta^{a,b}}}=\sup_{c_0,c_1 \in \mathbb{R}}\frac{\bigl|\bigl\langle\Pi_{a,b}\,\xi^{a,b},c_0\mathbf{1}^{a,b}+c_1 x^{a,b}-\theta^{a,b}\bigr\rangle\bigr|}{\norm{c_0\mathbf{1}^{a,b}+c_1 x^{a,b}-\theta^{a,b}}}\leq\norm{\Pi_{a,b}\,\xi^{a,b}}.
\end{equation}
Now let $\{v_1,\dotsc,v_d\}$ be an orthonormal basis of $L^{a,b}$, so that $\norm{\Pi_{a,b}\,\xi^{a,b}}=\bigl(\sum_{j=1}^d\,\ipr{\xi^{a,b}}{v_j}^2\bigr)^{1/2}\leq\sqrt{3}\max_{1\leq j\leq d}\,\abs{\ipr{\xi^{a,b}}{v_j}}$. For each $j\in [d]$, we have $\E\bigl(e^{t\ipr{\xi^{a,b}}{v_j}}\bigr)\leq e^{\norm{tv_j}^2/2}=e^{t^2/2}$ for all $t\in\R$ by Assumption~\ref{ass:inflection}, so $\ipr{\xi^{a,b}}{v_j}$ is sub-Gaussian with parameter 1. Thus, for each $(a,b)\in\mathcal{T}_n$ and every $c>0$, we have
\begin{equation}
\label{eq:suplinpr}
\Pr\rbr{\sup_{c_0,c_1 \in \mathbb{R}}\frac{\bigl|\bigl\langle\xi^{a,b},c_0\mathbf{1}^{a,b}+c_1 x^{a,b}-\theta^{a,b}\bigr\rangle\bigr|}{\norm{c_0\mathbf{1}^{a,b}+c_1 x^{a,b}-\theta^{a,b}}}
\geq\sqrt{6c\log n}}\leq\Pr\bigl(\norm{\Pi_{a,b}\,\xi^{a,b}}\geq\sqrt{6c\log n}\bigr)\leq 6n^{-c}.
\end{equation}
Since $\abs{\mathcal{T}_n}<n^2$, we can take $c=3\;(>2)$ in~\eqref{eq:suplinpr} and apply a union bound to deduce from~\eqref{eq:Lambdan2} that there are events $(E_{n2})$ with $\Pr(E_{n2}^c)\to 0$ such that
\begin{equation}
\label{eq:Lambdan2quad}
\Lambda_{n2}\geq\inf_{(a,b)\in\mathcal{T}_n}\,\inf_{c_0,c_1 \in \mathbb{R}}\,\bigl\{\norm{\theta^{a,b}-c_0\mathbf{1}^{a,b}-c_1 x^{a,b}}^2-2\sqrt{18\log n}\,\norm{\theta^{a,b}-c_0\mathbf{1}^{a,b}-c_1 x^{a,b}}\bigr\}
\end{equation}
on $E_n^+\cap E_{n1}\cap E_{n2}$, for each $n$. Note that the quadratic function $t\mapsto t^2-2t\sqrt{18\log n}$ attains its minimum at $t=\sqrt{18\log n}$ and is increasing on $[\sqrt{18\log n},\infty)$. In addition, using the local smoothness condition on $f_0$ in Assumption~\ref{ass:inflection} and the fact that $x_b-x_a\geq u_n-t_n=2^{-1}C_n(n/\log n)^{-1/(2\alpha+1)}\bigl(1+o(1)\bigr)$ for all $(a,b)\in\mathcal{T}_n$, we can show that there exists $\rho_\alpha>0$, depending only on $\alpha$, such that the following holds:
\begin{lemma}
\label{lem:step2lbd}
$\inf_{(a,b)\in\mathcal{T}_n}\inf_{c_0,c_1 \in \mathbb{R}}\norm{\theta^{a,b}-c_0\mathbf{1}^{a,b}-c_1 x^{a,b}}^2\geq\rho_\alpha B^2\,n u_n^{2\alpha+1}\geq\rho_\alpha B^2(C_n/4)^{2\alpha+1}\log n$ for all sufficiently large $n$.
\end{lemma}

\unparskip
Since $C_n\to\infty$, this means that $\inf_{(a,b)\in\mathcal{T}_n}\inf_{c_0,c_1 \in \mathbb{R}}\norm{\theta^{a,b}-c_0\mathbf{1}^{a,b}-c_1 x^{a,b}}\geq \sqrt{18\log n}$ for all sufficiently large $n$, so it follows from~\eqref{eq:Lambdan2quad} that
\begin{equation}
\label{eq:step2a}
\Lambda_{n2}\geq\rho_\alpha B^2\rbr{\frac{C_n}{4}}^{2\alpha+1}\log n-2\sqrt{(18\log n)\,\rho_\alpha B^2\rbr{\frac{C_n}{4}}^{2\alpha+1}\log n}\geq\frac{\rho_\alpha B^2}{2}\rbr{\frac{C_n}{4}}^{2\alpha+1}\log n
\end{equation}
on $E_n^+\cap E_{n1}\cap E_{n2}$, for all sufficiently large $n$.

\textbf{Step 2b -- bounding $\Lambda_{n1}$ and $\Lambda_{n3}$}: For each $n$, note that $\hat{g}_{n,0}-f_0$ is convex on $\mathcal{I}_{01}:=(m_0,\hat{\tau}_{0L}]$ and $\mathcal{I}_{03}=\bigl[(m_0+\tilde{m}_+)/2,\tilde{m}_+\bigr)$. For $j=1,3$, writing $\tilde{g}_{nj}$ for the convex LSE based on $\{(x_i,\xi_i):x_i\in\mathcal{I}_{0j}\}$, we see from~\eqref{eq:Lambdanj} that $\Lambda_{nj}\geq\sum_{i: x_i\in\mathcal{I}_{0j}}\bigl\{\bigl(\xi_i-\tilde{g}_{nj}(x_i)\bigr)^2-\xi_i^2\bigr\}$. To handle the randomness of $\mathcal{I}_{0j}$, let $\mathcal{T}_n':=\{(a,b):1\leq a\leq b\leq n\}$
and for $(a,b)\in\mathcal{T}_n'$, denote by $\hat{\xi}^{a,b}:=\argmin_{v\in K^{a,b}}\norm{\xi^{a,b}-v}$ the projection of $\xi^{a,b}=(\xi_a,\xi_{a+1},\dotsc,\xi_b)$ onto the closed, convex cone $K^{a,b}$ of convex sequences based on $x_a,\dotsc,x_b$, as defined at the start of Section~\ref{sec:lsregmainproofs}. Then $\norm{\xi^{a,b}}^2-\norm{\xi^{a,b}-\hat{\xi}^{a,b}}^2=\norm{\hat{\xi}^{a,b}}^2$ by Lemma~\ref{lem:cone}, and for every $c>0$, we have
\begin{equation}
\label{eq:bellec1}
\Pr\bigl\{\norm{\hat{\xi}^{a,b}}^2\geq 16\log\bigl(e(b-a+1)\bigr)+4c\log n\bigr\}\leq n^{-c}.
\end{equation}
This can be seen by taking $\mu=u=0$ in~\citet[Theorem~4.3]{Bel18}, an oracle inequality for convex LSEs that holds under the sub-Gaussian condition on the errors in Assumption~\ref{ass:inflection}, in view of~\citet[Remark~2.2, Proposition~6.2 and Proposition~6.4]{Bel18}. Since $\abs{\mathcal{T}_n'}<n^2$, we now take $c=3$ in~\eqref{eq:bellec1} and apply a union bound to conclude that there are events $(E_{n3})$ with $\Pr(E_{n3}^c)\to 0$ such that
\begin{align}
\Lambda_{nj}\geq\sum_{i: x_i\in\mathcal{I}_{0j}}\bigl\{\bigl(\xi_i-\tilde{g}_{nj}(x_i)\bigr)^2-\xi_i^2\bigr\}&\geq-\max_{(a,b)\in\mathcal{T}_n'}\bigl\{\norm{\xi^{a,b}-\hat{\xi}^{a,b}}^2-\norm{\xi^{a,b}}^2\bigr\}\notag\\
\label{eq:step2b}
&=-\max_{(a,b)\in\mathcal{T}_n'}\norm{\hat{\xi}^{a,b}}^2\geq -28\log(en)
\end{align}
for $j=1,3$ on $E_n^+\cap E_{n1}\cap E_{n3}$, for each $n$.

\textbf{Step 3 -- bounding $\tilde{\Delta}_{n,A}$ for $A\in\{-1,1\}$}: The techniques we apply here are broadly similar to those used in Step 2b, but the arguments are a little more involved. For each $n$, we now consider $\mathcal{I}_{-1}=[\hat{\tau}_L,m_0)$ and $\mathcal{I}_1=(\tilde{m}_+,\hat{\tau}_R]$, where $\hat{\tau}_L=x_{\hat{\ell}}$ and $\hat{\tau}_R=x_{\hat{k}}$ are as given in Step 1. Let $\hat{\ell}_+,\hat{k}_-\in [n]$ be such that $x_{\hat{\ell}_+}$ is the smallest knot of $\hat{f}_n^{m_0}$ in $(\hat{\tau}_L,1]$ and $x_{\hat{k}_-}$ is the largest knot of $\hat{f}_n^{m_0}$ in $[0,\hat{\tau}_R)$. Then $\{i:x_i\in\mathcal{I}_{-1}\}\subseteq\{\hat{\ell},\hat{\ell}+1,\dotsc,\hat{\ell}_+\}$ and $\{i:x_i\in\mathcal{I}_1\}\subseteq\{\hat{k}_-,\hat{k}_- +1,\dotsc,\hat{k}\}$ in all cases, in view of the definitions of $\hat{\tau}_L,\hat{\tau}_R$. Later on, we will apply Lemma~\ref{lem:step3} to $x_{\hat{\ell}},x_{\hat{\ell}_+}$ and $x_{\hat{k}_-},x_{\hat{k}}$, which are pairs of successive knots of $\hat{f}_n^{m_0}$.

Recalling from~\eqref{eq:step1} that we defined $\tilde{\Delta}_{n,\pm 1}$ as weighted sums of squares, we start by bounding these from below by unweighted sums that do not feature the (random) `boundary weights' $\underline{w}_{\hat{k}},\overline{w}_{\hat{\ell}}\in [0,1]$ from~\eqref{eq:weights}. For $1\leq a\leq b\leq n$, let $K^{a,b}$ be as in Step 2b, so that $K^{a,b}$ and $-K^{a,b}$ are the cones of convex and concave sequences respectively based on $x_a,\dotsc,x_b$, and let $\xi^{a,b}=(\xi_a,\xi_{a+1},\dotsc,\xi_b)$. Denote by $\check{\theta}^{a,b}:=\argmin_{v\in K^{a,b}}\norm{Y^{a,b}-v}$ and $\hat{\theta}^{a,b}:=\argmin_{v\in -K^{a,b}}\norm{Y^{a,b}-v}$ the projections of $Y^{a,b}:=(Y_a,Y_{a+1},\dotsc,Y_b)$ onto $K^{a,b}$ and $-K^{a,b}$ respectively. Let $\hat{a}:=\floor{n\tilde{m}_+}+1$ and $\check{b}:=\ceil{nm_0}-1$, so that $x_{\hat{a}-1}\leq{\tilde{m}_+}<x_{\hat{a}}$ and $x_{\check{b}}<m_0\leq x_{\check{b}+1}$, and define $\check{a},\hat{b}\in [n]$ by
\[\check{a}:=
\begin{cases}
\,\hat{\ell}&\;\text{if }\bigl(Y_{\hat{\ell}}-\tilde{f}_n(x_{\hat{\ell}})\bigr)^2\leq\xi_{\hat{\ell}}^2\\
\,\hat{\ell}+1&\;\text{otherwise}
\end{cases}
\qquad
\text{and}
\qquad
\hat{b}:=
\begin{cases}
\,\hat{k}&\;\text{if }\bigl(Y_{\hat{k}}-\tilde{f}_n(x_{\hat{k}})\bigr)^2\leq\xi_{\hat{k}}^2\\
\,\hat{k}-1&\;\text{otherwise}.
\end{cases}
\]
Then $[x_{\check{a}},x_{\check{b}}]\subseteq\mathcal{I}_{-1}$ and $[x_{\hat{a}},x_{\hat{b}}]\subseteq\mathcal{I}_1$, so $\hat{\ell}\leq\check{a}\leq\check{b}\leq\hat{\ell}_+$ and $\hat{k}_-\leq\hat{a}\leq\hat{b}\leq\hat{k}$. On the event $E_n^+$, the function $\tilde{f}_n$ is convex on $[x_{\check{a}},x_{\check{b}}]\subseteq\mathcal{I}_{-1}$ and concave on $[x_{\hat{a}},x_{\hat{b}}]\subseteq\mathcal{I}_1$, so it follows from the definitions above that
\begin{equation}
\label{eq:Deltanpm1}
\tilde{\Delta}_{n,A}=\sum_{i: x_i\in\mathcal{I}_A}w_i^A\bigl\{\bigl(Y_i-\tilde{f}_n(x_i)\bigr)^2-\xi_i^2\bigr\}\geq
\begin{cases}
\,\norm{Y^{\check{a},\check{b}}-\check{\theta}^{\check{a},\check{b}}}^2-\norm{\xi^{\check{a},\check{b}}}^2\;&\text{for }A=-1\\
\,\norm{Y^{\hat{a},\hat{b}}-\hat{\theta}^{\hat{a},\hat{b}}}^2-\norm{\xi^{\hat{a},\hat{b}}}^2\;&\text{for }A=1
\end{cases}
\end{equation}
on $E_n^+$. Next, we develop these bounds further using some orthogonality properties and the oracle inequality stated as~\citet[Theorem~4.3]{Bel18} once again, taking into account the randomness of $\check{a},\check{b},\hat{a},\hat{b}$. For $1\leq a<b\leq n$, let $\mathbf{1}^{a,b},x^{a,b}\in\R^{b-a+1}$ and $\theta^{a,b}=\bigl(f_0(x_a),f_0(x_{a+1})\dotsc,f_0(x_b)\bigr)$ be as in Step 2a, and write $A^{a,b}:=\Span\{\mathbf{1}^{a,b},x^{a,b}\}$ for the subspace of affine sequences of length $b-a+1$ based on $x_a,x_{a+1},\dotsc,x_b$. Then $\bar{\theta}^{a,b}:=\argmin_{v\in A^{a,b}}\norm{\theta^{a,b}-v}$ satisfies $\ipr{\theta^{a,b}-\bar{\theta}^{a,b}}{\theta^{a,b}-\theta}=0$ for all $\theta\in A^{a,b}$. Moreover, for all sufficiently small $\eta>0$, we have $\check{\theta}^{a,b}\pm\eta(\check{\theta}^{a,b}-\bar{\theta}^{a,b})\in K^{a,b}$ and $\hat{\theta}^{a,b}\pm\eta(\hat{\theta}^{a,b}-\bar{\theta}^{a,b})\in -K^{a,b}$, so it follows from~\eqref{eq:projconv} or Lemma~\ref{lem:projconv}(a) that
$\ipr{Y^{a,b}-\check{\theta}^{a,b}}{\bar{\theta}^{a,b}-\check{\theta}^{a,b}}=0$ and $\ipr{Y^{a,b}-\hat{\theta}^{a,b}}{\bar{\theta}^{a,b}-\hat{\theta}^{a,b}}=0$. Therefore, writing $Y^{a,b}=\theta^{a,b}+\xi^{a,b}$, we deduce that
\begin{align}
\norm{Y^{a,b}-\check{\theta}^{a,b}}^2-\norm{\xi^{a,b}}^2&=\norm{Y^{a,b}-\bar{\theta}^{a,b}}^2-\norm{\check{\theta}^{a,b}-\bar{\theta}^{a,b}}^2-\norm{\xi^{a,b}}^2\notag\\
&=\norm{\xi^{a,b}+(\theta^{a,b}-\bar{\theta}^{a,b})}^2-\norm{\xi^{a,b}}^2-\norm{\check{\theta}^{a,b}-\bar{\theta}^{a,b}}^2\notag\\
&=2\,\ipr{\xi^{a,b}}{\theta^{a,b}-\bar{\theta}^{a,b}}-\norm{\check{\theta}^{a,b}-\bar{\theta}^{a,b}}^2+\norm{\theta^{a,b}-\bar{\theta}^{a,b}}^2\notag\\
\label{eq:Deltanpm1ipr}
&\geq 2\,\ipr{\xi^{a,b}}{\theta^{a,b}-\bar{\theta}^{a,b}}-2\,\norm{\check{\theta}^{a,b}-\theta^{a,b}}^2-\norm{\theta^{a,b}-\bar{\theta}^{a,b}}^2,
\end{align}
where the final inequality follows since $\norm{z+z'}^2\leq 2(\norm{z}^2+\norm{z'}^2)$ for $z,z'\in\R^{b-a+1}$. Similarly,
\[\norm{Y^{a,b}-\hat{\theta}^{a,b}}^2-\norm{\xi^{a,b}}^2\geq 2\,\ipr{\xi^{a,b}}{\theta^{a,b}-\bar{\theta}^{a,b}}-2\,\norm{\hat{\theta}^{a,b}-\theta^{a,b}}^2-\norm{\theta^{a,b}-\bar{\theta}^{a,b}}^2,\]
and we now address each of the three terms on the right-hand side in turn. Firstly, letting $v^{a,b}:=(\theta^{a,b}-\bar{\theta}^{a,b})/\norm{\theta^{a,b}-\bar{\theta}^{a,b}}$, we have $\norm{v^{a,b}}=1$, so $\ipr{\xi^{a,b}}{v^{a,b}}$ is sub-Gaussian with parameter 1. Therefore, for any $a,b$ with $1\leq a\leq b\leq n$ and for every $c>0$, we have
\begin{equation}
\label{eq:maxvab}
\Pr\bigl(\abs{\ipr{\xi^{a,b}}{\theta^{a,b}-\bar{\theta}^{a,b}}}\geq\sqrt{2c\log n}\,\norm{\theta^{a,b}-\bar{\theta}^{a,b}}\bigr)\leq 2n^{-c}.
\end{equation}
Secondly, by taking $\mu=\theta^{a,b}$ and $u=\bar{\theta}^{a,b}\in A^{a,b}$ in~\citet[Theorem~4.3]{Bel18} and applying this result to the closed, convex cones $\pm K^{a,b}$, we find that
\begin{equation}
\label{eq:bellec2}
\Pr\bigl\{\norm{\check{\theta}^{a,b}-\theta^{a,b}}^2\vee\norm{\hat{\theta}^{a,b}-\theta^{a,b}}^2\geq\norm{\theta^{a,b}-\bar{\theta}^{a,b}}^2+16\log\bigl(e(b-a+1)\bigr)+4c\log n\bigr\}\leq 2n^{-c}
\end{equation}
for any $a,b$ with $1\leq a\leq b\leq n$ and for every $c>0$. Finally, we can establish the following for each $n$ (without using the local smoothness condition on $f_0$ in Assumption~\ref{ass:inflection}):
\begin{lemma}
\label{lem:step3}
For $1\leq a\leq b\leq n$, define $\breve{\theta}^{a,b}:=\bigl(\hat{f}_n^{m_0}(x_a),\hat{f}_n^{m_0}(x_{a+1}),\dotsc,\hat{f}_n^{m_0}(x_b)\bigr)$, and as in Step 2a, let $\Pi_{a,b}\,\xi^{a,b}:=\argmin_{v\in L^{a,b}}\norm{\xi^{a,b}-v}$ be the projection of $\xi^{a,b}$ onto the subspace $L^{a,b}=\Span\{\theta^{a,b},\mathbf{1}^{a,b},x^{a,b}\}$. If $x_{\tilde{k}}<x_{\tilde{k}'}$ are successive knots of $\hat{f}_n^{m_0}$, then
\[\norm{\theta^{\tilde{k},\tilde{k}'}-\breve{\theta}^{\tilde{k},\tilde{k}'}}\leq\max_{1\leq a\leq b\leq n}\,\norm{\Pi_{a,b}\,\xi^{a,b}}+2\max_{1\leq i\leq n}\,\abs{\xi_i}+2\max_{1\leq i\leq n}\,\abs{(\hat{f}_n^{m_0}-f_0)(x_i)}=:\Xi.\]
\end{lemma}

\unparskip
At the start of Step 3, $x_{\hat{\ell}},x_{\hat{\ell}_+}$ and $x_{\hat{k}_-},x_{\hat{k}}$ were defined to be pairs of successive knots of $\hat{f}_n^{m_0}$. Therefore, since $\hat{\ell}\leq\check{a}\leq\check{b}\leq\hat{\ell}_+$ and $\bigl(\breve{\theta}_i^{\hat{\ell},\hat{\ell}_+}(x_{\check{a}}),\dotsc,\breve{\theta}_i^{\hat{\ell},\hat{\ell}_+}(x_{\check{b}})\bigr)\in A_{\check{a},\check{b}}$, it follows that
\[\norm{\theta^{\check{a},\check{b}}-\bar{\theta}^{\check{a},\check{b}}}^2=\min_{\,\theta\in A_{\check{a},\check{b}}}\norm{\theta^{\check{a},\check{b}}-\theta}^2\leq\sum_{i=\check{a}}^{\check{b}}\,\bigl(\theta_i^{\hat{\ell},\hat{\ell}_+}-\breve{\theta}_i^{\hat{\ell},\hat{\ell}_+}\bigr)^2\leq\norm{\theta^{\hat{\ell},\hat{\ell}_+}-\breve{\theta}^{\hat{\ell},\hat{\ell}_+}}^2\leq\Xi^2.\]
Similarly, $\norm{\theta^{\hat{a},\hat{b}}-\bar{\theta}^{\hat{a},\hat{b}}}\leq\norm{\theta^{\hat{k}_-,\hat{k}}-\breve{\theta}^{\hat{k}_-,\hat{k}}}\leq\Xi$. Now recall the tail bound~\eqref{eq:suplinpr} for $\norm{\Pi_{a,b}\,\xi^{a,b}}$, which applies to all $1\leq a\leq b\leq n$, and our Assumption~\ref{ass:inflection} that $\xi_1,\dotsc,\xi_n$ are sub-Gaussian random variables with parameter 1. Applying a union bound, we see that
\begin{align*}
\Pr\,\Bigl(\max_{1\leq a\leq b\leq n}\,\norm{\Pi_{a,b}\,\xi^{a,b}}+2\max_{1\leq i\leq n}\,\abs{\xi_i}\geq 6\sqrt{c\log n}\Bigr)&\\
\leq\Pr\,\Bigl(\max_{1\leq a\leq b\leq n}\,\norm{\Pi_{a,b}\,\xi^{a,b}}\geq\sqrt{6c\log n}\Bigr)&+\Pr\,\Bigl(\max_{1\leq i\leq n}\,\abs{\xi_i}\geq\sqrt{2c\log n}\Bigr)\leq 6n^{2-c}+2n^{1-c}
\end{align*}
for every $c>0$. Since $\max_{x\in [0,1]}\,\abs{(\hat{f}_n^{m_0}-f_0)(x)}=o_p(1)$ by Proposition~\ref{cor:consistency}(c), it follows from Lemma~\ref{lem:step3} that
\begin{equation}
\label{eq:thetanorm}
\Pr\bigl(\norm{\theta^{\check{a},\check{b}}-\bar{\theta}^{\check{a},\check{b}}}\vee\norm{\theta^{\hat{a},\hat{b}}-\bar{\theta}^{\hat{a},\hat{b}}}\geq 6\sqrt{c\log n}+\eta\bigr)\leq\Pr\bigl(\Xi\geq 6\sqrt{c\log n}+\eta\bigr) \rightarrow 0
\end{equation}
as $n \rightarrow \infty$, for any $c>2$ and $\eta>0$. We now combine~\eqref{eq:thetanorm} with~\eqref{eq:maxvab} and~\eqref{eq:bellec2}, where we take $c=3\;(>2)$ and apply a union bound to handle all pairs $(a,b)$ with $1\leq a\leq b\leq n$. Together with~\eqref{eq:Deltanpm1} and~\eqref{eq:Deltanpm1ipr}, these imply that there exist a universal constant $\rho'>0$ and events $(E_{n4})$ with $\Pr(E_{n4}^c)\to 0$ such that
\begin{align}
\tilde{\Delta}_{n,-1}\geq\norm{Y^{\check{a},\check{b}}-\check{\theta}^{\check{a},\check{b}}}^2-\norm{\xi^{\check{a},\check{b}}}^2&\geq 2\,\ipr{\xi^{\check{a},\check{b}}}{\theta^{\check{a},\check{b}}-\bar{\theta}^{\check{a},\check{b}}}-2\,\norm{\check{\theta}^{\check{a},\check{b}}-\theta^{\check{a},\check{b}}}^2-\norm{\theta^{\check{a},\check{b}}-\bar{\theta}^{\check{a},\check{b}}}^2\notag\\
&\geq -2\sqrt{2\log n}\,\norm{\theta^{\check{a},\check{b}}-\bar{\theta}^{\check{a},\check{b}}}-56\log(en)-3\,\norm{\theta^{\check{a},\check{b}}-\bar{\theta}^{\check{a},\check{b}}}^2\notag\\
\label{eq:step3a}
&\geq -\rho'\log(en)
\end{align}
and
\begin{align}
\tilde{\Delta}_{n,1}\geq\norm{Y^{\hat{a},\hat{b}}-\hat{\theta}^{\hat{a},\hat{b}}}^2-\norm{\xi^{\hat{a},\hat{b}}}^2&\geq 2\,\ipr{\xi^{\hat{a},\hat{b}}}{\theta^{\hat{a},\hat{b}}-\bar{\theta}^{\hat{a},\hat{b}}}-2\,\norm{\hat{\theta}^{\hat{a},\hat{b}}-\theta^{\hat{a},\hat{b}}}^2-\norm{\theta^{\hat{a},\hat{b}}-\bar{\theta}^{\hat{a},\hat{b}}}^2\notag\\
&\geq -2\sqrt{2\log n}\,\norm{\theta^{\hat{a},\hat{b}}-\bar{\theta}^{\hat{a},\hat{b}}}-56\log(en)-3\,\norm{\theta^{\hat{a},\hat{b}}-\bar{\theta}^{\hat{a},\hat{b}}}^2\notag\\
\label{eq:step3b}
&\geq -\rho'\log(en)
\end{align}
on $E_n^+\cap E_{n1}\cap E_{n4}$, for each $n$.

Having carried out Steps 1--3 above, we finally define the events $\Omega_n:=\bigcap_{\,j=1}^{\,4} E_{nj}$ for $n\in\N$, which satisfy $\Pr(\Omega_n^c)\to 0$. We conclude from~\eqref{eq:step1},~\eqref{eq:step2a},~\eqref{eq:step2b},~\eqref{eq:step3a} and~\eqref{eq:step3b} that
\begin{align}
\label{eq:Deltan}
\Delta_n=S_n(\tilde{f}_n)-S_n(\hat{f}_n^{m_0})&\geq\sum_{A=-1}^1\tilde{\Delta}_{n,A}=\Lambda_{n2}+\Lambda_{n1}+\Lambda_{n3}+\tilde{\Delta}_{n,-1}+\tilde{\Delta}_{n,1}\\
&\geq 2^{-1}\rho_\alpha B^2\,(C_n/4)^{2\alpha+1}\log n-2\,(28\log n)-2\rho'\log (en)>0\notag
\end{align}
on $E_n^+\cap\Omega_n$, for all sufficiently large $n$. As mentioned at the start of the proof, this means that $\Pr(E_n^+)\leq\Pr(E_n^+\cap\Omega_n)+\Pr(\Omega_n^c)\to 0$, as desired.
\end{proof}

\deparskip
\begin{proof}[Proof of Proposition~\ref{prop:lamlb}]
Fix $\tau\in (0,1)$. First, we consider the case where $f_0\in\mathcal{F}^{m_0}$ satisfies Assumption~\ref{ass:inflection} for some $\alpha>1$. By suitably perturbing $f_0$, we construct for each (sufficiently large) $n$ a function $f_{\delta_n}\in\mathcal{F}(f_0,\tau/\sqrt{n})$ that has a unique inflection point at distance of order $\delta_n\asymp (\tau^2/n)^{1/(2\alpha+1)}$ from~$m_0$. The local asymptotic minimax lower bound~\eqref{eq:lamlb} is then obtained by applying (the proof of) Le Cam's two-point lemma to $\{f_0,f_{\delta_n}\}$. We will write $\dtv(P,Q)$ 
for the total variation distance between probability measures $P,Q$.

To this end, for each $\delta\in (0,1-m_0)$, let $u(m_0+\delta)$ be a subgradient of the concave function $\restr{f_0}{[m_0,1]}$ at $m_0+\delta$, so that $u(m_0+\delta)<f_0'(m_0)$ and \[f_0(x)\leq f_0(m_0+\delta)+u(m_0+\delta)(x-m_0-\delta)=:f_{1,\delta}(x)\quad\text{for all } x\in [m_0,1].\] 
Define
\[f_{2,\delta}(x):=f_0(m_0)+f_0'(m_0)(x-m_0)+\delta(x-m_0)^\alpha\quad\text{for }x\in [m_0,1],\]
so that $f_{2,\delta}$ is strictly convex on $[m_0,1]$ (thanks to the inclusion of the final term $\delta(x-m_0)^\alpha$) and $f_{2,\delta}(x)> f_0(m_0)+f_0'(m_0)(x-m_0)\geq f_0(x)$ for all $x\in (m_0,1]$. Note in particular that $f_{1,\delta}(m_0)>f_0(m_0)=f_{2,\delta}(m_0)$ and $f_{1,\delta}(m_0+\delta)=f_0(m_0+\delta)<f_{2,\delta}(m_0+\delta)$. Consequently, defining $f_\delta\colon [0,1]\to\R$ by
\begin{equation}
\label{eq:fdelta}
f_\delta(x):=
\begin{cases}
f_0(x)\quad&x\in [0,m_0]\cup [m_0+\delta,1]\\
f_{1,\delta}(x)\wedge f_{2,\delta}(x)&x\in (m_0,m_0+\delta),
\end{cases}
\end{equation}
we deduce that there exists a unique $c_\delta\in (0,1)$ such that $f_\delta=f_{2,\delta}$ on $[m_0,m_0+\delta c_\delta]$ and $f_\delta=f_{1,\delta}$ on $[m_0+\delta c_\delta,m_0+\delta]$. Moreover, since $f_{2,\delta}$ is strictly convex and $f_{1,\delta}(m_0+\delta c_\delta)>f_0'(m_0)>u(m_0+\delta)=f_{2,\delta}(m_0+\delta c_\delta)$, it follows that $f_\delta$ lies in $\mathcal{F}$ and has a unique inflection point at $m_\delta:=m_0+\delta c_\delta$.

Now for any sequence $(\delta_n)$ with $\delta_n\to 0$ and $n\delta_n\to\infty$, it follows from Assumption~\ref{ass:inflection} and some elementary analytic arguments that the following holds as $n\to\infty$; see Section~\ref{subsec:inflectproofs}.
\begin{lemma}
\label{lem:lamlb}
For $\alpha>1$, we have $\norm{f_{\delta_n}-f_0}_n^2=\bigl(1+o(1)\bigr)\int_0^1\,(f_{\delta_n}-f_0)^2=\bigl(1+o(1)\bigr)C_\alpha B^2\delta_n^{2\alpha+1}$ and $c_{\delta_n}=\bigl(1+o(1)\bigr)(1-\alpha^{-1})$ as $n\to\infty$, where $C_\alpha:=\int_0^1\,\bigl\{t^\alpha-\bigl(1-(1-t)\alpha\bigr)^+\bigr\}^2\,dt>0$. 
\end{lemma}

\unparskip
Thus, setting $\delta_n:=\bigl(1+o(1)\bigr)(2C_\alpha B^2 n\tau^{-2})^{-1/(2\alpha+1)}$, we deduce that $f_{\delta_n}\in\mathcal{F}(f_0,\tau/\sqrt{n})$ for all sufficiently large $n$. For all such $n$, write $P_{0,n}^Y,P_{1,n}^Y$ for the distributions of $(Y_{n1},\dotsc,Y_{nn})$ under the data generating mechanisms $Y_{ni}=f_0(x_{ni})+\xi_{ni}$ and $Y_{ni}=f_{\delta_n}(x_{ni})+\xi_{ni}$ respectively. Since $\xi_{n1},\dotsc,\xi_{nn}\iid N(0,1)$ by assumption, we have by Pinsker's inequality that $\dtv^2(P_{0,n}^Y,P_{1,n}^Y)\leq\KL(P_{0,n}^Y,P_{1,n}^Y)/2=n\norm{f_{\delta_n}-f_0}_n^2/2$, so for all sufficiently large $n$, the minimax risk can be bounded from below using Le Cam's two point lemma:
\begin{align*}
\inf_{\breve{m}_n}\,\sup_{f\in\mathcal{F}(f_0,\tau/\sqrt{n})}\E_f\bigl(d(\breve{m}_n,\mathcal{I}_f)\bigr)&\geq\frac{1}{2}\,\inf_{\breve{m}_n}\,\{\E_{f_0}(\abs{\breve{m}_n-m_0})+\E_{f_{\delta_n}}(\abs{\breve{m}_n-m_{\delta_n}})\}\\
&\geq\frac{\abs{m_{\delta_n}-m_0}}{2}\,\bigl(1-\dtv(P_{0,n}^Y,P_{1,n}^Y)\bigr)\\
&\geq\bigl(1+o(1)\bigr)\,\frac{\alpha-1}{2\alpha}\biggl(\frac{2C_\alpha B^2n}{\tau^2}\biggr)^{-1/(2\alpha+1)}\biggl(1-\frac{\tau}{2^{1/2}}\biggr).
\end{align*}
This yields~\eqref{eq:lamlb}, as desired.

In the case $\alpha\in (0,1)$, we instead define $f_\delta\colon [0,1]\to\R$ for $\delta\in (0,1-m_0)$ by
\[f_\delta(x)=
\begin{cases}
f_0(x)\wedge \bigl\{f_0(m_0)+(1-\delta)\frac{f_0(m_0+\delta)-f_0(m_0)}{\delta}(x-m_0)\bigr\}\quad&\text{for }x\in [0,m_0]\\
f_0(m_0)+\bigl(f_0(m_0+\delta)-f_0(m_0)\bigr)\bigl\{(1-\delta)\frac{x-m_0}{\delta}+\delta\bigl(\frac{x-m_0}{\delta}\bigr)^2\bigr\}&\text{for }x\in (m_0,m_0+\delta)\\
f_0(x)\quad&\text{for }x\in [m_0+\delta,1],
\end{cases}
\]
so that $f_\delta\in\mathcal{F}$ and $m_0+\delta$ is the unique inflection point of $f_\delta$ (since $x\mapsto (x-m_0)^2$ is strictly convex). Then based on similar (and slightly simpler) calculations to those for Lemma~\ref{lem:lamlb}, we can apply Le Cam's two point lemma as above to obtain the conclusion of Proposition~\ref{prop:lamlb} when $\alpha\in (0,1)$.
\end{proof}

\umparskip
\section{Projections onto classes of S-shaped functions}
\label{sec:proj2}
The purpose of this section is to introduce the general projection framework that underpins our estimation methodology, and to study the continuity properties of this projection.  This allows us to deduce not only the consistency guarantees for our estimator, as stated in Proposition~\ref{cor:consistency}, but also to ensure its robustness to model misspecification; see Proposition~\ref{prop:consistency} below.

For a finite Borel measure $\nu$ on $[0,1]$, we say that $x\in\supp\nu$ is an \emph{isolated point} of $\supp\nu$ if there exists an open neighbourhood $U$ of $x$ such that $U\cap\supp\nu=\{x\}$. Denote by $\csupp\nu:=\conv(\supp\nu)$ the \emph{convex support} of $\nu$, which is the smallest closed, convex set $C$ with $\nu(C^c)=0$. 
For Lebesgue measurable functions $f,g\colon[0,1]\to\R\cup\{\pm\infty\}$, we write $f\sim_\nu g$ if $f=g$ $\nu$-almost everywhere, and noting that $\sim_\nu$ defines an equivalence relation on the set of such measurable functions, we denote by $[f]_\nu$ the $\sim_\nu$ equivalence class of $f$.

For $q\in [1,\infty)$, we write $L^q(\nu)\equiv L^q([0,1],\nu)$ for the space of Lebesgue measurable functions $f\colon [0,1]\to\R\cup\{\pm\infty\}$ such that $\norm{f}_{L^q(\nu)}:=\bigl(\int_{[0,1]}\,\abs{f}^q\,d\nu\bigr)^{1/q}<\infty$, and define $\mathcal{L}^q(\nu)\equiv\mathcal{L}^q([0,1],\nu):=\{[f]_\nu:f\in L^q(\nu)\}$. When $q=2$, recall that the bilinear form $\ipr{\cdot\,}{\cdot}_{L^2(\nu)}$ on $L^2(\nu)$ defined by $\ipr{f}{g}_{L^2(\nu)}:=\int_{[0,1]} fg\,d\nu$ induces a Hilbert space structure on $\mathcal{L}^2(\nu)$.

For a Borel set $A\subseteq [0,1]$ and a Lebesgue measurable function $f\colon [0,1]\to\R$, let $\norm{f}_{L^\infty(A,\nu)}:=\inf\{B\geq 0:\abs{f(x)}\leq B\text{ for }\nu\text{-almost every }x\in A\}$, where we adopt the convention that $\inf\emptyset=\infty$. A function $f\in [0,1]\to\R$ is said to be \emph{locally bounded} at $x\in [0,1]$ if there exists $\varepsilon>0$ such that $f$ is bounded on $(x-\varepsilon,x+\varepsilon)\cap [0,1]$.

The following proposition provides some basic structural properties of the classes $\mathcal{F}^m$. See Section~\ref{sec:projproofs} in the supplementary material for the proofs of all results in this subsection.
\begin{proposition}
\label{prop:clconv}
If $m\in [0,1]$ and $\nu$ is a Borel probability measure on $[0,1]$, then $\mathcal{F}_\nu^m:=\{[f]_\nu:f\in\mathcal{F}^m\}$ is a convex cone in $\mathcal{L}^2(\nu)$. Moreover, the following hold for all $m\in [0,1]$:

\unparskip
\begin{enumerate}[label=(\alph*)]
\item $\{[f]_\nu:f\in\mathcal{F}^m\text{ is Lipschitz}\}$ is dense in $\mathcal{F}_\nu^m$ (with respect to the topology induced by $\norm{{\cdot}}_{L^2(\nu)}$).
\item Let $\tilde{m}:=\argmin_{x\in\csupp\nu}\,\abs{x-m}$. Then $\mathcal{F}_\nu^m$ is a dense subset of $\mathcal{F}_\nu^{\tilde{m}}$.
\item $\mathcal{F}_\nu^m$ is closed in $\mathcal{L}^2(\nu)$ if and only if at least one of the following conditions is satisfied: 

\unparskip
\begin{enumerate}[label=(\roman*)]
\item $\nu([0,m])>0$ and $\nu([m,1])>0$;
\item $\max(\supp\nu)<m$ and $\max(\supp\nu)$ is an isolated point of $\supp\nu$;
\item $\min(\supp\nu)>m$ and $\min(\supp\nu)$ is an isolated point of $\supp\nu$.
\end{enumerate}

\unparskip
\item Suppose that none of the conditions (i)--(iii) hold, and let $E_\nu$ be the interval containing all $x\in (\min(\supp\nu),\max(\supp\nu))$
as well as those $x\in\{\min(\supp\nu),\max(\supp\nu)\}$ for which $\nu(\{x\})>0$. Denote by $\Cl\mathcal{F}_\nu^m$ the closure of $\mathcal{F}_\nu^m$ in $\mathcal{L}^2(\nu)$. Then
\[\Cl\mathcal{F}_\nu^m=
\begin{cases}
\,\{[f]_{\nu}:f\in L^2(\nu)\text{ and}\restr{f}{E_\nu}\!\text{ is convex and increasing}\}\;&\text{if }\max(\supp\nu)\leq m\\
\,\{[f]_{\nu}:f\in L^2(\nu)\text{ and}\restr{f}{E_\nu}\!\text{ is concave and increasing}\}\;&\text{if }\min(\supp\nu)\geq m.
\end{cases}
\]
\end{enumerate}

\unparskip
\end{proposition}

\unparskip
For example, if $\nu$ is Lebesgue measure on $[0,1]$, then $\mathcal{F}_\nu^m$ is a closed subset of $\mathcal{L}^2(\nu)$ if and only if $m\in (0,1)$. 

Let $\mathcal{P}$ be the class of probability distributions $P$ on $[0,1]\times\R$ such that $\int_{[0,1]\times\R}\,y^2\,dP(x,y)<\infty$. For $P\in\mathcal{P}$, denote by $P^X$ the marginal distribution on $[0,1]$ induced by the coordinate projection $(x,y)\mapsto x$, and for $f\in L^2(P^X)$, define
\begin{equation}
\label{eq:L2P}
L(f,P):=\int_{[0,1]\times\R}\,\bigl(y-f(x)\bigr)^2\,dP(x,y).
\end{equation}
Introducing $(X,Y)\sim P$, we say that $f_P\colon [0,1]\to\R$ is a \emph{regression function for $P$} if $f_P(X)$ is a version of $\E(Y|\,X)$. Then $f_P\in L^2(P^X)$ and 
\begin{align}
L(f,P)=\E\bigl(\{Y-f(X)\}^2\bigr)&=\E\bigl(\{Y-f_P(X)\}^2\bigr)+\E\bigl(\{f_P(X)-f(X)\}^2\bigr)\notag\\
\label{eq:orth}
&=\E\bigl(\{Y-f_P(X)\}^2\bigr)+\norm{f_P-f}_{L^2(P^X)}^2
\end{align}
for all $f\in L^2(P^X)$. Note that by~\eqref{eq:orth}, we have $L(f_n,P)\to L(f,P)$ whenever $\norm{f_n-f}_{L^2(P^X)}\to 0$. Thus, for each $m\in [0,1]$, it follows from Proposition~\ref{prop:clconv}(a) that $L_m^*(P):=\inf_{f\in\mathcal{F}^m}L(f,P)$ is the infimum of $f\mapsto L(f,P)$ over all Lipschitz $f\in\mathcal{F}^m$. 

For $\delta\geq 0$, let $\psi_m^\delta(P):=\{f\in\mathcal{F}^m:L(f,P)\leq L_m^*(P)+\delta\}$, which is a non-empty set when $\delta>0$. In Corollary~\ref{cor:proj}(d) below, we give sufficient conditions for $\psi_m^0(P)$ to be non-empty, i.e.\ for $f\mapsto L(f,P)$ to attain its infimum $L_m^*(P)$ over $\mathcal{F}^m$.

Recall that if $E$ is a closed, convex subset of a Hilbert space $(H,\norm{{\cdot}})$, then for each $x\in H$, there is a unique $y\in E$ such that $\norm{x-y}=\min_{w\in E}\norm{x-w}$, namely the \emph{projection} of $x$ onto $E$~\citep[e.g.][Theorem~4.10]{Rud87}. In view of this and Proposition~\ref{prop:clconv}, we can now define projection maps $\psi_m^*\colon\mathcal{P}\to\mathcal{L}^2(P^X)$ associated with the convex function classes $\mathcal{F}^m$.
\begin{corollary}
\label{cor:proj}
Fix $P\in\mathcal{P}$ and denote by $P^X$ the corresponding marginal distribution on $[0,1]$. For each $m\in [0,1]$, let $\psi_m^*(P)$ be the collection of all $f\colon [0,1]\to\R$ such that $[f]_{P^X}\in\Cl\mathcal{F}_{P^X}^m\subseteq\mathcal{L}^2(P^X)$ and $L(f,P)=L_m^*(P)$. Then the following hold for all $m\in [0,1]$:

\unparskip
\begin{enumerate}[label=(\alph*)]
\item $\psi_m^*(P)$ is a non-empty $\sim_{P^X}$ equivalence class containing $\psi_m^0(P)$.
\item $\psi_m^*(P)\in\mathcal{L}^2(P^X)$, and if $\psi_m^0(P)\neq\emptyset$, then $\psi_m^*(P)\in\mathcal{L}^q(P^X)$ for all $q\in [1,\infty)$.
\item Defining $\tilde{m}=\argmin_{x\in\csupp P^X}\,\abs{x-m}$ as in Proposition~\ref{prop:clconv}(b), we have $\psi_m^*(P)=\psi_{\tilde{m}}^*(P)$ 
and $L_m^*(P)=L_{\tilde{m}}^*(P)$. If $\psi_m^0(P)$ is non-empty, then so is $\psi_{\tilde{m}}^0(P)$.
\item $\psi_m^0(P)=\psi_m^*(P)\cap\mathcal{F}^m$ is non-empty if at least one of the following holds: 

\unparskip
\begin{enumerate}[label=(\roman*)]
\item $\mathcal{F}_{P^X}^m$ is a closed subset of $\mathcal{L}^2(P^X)$, i.e.\ if $m,P^X$ satisfy at least one of the conditions (i)--(iii) in Proposition~\ref{prop:clconv}(c);
\item $m=\tilde{m}$ and $P$ has a regression function $f_P$ that is locally bounded at $\tilde{m}$.
\end{enumerate}

\unparskip
\item All functions in $\psi_m^0(P)$ agree on $(\supp P^X)\setminus\{m\}$. If in addition $P^X(\{m\})>0$ or all elements of $\psi_m^0(P)$ are continuous (at $m$), then they all agree on $\supp P^X$.
\item When $m\in\supp P^X$ and $P^X(\{m\})=0$, all elements of $\psi_m^0(P)$ agree on $\supp P^X$ if and only if all elements of $\psi_m^0(P)$ are continuous (at $m$).
\item Suppose that at least one element of $\psi_m^0(P)$ is continuous (at $m$), and moreover that $\supp P^X$ has non-empty intersection with both $(m-\varepsilon,m)$ and $(m,m+\varepsilon)$ for all $\varepsilon>0$. Then all elements of $\psi_m^0(P)$ are continuous (at $m$) and agree on $\supp P^X$.
\end{enumerate}
\unparskip
\end{corollary}
\begin{remark}
If $m\in\Int(\csupp P^X)$, then $m,P^X$ satisfy condition (i) in Proposition~\ref{prop:clconv}(c), so $\psi_m^0(P)\neq\emptyset$ in this case by (d) above. Moreover, in condition (ii) in part (d), we need only insist that $f_P$ is bounded on $\Int(\csupp P^X)\cap (\tilde{m}-\varepsilon,\tilde{m}+\varepsilon)$ for some $\varepsilon>0$; indeed, setting $\tilde{z}:=\argmin_{x\in\csupp P^X}\,\abs{z-x}$ for $z\in [0,1]$, we can instead work with $\tilde{f}_P\colon z\mapsto f_P(\tilde{z})$, which is another regression function for $P$ that is bounded on $(\tilde{m}-\varepsilon,\tilde{m}+\varepsilon)$. As for (e,\,f,\,g), recall that all elements of $\mathcal{F}^m$ are continuous on $[0,1]\setminus\{m\}$, so being continuous on $[0,1]$ is equivalent to being continuous at $m$ for all such functions.
\end{remark}

\unparskip
Next, we investigate the continuity of the maps $(m,P)\mapsto L_m^*(P)$ and $(m,P)\mapsto\psi_m^*(P)$ with respect to a suitable topology on $[0,1]\times\mathcal{P}$. Recall that for $q\in [1,\infty)$ and $d\in\N$, the \emph{$q$-Wasserstein distance} between probability measures $P_1,P_2$ on $\R^d$ is defined by $W_q(P_1,P_2):=\inf_{(X,Y)}\E(\norm{X-Y}^q)^{1/q}$, where the infimum is taken over all pairs of random variables $X,Y$ defined on a common probability space with $X\sim P_1$ and $Y\sim P_2$. It is a standard fact that $W_q(P_n,P)\to 0$ if and only if $P_n\cvd P$ and $\int_{\R^d}\norm{w}^q\,dP_n(w)\to\int_{\R^d}\norm{w}^q\,dP(w)$. 

In the result below, we equip $[0,1]\times\mathcal{P}$ with the product topology induced by the Euclidean metric on $[0,1]$ and the $W_2$ metric on $\mathcal{P}$. 
\begin{proposition}
\label{prop:cont}
Let $(m_n)_{n=1}^\infty$ be a sequence in $[0,1]$ that converges to some $m_0\in [0,1]$. Fix $P\in\mathcal{P}$ and the corresponding marginal distribution $P^X$ on $[0,1]$. Let $\tilde{m}_0:=\argmin_{x\in\csupp P^X}\,\abs{x-m_0}$. Let $(P_n)_{n=1}^\infty$ be any sequence of probability measures in $\mathcal{P}$ such that $W_2(P_n,P)\to 0$. Then

\unparskip
\begin{enumerate}[label=(\alph*)]
\item $\limsup_{n\to\infty}L_{m_n}^*(P_n)\leq L_{m_0}^*(P)$;
\item $\liminf_{n\to\infty}L_{m_n}^*(P_n)\geq L_{m_0}^*(P)$ provided that $P^X(\{\tilde{m}_0\})=0$;
\item $\lim_{n\to\infty}L_{m_n}^*(P)=L_{m_0}^*(P)$.
\end{enumerate}

\unparskip
Thus, for all $Q\in\mathcal{P}$, the map $m\mapsto L_m^*(Q)$ is continuous on $[0,1]$ and $L^*(Q):=\min_{m\in [0,1]}L_m^*(Q)$ is well-defined. Moreover,

\unparskip
\begin{enumerate}[label=(\alph*),resume]
\item $\sup_{m\in [0,1]}\,\abs{L_m^*(P_n)-L_m^*(P)}\to 0$ and $L^*(P_n)\to L^*(P)$ as $n\to\infty$ if $P^X(\{m\})=0$ for all $m\in [0,1]$.
\end{enumerate}

\unparskip
For $\delta\geq 0$ and $Q\in\mathcal{P}$, define
$\mathcal{I}^{\delta}(Q):=\{m\in [0,1]:L_m^*(Q)\leq L^*(Q)+\delta\}$ and $\mathcal{I}^*(Q):=\mathcal{I}^0(Q)=\argmin_{m\in [0,1]}L_m^*(Q)$. Let $(\delta_n)$ be any deterministic, non-negative sequence such that $\delta_n\to 0$. Then

\unparskip	
\begin{enumerate}[label=(\alph*),resume]
\item $\sup_{m_n'\in \mathcal{I}^{\delta_n}(P_n)}\inf_{m^*\in \mathcal{I}^*(P)}\,\abs{m_n'-m^*}\to 0$ as $n\to\infty$ provided that $P^X(\{m\})=0$ for all $m\in [0,1]$.
\end{enumerate}

\unparskip	
If $P^X(\{\tilde{m}_0\})=0$, then the following hold as $n\to\infty$:

\unparskip	
\begin{enumerate}[label=(\alph*),resume]
\item $\sup_{f_n\in\psi_{m_n}^{\delta_n}(P_n)}\sup_{f^*\in\psi_{m_0}^*(P)}\norm{f_n-f^*}_{L^\infty(A,P^X)}\to 0$ for all closed sets $A\subseteq (\supp P^X)\setminus\{\tilde{m}_0\}$;
\item If $\psi_{\tilde{m}_0}^0(P)\neq\emptyset$, then $\sup_{f_n\in\psi_{m_n}^{\delta_n}(P_n)}\sup_{f^*\in\psi_{\tilde{m}_0}^0(P)}\sup_{x\in A}\,\abs{f_n(x)-f^*(x)}\to 0$ for all closed sets $A\subseteq (\supp P^X)\setminus\{\tilde{m}_0\}$.
\end{enumerate}

\unparskip	
Suppose further that $m_0\in\Int(\csupp P^X)$ and $P^X(\{m_0\})=0$. Then $\psi_{m_0}^0(P)\neq\emptyset$ and the following hold as $n\to\infty$:

\unparskip	
\begin{enumerate}[label=(\alph*),resume]
\item $\sup_{f_n\in\psi_{m_n}^{\delta_n}(P_n)}\sup_{f^*\in\psi_{m_0}^*(P)}\norm{f_n-f^*}_{L^q(P^X)}\to 0$ for all $q\in [1,\infty)$;
\item $\sup_{f_n\in\psi_{m_n}^{\delta_n}(P_n)}\sup_{f^*\in\psi_{m_0}^0(P)}\sup_{x\in \supp P^X}\abs{(f_n-f^*)(x)}\to 0$ provided that all elements of $\psi_{m_0}^0(P)$ agree on $\supp P^X$.
\end{enumerate}

\unparskip
\end{proposition}
\begin{remark}
Since $[0,1]$ is compact, the conclusion of (e) is equivalent to the following: for any sequence $(m_n')$ with $m_n'\in\mathcal{I}^{\delta_n}(P_n)$ for all $n$, every subsequence of $(m_n')$ has a further subsequence that converges to an element of $\mathcal{I}^*(P)$.
When $m_0\notin\supp P^X$, the conclusion of (i) is implied by (g). If instead $m_0\in\supp P^X$, then by Corollary~\ref{cor:proj}(f), the condition in (i) is satisfied if and only if all elements of $\psi_{m_0}^0(P)$ are continuous, and Corollary~\ref{cor:proj}(g) provides a sufficient criterion for this.
\end{remark}

\unparskip
Recall that $\mathcal{F}=\bigcup_{m\in [0,1]}\mathcal{F}^m$ denotes the set of S-shaped functions on $[0,1]$. For $P\in\mathcal{P}$ and $\delta\geq 0$, define $\psi^\delta(P):=\{f\in\mathcal{F}:L(f,P)\leq L^*(P)+\delta\}$, which is non-empty when $\delta>0$, and note that $\psi^\delta(P)\subseteq\bigcup_{m\in\mathcal{I}^\delta(P)}\psi_m^\delta(P)$. 
Also, let $\psi^*(P):=\bigcup_{m\in\mathcal{I}^*(P)}\psi_m^*(P)\supseteq\bigcup_{m\in\mathcal{I}^*(P)}\psi_m^0(P)=\psi^0(P)$. 
\begin{corollary}
\label{cor:cont}
Fix $P\in\mathcal{P}$ and the corresponding marginal distribution $P^X$ on $[0,1]$. Then

\unparskip
\begin{enumerate}[label=(\alph*)]
\item $\psi^0(P)\neq\emptyset$ if and only if $\psi_m^0(P)\neq\emptyset$ for some $m\in\mathcal{I}^*(P)\cap\csupp P^X$, which is guaranteed if at least one of the following holds:

\unparskip
\begin{enumerate}[label=(\roman*)]
\item $\mathcal{I}^*(P)\cap\Int(\csupp P^X)\neq\emptyset$;
\item $P$ has a regression function $f_P\colon [0,1]\to\R$ that is locally bounded at each $m\in\mathcal{I}^*(P)\cap\{\min(\csupp P^X),\max(\csupp P^X)\}$ for which $P^X(\{m\})=0$. 
\end{enumerate}

\unparskip
If (ii) holds, then $\psi_m^0(P)\neq\emptyset$ for all $m\in\mathcal{I}^*(P)\cap\csupp P^X$.
\end{enumerate}

\unparskip
Suppose that $P^X(\{m\})=0$ for all $m\in [0,1]$. Let $(P_n)_{n=1}^\infty$ be a sequence in $\mathcal{P}$ with $W_2(P_n,P)\to 0$ and let $(\delta_n)$ be any deterministic, non-negative sequence such that $\delta_n\to 0$. Then the following hold as $n\to\infty$:

\unparskip
\begin{enumerate}[label=(\alph*),resume]
\item $\sup_{f_n\in\psi^{\delta_n}(P_n)}\inf_{f^*\in\psi^*(P)}\norm{f_n-f^*}_{L^\infty(A,P^X)}\to 0$ for all closed sets $A\subseteq (\supp P^X)\setminus\mathcal{I}^*(P)$.
\item Assume that $\psi_m^0(P)\neq\emptyset$ for all $m\in\mathcal{I}^*(P)\cap\csupp P^X$ and let $\tilde{\mathcal{I}}^*(P)$ be the set of $m^*\in\mathcal{I}^*(P)$ such that either $m^*\notin\Int(\csupp P^X)$ or not all elements of $\psi_{m^*}^0(P)$ are continuous. Then $\sup_{f_n\in\psi^{\delta_n}(P_n)}\inf_{f^*\in\psi^0(P)}\sup_{x\in A}\,\abs{(f_n-f^*)(x)}\to 0$ for all closed sets $A\subseteq (\supp P^X)\setminus\tilde{\mathcal{I}}^*(P)$.
\item If $\mathcal{I}^*(P)\subseteq\Int(\csupp P^X)$, then $\sup_{f_n\in\psi^{\delta_n}(P_n)}\inf_{f^*\in\psi^0(P)}\norm{f_n-f^*}_{L^q(P^X)}\to 0$ for all $q\in [1,\infty)$.
\end{enumerate}

\unparskip
\end{corollary}
\begin{remark}
In assertions (b)--(d), we see that $\psi^*(P)$ or $\psi^0(P)$ can be regarded as a `limiting set' $\mathcal{M}$ to which the sets $\mathcal{M}_n=\psi^{\delta_n}(P_n)$ converge, in the sense that $\sup_{f_n\in\mathcal{M}_n}\inf_{f\in\mathcal{M}}\,\rho(f_n,f)\to 0$ for each of three different pseudometrics $\rho$. In the proof, we establish a slightly stronger conclusion for each $\rho$: for any sequence $(f_n)$ with $f_n\in\mathcal{M}_n$ for all $n$, every subsequence of $(f_n)$ has a further subsequence that converges to an element of $\mathcal{M}$ with respect to $\rho$. 
Note that unlike in Proposition~\ref{prop:cont}(f)--(i), we take an infimum rather than a supremum over $\mathcal{M}$ in the convergence statements above. This is because we do not in general have $\rho(f,g)=0$ for all $f,g\in\mathcal{M}$ (in contrast to the sets $\psi_m^*(P)$ for each fixed $m\in [0,1]$). In Section~\ref{sec:contproof}, we also demonstrate through Examples~\ref{ex:1} and~\ref{ex:2} that if certain technical conditions in Proposition~\ref{prop:cont} and Corollary~\ref{cor:cont} are dropped, then some of the conclusions fail to hold in general.
\end{remark}

\unparskip
For a regression function $f_0\colon [0,1]\to\R$ (that need not be S-shaped) and a sequence of models~\eqref{eq:Model} indexed by $n\in\N$, we can now establish asymptotic convergence
results for S-shaped LSEs and their inflection points, including under model misspecification. To this end, we apply the continuity results from the general projection theory above (specifically Proposition~\ref{prop:cont} and Corollary~\ref{cor:cont}) to the empirical distributions $\Pr_n:=n^{-1}\sum_{i=1}^n\delta_{(x_{ni},Y_{ni})}$. Recall that we write $\Pr_n^X:=n^{-1}\sum_{i=1}^n\delta_{x_{ni}}$.
\begin{proposition}
\label{prop:consistency}
Suppose that the following conditions hold:

\unparskip
\begin{enumerate}[label=(\roman*)]
\item $(\Pr_n^X)$ converges weakly to a distribution $P_0^X$ on $[0,1]$ satisfying $P_0^X(\{m\})=0$ for all $m\in [0,1]$; 
\item For some distribution $P_\xi$ with mean 0 and finite variance, we have $\xi_{n1},\dotsc,\xi_{nn}\iid P_\xi$ for each $n$;
\item $f_0$ is bounded on $[0,1]$ and continuous $P_0^X$-almost everywhere (i.e.\ the set of discontinuities of $f_0$ has $P_0^X$ measure 0).
\end{enumerate}

\unparskip
Let $P_0\in\mathcal{P}$ be the distribution of $(X,f_0(X)+\xi)$, where $X\sim P_0^X$ and $\xi\sim P_\xi$ are independent, 
and define $L_m^*(P_0)$ for $m\in [0,1]$ and $L^*(P_0),\psi^0(P_0)$, $\mathcal{I}^*(P_0)$, $\tilde{\mathcal{I}}^*(P_0)$ as in Proposition~\ref{prop:cont} and Corollary~\ref{cor:cont}. Then $\psi^0(P_0)\neq\emptyset$, and

\unparskip
\begin{enumerate}[label=(\alph*)]
\item $\sup_{m\in[0,1]}\,\abs{L_m^*(\Pr_n)-L_m^*(P_0)}\cvp 0$ 
and $L^*(\Pr_n)-L^*(P_0)\cvp 0$ as $n\to\infty$, where $L_m^*(\Pr_n)=S_n(\hat{f}_n^m)/n$ and $L^*(\Pr_n)=\min_{1\leq j\leq n}S_n(\hat{f}_n^{x_{nj}})/n$ for $m\in [0,1]$ and $n\in\N$.
\end{enumerate}

\unparskip
For each $n$, fix an LSE $\tilde{f}_n$ over $\mathcal{F}$, so that $\tilde{f}_n\in\psi^0(\Pr_n)$, and let $\tilde{m}_n$ be any inflection point of $\tilde{f}_n$. Then the following hold as $n\to\infty$:

\unparskip
\begin{enumerate}[label=(\alph*),resume]
\item $\inf_{m^*\in\mathcal{I}^*(P_0)}\,\abs{\tilde{m}_n-m^*}\cvp 0$;
\item $\inf_{f^*\in\psi^0(P_0)}\sup_{x\in A}\,\abs{(\tilde{f}_n-f^*)(x)}\overset{p\ast}{\to}0$ for any closed set $A\subseteq\supp P_0^X\setminus\tilde{\mathcal{I}}^*(P_0)$;
\item $\inf_{f^*\in\psi^0(P_0)}\,\norm{\tilde{f}_n-f^*}_{L^q(P_0^X)}\overset{p\ast}{\to}0$ for all $q\in [1,\infty)$, provided that $\mathcal{I}^*(P_0)\subseteq\Int(\csupp P_0^X)$.
\end{enumerate}

\unparskip
\end{proposition}

\unparskip
Using the full strength of Proposition~\ref{prop:cont} and Corollary~\ref{cor:cont}, we see that for a sequence of non-negative tolerances $\delta_n\to 0$, the conclusions above extend to
sequences $(\tilde{f}_n)$ where each $\tilde{f}_n$ takes values in $\psi^{\delta_n}(\Pr_n)$, the set of approximate $\delta_n$-minimisers of $f\mapsto S_n(f)$ over $\mathcal{F}$.

In the correctly specified setting where $f_0\in\mathcal{F}^{m_0}$ for some unique $m_0\in [0,1]$, Proposition~\ref{prop:consistency} specialises to the consistency result stated as Proposition~\ref{cor:consistency} in the main text.

\hfparskip
\section{Proofs for Section~\ref{sec:appendix}}
\label{sec:compproofs}
For reference, we state a result from convex analysis that generalises Lemmas~\ref{lem:cone} and~\ref{lem:projactive}.
\begin{lemma}
\label{lem:projconv}
Let $\Lambda\subseteq\R^n$ be a non-empty closed, convex set. For each $y\in\R^n$, there exists a unique projection of $y$ onto $\Lambda$, given by $\Pi_\Lambda(y)=\argmin_{u\in\Lambda}\norm{u-y}$, and we have the following:

\unparskip
\begin{enumerate}[label=(\alph*)]
\item $\Pi_\Lambda(y)$ is the unique $\hat{y}\in\Lambda$ for which $\ipr{u-\hat{y}}{y-\hat{y}}\leq 0$ for all $u\in\Lambda$.
\item For each $u\in\Lambda$, we have $\Pi_\Lambda^{-1}(\{u\})=u+N_\Lambda(u)$, where $N_\Lambda(u):=\{v\in\R^n\setminus\{0\}:\ipr{v}{\tilde{u}}\leq\ipr{v}{u}\text{ for all }\tilde{u}\in\Lambda\}\cup\{0\}$ is the \emph{normal cone} of $\Lambda$ at $u$.
\end{enumerate}

\unparskip
Furthermore, each element of $\Lambda$ is contained in the relative interior of a unique face of $\Lambda$. For each face $F\subseteq\Lambda$, we have the following:

\unparskip
\begin{enumerate}[label=(\alph*),resume]
\item There is a closed convex cone $N_\Lambda(F)$ such that $N_\Lambda(u)=N_\Lambda(F)$ for all $u\in\relint F$, and $\Pi_\Lambda^{-1}(\relint F)=(\relint F)+N_\Lambda(F)$. If $u\in\relint F$ and $v\in N_\Lambda(F)$, then $\Pi_\Lambda(u+v)=u$.
\item For all $y\in\Pi_\Lambda^{-1}(\relint F)$, we have $\Pi_\Lambda(y)=\Pi_{\aff(F)}(y)$, where $\aff(F)$ denotes the affine hull of $F$, i.e.\ the smallest affine subspace containing $F$.
\item If in addition $\Lambda$ is a finitely generated cone, then $F$ and $N_\Lambda(F)$ are also finitely generated cones, and $\Span(F)$ and $\Span\bigl(N_\Lambda(F)\bigr)$ are complementary orthogonal subspaces. Thus, $\Pi_\Lambda^{-1}(\relint F)=(\relint F)+N_\Lambda(F)$ is an $n$-dimensional convex cone (with non-empty interior).
\end{enumerate}

\unparskip
\end{lemma}

\unparskip
\begin{proof}
For (a,\,b) and the first assertion in (c), see~\citet[Section~1.2]{Sch14}, and (2.3) and Lemma~2.2.2 in~\citet[Section~2.2]{Sch14}. Using these, we now complete the proofs of (c,\,d,\,e).

(c) By the definition of $N_\Lambda(F)$, we have $\ipr{v}{\tilde{u}-u}\leq 0$ for all $v\in N_\Lambda(F)$, $\tilde{u}\in\relint F$ and $u\in\Lambda$, with equality when $u\in\relint F$. It now follows from (a) that $\Pi_\Lambda(u+v)=u$ for all $u\in\relint F$ and $v\in N_\Lambda(F)$.

(d) Take any $u\in\aff(F)$. Since $\Pi_\Lambda(y)\in\relint F$, we have $\Pi_\Lambda(y)+\lambda u$ for some sufficiently small $\lambda>0$, so $\ipr{\lambda u}{y-\Pi_\Lambda(y)}=0$. Thus, $\Pi_\Lambda(y)\in\aff(F)$ and $\ipr{u}{y-\Pi_\Lambda(y)}=0$ for all $u\in\aff(F)$, so indeed $\Pi_\Lambda(y)=\Pi_{\aff(F)}(y)$.

(e) For a finitely generated cone $\Lambda$, this follows from Theorem~2.4.9 and (2.25) in~\citet[Section~2.2]{Sch14}. 
\end{proof}

\unparskip
By applying Lemma~\ref{lem:cone}, we can give an alternative self-contained proof of Lemma~\ref{lem:projactive} for the cone $\Lambda^j\subseteq\R^j$ of increasing convex sequences based on $x_1,\dotsc,x_j$, whose generators $\pm u^0,u^1,\dotsc,u^{j-1}$ are specified in the paragraph below~\eqref{eq:Thetak}. For $A\subseteq [j-1]$, recall from Remark~\ref{rem:conefaces} that we write $P_A\in\R^{j\times j}$ for the matrix that represents the orthogonal projection $\Pi_{\mathcal{L}_A}$ onto $\mathcal{L}_A=\Span\{u^{\ell'}:\ell'\in A\cup\{0\}\}$, the subspace consisting of all $v\in\R^j$ whose knots lie in $A$. 

\unparskip
\begin{proof}[Proof of Lemma~\ref{lem:projactive}]
For each $\tilde{v}\in\{v',v''\}$, we have $\Pi_\Lambda(\tilde{v})\in F_A\subseteq\Img(U_A)$, and Lemma~\ref{lem:cone} implies that $\ipr{u^\ell}{\tilde{v}-\Pi_\Lambda(\tilde{v})}=0$ for all $\ell\in A\cup\{0\}$. This shows that $P_A(\tilde{v})=\Pi_\Lambda(\tilde{v})\in\relint F_A$ for $\tilde{v}\in\{v',v''\}$. Now fix $t\in [0,1]$ and $v:=(1-t)v'+tv''$. Then $P_A(v)=(1-t)P_A(v')+tP_A(v'')=(1-t)\Pi_\Lambda(v')+t\Pi_\Lambda(v'')\in\relint F_A$ by the convexity of $\relint F_A$. In addition, by applying Lemma~\ref{lem:cone} to $\Pi_\Lambda(v'),\Pi_\Lambda(v'')$, we deduce that
\[\ipr{u^\ell}{v-P_A(v)}=(1-t)\ipr{u^\ell}{v'-\Pi_\Lambda(v')}+t\ipr{u^\ell}{v''-\Pi_\Lambda(v'')}\leq 0\]
for all $0\leq\ell\leq j-1$, with equality if $\ell\in A\cup\{0\}$ (in view of the definition of $P_A$). It follows from Lemma~\ref{lem:cone} that $\Pi_\Lambda(v)=P_A(v)\in\relint F_A$, as required.
\end{proof}

\deparskip
\begin{proof}[Proof of Lemma~\ref{lem:terminate}]
For (iv), we know from (iii) that $\Pi_{\Lambda^j}\bigl(v(t)\bigr)=P_{A_r}v(t)$ for all $t\in [t_r,t_{r+1}]$ after an iteration of~(II). By applying Lemma~\ref{lem:projactive}, we deduce that there exist $\eta>0$ and $A'\subseteq [j-1]$ with $A' \neq A_r$ such that $\hat{v}(t):=\Pi_{\Lambda^j}\bigl(v(t)\bigr)=P_{A'}v(t)\in\Lambda^j\cap\mathcal{L}_{A'}=\{u\in\Lambda^j:A(u)\subseteq A'\}$ for all $t\in [t_{r+1},t_{r+1}+\eta]$. Now in (IV), note that 

\unparskip
\begin{itemize}[leftmargin=0.4cm]
\item If $\ell\in A_r\setminus A_r^-$, then $\lambda_\ell\bigl(\hat{v}(t_{r+1})\bigr)=\beta_\ell(t_{r+1})>0$, so $\ell\in A\bigl(\hat{v}(t_{r+1})\bigr)\subseteq A'$;
\item If $\ell\in A'$, then $\gamma_\ell(t_{r+1})=\ipr{u^\ell}{v(t_{r+1})-\hat{v}(t_{r+1})}=\ipr{u^\ell}{(I-P_{A'})v(t_{r+1})}=0$, so $\ell\in\{1\leq\ell'\leq j-1:\gamma_{\ell'}(t_{r+1})=0\}=A_r\cup A_r^+$.
\end{itemize}

\unparskip
Thus, $A_r\setminus A_r^-\subseteq A'\subseteq A_r\cup A_r^+$, so in all cases, (IV) is guaranteed to find subsets $A^\pm\subseteq A_r^\pm$ such that when we take $A_{r+1}=(A_r\setminus A^-)\cup A^+$, the next iteration of (II) strictly increases $t$. In particular, this always happens in scenario (a) where $\abs{A_r^-\cup A_r^+}=1$, 
since we necessarily have $A'=(A_r\setminus A_r^-)\cup A_r^+=A_{r+1}$ in this case.

For (v), recall again that for each $r\in\N_0$, we have $\Pi_{\Lambda^j}\bigl(v(t)\bigr)=P_{A_r}v(t)\equiv\hat{v}_r(t)$ for all $t\in [t_r,t_{r+1}]$. As noted in (II), $\lambda_\ell\bigl(\hat{v}_r(t)\bigr)$ and $\ipr{u^\ell}{v(t)-\hat{v}_r(t)}$ vary linearly with $t$ for all $\ell\in [j-1]$, so it follows from this and~\eqref{eq:tr+1} that
\[t_{r+1}=\sup\,\bigl\{t\geq t_r:\lambda_\ell\bigl(\hat{v}_r(t)\bigr)\geq 0,\,\ipr{u^\ell}{v(t)-\hat{v}_r(t)}\leq 0\text{ for all }1\leq\ell\leq j-1\bigr\}.\]
Thus, by Lemma~\ref{lem:cone}, we cannot have $\hat{v}_r(t)=\Pi_{\Lambda^j}\bigl(v(t)\bigr)$ for any $t>t_{r+1}$, so $A_r\neq A_{r'}$ for all $r'>r$, as claimed.
\end{proof}

\unparskip
\textbf{Degeneracies in Algorithm~\ref{alg:coneproj}}: It can be verified that there is a set $A\subseteq\R^j\times\R^j$ of Lebesgue measure 0 such that if $(v(0),v(1))\notin A$, then no degeneracies occur on the trajectory of the algorithm.  Thus, degeneracies are rarely an issue when $v(0),v(1)$ are obtained from simulated or real data rather than artificially constructed (see Example~\ref{ex:noiseless}). To avoid them in practice,~\citet[page~73]{FM89} and~\citet[page~28]{Mey99} suggest slightly perturbing $v(0),v(1)$ or some intermediate $v(t_r)$. The approach we outline in Stage (IV) covers all eventualities in the degenerate scenario (b), but this can be time-consuming when $\abs{A_r^-\cup A_r^+}$ is large. 

In the special case where $u=v(0)-v(1)$ is a positive multiple of $u^{j-1}=e_j=(0,\dotsc,0,1)\in\R^j$ (which is of particular relevance in the \texttt{SeqConReg} procedure in Section~\ref{sec:computation}), a more efficient alternative to (IV) is as follows:

\unparskip
\begin{enumerate}[label=(\Roman*')]
\setcounter{enumi}{3}
\item Instead define $A_r^-:=\{\ell\in A_r:\beta_\ell(t_{r+1})=0,\,\hat{\lambda}_\ell^{A_r}(u)>0\}$ and $A_r^+:=\{\ell\in A_r^c:\gamma_\ell(t_{r+1})=0,\,\hat{\zeta}_\ell^{A_r}(u)<0\}$, and let $\ell_{\max}:=\max(A_r^-\cup A_r^+)$. Let 
$A_{r+1}:=A_r\setminus\{\ell_{\max}\}$ if $\ell_{\max}\in A_r^-$ and otherwise let $A_{r+1}:=A_r\cup\{\ell_{\max}\}$ if $\ell_{\max}\in A_r^+$. Then execute (II) and (III) with this $A_{r+1}$ (and $r+1$ in place of $r$ throughout).
\end{enumerate}

\unparskip
If there is a degeneracy at $t_{r+1}$, then when we run this modified algorithm, there may be several subsequent iterations of (II) in which $t$ does not increase (i.e.\ we remain at $t_{r+1}$). Nevertheless, the choice of $\ell_{\max}=\max(A_r^-\cup A_r^+)$ in (IV') ensures that property (iv) still holds, and hence that the algorithm terminates with the exact solution (usually after fewer iterations than in the original). 
\begin{proposition}
\label{prop:ejactive}
Suppose that $u=v(0)-v(1)$ is a positive multiple of $u^{j-1}=e_j\in\R^j$. Then with modification (IV'), Algorithm~\ref{alg:coneproj} terminates with the correct solution after finitely many steps, and the following hold for any $r\in\N_0$:

\unparskip
\begin{enumerate}[label=(\alph*)]
\item $\max A_r\geq\max A_{r+1}$; in other words, if $\max A_r<\ell\leq j-1$, then $\ell\notin A_{r'}$ for any $r'>r$.
\item Let $\ell_r:=\max(\{\ell\in A_r:\ell+1\in A_r\}\cup\{0\})$. Then in~\eqref{eq:algconeproj}, we have $\hat{\lambda}_\ell^{A_r}(u)\equiv\lambda_\ell(P_{A_r}u)=0$ for all $0\leq\ell\leq\ell_r-1$ and $\hat{\zeta}_\ell^{A_r}(u)\equiv\ipr{u^\ell}{(I-P_{A_r})u}=0$ for all $0\leq\ell\leq\ell_r+1$.
\end{enumerate}

\unparskip
\end{proposition}

\unparskip
This follows from Lemma~\ref{lem:resipr} below, which captures some specific structural features of the generators $\pm u^0,u^1,\dotsc,u^{j-1}$ of $\Lambda^j$. The facts in (a) and (b) lead to some additional computational shortcuts in Algorithm~\ref{alg:coneproj} when $u$ is a positive multiple of $e_j$. Specifically, when computing $t_{r+1}$ in Stage (II) of the procedure, it follows from Proposition~\ref{prop:ejactive} that we need only compute the ratios in~\eqref{eq:algconeproj} for $\ell_r<\ell\leq\max A_r$. Thus, when $t\geq t_r$, we can drop all $\beta_\ell(t)$ and $\gamma_\ell(t)$ with $\ell>\max A_r$, and when updating the primal and dual variables for use in subsequent iterations, no calculations are needed to see that $\beta_\ell(t_{r+1})=\beta_\ell(t_r)$ and $\gamma_\ell(t_{r+1})=\gamma_\ell(t_r)$ for all $1\leq\ell\leq\ell_r$. 
\begin{example}
\label{ex:noiseless}
We can actually write down explicitly the sequence of `active sets' $A_0,A_1,\dotsc$ obtained by Algorithm~\ref{alg:coneproj} in the special case where $v(0)\in\Lambda^j$ and $u=v(0)-v(1)$ is a positive multiple of $e_j$. This can happen if for example in~\eqref{eq:seqconv}, the observations $Y_1,\dotsc,Y_j$ are drawn according to a \emph{noiseless} regression model~\eqref{eq:Model} in which $f_0$ is increasing and convex on $[x_1,x_{j-1}]$. With $A_0=A\bigl(v(0)\bigr)$, it turns out that for $r\in\N_0$, we have
\begin{equation}
\label{eq:Ar+1}
A_{r+1}=
\begin{cases}
A_r\setminus\{\max A_r\}&\quad\text{if }\max A_r-1\in A_r\text{ or }\max A_r=j-1\\
A_r\cup\{\max A_r-1\}&\quad\text{if }\max A_r-1\notin A_r\text{ and }\max A_r<j-1.
\end{cases}
\end{equation}
Indeed, given that $v(0)\in\Lambda^j$ and hence that $\gamma_\ell(0)=0$ for all $0\leq\ell\leq j-1$, we can apply Proposition~\ref{prop:ejactive} to establish inductively that $\gamma_\ell(t_{r+1})=0$ for all $0\leq\ell\leq\max A_r$ and $r\in\N_0$. In particular, we always have $A_r^+=\{1,\dotsc,\max A_r-1\}\cap A_r^c$, and

\unparskip
\begin{itemize}[leftmargin=0.4cm]
\item If $\max A_r-1\in A_r$ or $\max A_r=j-1$, then $t_r<t_{r+1}$ and $A_r^-=\{\max A_r\}$;
\item If $\max A_r-1\notin A_r$ and $\max A_r<j-1$, then $t_r=t_{r+1}$ and $A_r^-=\emptyset$.
\end{itemize}

\unparskip
Note that unless $\{1,\dotsc,\max A_r-2\}\subseteq A_r$, there is a degeneracy at $t_{r+1}$, so we use (IV') above to form the next `active set' $A_{r+1}$. In addition, we have $\max A_r>\max A_{r+2}$ for all $r\in\N_0$ in view of~\eqref{eq:Ar+1}, so the number of distinct `active sets' on the trajectory of Algorithm~\ref{alg:coneproj} is at most $2(j-1)$.
This is much less than $2^{j-1}$, the total number of subsets of $[j-1]$, and an open question is whether for general $v(0)\in\R^j$ (and $u=v(0)-v(1)$ as above), the number of `active sets' is necessarily bounded above by a polynomial in $j$. If this is always true (or true in `most' cases), then our sequential procedure for increasing convex regression is guaranteed to have a worst-case (or average-case) complexity that is at most polynomial in the number of observations $n$. 
\end{example}

\unparskip
For fixed $j\in [n]$, $\ell\in [j-1]$ and $A\subseteq [j-1]$, Lemma~\ref{lem:resipr} determines the signs of the entries of $(I-P_A)u^\ell\in\R^j$ indexed by $A\cup\{\ell\}$. This yields useful information on how the primal and dual variables change in Algorithm~\ref{alg:coneproj} when the vector $u=v(0)-v(1)$ therein is a positive multiple of $e_j$. This enables us to justify the more efficient implementation (IV') of Stage (IV) of this procedure, as well as assertions (a) and (b) in Proposition~\ref{prop:ejactive} on the composition of the resulting active sets.

We write $e_1,\dotsc,e_j$ for the standard basis vectors in $\R^j$ and $\ipr{\cdot\,}{\cdot}$ for the standard Euclidean inner product. For $t\in\R$, let $\sgn(t):=(\abs{t}/t)\Ind_{\{t\neq 0\}}$. 
\begin{lemma}
\label{lem:resipr}
For $A\subseteq [j-1]$, enumerate the elements of $A$ as $a_1>a_2>\cdots>a_m$, and let $a_0=j$ and $a_{m+1}=0$. Fix $\ell\in [j-1]$. Then $(I-P_A)u^\ell\in\mathcal{L}_{A\cup\{\ell\}}$, and $(I-P_A)u^\ell=0$ if and only if $\ell\in A$. 

Suppose now that $\ell\notin A$ and let $q\in [m+1]$ be such that $a_q<\ell<a_{q-1}$. Define $q_-:=\max(\{1\leq\tilde{q}\leq q-1:a_{\tilde{q}-1}=a_{\tilde{q}}+1\}\cup\{0\})$ and $q_+:=\min(\{q+1\leq\tilde{q}\leq m:a_{\tilde{q}-1}=a_{\tilde{q}}+1\}\cup\{m+1\})$. Then  $\ipr{(I-P_A)u^\ell}{e_\ell}<0$, $\ipr{P_A u^\ell}{e_1}=\ipr{P_A u^\ell}{e_{a_m}}$, and for $0\leq s\leq m$, we have 
\begin{equation}
\label{eq:resipr}
\sgn\bigl(\ipr{(I-P_A)u^\ell}{e_{a_s}}\bigr)=
\begin{cases}
(-1)^{s-q}&\quad\text{if }q\leq s<q_+\\
(-1)^{q-1-s}&\quad\text{if }q_-\leq s\leq q-1\\
0&\quad\text{if }s<q_+\text{ or }s\geq q_+.
\end{cases}
\end{equation}
\end{lemma}

\deparskip
\begin{proof}[Proof of Lemma~\ref{lem:resipr}]
The generators $u^0,u^1,\dotsc,u^{j-1}$ of the cone $\Lambda^j$ are linearly independent, so $u^\ell\equiv(u_1^\ell,\dotsc,u_j^\ell)=\bigl((x_i-x_\ell)^+:1\leq i\leq j\bigr)\in\mathcal{L}_A$ (i.e.\ $(I-P_A)u^\ell=0$) if and only if $\ell\in A$. Suppose henceforth that $\ell\notin A$ and let $\tilde{z}\equiv (\tilde{z}_1,\dotsc,\tilde{z}_j):=P_A u^\ell=\argmin_{z\in\mathcal{L}_A}\norm{u^\ell-z}$. Then $u^\ell-\tilde{z}=(I-P_A)u^\ell\in\Span(\{u^\ell\}\cup\mathcal{L}_A)=\mathcal{L}_{A\cup\{\ell\}}$, so $u^\ell-\tilde{z}$ is determined by $\{u_i^\ell-\tilde{z}_i=\ipr{(I-P_A)u^\ell}{e_i}:i\in A\cup\{\ell\}\}$. 

To establish~\eqref{eq:resipr}, we make the following additional definitions. For $z\in\R^j$ and $J=\{b,b+1,\dotsc,b'\}$ with $1\leq b\leq b'\leq j$, we write $z_J=(z_b,z_{b+1},\dotsc,z_{b'})$ for the subvector of $z$ indexed by $J$. For $s\in [m]$, partition $[j]$ into the subsets \[J_s^+:=\{1,\dotsc,a_s\},\quad J_s:=\{a_s+1,\dotsc,a_{s-1}-1\},\quad J_s^-:=\{a_{s-1},\dotsc,j\},\]
and for $i\in J_s$, let $t_i^{J_s}:=(x_i-x_{a_s})/(x_{a_{s-1}}-x_{a_s})\in (0,1)$, so that any affine sequence based on $(x_i:i\in J_s)$ can be written in the form $v^{J_s}(\lambda,\vartheta):=\bigl((1-t_i^{J_s})(u_{a_s}^\ell-\lambda)+t_i^{J_s}(u_{a_{s-1}}^\ell-\vartheta):a_s+1\leq i\leq a_{s-1}-1\bigr)$ for some $\lambda,\vartheta\in\R$. Moreover,
\[\mathcal{L}_{A,s}^+(\lambda):=\{z_{J_s^+}:z\in\mathcal{L}_A,\,z_{a_s}=\lambda\}\subseteq\R^{\abs{J_s^+}}\;\text{ and }\;
\mathcal{L}_{A,s}^-(\lambda):=\{z_{J_s^-}:z\in\mathcal{L}_A,\,z_{a_{s-1}}=\lambda\}\subseteq\R^{\abs{J_s^-}}\]
are affine subspaces for each $\lambda\in\R$, and $\mathcal{L}_{A,s}^\pm(\lambda)=\lambda\mathcal{L}_{A,s}^\pm(1)$ if $\lambda\neq 0$, so the vectors $\tilde{v}^{J_s^\pm}:=\argmin_{v\in\mathcal{L}_{A,s}^\pm(1)}\norm{v}$ and $\lambda\tilde{v}^{J_s^\pm}=\argmin_{v\in\mathcal{L}_{A,s}^\pm(\lambda)}\norm{v}$ are well-defined for all $\lambda\in\R$. 
\setcounter{claim}{0}
\begin{claim}
\label{cl:resipr1}
Let $\star\in\{+,-\}$ and $s\in [m-1]$ be such that $a_s,a_{s\star 1}\in J_q^\star$. If $\abs{a_{s}-a_{s\star 1}}=1$, then $\tilde{v}_{a_{s\star 1}}^{J_q^\star}=0$. Otherwise, if $\abs{a_{s}-a_{s\star 1}}>1$, then $\sgn\bigl(\tilde{v}_{a_{s\star 1}}^{J_q^\star}\bigr)=-\sgn\bigl(\tilde{v}_{a_s}^{J_q^\star}\bigr)$. Thus,
\[\sgn\bigl(\tilde{v}_{a_s}^{J_q^+}\bigr)=
\begin{cases}
(-1)^{s-q}&\quad\text{if }q\leq s<q_+\\
0&\quad\text{if }q_+\leq s\leq m;
\end{cases}
\quad
\sgn\bigl(\tilde{v}_{a_s}^{J_q^-}\bigr)=
\begin{cases}
(-1)^{q-1-s}&\quad\text{if }q_-\leq s\leq q-1\\
0&\quad\text{if }0\leq s<q_-.
\end{cases}
\]
\end{claim}

\deparskip
\begin{proof}[Proof of Claim~\ref{cl:resipr1}]
We focus on the case $\star=+$; the arguments for $\star=-$ are similar. Note that $a_s,a_{s+1}\in J_s^+$ precisely when $q\leq s\leq m-1$. For any such $s$ and $\mu\in\R$, define $\tilde{v}^{J_q^+}(s;\mu)\equiv\bigl(\tilde{v}_i^{J_q^+}(s,\mu):1\leq i\leq a_q\bigr)$ by
\[\tilde{v}_i^{J_q^+}(s;\mu):=
\begin{cases}
\mu\tilde{v}_i^{J_{s+1}^+}=\argmin_{v\in\mathcal{L}_{A,s+1}^+(\mu)}\norm{v}\;\:&\text{for }i\in J_{s+1}^+=\{1,\dotsc,a_{s+1}\}\\ (1-t_i^{J_{s+1}})\mu+t_i^{J_{s+1}}\tilde{v}_{a_s}^{J_q^+}=v_i^{J_{s+1}}\bigl(-\mu,-\tilde{v}_{a_s}^{J_q^+}\bigr)\;\:&\text{for }i\in J_{s+1}\\
\tilde{v}_i^{J_q^+}\;\:&\text{for }i\in J_{s+1}^-\cap J_q^+=\{a_s,\dotsc,a_q\}. 
\end{cases}
\]
Then $\tilde{v}_{a_s}^{J_q^+}(s;\mu)=\tilde{v}_{a_s}^{J_q^+}$, and since $\mu\tilde{v}^{J_{s+1}^+}\in\mathcal{L}_{A,s+1}^+(\mu)$, we have $\tilde{v}_{a_{s+1}}^{J_q^+}(s;\mu)=\mu$, so $\tilde{v}^{J_q^+}(s;\mu)\in\mathcal{L}_{A,q}^+(1)$. Observe in addition that $\tilde{v}^{J_q^+}\bigl(s;\tilde{v}_{a_{s+1}}^{J_q^+}\bigr)=\tilde{v}^{J_q^+}$; indeed, $\tilde{v}_i^{J_q^+}\bigl(s;\tilde{v}_{a_{s+1}}^{J_q^+}\bigr)=\tilde{v}_i^{J_q^+}$ for $a_{s+1}+1\leq i\leq a_q$ and 
\begin{align*}
\tilde{\mathcal{L}}_{A,q,s+1}^+&:=\bigl\{v\equiv (v_1,\dotsc,v_{a_q})\in\mathcal{L}_{A,q}^+:v_i=\tilde{v}_i^{J_q^+}\text{ for }a_{s+1}+1\leq i\leq a_q\bigr\}\\
&\phantom{:}=\bigl\{(v_1,\dotsc,v_{a_q})\in\R^{a_q}:v_{J_{s+1}^+}\in\mathcal{L}_{A,s+1}^+\bigl(\tilde{v}_{a_{s+1}}^{J_q^+}\bigr),\,v_i=\tilde{v}_i^{J_q^+}\text{ for }a_{s+1}+1\leq i\leq a_q\bigr\},
\end{align*}
so $\tilde{v}^{J_q^+}=\argmin_{v\in\mathcal{L}_{A,q}^+(1)}\norm{v}=\argmin_{v\in\tilde{\mathcal{L}}_{A,q,s+1}^+}\!\norm{v}$ satisfies \[(\tilde{v}^{J_q^+})_{J_{s+1}^+}=\argmin_{v'\in\mathcal{L}_{A,s+1}^+\bigl(\tilde{v}_{a_{s+1}}^{J_q^+}\bigr)}\!\norm{v'}=\tilde{v}_{a_{s+1}}^{J_q^+}\tilde{v}^{J_{s+1}^+}=\bigl(\tilde{v}^{J_q^+}(s;\tilde{v}_{a_{s+1}}^{J_q^+})\bigr)_{J_{s+1}^+}.\]
Thus,
\[\mu\mapsto r_s(\mu):=\norm{\tilde{v}^{J_q^+}(s;\mu)}^2=\mu^2\,\norm{\tilde{v}^{J_{s+1}^+}}^2+\sum_{i\in J_{s+1}}\bigl((1-t_i^{J_{s+1}})\mu+t_i^{J_{s+1}}\tilde{v}_{a_s}^{J_q^+}\bigr)^2+\sum_{i=a_s}^{a_q}\bigl(\tilde{v}_i^{J_q^+}\bigr)^2\]
is a quadratic function with $r_s\bigl(\tilde{v}_{a_{s+1}}^{J_q^+}\bigr)=\norm{\tilde{v}^{J_q^+}}^2=\min_{v\in\mathcal{L}_{A,q}^+(1)}\norm{v}^2=\min_{\mu\in\R}r_s(\mu)$, so
\begin{equation}
\label{eq:gradRs}
0=r_s'\bigl(\tilde{v}_{a_{s+1}}^{J_q^+}\bigr)=2\tilde{v}_{a_{s+1}}^{J_q^+}\,\norm{\tilde{v}^{J_{s+1}^+}}^2+2\sum_{i\in J_{s+1}}\bigl((1-t_i^{J_{s+1}})\tilde{v}_{a_{s+1}}^{J_q^+}+t_i^{J_{s+1}}\tilde{v}_{a_s}^{J_q^+}\bigr)\bigl(1-t_i^{J_{s+1}}\bigr).
\end{equation}
Now $\norm{\tilde{v}^{J_{s+1}^+}}^2>0$ by the definition of $\tilde{v}^{J_{s+1}^+}\in\mathcal{L}_{A,s+1}^+(1)$, so if $a_s=a_{s+1}+1$, i.e.\ $J_{s+1}=\emptyset$, then $\tilde{v}_{a_{s+1}}^{J_q^+}=0$ by~\eqref{eq:gradRs}. On the other hand, suppose instead that $a_s>a_{s+1}+1$, in which case $J_{s+1}\neq\emptyset$. If $\tilde{v}_{a_s}^{J_q^+}=0$, then $\tilde{v}_{a_{s+1}}^{J_q^+}=0$ by~\eqref{eq:gradRs}. When $\tilde{v}_{a_s}^{J_q^+}>0$, we must have $\tilde{v}_{a_{s+1}}^{J_q^+}<0$, since otherwise the first term on the right-hand side of~\eqref{eq:gradRs} would be non-negative and each summand in the second term would be strictly positive, contradicting the fact that $r_s'\bigl(\tilde{v}_{a_{s+1}}^{J_q^+}\bigr)=0$. Similarly, if $\tilde{v}_{a_s}^{J_q^+}<0$, then $\tilde{v}_{a_{s+1}}^{J_q^+}>0$. This completes the proof of the claim. 
\end{proof}

\unparskip
Next, note that $u_{J_q^+}^\ell=0\in\mathcal{L}_{A,q}^+(0)=\mathcal{L}_{A,q}^+(u_{a_q}^\ell)$ and $u_{J_q^-}^\ell=(x_i-x_\ell:a_{q-1}\leq i\leq j)\in\mathcal{L}_{A,q}^-(u_{a_{q-1}}^\ell)$. Thus, for each $(\lambda,\vartheta)\in\R^2$, we have
\begin{align}
\mathcal{L}_A^q(\lambda,\vartheta)&:=\{z\in\mathcal{L}_A:(u^\ell-z)_{a_q}=\lambda,\,(u^\ell-z)_{a_{q-1}}=\vartheta\}\notag\\
\label{eq:LAq}
&\phantom{:}=\bigl\{z\in\mathcal{L}_A:(u^\ell-z)_{J_q^+}\in\mathcal{L}_{A,q}^+(\lambda),\,(u^\ell-z)_{J_q^-}\in\mathcal{L}_{A,q}^-(\vartheta),\,z_{J_q}=v^{J_q}(\lambda,\vartheta)\bigr\},
\end{align}
and the unique minimiser $\tilde{z}(\lambda,\vartheta)$ of $z\mapsto\norm{u^\ell-z}^2=\norm{(u^\ell-z)_{J_q^+}}^2+\norm{(u^\ell-z)_{J_q}}^2+\norm{(u^\ell-z)_{J_q^-}}^2$ over $\mathcal{L}_A^q(\lambda,\vartheta)$ satisfies $\tilde{z}(\lambda,\vartheta)_{J_q}=v^{J_q}(\lambda,\vartheta)$, $\bigl(u^\ell-\tilde{z}(\lambda,\vartheta)\bigr)_{J_q^+}=\argmin_{v\in\mathcal{L}_{A,q}^+(\lambda)}\norm{v}=\lambda\tilde{v}^{J_q^+}$ and $\bigl(u^\ell-\tilde{z}(\lambda,\vartheta)\bigr)_{J_q^-}=\argmin_{v\in\mathcal{L}_{A,q}^-(\vartheta)}\norm{v}=\vartheta\tilde{v}^{J_q^-}$. Let
\begin{align*}
r(\lambda,\vartheta)&:=\min_{z\in\mathcal{L}_A^q(\lambda,\vartheta)}\norm{u^\ell-z}^2=\norm{u^\ell-\tilde{z}(\lambda,\vartheta)}^2=\norm{\lambda\tilde{v}^{J_q^+}}^2+\norm{u_{J_q}^\ell-v^{J_q}(\lambda,\vartheta)}^2+\norm{\vartheta\tilde{v}^{J_q^-}}^2\\
&\phantom{:}=\lambda^2\,\norm{\tilde{v}^{J_q^+}}^2+\sum_{i\in J_q}\bigl(u_i^\ell-(1-t_i^{J_q})(u_{a_q}^\ell-\lambda)-t_i^{J_q}(u_{a_{q-1}}^\ell-\vartheta)\bigr)^2+\vartheta^2\,\norm{\tilde{v}^{J_q^-}}^2,
\end{align*}
so that $(\lambda,\vartheta)\mapsto r(\lambda,\vartheta)$ is a quadratic form with
\begin{equation}
\label{eq:gradR}
\nabla r(\lambda,\vartheta)=2\lambda\,\norm{\tilde{v}^{J_q^+}}^2\begin{pmatrix}1\\ 0\end{pmatrix}+2\sum_{i\in J_q}\bigl(u_i^\ell-v_i^{J_q}(\lambda,\vartheta)\bigr)\begin{pmatrix}1-t_i^{J_q}\\ t_i^{J_q}\end{pmatrix}+2\vartheta\,\norm{\tilde{v}^{J_q^-}}^2\begin{pmatrix}0\\ 1\end{pmatrix}
\end{equation}
for each $(\lambda,\vartheta)\in\R^2$, and $\tilde{z}=P_A u^\ell$ satisfies $\norm{u^\ell-\tilde{z}}^2=\min_{z\in\mathcal{L}_A}\norm{u^\ell-z}^2=\min_{(\lambda,\vartheta)\in\R^2}r(\lambda,\vartheta)$. Thus, writing $\tilde{\lambda}:=(u^\ell-\tilde{z})_{a_q}=\ipr{(I-P_A)u^\ell}{e_{a_q}}$ and $\tilde{\vartheta}:=(u^\ell-\tilde{z})_{a_{q-1}}=\ipr{(I-P_A)u^\ell}{e_{a_{q-1}}}$, we have $\tilde{z}=\tilde{z}(\tilde{\lambda},\tilde{\vartheta})$ and $(\tilde{\lambda},\tilde{\vartheta})=\argmin_{(\lambda,\vartheta)\in\R^2}r(\lambda,\vartheta)$,
whence $\nabla r(\tilde{\lambda},\tilde{\vartheta})=0$.
\begin{claim}
\label{cl:resipr2}
$\tilde{\lambda},\tilde{\vartheta}>0$ and $\ipr{(I-P_A)u^\ell}{e_\ell}=(u^\ell-\tilde{z})_\ell=u_\ell^\ell-v^{J_q}(\tilde{\lambda},\tilde{\vartheta})_\ell<0$.
\end{claim}

\deparskip	
\begin{proof}[Proof of Claim~\ref{cl:resipr2}]
It suffices to show that if $(\lambda,\vartheta)\in\R^2$ is such that either $\lambda\leq 0$, $\vartheta\leq 0$ or $u_\ell^\ell-v^{J_q}(\lambda,\vartheta)_\ell\geq 0$, then $\nabla r(\lambda,\vartheta)\neq 0$. For any such $(\lambda,\vartheta)$, it is enough to prove that there exist $\lambda',\vartheta'\in\R$ such that $\lambda\lambda'\geq 0$, $\vartheta\vartheta'\geq 0$ and $\bigl(u_i^\ell-v_i^{J_q}(\lambda,\vartheta)\bigr)\bigl((1-t_i^{J_q})\lambda'+t_i^{J_q}\vartheta'\bigr)\geq 0$ for all $i\in J_q$, with at least one of these inequalities being strict, since then
\[\nabla r(\lambda,\vartheta)^\top\begin{pmatrix}\lambda'\\\vartheta'\end{pmatrix}=2\lambda\lambda'\,\norm{\tilde{v}^{J_q^+}}^2+2\sum_{i\in J_q}\bigl(u_i^\ell-v_i^{J_q}(\lambda,\vartheta)\bigr)\bigl((1-t_i^{J_q})\lambda'+ t_i^{J_q}\vartheta'\bigr)+2\vartheta\vartheta'\,\norm{\tilde{v}^{J_q^-}}^2>0\]
by~\eqref{eq:gradR}. To this end, define the convex function $g\colon x\mapsto (x-x_\ell)^+$ on $[x_{a_q},x_{a_{q-1}}]$, and let $h$ be the unique affine function with $h(x_{a_{q-1}})=g(x_{a_{q-1}})-\lambda$ and $h(x_{a_q})=g(x_{a_q})-\vartheta$, so that $g(x_i)=u_i^\ell$ and $h(x_i)=v_i^{J_q}(\lambda,\vartheta)$ for $a_q\leq i\leq a_{q-1}$.
Since $(\lambda,\vartheta)$ satisfies at least one of the three conditions above, the possibilities for $I:=\{x\in (x_{a_q},x_{a_{q-1}}):h(x)>g(x)\}$ are as follows. In each case, we verify that there is an affine function $\tilde{h}$ such that $\bigl(g(x)-h(x)\bigr)\bigl(h(x)-\tilde{h}(x)\bigr)\geq 0$ for all $x\in [x_{a_q},x_{a_{q-1}}]$, with strict inequality for some $x\in\{x_{a_q},x_{a_q+1},\dotsc,x_{a_{q-1}}\}$:

\unparskip
\begin{itemize}[leftmargin=0.4cm]
\item $I=\emptyset$: then $g(x)\geq h(x)$ for all $x\in [x_{a_q},x_{a_{q-1}}]$, and strict inequality holds for some $x\in\{x_{a_q},x_{a_q+1},\dotsc,x_{a_{q-1}}\}$, so we can take $\tilde{h}$ to be any affine function such that $\tilde{h}<h$ on $[x_{a_q},x_{a_{q-1}}]$.
\item $I=(x_{a_q},x_{a_{q-1}})$: by the continuity of $g,h$, we have $g(x)\leq h(x)$ for all $x\in [x_{a_q},x_{a_{q-1}}]$, with strict inequality for some $x\in\{x_{a_q},x_{a_q+1},\dotsc,x_{a_{q-1}}\}$, and we can take $\tilde{h}$ to be any affine function such that $\tilde{h}>h$ on $[x_{a_q},x_{a_{q-1}}]$.
\item $I=(x_{a_q},\tilde{x})$ for some $\tilde{x}\in (x_{a_q},x_{a_{q-1}})$: by continuity, $g(\tilde{x})=h(\tilde{x})$, and we must have $g(x_{a_{q-1}})>h(x_{a_{q-1}})$ since $I\neq (x_{a_q},x_{a_{q-1}})$. Thus, we can take $\tilde{h}$ to be any affine function satisfying $\tilde{h}(\tilde{x})=h(\tilde{x})$ and $\tilde{h}(x_{a_q})>h(x_{a_q})$, so that $g\leq h\leq\tilde{h}$ on $[x_{a_q},\tilde{x}]$, $g\geq h\geq\tilde{h}$ on $[\tilde{x},x_{a_{q-1}}]$ and $g(x_{a_{q-1}})>h(x_{a_{q-1}})>\tilde{h}(x_{a_{q-1}})$.
\item $I=(\tilde{x},x_{a_{q-1}})$ for some $\tilde{x}\in (x_{a_q},x_{a_{q-1}})$: similarly, we can take $\tilde{h}$ to be any affine function satisfying $\tilde{h}(\tilde{x})=h(\tilde{x})$ and $\tilde{h}(x_{a_{q-1}})>h(x_{a_{q-1}})$.
\end{itemize}

\unparskip
Now let $\lambda':=(h-\tilde{h})(x_{a_q})$ and $\vartheta':=(h-\tilde{h})(x_{a_{q-1}})$. Then for each $i\in J_q$, we have $(h-\tilde{h})(x_i)=(1-t_i^{J_q})\lambda'+t_i^{J_q}\vartheta'$ since $h-\tilde{h}$ is an affine function, and recall that $(g-h)(x_i)=u_i^\ell-v_i^{J_q}(\lambda,\vartheta)$. Thus, $\lambda',\vartheta'$ have the required properties.
\end{proof}

\unparskip
In conclusion, by the observation after~\eqref{eq:LAq} and Claim~\ref{cl:resipr2}, we have
\[\sgn\bigl(\ipr{(I-P_A)u^\ell}{e_{a_s}}\bigr)=\sgn\bigl((u^\ell-\tilde{z}(\tilde{\lambda},\tilde{\vartheta}))_{a_s}\bigr)=
\begin{cases}
\sgn\bigl(\tilde{\lambda}\tilde{v}_{a_s}^{J_q^+}\bigr)=\sgn\bigl(\tilde{v}_{a_s}^{J_q^+}\bigr)&\quad\text{if }q\leq s\leq m\\
\sgn\bigl(\tilde{\vartheta}\tilde{v}_{a_s}^{J_q^-}\bigr)=\sgn\bigl(\tilde{v}_{a_s}^{J_q^-}\bigr)&\quad\text{if }0\leq s\leq q-1,
\end{cases}\]
which together with Claim~\ref{cl:resipr1} implies~\eqref{eq:resipr}, as desired.
\end{proof}

\deparskip
\begin{proof}[Proof of Proposition~\ref{prop:ejactive}]
For fixed $A\subseteq [j-1]$ and $\ell\in [j-1]$, 
let $\hat{\lambda}_\ell^A(u)=\lambda_\ell(P_A u)$ and $\hat{\zeta}_\ell^A(u)=\ipr{u^\ell}{(I-P_A)u}$ be as in~\eqref{eq:algconeproj}, where $u$ is some positive multiple of $e_j$. Enumerate the elements of $A$ as $j=a_0>a_1>\cdots>a_m>a_{m+1}=0$ and let $q':=\min(\{2\leq\tilde{q}\leq m:a_{\tilde{q}-1}=a_{\tilde{q}}+1\}\cup\{m+1\})$. Now $P_A u\in\mathcal{L}_A$, and if $j-1\notin A$, then for all $s\in\{0,\dotsc,m\}$, it follows by taking $\ell=j-1$ and $q=1$ in~\eqref{eq:resipr} that
\[\ipr{P_A u}{e_{a_s}}\:
\begin{cases}
>0&\;\text{if }s<q'\text{ and }s\text{ is odd}\\
<0&\;\text{if }s<q'\text{ and }s\text{ is even}\\
=0&\;\text{if }s\geq q'.
\end{cases}
\]
For $\ell\in [j-1]$, we deduce from this and~\eqref{eq:primal} that 
\begin{equation}
\label{eq:primalipr}
\hat{\lambda}_\ell^A(u)=\lambda_\ell(P_A u)\:
\begin{cases}
>0&\;\text{if }\ell=a_s\text{ for some odd }1\leq s\leq q'\\
<0&\;\text{if }\ell=a_s\text{ for some even }1\leq s\leq q'\\
=0&\;\text{otherwise}.
\end{cases}
\end{equation}
Moreover, if $j-1\notin A$, then for $\ell\in [j-1]$, it follows by taking $s=0$ in~\eqref{eq:resipr} that
\begin{equation}
\label{eq:dualipr}
\hat{\zeta}_\ell^A(u)=\ipr{(I-P_A)u^\ell}{u}\:
\begin{cases}
>0&\;\text{if }a_q<\ell<a_{q-1}\text{ for some odd }q\in [q']\\
<0&\;\text{if }a_q<\ell<a_{q-1}\text{ for some even }q\in [q']\\
=0&\;\text{if }\ell\leq a_{q'}=\ell_r\text{ or }\ell\in A.
\end{cases}
\end{equation}
We are now in a position to show that under modification (IV'), Algorithm~\ref{alg:coneproj} cannot remain indefinitely at any of the thresholds $t_r$. To this end, it suffices to verify that if $r\in\N$ is such that $t_r=t_{r+1}=t_{r+2}$, then $\ell_{\max}:=\max(A_r^-\cup A_r^+)>\max(A_{r+1}^-\cup A_{r+1}^+)$. First, we prove that $\ell_{\max}\notin A_{r+1}^-\cup A_{r+1}^+$. Enumerating the elements of $A\equiv A_r$ as $a_1>\cdots>a_m$ and defining $a_0,q'$ as above, we consider separately the cases $\ell_{\max}\in A_r^-$ and $\ell_{\max}\in A_r^+$.

\unparskip
\begin{itemize}[leftmargin=0.4cm]
\item If $\ell_{\max}\in A_r^-$, then $\beta_{\ell_{\max}}(t_{r+1})=0$ and $\hat{\lambda}_{\ell_{\max}}^{A_r}(u)>0$. This means that $\ell_{\max}=a_s$ for some odd $s\in [q']$. Indeed, when $j-1\notin A_r$, this follows from~\eqref{eq:primalipr}, and otherwise if $j-1\in A_r$,
then $A_r^-=\{j-1\}$ and $\ell_{\max}=j-1=a_1$. Now $A_{r+1}=A_r\setminus\{\ell_{\max}\}\subseteq [j-2]$ under (IV'), so $\ell_{\max}\notin A_{r+1}^-\subseteq A_{r+1}$, and enumerating the elements of $A_{r+1}$ as $a_1>\cdots>a_{s-1}>a_s'>a_{s+1}'>\cdots>a_{m-1}'$, we have $a_s'<\ell_{\max}<a_{s-1}$. Since $s\;(\leq q')$ is odd, we deduce from~\eqref{eq:dualipr} that $\hat{\zeta}_{\ell_{\max}}^{A_{r+1}}(u)>0$, and hence that $\ell_{\max}\notin A_{r+1}^+$.
\item Otherwise, if $\ell_{\max}\in A_r^+$, then $\gamma_{\ell_{\max}}(t_{r+1})=0$ and $\hat{\zeta}_{\ell_{\max}}^{A_r}(u)<0$. In this case, we necessarily have $j-1\notin A_r$, since otherwise $A_r^+=\emptyset$, so it follows from~\eqref{eq:dualipr} that $a_s<\ell_{\max}<a_{s-1}$ for some even $s\in [q'-1]$. Now $A_{r+1}=A_r\cup\{\ell_{\max}\}\subseteq [j-2]$ under (IV'), so $\ell_{\max}\notin A_{r+1}^+\subseteq A_{r+1}^c$, and $A_{r+1}$ can be enumerated as $a_1>\cdots>a_{s-1}>a_s'>a_{s+1}'>\cdots>a_{m+1}'$, where $a_s'=\ell_{\max}$. Since $s$ is even, we deduce from~\eqref{eq:primalipr} that $\hat{\lambda}_{\ell_{\max}}^{A_{r+1}}(u)\leq 0$,
and hence that $\ell_{\max}\notin A_{r+1}^-$.
\end{itemize}

\unparskip
It remains to show that $A_{r+1}^-\cup A_{r+1}^+$ does not contain any $\ell\in\{\ell_{\max}+1,\dotsc,j-1\}$. If $\ell_{\max}=j-1$, then there is nothing to prove, so we assume that $\ell_{\max}<j-1$, in which case $j-1\notin A_r$ by the arguments above. Writing $a_q'$ for the $q^{th}$ largest element of $A_{r+1}$, we see that in both cases above, $a_1'=a_1>\cdots>a_{s-1}'=a_{s-1}$ are precisely the indices in $A_{r+1}$ that are strictly greater than $\ell_{\max}$, where $s\leq q'$. Now fix $j-1\geq\ell>\ell_{\max}$ and note that since $t_{r+1}=t_{r+2}$ by assumption, we have the following:

\unparskip
\begin{itemize}[leftmargin=0.4cm]
\item Suppose that $\ell\in A_{r+1}$ and $\beta_\ell(t_{r+2})=0$, in which case $\beta_\ell(t_{r+1})=0$, $\ell\in A_r$ and $\ell\notin A_r^-$ by the definition of $\ell_{\max}$. Thus, $\hat{\lambda}_\ell^{A_r}(u)\leq 0$, so by applying~\eqref{eq:primalipr} to $A_r$, 
we deduce that $\ell\neq a_q$ for any odd $q\in [q']$. Since $a_q=a_q'$ for $q\leq s-1$ and $\ell>a_q'$ for $q\geq s$, this means that $\ell\neq a_q'$ for any odd $q$. Applying~\eqref{eq:primalipr} once again to $A_{r+1}$, we conclude that $\hat{\lambda}_\ell^{A_{r+1}}(u)\leq 0$, whence $\ell\notin A_{r+1}^-$.
\item Suppose that $\ell\notin A_{r+1}$ and $\gamma_\ell(t_{r+2})=0$, in which case $\gamma_\ell(t_{r+1})=0$, $\ell\notin A_r$ and $\ell\notin A_r^+$ by the definition of $\ell_{\max}$. Thus, $\hat{\zeta}_\ell^{A_r}(u)\geq 0$, so in view of~\eqref{eq:dualipr}, we cannot have $a_q<\ell<a_{q-1}$ for any even $q\in [q'-1]$. As above, it follows that we cannot have $a_q'<\ell<a_{q-1}'$ for any even $q$. Applying~\eqref{eq:dualipr} once again to $A_{r+1}$, we conclude that $\hat{\zeta}_\ell^{A_{r+1}}(u)\geq 0$, whence $\ell\notin A_{r+1}^+$.
\end{itemize}

\unparskip
This completes the justification of (IV'). Finally, we obtain both assertions of Proposition~\ref{prop:ejactive} as straightforward consequences of~\eqref{eq:primalipr} and~\eqref{eq:dualipr}.

(a) By taking $q=1$ in the first line of~\eqref{eq:dualipr}, we see that $\hat{\lambda}_\ell^{A_r}(u)>0$ for all $\max A_r=a_1<\ell\leq j-1$. Thus, in Algorithm~\ref{alg:coneproj} with modification (IV'), $A_{r+1}\subseteq A_r\cup A_r^+\subseteq\{1,\dotsc,\max A_r\}$.

(b) Since $\ell_r=a_{q'}$ here, this follows immediately from the final lines of~\eqref{eq:primalipr} and~\eqref{eq:dualipr}.
\end{proof}

\umparskip
\section{Auxiliary results for Section~\ref{sec:lsregmainproofs}}

\subsection{Auxiliary results for Section~\ref{subsec:oracle}}
The proof of Theorem~\ref{thm:worstcase} relies on the following bound on the localised Gaussian widths of the cone $\Gamma[\mathcal{D}]=\bigl\{\bigl(f(x_1),\dotsc,f(x_n)\bigr):f\in\mathcal{F}\bigr\}\subseteq\R^n$, where $\mathcal{D}$ is a set of design points $x_1<\cdots<x_n$ in $[0,1]$ and $\mathcal{F}$ is the class of all S-shaped functions on $[0,1]$. For $\theta\equiv (\theta_1,\dotsc,\theta_n)\in\R^n$ and $r>0$, recall from Section~\ref{sec:lsregmainproofs} that we defined $V(\theta)=\theta_n-\theta_1$ and $\Gamma(\theta,r)\equiv\Gamma(\theta,r)[\mathcal{D}]=\{v\in\Gamma[\mathcal{D}]:\norm{v-\theta}\leq r\}$.
\begin{lemma}
\label{lem:lgwTheta}
For a set $\mathcal{D}\subseteq [0,1]$ of design points $x_1<\cdots<x_n$ with $n \geq 2$, define $\tilde{R}(\mathcal{D})$ as in~\eqref{eq:RD}. Fix $\theta\in\Gamma[\mathcal{D}]$ and $r>0$. If $Z\sim N_n(0,I_n)$, then for all $\tilde{C}\geq 1$, we have
\begin{align}
\label{eq:lgwTheta}
\E\,\biggl(\sup_{v\in\Gamma(\theta,r)}\:\abs{Z^\top (v-\theta)}\biggr)&\lesssim\frac{r^2}{\tilde{C}}+r\sqrt{\log n}+\bigl(V(\theta)+\tilde{C}\bigr)^{1/4}\tilde{R}(\mathcal{D})^{1/8}\,r^{3/4}\notag\\
&\lesssim\frac{r^2}{\tilde{C}}+r\sqrt{\log n}+\bigl(V(\theta)+\tilde{C}\bigr)^{1/4}\biggl(\frac{x_n-x_1}{\min_{2\leq i\leq n}(x_i-x_{i-1})}\biggr)^{1/8}\,r^{3/4}.
\end{align}
\end{lemma}

\unparskip
We will derive this result from the bounds in Lemma~\ref{lem:coverTheta} and~\ref{lem:convcover} on the covering numbers of 
\begin{align*}
\Gamma_{A,B}[\mathcal{D}]&:=\{(v_1,\dotsc,v_n)\in\Gamma[\mathcal{D}]:A\leq v_i\leq B\text{ for all }i\}\\
K_{A,B}[\mathcal{D}]&:=\bigl\{\bigl(f(x_1),\dotsc,f(x_n)\bigr),\,f\text{ is convex },\,A\leq f\leq B\bigr\},
\end{align*}
where $-\infty<A<B<\infty$. For $\varepsilon>0$ and $U\subseteq\R^n$, recall that $\mathcal{N}\subseteq\R^n$ is said to be an \emph{$\varepsilon$-cover} of $U$ (with respect to the Euclidean norm $\norm{{\cdot}}$) if $U\subseteq\bigcup_{u\in\mathcal{N}}\bar{B}(u,\varepsilon)$, where $\bar{B}(u,\varepsilon):=\{v\in\R^n:\norm{v-u}\leq\varepsilon\}$. We denote by $N(\varepsilon,U):=\inf\{\abs{\mathcal{N}}:\mathcal{N}\text{ is an }\varepsilon\text{-cover of }U\}\in\N\cup\{\infty\}$ the \emph{$\varepsilon$-covering number} of $U$.
\begin{lemma}
\label{lem:coverTheta}
In the setting of Lemma~\ref{lem:lgwTheta}, the following holds for all $-\infty<A<B<\infty$ and $\varepsilon>0$:
\begin{equation}
\label{eq:coverTheta}
\log N(\varepsilon,\Gamma_{A,B}[\mathcal{D}])\lesssim\log n+(B-A)^{1/2}\,\frac{\tilde{R}(\mathcal{D})^{1/4}}{\varepsilon^{1/2}}.
\end{equation}
\end{lemma}
\begin{lemma}
\label{lem:convcover}
For any finite set $\mathcal{D}\subseteq [0,1]$ and every $\varepsilon>0$, we have
\begin{equation}
\label{eq:convcoverR}
\log N(\varepsilon,K_{0,1}[\mathcal{D}])\lesssim\frac{\tilde{R}(\mathcal{D})^{1/4}}{\varepsilon^{1/2}}.
\end{equation}
\end{lemma}

\unparskip
We first give the proof of~\eqref{eq:convcoverR}, which gives rise to the definition of $\tilde{R}(\mathcal{D})$ in~\eqref{eq:RD}, and then deduce Lemmas~\ref{lem:coverTheta} and~\ref{lem:lgwTheta} in that order.

\unparskip
\begin{proof}[Proof of Lemma~\ref{lem:convcover}]
We proceed by induction on $n=\abs{\mathcal{D}}$: for $n=1$, the bound clearly holds since $\tilde{R}(\mathcal{D})=1$ by definition, so suppose now that $\mathcal{D}=\{x_1<\cdots<x_n\}$ for some $n\geq 2$. First, by taking $c_1=n\min_{2\leq i\leq n}(x_i-x_{i-1})$ in the second bound in~\citet[Lemma~A.4]{GS13} and then arguing as in the proof of~\citet[Lemma~3.3]{Cha16} ~\citet[Lemma~3.3]{Cha16}, we see that
\begin{equation}
\label{eq:convcoverRn1}
\log N(\varepsilon,K_{0,1}[\mathcal{D}])\lesssim\frac{1}{\varepsilon^{1/2}}\biggl(\frac{(x_n-x_1)}{\min_{2\leq i\leq n}(x_i-x_{i-1})}\biggr)^{1/4}.
\end{equation}
In addition, for a fixed partition of $\mathcal{D}$ into $k\geq 2$ non-empty sets $\mathcal{D}_1,\dotsc,\mathcal{D}_k$, we now define $\tilde{a}_\ell:=\tilde{R}(\mathcal{D}_\ell)^{1/10}/\bigl(\sum_{\ell'=1}^k \tilde{R}(\mathcal{D}_{\ell'})^{1/5}\bigr)^{1/2}$ for $\ell\in [k]$. Given $\varepsilon>0$, let $\varepsilon_\ell:=\varepsilon\tilde{a}_\ell$ for each $\ell$, so that $\varepsilon^2=\sum_{\ell=1}^k\varepsilon_\ell^2$. Then since $\abs{\mathcal{D}_1},\dotsc,\abs{\mathcal{D}_k}<\abs{\mathcal{D}}$, it follows by induction that
\begin{equation}
\label{eq:convcoverRn2}
\log N(\varepsilon,K_{0,1}[\mathcal{D}])\leq\sum_{\ell=1}^k\log N(\varepsilon_\ell,K_{0,1}[\mathcal{D}_\ell])\lesssim\sum_{\ell=1}^k\,\frac{\tilde{R}(\mathcal{D}_\ell)^{1/4}}{\varepsilon_\ell^{1/2}}=\frac{\bigl(\sum_{\ell=1}^k \tilde{R}(\mathcal{D}_\ell)^{1/5}\bigr)^{5/4}}{\varepsilon^{1/2}},
\end{equation}
where by Lemma~\ref{lem:optim} below (with $b_\ell=\tilde{R}(\mathcal{D}_\ell)^{1/2}$ for all $\ell$), our choice of $\varepsilon_1,\dotsc,\varepsilon_k$ minimises the penultimate expression above subject to the constraint $\varepsilon^2=\sum_{\ell=1}^k\varepsilon_\ell^2$. Minimising the right hand side of~\eqref{eq:convcoverRn2} over all partitions of $\mathcal{D}$ into $k\geq 2$ non-empty subsets, we can combine~\eqref{eq:convcoverRn1} and~\eqref{eq:convcoverRn2} to complete the inductive step for~\eqref{eq:convcoverR}, in view of the definition of $\tilde{R}(\mathcal{D})$ in~\eqref{eq:RD}.
\end{proof}

\unparskip
\begin{lemma}
\label{lem:optim}
For fixed $b_1,\dotsc,b_k>0$, the unique solution to the optimisation problem
\[\min\;\sum_{\ell=1}^k\rbr{\frac{b_\ell}{a_\ell}}^{1/2}\quad\text{subject to }\;\;\sum_{\ell=1}^k a_\ell^2 =1,\,a_\ell>0\text{ for }\ell\in [k]\]
is given by $\sum_{\ell=1}^k(b_\ell/a_\ell^*)^{1/2}=\bigl(\sum_{\ell=1}^k b_\ell^{2/5}\bigr)^{5/4}$, where $a_\ell^*:=b_\ell^{1/5}/\bigl(\sum_{\ell=1}^k b_\ell^{2/5}\bigr)^{1/2}$ for each $\ell$.
\end{lemma}

\deparskip
\begin{proof}[Proof of Lemma~\ref{lem:optim}]
Let $\tau=2/5$, $p=2/(2-\tau)$ and $q=2/\tau$, so that $1/p+1/q=1$, $\tau p=1/2$ and $\tau q=2$. Then by H\"older's inequality,
\[\sum_{\ell=1}^k b_\ell^{2/5}=\sum_{\ell=1}^k b_\ell^\tau\leq\cbr{\sum_{\ell=1}^k \rbr{\frac{b_\ell^\tau}{a_\ell^\tau}}^p}^{1/p}\rbr{\sum_{\ell=1}^k a_\ell^{\tau q}}^{1/q}=\rbr{\sum_{\ell=1}^k b_\ell^{2/5}}^{5/4},
\]
with equality if and only if $a_\ell=b_\ell^{p/(p+q)}/\lambda=b_\ell^{1/5}/\lambda$ for all $\ell$, where taking $\lambda=\bigl(\sum_{\ell=1}^k b_\ell^{2/5}\bigr)^{1/2}$ ensures that $\sum_{\ell=1}^k a_\ell^2=1$. 
\end{proof}

\deparskip
\begin{proof}[Proof of Lemma~\ref{lem:coverTheta}]
By a scaling argument, it suffices to show that
\begin{equation}
\label{eq:cover01}
\log N(\varepsilon,\Gamma_{0,1}[\mathcal{D}])\lesssim\log n+\frac{\tilde{R}(\mathcal{D})^{1/4}}{\varepsilon^{1/2}}
\end{equation}
for all $\varepsilon>0$, i.e.\ that~\eqref{eq:coverTheta} holds when $A=0$ and $B=1$. Indeed, for general $-\infty<A<B<\infty$, define the invertible affine map $L_{A,B}\colon\R^n\to\R^n$ by $L_{A,B}(v)_i:=A+(B-A)v_i$ for $v\equiv (v_1,\dotsc,v_n)\in\R^n$ and $i\in [n]$, so that $\Gamma_{A,B}[\mathcal{D}]=\{L_{A,B}(v):v\in\Gamma_{0,1}[\mathcal{D}]\}$. If~\eqref{eq:cover01} holds, then for any $\varepsilon>0$, we can find an $\varepsilon/(B-A)$-cover $\mathcal{N}$ of $\Gamma_{0,1}[\mathcal{D}]$ with $\log\,\abs{\mathcal{N}}\lesssim\log n+(B-A)^{1/2}\,(\sqrt{Rn}/\varepsilon)^{1/2}$. For any $\theta\in\Gamma_{A,B}[\mathcal{D}]$, there exists $\theta^*\in\mathcal{N}$ satisfying $\norm{\theta-L_{A,B}(\theta^*)}=(B-A)\norm{L_{A,B}^{-1}(\theta)-\theta^*}\leq\varepsilon$, so $\mathcal{N}_{A,B}:=\{L_{A,B}(v):v\in\mathcal{N}\}$ is an $\varepsilon$-cover of $\Gamma_{A,B}[\mathcal{D}]$ with $\log\,\abs{\mathcal{N}_{A,B}}=\log\,\abs{\mathcal{N}}\lesssim\log n+(B-A)^{1/2}\,(\sqrt{Rn}/\varepsilon)^{1/2}$, as desired.

To establish~\eqref{eq:cover01}, fix $\varepsilon>0$ and let $\Gamma_{0,1}^m[\mathcal{D}]:=\bigl\{\bigl(f(x_1),\dotsc,f(x_n)\bigr):f\in\mathcal{F}^m,\,0\leq f\leq 1\bigr\}$ for $m\in [0,1]$, so that $\Gamma_{0,1}[\mathcal{D}]=\bigcup_{j=1}^{\,n}\Gamma_{0,1}^{x_j}[\mathcal{D}]$ and \[N(\varepsilon,\Gamma_{0,1}[\mathcal{D}])\leq\sum_{i=1}^n N(\varepsilon,\Gamma_{0,1}^{x_j}[\mathcal{D}])\leq n\max_{1\leq j\leq n}N(\varepsilon,\Gamma_{0,1}^{x_j}[\mathcal{D}]).\]
Now for $j\in [n]$, let $\mathcal{D}_j^-:=\{x_i:1\leq i\leq j\}$ and $\mathcal{D}_j^+:=\{x_i:j+1\leq i\leq n\}$. Then $\Gamma_{0,1}[\mathcal{D}]\subseteq K_{0,1}[\mathcal{D}_j^-]\times \bigl(-K_{0,1}[\mathcal{D}_j^+]\bigr)$ and $\tilde{R}(\mathcal{D}_j^\pm)\leq \tilde{R}(\mathcal{D})$ by~\eqref{eq:RD}, so it follows from Lemma~\ref{lem:convcover} that
\[\log N(\varepsilon,\Gamma_{0,1}^{x_j}[\mathcal{D}])\leq\log N\bigl(\varepsilon/\sqrt{2},K_{0,1}[\mathcal{D}_j^-]\bigr)+\log N\bigl(\varepsilon/\sqrt{2},K_{0,1}[\mathcal{D}_j^+]\bigr)\lesssim\frac{\tilde{R}(\mathcal{D})^{1/4}}{\varepsilon^{1/2}}.\]
We conclude that
\[\log N(\varepsilon,\Gamma_{0,1}[\mathcal{D}])\leq\log n+\max_{1\leq j\leq n}\log N(\varepsilon_\ell,\Gamma_{0,1}^{x_j}[\mathcal{D}])\lesssim\log n+\frac{\tilde{R}(\mathcal{D})^{1/4}}{\varepsilon^{1/2}},\]
as required.
\end{proof}

\unparskip
When $\varepsilon\gg B-A$, it turns out that in the proof above, we do not have to construct separate $\varepsilon$-covers for each of the sets $\Gamma_{A,B}^{x_1}[\mathcal{D}],\dotsc,\Gamma_{A,B}^{x_n}[\mathcal{D}]$ individually. This is because elements of $\Gamma_{A,B}^{x_j}[\mathcal{D}]$ can be approximated to accuracy $\varepsilon$ by those in a covering set for $\Gamma_{A,B}^{x_{j'}}[\mathcal{D}]$ with $j'$ close to $j$. In general, we can improve the first $\log n$ term in~\eqref{eq:coverTheta} to $\log\bigl(1\vee\{n(B-A)^2/\varepsilon^2\}\wedge n\bigr)$, and hence obtain an overall bound in Lemma~\ref{lem:coverTheta} that tends to 0 as $\varepsilon\rightarrow \infty$.  We omit further details of these additional arguments, since this improved result leads to the same worst-case oracle inequality~\eqref{eq:worstpr} as in Theorem~\ref{thm:worstcase} (possibly with a slightly smaller universal constant $C$).

\unparskip
\begin{proof}[Proof of Lemma~\ref{lem:lgwTheta}]
Fix $\theta\in\Gamma$ and let $Z\sim N_n(0,I_n)$. For every $k\in\N$, let $A_k:=\theta_1-2^k=\min_{1\leq i\leq n}\theta_i-2^k$ and $B_k:=\theta_n+2^k=\max_{1\leq i\leq n}\theta_i+2^k$, and define $\pi_k(s):=s\vee A_k\wedge B_k$ for $s\in\R$. Note that if $v\in\Gamma$, then $\pi_k(v):=(\pi_k(v_1),\dotsc,\pi_k(v_n))\in\Gamma_{A_k,B_k}=:\tilde{\Gamma}_k$. Moreover, $\theta\in\tilde{\Gamma}_k$ in view of our choice of $A_k,B_k$, and if $v\in\Gamma(\theta,r)$ for some $r>0$, then $\pi_k(v)\in\tilde{\Gamma}_k(\theta,r)$. Consequently, for any $r>0$ and $k\in\N$, we have
\begin{align}
\E\,\biggl(\sup_{v\in\Gamma(\theta,r)}\:\abs{Z^\top (v-\theta)}\biggr)&\leq\E\,\biggl(\sup_{v\in\Gamma(\theta,r)}\:\bigl|Z^\top \bigl(\pi_k(v)-\theta\bigr)\bigr|\biggr)+\E\,\biggl(\sup_{v\in\Gamma(\theta,r)}\:\bigl|Z^\top \bigl(v-\pi_k(v)\bigr)\bigr|\biggr)\notag\\
\label{eq:lgwsum}
&\leq\E\,\biggl(\sup_{u\in\tilde{\Gamma}_k(\theta,r)}\:\abs{Z^\top (u-\theta)}\biggr)+\E\,\biggl(\sup_{v\in\Gamma(\theta,r)}\:\bigl|Z^\top \bigl(v-\pi_k(v)\bigr)\bigr|\biggr).
\end{align}
To bound the first term in~\eqref{eq:lgwsum}, observe first that by the triangle inequality, $\tilde{\Gamma}_k(\theta,r)$ has diameter $d:=\sup\{\norm{v-v'}:v,v'\in\tilde{\Gamma}_k(\theta,r)\}\leq 2r$. We can now apply Lemma~\ref{lem:coverTheta} in conjunction with Dudley's metric entropy bound for Gaussian processes~\citep[e.g.][Theorem~2.3.7]{GN15} to see that
\begin{align}
\E\,\biggl(\sup_{u\in\tilde{\Gamma}_k(\theta,r)}\:\abs{Z^\top (u-\theta)}\biggr)&\leq 4\sqrt{2}\int_0^{d/2}\sqrt{\log 2N\bigl(\varepsilon,\tilde{\Gamma}_k(\theta,r)\bigr)}\,d\varepsilon\leq 4\sqrt{2}\int_0^r\sqrt{\log 2N(\varepsilon,\Gamma_{A_k,B_k})}\,d\varepsilon\notag\\
&\lesssim\int_0^r\,\bigl\{\sqrt{\log n}+(B_k-A_k)^{1/4}\tilde{R}(\mathcal{D})^{1/8}\,\varepsilon^{-1/4}\bigr\}\,d\varepsilon\notag\\
\label{eq:entint}
&\lesssim r\sqrt{\log n}+\bigl(V(\theta)+2^k\bigr)^{1/4}\tilde{R}(\mathcal{D})^{1/8}\,r^{3/4}.
\end{align}
As for the second term in~\eqref{eq:lgwsum}, we define $I_{1,\ell}(v):=\{1\leq i\leq n:A_{\ell+1}<v_i\leq A_\ell\}$ and $I_{2,\ell}(v):=\{1\leq i\leq n:B_\ell\leq v_i<B_{\ell+1}\}$ for $v\equiv (v_1,\dotsc,v_n)\in\R^n$ and $\ell\in\N$. Note that if $j\in I_{1,\ell}(v)$ for some $\ell\geq k$, then $\theta_j-v_j\geq\theta_1-A_\ell=2^\ell$ and $0\leq\pi_k(v_j)-v_j<\theta_1-A_{\ell+1}=2^{\ell+1}$. Similarly, $v_j-\theta_j\geq 2^\ell$ and $0\leq v_j-\pi_k(v_j)<2^{\ell+1}$ for all $j\in I_{2,\ell}(v)$. Thus, if $v\in\Gamma(\theta,r)$, then $\sum_{i=1}^n (\theta_i-v_i)^2\leq r^2$, so $\abs{I_{1,\ell}(v)}\vee\abs{I_{2,\ell}(v)}\leq r^2/2^{2\ell}$; in fact, since $v_1\leq\cdots\leq v_n$, this means that $I_{1,\ell}(v)\subseteq\{1,\dotsc,\floor{r^2/2^{2\ell}}\}$ and $I_{2,\ell}(v)\subseteq\{n+1-i:1\leq i\leq\floor{r^2/2^{2\ell}}\}$. Consequently, for every $v\in\Gamma(\theta,r)$, we have
\begin{align*}
\bigl|Z^\top\bigl(v-\pi_k(v)\bigr)\bigr|\leq\sum_{\ell=k}^\infty\,\sum_{i\in I_{1,\ell}(v)\,\cup\, I_{2,\ell}(v)}\abs{Z_i}\,\abs{v_i-\pi_k(v_i)}&\leq\sum_{\ell=k}^\infty\,2^{\ell+1}\sum_{i\in I_{1,\ell}(v)\,\cup\, I_{2,\ell}(v)}\abs{Z_i}\\
&\leq\sum_{\ell=k}^\infty\,2^{\ell+1}\sum_{i=1}^{\floor{r^2/2^{2\ell}}}(\abs{Z_i}+\abs{Z_{n+1-i}}),
\end{align*}
so
\begin{equation}
\label{eq:lgw2}
\E\,\biggl(\sup_{v\in\Gamma(\theta,r)}\:\bigl|Z^\top \bigl(v-\pi_k(v)\bigr)\bigr|\biggr)\leq\sum_{\ell=k}^\infty\,2^{\ell+2}\sum_{i=1}^{\floor{r^2/2^{2\ell}}}\E(\abs{Z_i})\lesssim\sum_{\ell=k}^\infty\frac{r^2}{2^\ell}\lesssim\frac{r^2}{2^k}.
\end{equation}
Finally, for any $\tilde{C}\geq 1$, let $k\in\N$ be such that $2^{k-1}\leq\tilde{C}<2^k$. The desired bound~\eqref{eq:lgwTheta} then follows from~\eqref{eq:lgwsum},~\eqref{eq:entint} and~\eqref{eq:lgw2}.
\end{proof}

\umparskip
\subsection{Auxiliary results for Section~\ref{subsec:inflection}}
\label{subsec:inflectproofs}
Here, we establish the key technical Lemmas~\ref{lem:kinkdist}--\ref{lem:step3} that form part of the proof of Theorem~\ref{thm:inflection}, as well as Lemma~\ref{lem:lamlb} from the proof of Proposition~\ref{prop:lamlb}. 
Lemma~\ref{lem:convslopes} below is the starting point for the proof of Lemma~\ref{lem:kinkdist}, and applies to general configurations of design points $x_1<\cdots<x_n$ (which need not be equispaced).
\begin{lemma}
\label{lem:convslopes}
Let $x_k$ be a kink of the convex LSE $\hat{g}_n$ based on $(x_1,Y_1),\dotsc,(x_n,Y_n)$. Let $\bar{x}_L:=k^{-1}\sum_{i=1}^k x_i$ and $\bar{x}:=n^{-1}\sum_{i=1}^n x_i$. Then \[\frac{\sum_{i=1}^k(x_i-\bar{x}_L)Y_i}{\sum_{i=1}^k(x_i-\bar{x}_L)^2}\leq\frac{\sum_{i=1}^n(x_i-\bar{x})Y_i}{\sum_{i=1}^n(x_i-\bar{x})^2}.\]
In other words, the slope of the regression line fitted using $\{(x_i,Y_i):1\leq i\leq k\}$ is at most that of the regression line fitted using $\{(x_i,Y_i):1\leq i\leq n\}$.
\end{lemma}

\deparskip
\begin{proof}
Let $S_L^2:=\sum_{i=1}^k(x_i-\bar{x}_L)^2$ and $S^2:=\sum_{i=1}^n(x_i-\bar{x})^2$. Then 
\[S^2\geq\sum_{i=1}^k(x_i-\bar{x})^2=\sum_{i=1}^k(x_i-\bar{x}_L)^2+k(\bar{x}_L-\bar{x})^2\geq S_L^2>0,\] 
since $k\geq 2$. The linear functions $h_L,h_R\colon\R\to\R$ defined by $h_L(x):=S_L^{-2}(x-\bar{x}_L)-S^{-2}(x-\bar{x})$ and $h_R(x):=-S^{-2}(x-\bar{x})$ have slopes $S_L^{-2}-S^{-2}\geq 0$ and $-S^{-2}<0$ respectively. Now let $h\in\mathcal{G}$ be such that $h(x_i)=h_L(x_i)$ for $i\in [k]$ and $h(x_i)=h_R(x_i)$ for $k+1\leq i\leq n$. Since $h_L(x_k)=S_L^{-2}(x_k-\bar{x}_L)-S^{-2}(x_k-\bar{x})\geq -S^{-2}(x_k-\bar{x})=h_R(x_k)$, this means that $h$ is convex on both $[x_1,x_k]$ and $[x_k,x_n]$ (and locally concave at $x_k$, a kink of $\hat{g}_n$). Therefore, $\hat{g}_n+\eta h\in\mathcal{G}$ is convex for sufficiently small $\eta>0$, whence $\sum_{i=1}^n h(x_i)\bigl(Y_i-\hat{g}_n(x_i)\bigr)\leq 0$ by~\eqref{eq:projconv} or Lemma~\ref{lem:projconv}. 

To establish that $\sum_{i=1}^n h(x_i)\,Y_i=\sum_{i=1}^k (h_L-h_R)(x_i)\,Y_i+\sum_{i=1}^n h_R(x_i)\,Y_i\leq 0$, as claimed in the lemma, it therefore suffices to show that \[\sum_{i=1}^n h(x_i)\,\hat{g}_n(x_i)=\sum_{i=1}^k (h_L-h_R)(x_i)\,\hat{g}_n(x_i)+\sum_{i=1}^n h_R(x_i)\,\hat{g}_n(x_i)\leq 0,\]
i.e.\ that the slope of the regression line fitted using $\{(x_i,\hat{g}_n(x_i)):1\leq i\leq k\}$ is at most that of the regression line fitted using $\{(x_i,\hat{g}_n(x_i)):1\leq i\leq n\}$. To this end, for $j\in [n]$, let $K^{1,j}\subseteq\R^j$ be the closed, convex cone of convex sequences based on $x_1,\dotsc,x_j$, as defined at the start of Section~\ref{sec:lsregmainproofs}, and define $\hat{v}^j:=\bigl(\hat{g}_n(x_1),\dotsc,\hat{g}_n(x_j)\bigr)\in K^{1,j}$. Let $\pm u^{j,0},\pm u^{j,1},u^{j,2},\dotsc,u^{j,j-1}\in K^{1,j}$ be its generators, where $u_i^{j,0}=1$ and $u_i^{j,\ell}=(x_i-x_\ell)^+$ for all $i\in [j]$ and $\ell\in [j-1]$ as in the paragraph containing~\eqref{eq:Thetak}. Since $\hat{v}^n\in K^{1,n}$, we can write $\hat{v}^n=\sum_{\ell=0}^{n-1}\hat{\lambda}_\ell u^\ell$ for some $\hat{\lambda}_0,\dotsc,\hat{\lambda}_{n-1}\in\R$ with $\hat{\lambda}_2,\dotsc,\hat{\lambda}_{n-1}\geq 0$. Let $A_L:=\{0,1\}\cup\{2\leq\ell\leq k-1:\hat{\lambda}_\ell>0\}$ and $A_R:=\{k\leq\ell\leq n-1:\hat{\lambda}_\ell>0\}$, so that \[\hat{v}^k=\sum_{\ell\in A_L}\hat{\lambda}_\ell u^{k,\ell}\quad\text{and}\quad\hat{v}^n=\sum_{\ell\in A_L}\hat{\lambda}_\ell u^{n,\ell}+\sum_{\ell\in A_R}\hat{\lambda}_\ell u^{n,\ell}.\]
For $j\in [n]$, let $\tilde{P}_j\in\R^{j\times j}$ represent the orthogonal projection onto $L_j:=\Span\{u^{j,0},u^{j,1}\}$, 
so that if $z\in\R^j$, then $\tilde{P}_jz$ is the vector of fitted values from ordinary least squares regression based on $\{(x_i,z_i):1\leq i\leq j\}$. We say that $v\in L_j$ has \emph{slope} $b$ if $v_i-v_{i-1}=b(x_i-x_{i-1})$ for $2\leq i\leq j$, and denote by $b_{j\ell}$ the slope of $\tilde{P}_j u^{j,\ell}$ for $0\leq\ell\leq j-1$. Since $k\leq n$, observe that $0\leq b_{k\ell}\leq b_{n\ell}\leq 1$ for all $0\leq\ell\leq k-1$ and $b_{k\ell}=b_{n\ell}$ for $r\in\{0,1\}$. Writing $b_k$ and $b_n$ for the slopes of $\tilde{P}_k\hat{v}^k=\sum_{\ell\in A_L}\hat{\lambda}_\ell\tilde{P}_k u^{k,\ell}$ and $\tilde{P}_n\hat{v}^n=\sum_{\ell\in A_L}\hat{\lambda}_\ell\tilde{P}_n u^{n,\ell}+\sum_{\ell\in A_R}\hat{\lambda}_\ell\tilde{P}_n u^{n,\ell}$ respectively, we conclude that 
\[b_k=\sum_{\ell\in A_L}\hat{\lambda}_\ell\,b_{k\ell}\leq\sum_{\ell\in A_L}\hat{\lambda}_\ell\,b_{n\ell}+\sum_{\ell\in A_R}\hat{\lambda}_\ell\,b_{n\ell}=b_n.\]
This completes the proof.
\end{proof}

\deparskip
\begin{proof}[Proof of Lemma~\ref{lem:kinkdist}]
If $\hat{\tau}_{0L}\neq m_0$ (i.e.\ $\{i:x_i\in\mathcal{I}_{01}\}$ is non-empty), then $\hat{\tau}_{0L}$ is a kink of $\hat{g}_{n,0}$ in $(m_0,(m_0+\tilde{m}_+)/2]$. Define $N_{01}:=\abs{\{i:x_i\in\mathcal{I}_{01}\}}\vee 1$, $N_0:=\abs{\{i:x_i\in\mathcal{I}_0\}}$, $\bar{x}_{01}:=N_{01}^{-1}\sum_{i: x_i\in\mathcal{I}_{01}}x_i$ and $\bar{x}_0:=N_0^{-1}\sum_{i: x_i\in\mathcal{I}_0}x_i$, where we suppress the dependence on $n$ for convenience. We deduce from Lemma~\ref{lem:convslopes} that if $N_{01}\geq 2$, then
\[\frac{\sum_{i: x_i\in\mathcal{I}_{01}}(x_i-\bar{x}_{01})Y_i}{\sum_{i: x_i\in\mathcal{I}_{01}}(x_i-\bar{x}_{01})^2}\leq\frac{\sum_{i: x_i\in\mathcal{I}_0}(x_i-\bar{x}_0)Y_i}{\sum_{i: x_i\in\mathcal{I}_0}(x_i-\bar{x}_0)^2},\]
and hence that
\begin{align}
&\frac{\sum_{i: x_i\in\mathcal{I}_{01}}(x_i-\bar{x}_{01})f_0(x_i)}{\sum_{i: x_i\in\mathcal{I}_{01}}(x_i-\bar{x}_{01})^2}-\frac{\sum_{i: x_i\in\mathcal{I}_0}(x_i-\bar{x}_0)f_0(x_i)}{\sum_{i: x_i\in\mathcal{I}_0}(x_i-\bar{x}_0)^2}\notag\\
\label{eq:slopes}
&\hspace{3cm}\leq\frac{\sum_{i: x_i\in\mathcal{I}_0}(x_i-\bar{x}_0)\,\xi_i}{\sum_{i: x_i\in\mathcal{I}_0}(x_i-\bar{x}_0)^2}-\frac{\sum_{i: x_i\in\mathcal{I}_{01}}(x_i-\bar{x}_{01})\,\xi_i}{\sum_{i: x_i\in\mathcal{I}_{01}}(x_i-\bar{x}_{01})^2}.
\end{align}
Let $\beta_{L1},\beta_{L2}$ be equal to the first and second terms respectively on the right-hand side of~\eqref{eq:slopes} when $N_{01}\geq 2$ (and set $\beta_{L1}=\beta_{L2}=0$ otherwise). Taking into account the randomness of the intervals $\mathcal{I}_{01},\mathcal{I}_0$, we claim that $(\beta_{L1}-\beta_{L2})\sqrt{n(\hat{\tau}_{0L}-m_0)^3}=O_p(\sqrt{\log n})$. Indeed, for fixed $1\leq a<b\leq n$, define $\bar{x}_{a:b}:=(b-a+1)^{-1}\sum_{i=a}^b x_i$, $S_{ab}^2:=\sum_{i=a}^b(x_i-\bar{x}_{a:b})^2$ and $\tilde{\beta}_{ab}:=S_{ab}^{-2}\,\sum_{i=a}^b (x_i-\bar{x}_{a:b})\,\xi_i$. Under Assumption~\ref{ass:inflection}, the design points $x_i\equiv x_{ni}=i/n$ are equispaced and the errors $\xi_i$ are sub-Gaussian with parameter 1, so $S_{ab}^2\asymp (b-a)^3/n^2$ and $\tilde{\beta}_{ab}$ has sub-Gaussian parameter $S_{ab}^{-2}\asymp n^2/(b-a)^3=n^{-1}(x_b-x_a)^{-3}$. Therefore, $\tilde{\beta}_{\mathrm{max}}:=\max_{1\leq a<b\leq n}\,\abs{\tilde{\beta}_{ab}}\sqrt{n(x_b-x_a)^3}=O_p(\sqrt{\log n})$~\citep[e.g.][Lemma~2.3.4]{GN15}, so
\begin{align}
\sqrt{n(\hat{\tau}_{0L}-m_0)^3}\,\abs{\beta_{L1}}\leq 2^{3/2}\,\tilde{\beta}_{\mathrm{max}}&=O_p(\sqrt{\log n})\notag\\
\label{eq:bR}
\sqrt{n(\hat{\tau}_{0L}-m_0)^3}\,\abs{\beta_{L2}}\leq\sqrt{n(\tilde{m}_+ -m_0)^3}\,\abs{\beta_{L2}}\leq 2^{3/2}\,\tilde{\beta}_{\mathrm{max}}&=O_p(\sqrt{\log n}),
\end{align}
which justifies the claim above. Now let $b_{L1},b_{L2}$ be equal to the first and second terms respectively on the left-hand side of~\eqref{eq:slopes} when $N_{01}\geq 2$ (and set $b_{L1}=b_{L2}=0$ otherwise). For $\gamma>1$ and $x_a\in (m_0,1]$, let $s_\gamma(x_a):=n^{-1}\sum_{i:x_i\in (m_0,x_a]}(x_i-m_0)^\gamma$, and observe that if $x_{j-1}\leq m_0<x_j<x_a$, then
\begin{align}
\frac{(x_a-m_0)^{\gamma+1}}{\gamma+1}\leq\int_{x_{j-1}}^{x_a}(x-m_0)^\gamma\,dx\leq s_\gamma(x_a)&\leq\int_{x_j}^{x_{a+1}}(x_a\wedge x-m_0)^\gamma\,dx\notag\\
\label{eq:sgamma}
&\leq\frac{(x_a-m_0)^{\gamma+1}}{\gamma+1}\rbr{1+\frac{\gamma+1}{n(x_a-m_0)}}.
\end{align}
We claim that if $\hat{\tau}_{0L}-m_0\geq 2n^{-1/(2\alpha+1)}$, then
\begin{align}
b_{L1}&=f_0'(m_0)-B\bigl(1+o_p(1)\bigr)\,\frac{s_{\alpha+1}(\hat{\tau}_{0L})-2^{-1}s_\alpha(\hat{\tau}_{0L})(\hat{\tau}_{0L}-m_0)}{s_2(\hat{\tau}_{0L})-2^{-1}s_1(\hat{\tau}_{0L})(\hat{\tau}_{0L}-m_0)}\notag\\
\label{eq:bL1}
&\geq f_0'(m_0)-\frac{6\alpha B}{(\alpha+1)(\alpha+2)}\bigl(1+o_p(1)\bigr)(\hat{\tau}_{0L}-m_0)^{\alpha-1}
\end{align}
and
\begin{align}
b_{L2}&=f_0'(m_0)-B\bigl(1+o_p(1)\bigr)\,\frac{s_{\alpha+1}(\tilde{m}_+)-2^{-1}s_\alpha(\tilde{m}_+)(\tilde{m}_+-m_0)}{s_2(\tilde{m}_+)-2^{-1}s_1(\tilde{m}_+)(\tilde{m}_+-m_0)}\notag\\
&\leq f_0'(m_0)-\frac{6\alpha B}{(\alpha+1)(\alpha+2)}\bigl(1+o_p(1)\bigr)(\tilde{m}_+ -m_0)^{\alpha-1}\notag\\
\label{eq:bL2}
&\leq f_0'(m_0)-2^{\alpha-1}\frac{6\alpha B}{(\alpha+1)(\alpha+2)}\bigl(1+o_p(1)\bigr)(\hat{\tau}_{0L}-m_0)^{\alpha-1}.
\end{align}
To see this, recall that under~\eqref{eq:smoothness} in Assumption~\ref{ass:inflection}, we can write $f_0(x)=f_0(m_0)+f_0'(m_0)(x-m_0)-B\bigl(1+\eta(x-m_0)\bigr)\sgn(x-m_0)\abs{x-m_0}^\alpha$ for $x\in [0,1]$ when $\alpha>1$, where $\eta(x-m_0)\to 0$ as $x\to m_0$. Writing $x_i-\bar{x}_{01}=(x_i-m_0)-2^{-1}(\hat{\tau}_{0L}-m_0)$, we see that
\begin{align}
\textstyle\sum_{i: x_i\in\mathcal{I}_{01}}(x_i-\bar{x}_{01})\bigl(f_0(m_0)+f_0'(m_0)(x_i-m_0)\bigr)&\textstyle=f_0'(m_0)\sum_{i: x_i\in\mathcal{I}_{01}}(x_i-\bar{x}_{01})(x_i-m_0)\notag\\
\label{eq:bL1a}
&\textstyle=f_0'(m_0)\bigl(s_2(\hat{\tau}_{0L})-2^{-1}s_1(\hat{\tau}_{0L})(\hat{\tau}_{0L}-m_0)\bigr)\\
&\textstyle=f_0'(m_0)\sum_{i: x_i\in\mathcal{I}_{01}}(x_i-\bar{x}_{01})^2\notag\\[\parskip]
\label{eq:bL1b}
\text{and}\qquad\textstyle\sum_{i: x_i\in\mathcal{I}_{01}}(x_i-\bar{x}_{01})(x_i-m_0)^\alpha&=s_{\alpha+1}(\hat{\tau}_{0L})-2^{-1}s_\alpha(\hat{\tau}_{0L}).
\end{align}
Moreover, since $\omega(\delta):=\sup\,\{\abs{\eta(x-m_0)}:x\in [0,1],\,\abs{x-m_0}\leq\delta\}\to 0$ as $\delta\to 0$ and $\tilde{m}_+ -m_0=o_p(1)$ by Proposition~\ref{cor:consistency}(a), 
we have
\begin{align}
\textstyle\abs{\sum_{i: x_i\in\mathcal{I}_{01}}(x_i-\bar{x}_{01})\,\eta(x_i-m_0)(x_i-m_0)^\alpha}&\textstyle\leq\omega(\abs{\tilde{m}_+ -m_0})\sum_{i: x_i\in\mathcal{I}_{01}}\abs{x_i-\bar{x}_{01}}(x_i-m_0)^\alpha\notag\\
\label{eq:bL1c}
&\textstyle=o_p(1)\,\bigl(s_{\alpha+1}(\hat{\tau}_{0L})+2^{-1}s_\alpha(\hat{\tau}_{0L})(\hat{\tau}_{0L}-m_0)\bigr).
\end{align}
Combining~\eqref{eq:bL1a},~\eqref{eq:bL1b} and~\eqref{eq:bL1c}, we obtain the first equality in~\eqref{eq:bL1}. On the event $\{\hat{\tau}_{0L}-m_0\geq 2n^{-1/(2\alpha+1)}\}$, we find using~\eqref{eq:sgamma} that 
\begin{align*}
\frac{s_{\alpha+1}(\hat{\tau}_{0L})-2^{-1}s_\alpha(\hat{\tau}_{0L})(\hat{\tau}_{0L}-m_0)}{s_2(\hat{\tau}_{0L})-2^{-1}s_1(\hat{\tau}_{0L})(\hat{\tau}_{0L}-m_0)}&\leq\frac{\frac{\alpha}{2(\alpha+1)(\alpha+2)}+\frac{1}{n(\hat{\tau}_{0L}-m_0)}}{\frac{1}{12}-\frac{1}{2n(\hat{\tau}_{0L}-m_0)}}(\hat{\tau}_{0L}-m_0)^{\alpha-1}\\
&\leq\bigl(1+o(1)\bigr)\frac{6\alpha}{(\alpha+1)(\alpha+2)}(\hat{\tau}_{0L}-m_0)^{\alpha-1},
\end{align*}
which justifies the lower bound on $b_{L1}$ in~\eqref{eq:bL1}. We can derive~\eqref{eq:bL2} similarly by first establishing analogues of~\eqref{eq:bL1a},~\eqref{eq:bL1b} and~\eqref{eq:bL1c}, and then applying~\eqref{eq:sgamma} to see that
\begin{align*}
\frac{s_{\alpha+1}(\tilde{m}_+)-2^{-1}s_\alpha(\tilde{m}_+)(\tilde{m}_+-m_0)}{s_2(\tilde{m}_+)-2^{-1}s_1(\tilde{m}_+)(\tilde{m}_+-m_0)}&\geq\frac{\frac{\alpha}{2(\alpha+1)(\alpha+2)}-\frac{1}{2n(\tilde{m}_+-m_0)}}{\frac{1}{12}+\frac{1}{n(\tilde{m}_+-m_0)}}(\tilde{m}_+-m_0)^{\alpha-1}\\
&\geq\bigl(1+o(1)\bigr)\frac{6\alpha}{(\alpha+1)(\alpha+2)}(\tilde{m}_+-m_0)^{\alpha-1}
\end{align*}
on the event $E_n^+\supseteq\{\hat{\tau}_{0L}-m_0\geq 2n^{-1/(2\alpha+1)}\}$. Since $\tilde{m}_+-m_0\geq 2(\hat{\tau}_{0L}-m_0)$ and $\alpha>1$, this yields the upper bound on $b_{L2}$ in~\eqref{eq:bL2}.

Thus, on the event $\{\hat{\tau}_{0L}-m_0\geq 2n^{-1/(2\alpha+1)}\}$, we can apply~\eqref{eq:bL1},~\eqref{eq:bL2},~\eqref{eq:slopes} and~\eqref{eq:bR} in that order to deduce that
\begin{align}
\label{eq:slopesbd}
(2^{\alpha-1}-1)\frac{6\alpha B}{(\alpha+1)(\alpha+2)}\bigl(1+o_p(1)\bigr)\sqrt{n}(\hat{\tau}_{0L}-m_0)^{\alpha+\frac{1}{2}}&\leq (b_{L1}-b_{L2})\sqrt{n(\hat{\tau}_{0L}-m_0)^3}\\
&\leq (\beta_{L1}-\beta_{L2})\sqrt{n(\hat{\tau}_{0L}-m_0)^3}\notag\\
&=O_p(\sqrt{\log n}).\notag
\end{align}
We conclude that $\hat{\tau}_{0L}-m_0=O_p\bigl((n/\log n)^{-1/(2\alpha+1)}\bigr)$, as required.
\end{proof}

\deparskip
\begin{proof}[Proof of Lemma~\ref{lem:step2lbd}]
Since $C_n\to\infty$, we have \[t_n=\sqrt{C_n}\,(n/\log n)^{-1/(2\alpha+1)}< 4^{-1}C_n(n/\log n)^{-1/(2\alpha+1)}=u_n/2\]
for all sufficiently large $n$. For each $n$, let $a_n:=\ceil{n(m_0+u_n/2)}$ and $b_n:=\floor{n(m_0+u_n)}$, so that $x_{a_n}<m_0+u_n/2\leq x_{a_n+1}$ and $x_{b_n-1}\leq m_0+u_n<x_{b_n}$. Then for all sufficiently large $n$, we have $\inf_{(a,b)\in\mathcal{T}_n}\inf_{c_0,c_1}\norm{\theta^{a,b}-c_0\mathbf{1}^{a,b}-c_1 x^{a,b}}^2\geq\inf_{c_0,c_1}\norm{\theta^{a_n,b_n}-c_0\mathbf{1}^{a_n,b_n}-c_1 x^{a_n,b_n}}^2=:R_n$. 

For $x\in [m_0,1]$, recall from~\eqref{eq:smoothness} in Assumption~\ref{ass:inflection} that \[f_0(x)=
\begin{cases}
f_0(m_0)+f_0'(m_0)(x-m_0)-B\bigl(1+\eta(x-m_0)\bigr)(x-m_0)^\alpha\quad&\text{when }\alpha>1\\
f_0(m_0)+B\bigl(1+\eta(x-m_0)\bigr)(x-m_0)^\alpha\quad&\text{when }\alpha\in (0,1),
\end{cases}
\]
where $\eta(x-m_0)\to 0$ as $x\to m_0$. For each $n$, let $\widetilde{\theta}_i^{a_n,b_n}:=B\bigl(1+\eta(x_i-m_0)\bigr)(x_i-m_0)^\alpha$ 
for $m_0\leq a_n\leq i\leq b_n$ and  $\widetilde{\theta}^{a_n,b_n}:=\bigl(\widetilde{\theta}_{a_n}^{a_n,b_n},\dotsc,\widetilde{\theta}_{b_n}^{a_n,b_n}\bigr)$, so that $\widetilde{\theta}^{a_n,b_n}+\theta^{a_n,b_n}\in\Span\{\mathbf{1}^{a_n,b_n},x^{a_n,b_n}\}=:A^{a_n,b_n}$ when $\alpha>1$ and $\widetilde{\theta}^{a_n,b_n}-\theta^{a_n,b_n}\in A^{a_n,b_n}$ when $\alpha\in (0,1)$. In addition, let $\bar{x}_{a_n,b_n}:=(b_n-a_n+1)^{-1}\sum_{i=a_n}^{b_n} x_i$,
so that $\widetilde{x}^{a_n,b_n}:=x^{a_n,b_n}-\bar{x}_{a_n,b_n}\mathbf{1}^{a_n,b_n}$ satisfies $\ipr{\mathbf{1}^{a_n,b_n}}{\widetilde{x}^{a_n,b_n}}=0$ and $A^{a_n,b_n}=\Span\{\mathbf{1}^{a_n,b_n},\widetilde{x}^{a_n,b_n}\}$. Then
\begin{equation}
\label{eq:Rn}
R_n=\inf_{v\in A^{a_n,b_n}}\norm{\theta^{a_n,b_n}-v}^2=\inf_{v\in A^{a_n,b_n}}\norm{\widetilde{\theta}^{a_n,b_n}-v}^2=\norm{\widetilde{\theta}^{a_n,b_n}}^2-c_{n0}^2\,\norm{\mathbf{1}^{a_n,b_n}}^2-c_{n1}^2\,\norm{\widetilde{x}^{a_n,b_n}}^2,
\end{equation}
where $c_{n0}:=\ipr{\widetilde{\theta}^{a_n,b_n}}{\mathbf{1}^{a_n,b_n}}/\norm{\mathbf{1}^{a_n,b_n}}^2$ and $c_{n1}:=\ipr{\widetilde{\theta}^{a_n,b_n}}{\widetilde{x}^{a_n,b_n}}/\norm{\widetilde{x}^{a_n,b_n}}^2$. We will consider in turn the three terms on the right-hand side of~\eqref{eq:Rn}. For each $n$, let $M_n:=nu_n$ and $z_{n,i}:=(x_i-m_0)/u_n$ for $a_n\leq i\leq b_n$, where $x_i\equiv x_{ni}=i/n$ by Assumption~\ref{ass:inflection}. Then $z_{n,i+1}-z_{n,i}=1/M_n$ for all $a_n\leq i<b_n$, and $z_{n,a_n}=1/2+o(1/M_n)$ and $z_{n,b_n}=1+o(1/M_n)$ by the definitions of $a_n,b_n$. Moreover, let $\tilde{\eta}_n(z):=\bigl(1+\eta(u_n z)\bigr)^2-1$ for $z\in [1/2,1]$ and note that $\sup_{z\in [1/2,1]}\,\abs{\tilde{\eta}_n(z)}=o(1)$ as $n\to\infty$. The first term in~\eqref{eq:Rn} can now be written as
\begin{align}
\norm{\widetilde{\theta}^{a_n,b_n}}^2&=\sum_{i=a_n}^{b_n}B^2\bigl(1+\eta(x_i-m_0)\bigr)^2(x_i-m_0)^{2\alpha}\notag\\
&=nu_n^{2\alpha+1}\sum_{i=a_n}^{b_n}\frac{B^2}{nu_n}\bigl(1+\eta(x_i-m_0)\bigr)^2\rbr{\frac{x_i-m_0}{u_n}}^{2\alpha}\notag\\
\label{eq:Rn1a}
&=B^2nu_n^{2\alpha+1}\sum_{i=a_n}^{b_n}\inv{M_n}\bigl(1+\tilde{\eta}_n(z_{n,i})\bigr)\,z_{n,i}^{2\alpha}.
\end{align}
Defining $F(z):=z^{\alpha}$ for $z\in [1/2,1]$ and noting that $M_n=nu_n=2^{-1}C_n(n^{2\alpha}\log n)^{1/(2\alpha+1)}\to\infty$, we have 
\begin{equation}
\label{eq:Rn1b}
\textstyle\sum_{i=a_n}^{b_n}\inv{M_n}z_{n,i}^{2\alpha}=\sum_{j=0}^{b_n-a_n}\inv{M_n}\,F(z_{n,a_n}+j/M_n)^2=\int_{1/2}^1 F(z)^2\,dz+o(1)=\bigl(1+o(1)\bigr)\int_{1/2}^1 z^{2\alpha}\,dz
\end{equation}
as $n\to\infty$, by a Riemann sum approximation to the (uniformly) continuous function $F$ on $[1/2,1]$. Since $\bigl|\sum_{i=a_n}^{b_n}\inv{M_n}\,\tilde{\eta}_n(z_{n,i})\,z_{n,i}^{2\alpha}\bigr|\leq\sup_{z\in [1/2,1]}\,\abs{\tilde{\eta}_n(z)}\,\sum_{i=a_n}^{b_n}\inv{M_n}z_{n,i}^{2\alpha}=o(1)$, we deduce that
\begin{equation}
\label{eq:Rn1}
\textstyle\norm{\widetilde{\theta}^{a_n,b_n}}^2=B^2nu_n^{2\alpha+1}\bigl(1+o(1)\bigr)\int_{1/2}^1 z^{2\alpha}\,dz.
\end{equation}
For $\gamma\geq 0$, we can use the rescaled design points $z_{n,i}$ and argue as in~\eqref{eq:Rn1a} and~\eqref{eq:Rn1b} to see that
\[\textstyle s_{n,\gamma}:=\sum_{i=a_n}^{b_n}(x_i-m_0)^\gamma=nu_n^{\gamma+1}\textstyle\sum_{i=a_n}^{b_n}\inv{M_n}z_{n,i}^\gamma=nu_n^{\gamma+1}\bigl(1+o(1)\bigr)\textstyle\int_{1/2}^1 z^\gamma\,dz\]
and
\begin{align*}
\textstyle\tilde{s}_{n,\gamma}:=\sum_{i=a_n}^{b_n}\bigl(1+\eta(x_i-m_0)\bigr)(x_i-m_0)^\gamma&=nu_n^{\gamma+1}\textstyle\sum_{i=a_n}^{b_n}\inv{M_n}\bigl(1+\eta(u_n z_{n,i})\bigr)z_{n,i}^\gamma\\
&=nu_n^{\gamma+1}\bigl(1+o(1)\bigr)\textstyle\int_{1/2}^1 z^\gamma\,dz;
\end{align*}
see also~\eqref{eq:sgamma},~\eqref{eq:bL1b} and~\eqref{eq:bL1c} in the proof of Lemma~\ref{lem:kinkdist}. Now writing $x_i-\bar{x}_{a_n,b_n}=(x_i-m_0)-(\bar{x}_{a_n,b_n}-m_0)$ for $a_n\leq i\leq b_n$ and $\bar{x}_{a_n,b_n}-m_0=s_{n,1}/s_{n,0}$, we have
\begin{align*}
\ipr{\widetilde{\theta}^{a_n,b_n}}{\mathbf{1}^{a_n,b_n}}&=\textstyle\sum_{i=a_n}^{b_n}B\bigl(1+\eta(x_i-m_0)\bigr)(x_i-m_0)^\alpha=B\tilde{s}_{n\alpha}\\
\ipr{\widetilde{\theta}^{a_n,b_n}}{\widetilde{x}^{a_n,b_n}}&=\textstyle\sum_{i=a_n}^{b_n}B\bigl(1+\eta(x_i-m_0)\bigr)(x_i-\bar{x}_{a_n,b_n})(x_i-m_0)^\alpha=B(\tilde{s}_{n,\alpha+1}-\tilde{s}_{n,\alpha} s_{n,1}/s_{n,0})\\
\norm{\mathbf{1}^{a_n,b_n}}^2&=\textstyle\sum_{i=a_n}^{b_n}1=s_{n,0}\\
\norm{\widetilde{x}^{a_n,b_n}}^2&=\textstyle\sum_{i=a_n}^{b_n}(x_i-\bar{x}_{a_n,b_n})^2=\sum_{i=a_n}^{b_n}\{(x_i-m_0)^2-(\bar{x}_{a_n,b_n}-m_0)^2\}=s_{n,2}-s_{n,1}^2/s_{n,0}.
\end{align*}
Setting $\bar{z}:=3/4$, we note that $s_{n,1}/s_{n,0}=u_n\bigl(1+o(1)\bigr)\bar{z}$. Therefore, the second and third terms in~\eqref{eq:Rn} can be written as
\begin{alignat}{2}
c_{n0}^2\,\norm{\mathbf{1}^{a_n,b_n}}^2&=\frac{\ipr{\widetilde{\theta}^{a_n,b_n}}{\mathbf{1}^{a_n,b_n}}^2}{\norm{\mathbf{1}^{a_n,b_n}}^2}&&=\frac{B^2\tilde{s}_{n,\alpha}^2}{s_{n,0}}=B^2nu_n^{2\alpha+1}\bigl(1+o(1)\bigr)\!\rbr{\frac{\int_{1/2}^1 z^{\alpha}\,dz}{\int_{1/2}^1\,dz}}^2\notag\\
c_{n1}^2\,\norm{\widetilde{x}^{a_n,b_n}}^2&=\frac{\ipr{\widetilde{\theta}^{a_n,b_n}}{\widetilde{x}^{a_n,b_n}}^2}{\norm{\widetilde{x}^{a_n,b_n}}^2}&&=\frac{B^2(\tilde{s}_{n,\alpha+1}-\tilde{s}_{n,\alpha} s_{n,1}/s_{n,0})^2}{s_{n,2}-s_{n,1}^2/s_{n,0}}\notag\\
\label{eq:Rn2}
&&&=B^2nu_n^{2\alpha+1}\bigl(1+o(1)\bigr)\!\rbr{\frac{\int_{1/2}^1\, (z-\bar{z})z^{\alpha}\,dz}{\int_{1/2}^1\,(z-\bar{z})^2\,dz}}^2.
\end{alignat}
Now for $G,\tilde{G}\in L^2[1/2,1]$, let $\ipr{G}{\tilde{G}}_\ast:=\int_{1/2}^1 G(z)\,\tilde{G}(z)\,dz$ and $\norm{G}_\ast^2:=\ipr{G}{G}_\ast$. Moreover, define $G_0,G_1\colon [1/2,1]\to\R$ by $G_0(z):=1$ and $G_1(z):=z-\bar{z}$. These span the (closed) subspace $\mathcal{L}$ of affine functions $G\colon [1/2,1]\to\R$ and satisfy $\ipr{G_0}{G_1}_\ast=0$. Let $c_j^*:=\ipr{F}{G_j}_\ast/\norm{G_j}_\ast^2$ for $j=0,1$, so that $F^*:=c_0^*G_0+c_1^*G_1$ is the projection of $F\colon z\mapsto z^\alpha$ onto $\mathcal{L}$ with respect to $\ipr{\cdot\,}{\cdot}_\ast$. Since $\alpha\neq 1$ in Assumption~\ref{ass:inflection}, we have $F\notin\mathcal{L}$, so
\begin{equation}
\label{eq:Fstar}
\rho_\alpha:=\norm{F}_\ast^2-c_0^*\norm{G_0}_\ast^2-c_1^*\norm{G_1}_\ast^2=\norm{F-c_0^*G_0-c_1^*G_1}_\ast^2=\norm{F-F^*}_\ast^2=\textstyle\int_{1/2}^1\,(F-F^*)^2>0.
\end{equation}
Thus, for all sufficiently large $n$, we can combine~\eqref{eq:Rn},~\eqref{eq:Rn1},~\eqref{eq:Rn2} and~\eqref{eq:Fstar} to conclude that
\begin{align*}
\inf_{(a,b)\in\mathcal{T}_n}\inf_{c_0,c_1}\norm{\theta^{a,b}-c_0\mathbf{1}^{a,b}-c_1 x^{a,b}}^2\geq R_n&=B^2nu_n^{2\alpha+1}\bigl(1+o(1)\bigr)\bigl(\norm{F}_\ast^2-c_0^*\norm{G_0}_\ast^2-c_1^*\norm{G_1}_\ast^2\bigr)\\
&=\rho_\alpha B^2nu_n^{2\alpha+1}\bigl(1+o(1)\bigr).
\end{align*}
Since $u_n=2^{-1}C_n(n/\log n)^{-1/(2\alpha+1)}=2^{-1}C_n\bigl(1+o(1)\bigr)(n/\log n)^{-1/(2\alpha+1)}$ for all $n$, this completes the proof.
\end{proof}

\deparskip
\begin{proof}[Proof of Lemma~\ref{lem:step3}]
Let $h\in\mathcal{G}$ be such that $h(x_i):=\bigl(\theta^{\tilde{k},\tilde{k}'}-\breve{\theta}^{\tilde{k},\tilde{k}'}\bigr)_{\tilde{k}\vee i\wedge\tilde{k}'}$ for $i\in [n]$.
Since $x_{\tilde{k}},x_{\tilde{k}'}$ are successive knots of $\hat{f}_n^{m_0}\in\mathcal{F}^{m_0}$ by assumption, $\breve{\theta}^{\tilde{k},\tilde{k}'}=\bigl(\hat{f}_n^{m_0}(x_{\tilde{k}}),\dotsc,\hat{f}_n^{m_0}(x_{\tilde{k}'})\bigr)$ is an affine sequence. Recalling that $f_0\in\mathcal{F}^{m_0}$ and $\theta^{\tilde{k},\tilde{k}'}=\bigl(f_0(x_{\tilde{k}}),\dotsc,f_0(x_{\tilde{k}'})\bigr)$, we can verify that $\hat{f}_n^{m_0}+\eta h\in\mathcal{H}^{m_0}$ for all sufficiently small $\eta>0$. Defining $\breve{\theta}:=\bigl(\hat{f}_n^{m_0}(x_1),\dotsc,\hat{f}_n^{m_0}(x_n)\bigr)$, $\theta:=\bigl(f_0(x_1),\dotsc,f_0(x_n)\bigr)$ and $Y^{\tilde{k},\tilde{k}'}:=(Y_{\tilde{k}},\dotsc,Y_{\tilde{k}'})$, we deduce from~\eqref{eq:projconv} or Lemma~\ref{lem:projconv} that
\begin{align*}
0&\geq\sum_{i=1}^n h(x_i)(Y_i-\breve{\theta}_i)\\
&=\bigl(\theta_{\tilde{k}}-\breve{\theta}_{\tilde{k}}\bigr)\sum_{i=1}^{\tilde{k}-1}\,(Y_i-\breve{\theta}_i)+\ipr{\theta^{\tilde{k},\tilde{k}'}-\breve{\theta}^{\tilde{k},\tilde{k}'}}{Y^{\tilde{k},\tilde{k}'}-\breve{\theta}^{\tilde{k},\tilde{k}'}}+\bigl(\theta_{\tilde{k}'}-\breve{\theta}_{\tilde{k}'}\bigr)\sum_{i=\tilde{k}'}^n\,(Y_i-\breve{\theta}_i).
\end{align*}
Now since $x_{\tilde{k}},x_{\tilde{k}'}$ are knots of $\hat{f}_n^{m_0}$, it follows from~\eqref{eq:bdadj1} in the proof of Lemma~\ref{lem:bdadj} (specialised to the setting of Proposition~\ref{prop:srestr}) that \[\textstyle\abs{\sum_{i=1}^{\tilde{k}-1}\,(Y_i-\breve{\theta}_i)}\leq\abs{Y_{\tilde{k}}-\breve{\theta}_{\tilde{k}}}\leq\abs{\theta_{\tilde{k}}-\breve{\theta}_{\tilde{k}}}+\abs{\xi_{\tilde{k}}}\;\;\;\text{and}\;\;\;\textstyle\abs{\sum_{i=\tilde{k}'}^n\,(Y_i-\breve{\theta}_i)}\leq\abs{Y_{\tilde{k}'}-\breve{\theta}_{\tilde{k}'}}\leq\abs{\theta_{\tilde{k}'}-\breve{\theta}_{\tilde{k}'}}+\abs{\xi_{\tilde{k}'}}.\]
Therefore, writing $Y^{\tilde{k},\tilde{k}'}=\theta^{\tilde{k},\tilde{k}'}+\xi^{\tilde{k},\tilde{k}'}$, we have
\[\norm{\theta^{\tilde{k},\tilde{k}'}-\breve{\theta}^{\tilde{k},\tilde{k}'}}^2\leq\ipr{\xi^{\tilde{k},\tilde{k}'}}{\breve{\theta}^{\tilde{k},\tilde{k}'}-\theta^{\tilde{k},\tilde{k}'}}+\abs{\theta_{\tilde{k}}-\breve{\theta}_{\tilde{k}}}\,\bigl(\abs{\theta_{\tilde{k}}-\breve{\theta}_{\tilde{k}}}+\abs{\xi_{\tilde{k}}}\bigr)+\abs{\theta_{\tilde{k}'}-\breve{\theta}_{\tilde{k}'}}\,\bigl(\abs{\theta_{\tilde{k}'}-\breve{\theta}_{\tilde{k}'}}+\abs{\xi_{\tilde{k}'}}\bigr).\]
Since $\abs{\theta_{\tilde{k}}-\breve{\theta}_{\tilde{k}}}\vee\abs{\theta_{\tilde{k}'}-\breve{\theta}_{\tilde{k}'}}\leq\norm{\theta^{\tilde{k},\tilde{k}'}-\breve{\theta}^{\tilde{k},\tilde{k}'}}$, this implies that
\[\norm{\theta^{\tilde{k},\tilde{k}'}-\breve{\theta}^{\tilde{k},\tilde{k}'}}\leq\frac{\ipr{\xi^{\tilde{k},\tilde{k}'}}{\breve{\theta}^{\tilde{k},\tilde{k}'}-\theta^{\tilde{k},\tilde{k}'}}}{\norm{\theta^{\tilde{k},\tilde{k}'}-\breve{\theta}^{\tilde{k},\tilde{k}'}}}+\abs{\xi_{\tilde{k}}}+\abs{\xi_{\tilde{k}'}}+\abs{\theta_{\tilde{k}}-\breve{\theta}_{\tilde{k}}}+\abs{\theta_{\tilde{k}'}-\breve{\theta}_{\tilde{k}'}}.\]
Finally, $\breve{\theta}^{\tilde{k},\tilde{k}'}$ is an affine sequence belonging to $A_{\tilde{k},\tilde{k}'}=\Span\{\mathbf{1}^{\tilde{k},\tilde{k}'},x^{\tilde{k},\tilde{k}'}\}$, so $\theta^{\tilde{k},\tilde{k}'}-\breve{\theta}^{\tilde{k},\tilde{k}'}\in L_{\tilde{k},\tilde{k}'}=\Span\{\theta^{\tilde{k},\tilde{k}'},\mathbf{1}^{\tilde{k},\tilde{k}'},x^{\tilde{k},\tilde{k}'}\}$ and it follows as in~\eqref{eq:suplin} that \[\frac{\ipr{\xi^{\tilde{k},\tilde{k}'}}{\breve{\theta}^{\tilde{k},\tilde{k}'}-\theta^{\tilde{k},\tilde{k}'}}}{\norm{\theta^{\tilde{k},\tilde{k}'}-\breve{\theta}^{\tilde{k},\tilde{k}'}}}\leq\norm{\Pi_{\tilde{k},\tilde{k}'}\,\xi^{\tilde{k},\tilde{k}'}}\leq\max_{1\leq a\leq b\leq n}\,\norm{\Pi_{a,b}\,\xi^{a,b}}.\]
Since $\abs{\theta_{\tilde{k}}-\breve{\theta}_{\tilde{k}}}+\abs{\theta_{\tilde{k}'}-\breve{\theta}_{\tilde{k}'}}\leq 2\max_{1\leq i\leq n}\,\abs{\theta_i-\breve{\theta}_i}=2\max_{1\leq i\leq n}\,\abs{(\hat{f}_n^{m_0}-f_0)(x_i)}$, we conclude that \[\norm{\theta^{\tilde{k},\tilde{k}'}-\breve{\theta}^{\tilde{k},\tilde{k}'}}\leq\max_{1\leq a\leq b\leq n}\,\norm{\Pi_{a,b}\,\xi^{a,b}}+2\max_{1\leq i\leq n}\,\abs{(\hat{f}_n^{m_0}-f_0)(x_i)}+2\max_{1\leq i\leq n}\,\abs{\xi_i}=\Xi,\]
as required.
\end{proof}

\deparskip
\begin{proof}[Proof of Lemma~\ref{lem:lamlb}]
For $\delta\in (0,(1-m_0)/2]$ and $t\in [0,(1-m_0)/\delta]$, let \[g_{0,\delta}(t):=\delta^{-\alpha}\bigl(f_0(m_0)+f_0'(m_0)\delta t-f_0(m_0+\delta t)\bigr)\quad\text{and}\quad u_\delta:=\delta^{-(\alpha-1)}\bigl(f_0'(m_0)-u(m_0+\delta)\bigr),\]
where $\alpha>1$. Then $g_{0,\delta}$ is convex and non-negative on $[0,2]$ for each such $\delta$, and since $u(m_0+\delta)$ was taken to be a subgradient of $\restr{f_0}{[m_0,1]}$ at $m_0+\delta$, we see that $u_\delta$ is a subgradient of $g_{0,\delta}$ at $t=1$. 
Assumption~\ref{ass:inflection} ensures that $g_{0,\delta}$ converges uniformly on $[0,2]$ to the function $g_0\colon t\mapsto Bt^\alpha$ as $\delta\to 0$, and so by taking $C=(0,2)$ in~\citet[Lemma~3.10]{SS11}, we deduce further that $u_\delta\to g_0'(1)=B\alpha$ as $\delta\to 0$. 

Moreover, for $\delta\in (0,(1-m_0)/2]$ and $t\in [0,1]$, let 
\begin{align*}
g_{1,\delta}(t)&:=\delta^{-\alpha}\bigl(f_0(m_0)+f_0'(m_0)\delta t-f_{1,\delta}(m_0+\delta t)\bigr)\\
&\phantom{:}=\delta^{-\alpha}\bigl\{\bigl(f_0(m_0)+f_0'(m_0)\delta-f_0(m_0+\delta)\bigr)-\delta(1-t)\bigl(f_0'(m_0)-u(m_0+\delta)\bigr)\bigr\}\\
&\phantom{:}=g_{0,\delta}(1)-(1-t)u_\delta\leq g_{0,\delta}(t)\\[\parskip]
\text{and}\qquad g_{2,\delta}(t)&:=\delta^{-\alpha}\bigl(f_0(m_0)+f_0'(m_0)\delta t-f_{2,\delta}(m_0+\delta t)\bigr)=-\delta t^\alpha,
\end{align*}
where $f_{1,\delta},f_{2,\delta}$ are as defined in the proof of Proposition~\ref{prop:lamlb}. Recalling from~\eqref{eq:fdelta} that $f_\delta=f_{1,\delta}\wedge f_{2,\delta}$ on $[m_0,m_0+\delta]$ by definition, we have 
\begin{equation}
\label{eq:gdelta}
g_\delta(t):=\delta^{-\alpha}\bigl(f_0(m_0)+f_0'(m_0)\delta t-f_\delta(m_0+\delta t)\bigr)=g_{1,\delta}(t)\vee g_{2,\delta}(t)
\end{equation}
for all $t\in [0,1]$. Note that $g_{1,\delta}(0)<g_{0,\delta}(0)=0\leq g_{2,\delta}(0)$ and $g_{1,\delta}(1)=g_{0,\delta}(1)\geq 0>g_{2,\delta}(1)$. Thus, since $g_{1,\delta},g_{2,\delta}$ are continuous functions that are strictly increasing and strictly decreasing respectively, there is a unique $c_\delta\in (0,1)$ satisfying $g_{1,\delta}\leq g_{2,\delta}$ on $[0,c_\delta]$ and $g_{1,\delta}\geq g_{2,\delta}$ on $[c_\delta,1]$; in other words, $f_{1,\delta}\geq f_{2,\delta}$ on $[m_0,m_0+\delta c_\delta]$ and $f_{1,\delta}\leq f_{2,\delta}$ on $[m_0+\delta c_\delta,m_0+\delta]$, so this is consistent with the definition of $c_\delta$ in the proof of Proposition~\ref{prop:lamlb}. Since $g_{0,\delta}(1)-(1-c_\delta)u_\delta=g_{1,\delta}(c_\delta)=g_{2,\delta}(c_{\delta})\in [-\delta,0]$, we have
\[1-u_\delta^{-1}g_{0,\delta}(1)\geq c_\delta\geq 1-u_\delta^{-1}\bigl(g_{0,\delta}(1)+\delta\bigr).\]
In the limit as $\delta\to 0$, it was shown above that $g_{0,\delta}(1)\to g_0(1)=B$ and $u_\delta\to g_0'(1)=B\alpha$, so $c_\delta\to 1-\alpha^{-1}$ and $g_{1,\delta}$ converges uniformly on $[0,1]$ to the affine function $g_1\colon t\mapsto g_0(1)-(1-t)g_0'(1)=B\bigl(1-(1-t)\alpha\bigr)$. Consequently, $g_{\delta}=g_{1,\delta}\vee g_{2,\delta}\to g_1\vee 0\equiv g_1^+$ uniformly on $[0,1]$ as $\delta\to 0$.

Now let $(\delta_n)$ be any sequence such that $\delta_n\to 0$ and $n\delta_n\to 0$ as $n\to\infty$. Having already shown that $\lim_{n\to\infty}c_{\delta_n}=1-\alpha^{-1}$, we proceed to establish the claimed limiting expression for $\norm{f_{\delta_n}-f_0}_n^2$. For each $n$, let $z_{n,i}:=(x_i-m_0)/\delta_n$ for $i\in [n]$, where $x_i\equiv x_{ni}=i/n$ by Assumption~\ref{ass:inflection}, so that $z_{n,i+1}-z_{n,i}=1/(n\delta_n)$ for all $i\in [n-1]$. Then recalling from~\eqref{eq:fdelta} that $f_{\delta_n}=f_0$ on $[0,1]\setminus (m_0,m_0+\delta_n)$, we can use~\eqref{eq:gdelta} to write
\begin{equation}
\label{eq:fdeltanorm}
\norm{f_{\delta_n}-f_0}_n^2=\frac{1}{n}\sum_{i\,:\, m_0<x_i<m_0+\delta_n}(f_{\delta_n}-f_0)^2(x_i)=\delta_n^{2\alpha+1}\,\frac{1}{n\delta_n}\sum_{i\,:\,0<z_{n,i}<1}(g_{\delta_n}-g_{0,\delta_n})^2(z_{n,i})
\end{equation}
for each $n$. Since $\abs{\{i:0<z_{n,i}<1\}}=O(n\delta_n)$ and $(g_{\delta_n}-g_{0,\delta_n})^2\to (g_1^+ -g_0)^2$ uniformly on $[0,1]$ as $n\to\infty$, we have
\begin{align*}
&\frac{1}{n\delta_n}\sum_{i\,:\,0<z_{n,i}<1}\abs{(g_{\delta_n}-g_{0,\delta_n})^2(z_{n,i})-(g_1^+ -g_0)^2(z_{n,i})}\\
&\hspace{2cm}\leq\frac{\abs{\{i:0<z_{n,i}<1\}}}{n\delta_n}\sup_{x\in [0,1]}\bigl|(g_{\delta_n}-g_{0,\delta_n})^2-(g_1^+ -g_0)^2\bigr|=o(1).
\end{align*}
Thus, by a Riemann sum approximation to the (uniformly) continuous function $(g_1^+ -g_0)^2$ on $[0,1]$ and the fact that $n\delta_n\to\infty$,
we see that
\begin{align*}
\frac{1}{n\delta_n}\sum_{i\,:\,0<z_{n,i}<1}(g_{\delta_n}-g_{0,\delta_n})^2(z_{n,i})&=\frac{1}{n\delta_n}\sum_{i\,:\,0<z_{n,i}<1}(g_1^+ -g_0)^2(z_{n,i})+o(1)\\
&=\int_0^1\,(g_1^+ -g_0)^2+o(1)\\
&=B^2\int_0^1\,\bigl\{t^\alpha-\bigl(1-(1-t)\alpha\bigr)^+\bigr\}^2\,dt+o(1)\\
&=\bigl(1+o(1)\bigr)C_\alpha B^2
\end{align*}
as $n\to\infty$, where $C_\alpha:=\int_0^1\,\bigl\{t^\alpha-\bigl(1-(1-t)\alpha\bigr)^+\bigr\}^2\,dt$. Combining this with~\eqref{eq:fdeltanorm} yields the desired conclusion.
\end{proof}

\umparskip
\section{Proofs for Section~\ref{sec:proj2}}
\label{sec:projproofs}
\begin{proof}[Proof of Proposition~\ref{prop:clconv}]
(a) Since every $g\in\mathcal{F}^m$ is bounded on $[0,1]$, it follows that $\mathcal{F}^m\subseteq L^2(\nu)$. If $g,h\in\mathcal{F}^m$, then certainly $\lambda g\in\mathcal{F}^m$ for all $\lambda>0$ and $\lambda g+(1-\lambda)h\in\mathcal{F}^m$ for all $\lambda\in [0,1]$. Now fix $g\in\mathcal{F}^m$, and for each $n\in\N$, let $g_n\colon [0,1]\to\R$ be the Lipschitz function that agrees with $g$ on $[0,m(1-1/n)]\cup\{m\}\cup [m(1-1/n)+1/n,1]$, and is also linear on both $[m(1-1/n),m]$ and $[m,m(1-1/n)+1/n]$. Then $g_n\in\mathcal{F}^m$ and $g(0)\leq g_n\leq g(1)$ for all $n$, and since $g_n\to g$ pointwise, we have $\norm{g_n-g}_{L^2(\nu)}\to 0$ by the dominated convergence theorem.

(b) Let $M_\nu^L:=\min(\supp\nu)$ and $M_\nu^R:=\max(\supp\nu)$, so that $\csupp\nu=[M_\nu^L,M_\nu^R]$. There is nothing to prove when $m\in\csupp\nu$, so suppose now that $m>M_\nu^R$. Then $\tilde{m}=M_\nu^R$, and note that $g\colon [0,\tilde{m}]\to\R$ is increasing, convex and Lipschitz if and only if there exists a Lipschitz $f\in\mathcal{F}^m$ such that $g=\restr{f}{[0,\tilde{m}]}$. Thus, $\mathcal{F}_\nu^m=\{[f]_\nu:f\in\mathcal{F}^{\tilde{m}}\text{ is Lipschitz}\}$ is dense in $\mathcal{F}_\nu^{\tilde{m}}$ by (a). The case $m<M_\nu^L$ is similar.

(c,\,$\Leftarrow$) We first show that if $(f_n)_{n=1}^\infty$ is a sequence of functions in $\mathcal{F}^m$ such that $\norm{f_n-f}_{L^2(\nu)}\to 0$ for some $f\in L^2(\nu)$, then under any one of the conditions (i)--(iii) above, there exists $g\in\mathcal{F}^m$ such that $f\sim_\nu g$. To begin with, note that $f_n\to f$ in $\nu$-measure, so there exists a subsequence $(g_k)_{k=1}^\infty\equiv (f_{n_k})_{k=1}^\infty$ such that $g_k\to f$ $\nu$-almost everywhere. In each of the cases below, we will in fact show that there is some $g\in\mathcal{F}^m$ that agrees with $f$ on $A:=\{x\in\supp \nu:g_k(x)\to f(x)\}$, which is a dense subset of $\supp\nu$. Indeed, if $S\subseteq\supp\nu$ satisfies $\nu(S^c\cap\supp\nu)=0$, then by the definition of $\supp\nu$, the set $S^c\cap\supp\nu$ has empty interior; in other words, $S$ is dense in $\supp\nu$.

\textbf{Case 1 -- $\nu([0,m))\wedge\nu((m,1])>0$}: Since $A$ is dense in $\supp\nu$, there exist $a_L,a_R\in A\subseteq\supp\nu$ such that $a_L<m<a_R$ and $g_k\to f$ on $\{a_L,a_R\}$. Since the functions $g_k$ are convex on $[0,m]$, concave on $[m,1]$ and increasing on $[0,1]$, we have
\begin{align}
\liminf_{k\to\infty}g_k(0)&\geq\liminf_{k\to\infty}\,\frac{m\,g_k(a_L)-a_L\,g_k(m)}{m-a_L}\geq\frac{m\,f(a_L)-a_L\,f(a_R)}{m-a_L}\notag\\
\label{eq:lbconv}
\limsup_{k\to\infty}g_k(1)&\leq\limsup_{k\to\infty}\,\frac{(1-m)\,g_k(a_R)-(1-a_R)\,g_k(m)}{a_R-m}\leq\frac{(1-m)\,f(a_R)-(1-a_R)\,f(a_L)}{a_R-m},
\end{align}
so $\{g_k(x)\}_{k=1}^\infty$ is bounded for each $x\in [0,1]$. Therefore, by considering separately the intervals $(0,m),(m,1)$, we can apply \citet[Theorem~10.9]{Rock97} and extract a subsequence $(g_{k_\ell})$ of $(g_k)$ that converges pointwise on $(0,m)\cup (m,1)$. In fact, $(g_{k_\ell})$ converges pointwise on $[0,1]\setminus\{m\}$ by Lemma~\ref{lem:convincr}, and the limit function $g$ is convex on $[0,m)$, concave on $(m,1]$ and increasing on $[0,1]\setminus\{m\}$. 
If in addition $m\in A$, then \[g(z)=\lim_{\ell\to\infty}g_{k_\ell}(z)\leq \lim_{\ell\to\infty}g_{k_\ell}(m)=f(m)\leq\lim_{\ell\to\infty}g_{k_\ell}(w)=g(w)\]
for all $z\in (0,m)$ and $w\in (m,1)$, so $\lim_{x\nearrow\,m}g(x)\leq f(m)\leq\lim_{x\searrow\,m}g(x)$. Thus, we can extend $g$ to a function on $[0,1]$ that belongs to $\mathcal{F}^m$ by setting $g(m)=f(m)$. Otherwise, if $m\notin A$, then we can set $g(m)=\lim_{x\searrow\,m}g(x)$ for concreteness. In both cases, we have $g\in\mathcal{F}^m$ and $f=g$ on $A$, as required.

\textbf{Case 2 -- $\nu((m,1])=0$}: Here, we have $\supp\nu\subseteq [0,m]$. We also assume that $\supp\nu$ contains at least two points, since otherwise the result holds trivially. Note that $\conv(\Cl A)\supseteq\Int\csupp\nu=(M_\nu^L,M_\nu^R)$. By convexity arguments similar to those given in Case 1, it follows that $\{g_k(x)\}_{k=1}^\infty$ is bounded for all $x\in (M_\nu^L,M_\nu^R)$. 
Thus, again by \citet[Theorem~10.9]{Rock97} and Lemma~\ref{lem:convincr}, there exists a subsequence $(g_{k_\ell})$ of $(g_k)$ that converges pointwise on $[0,M_\nu^R)$ to some increasing convex function $g\colon [0,M_\nu^R)\to\R$. By the definition of $A$, we must have $f=g$ on $[0,M_\nu^R)\cap A$.

\unparskip
\begin{itemize}[leftmargin=0.4cm]
\item If condition (i) holds, then $M_\nu^R=m$ and $\nu(\{m\})>0$, so $m\in A$, i.e.\ $g_k(m)\to f(m)$. We now extend $g$ to $[0,1]$ by setting $g(x)=f(m)$ for all $x\in [m,1]$. Then $f=g$ on $A$, and for all $x\in [0,m)$, we have $g(x)=\lim_{\ell\to\infty}g_{k_\ell}(x)\leq\lim_{\ell\to\infty}g_{k_\ell}(m)=f(m)$, so $g\in\mathcal{F}^m$.
\item If condition (ii) holds, then $M_\nu^R\in (0,m)$ is an isolated point of $\supp\nu$, so $\nu(\{M_\nu^R\})>0$, $M_\nu^R\in A$ and $M_\nu':=\max(\supp\nu\setminus\{M_\nu^R\})<M_\nu^R<m$. Let $h\colon [0,1]\to\R$ be the function that agrees with $f$ on $[0,M_\nu']\cup\{M_\nu^R\}$ and is linear on $[M_\nu',1]$.
Then $h$ is linear on $[m,1]$ and convex and increasing on $[0,1]$, so $h\in\mathcal{F}^m$. Since $(M_\nu',1]\cap A=\{M_\nu^R\}$, the functions $f,g,h$ agree on $A$, as required.
\end{itemize}

\unparskip
The analogous case where $\supp\nu\subseteq [m,1]$ can be handled in much the same way, and so we have now demonstrated the sufficiency of each of the conditions (i), (ii) and (iii). 

(c,\,$\Rightarrow$) Supposing that none of the conditions (i)--(iii) hold, we now verify that $\Cl\mathcal{F}_\nu^m\supsetneqq\mathcal{F}_\nu^m$. We consider only the cases where $\supp\nu\subseteq [0,m]$; the arguments are similar if $\supp\nu\subseteq [m,1]$. 

\textbf{Case 1 -- $\nu(\{M_\nu^R\})=0$}: Note that there exists $f\in L^2(\nu)$ such that $\restr{f}{E_\nu}$ is convex and increasing, and $f(x)\to\infty$ as $x\nearrow M_\nu^R$. Indeed, a concrete example of such a function can be obtained via the following construction: since $M_\nu^R$ is not an isolated point of $\supp\nu$ by assumption, there exists a sequence $(a_n\in\supp\nu\setminus\{M_\nu^R\}:n\in\N)$ such that $a_n\nearrow M_\nu^R$ and $\nu((a_n,M_\nu^R))\leq 2^{-3n}$ for all $n$. For each $n$, let $h_n\colon[0,1]\to\R$ be such that $h_n=0$ on $[0,a_n]$, $h_n(M_\nu^R)=2^{n/2}$ and $h_n$ is linear on $[a_n,1]$. Then $h_n$ is convex and increasing, and $\norm{h_n}_{L^2(\nu)}\leq (2^n\cdot 2^{-3n})^{1/2}=2^{-n}$. Thus, the function $h\colon [0,1]\to\R$ defined by $h(x):=\sum_{n=1}^\infty h_n(x)\Ind_{\{x<M_\nu^R\}}$ is also convex and increasing. Moreover, $\norm{h}_{L^2(\nu)}\leq\sum_{n=1}^\infty\norm{h_n}_{L^2(\nu)}<\infty$ and $h(x)\to\infty$ as $x\nearrow M_\nu^R$. 

For any $f$ with the above properties, we now argue that $[f]_\nu\notin\mathcal{F}_\nu^m$, i.e.\ that there does not exist $g\in\mathcal{F}^m$ such that $f\sim_\nu g$. Indeed, if $g$ is a function that agrees with $f$ on a set $S\subseteq\supp\nu$ with the property that $\nu(S^c\cap\supp\nu)=0$, then recall from the second paragraph of the proof that $S$ is dense in $\supp\nu$. Thus, since $M_\nu^R$ is not an isolated point of $\supp\nu$, there exists a sequence $(s_n\in S\setminus\{M_\nu^R\}:n\in\N)$ such that $s_n\to M_\nu^R$, and we must have $g(s_n)=f(s_n)$ for all $n$. But this implies that $g(x)\to\infty$ as $x\nearrow M_\nu^R$, so $g$ cannot be extended to a finite convex function on $[0,m]$. 

\textbf{Case 2 -- $\nu(\{M_\nu^R\})>0$}: Consider any $f\in L^2(\nu)$ such that $\restr{f}{E_\nu}$ is convex and increasing, and $f$ is discontinuous at $M_\nu^R$. Since $M_\nu^R\in (0,m)$ is not an isolated point of $\supp\nu$, we deduce as before that if $f\sim_\nu g$ for some $g\colon [0,1]\to\R$, then there exists a sequence $(s_n\in\supp\nu\setminus\{M_\nu^R\}:n\in\N)$ such that $s_n\to M_\nu^R$ and $g(s_n)=f(s_n)$ for all $n$. But since $\nu(\{M_\nu^R\})>0$, we have $g(M_\nu^R)=f(M_\nu^R)\neq\lim_{n\to\infty}f(s_n)=\lim_{n\to\infty}g(s_n)$, so $g$ is not continuous at $M_\nu^R\in (0,m)$ and hence does not belong to $\mathcal{F}^m$. This shows that $[f]_\nu\notin\mathcal{F}_\nu^m$ and hence that $\Cl\mathcal{F}_\nu^m\supsetneqq\mathcal{F}_\nu^m$, as required.

(d) Suppose again that none of the conditions (i)--(iii) hold, assuming for the time being that $\supp\nu\subseteq [0,m]$. Then $m>0$ and $M_\nu^R\in [0,m]$ is not an isolated point of $\supp\nu$. Let $(f_n)_{n=1}^\infty$ be any sequence in $\mathcal{F}^m$ such that $\norm{f_n-f}_{L^2(\nu)}\to 0$ for some $f\in L^2(\nu)$. By a very similar argument to that given in the first bullet point in Case 2 of (c,\,$\Leftarrow$), there exists an increasing convex $g$ defined on $E_\nu$ that agrees $\nu$-almost everywhere with $f$; recall that $E_\nu$ contains $M_\nu^R$ if and only if $\nu(\{M_\nu^R\})>0$. Since $\mathcal{F}_\nu^m=\{[f]_\nu:f\in\mathcal{F}^m\}\subseteq\mathcal{L}^2(\nu)$ by (a), we deduce that \[\Cl\mathcal{F}_\nu^m\subseteq\{[f]_{\nu}:f\in L^2(\nu)\text{ and}\restr{f}{E_\nu}\!\text{ is convex and increasing}\}.\]
For the reverse inclusion, we split into the two cases considered in (c,\,$\Rightarrow$) above.

\textbf{Case 1 -- $\nu(\{M_\nu^R\})=0$}: Here, we have $M_\nu^R\notin E_\nu$. For a fixed $f\in L^2(\nu)$ such that $\restr{f}{E_\nu}$ is convex and increasing, we claim that there exists a sequence $(f_n)_{n=1}^\infty$ in $\mathcal{F}^m$ such that $\norm{f_n-f}_{L^2(\nu)}\to 0$. Indeed, fix a sequence $(x_n\in\supp\nu\setminus\{M_\nu^R\}:n\in\N)$ such that $x_n\nearrow M_\nu^R$, and for each $n$, observe that since $\restr{f}{E_\nu}$ has a finite and non-negative subgradient at $x_n$, there exists an increasing convex $f_n\in\mathcal{F}^m$ such that $f_n=f$ on $E_\nu\cap [0,x_n]$, $f_n$ is linear on $[x_n,1]$ and $f_n\leq f$ on $E_\nu$. Thus, since $\nu(E_\nu)=1$, $\inf_{E_\nu}f\leq f_n\leq\sup_{E_\nu}f$ on $E_\nu$ for all $n$ and $f_n\to f$ pointwise on $E_\nu$, it follows by the dominated convergence theorem that $\norm{f_n-f}_{L^2(\nu)}^2=\int_{E_\nu}\,\abs{f_n-f}^2\to 0$, as required.

\textbf{Case 2 -- $\nu(\{M_\nu^R\})>0$}: Note that $M_\nu^R<m$ and $M_\nu^R\in E_\nu$ in this case. As before, take any $f\in L^2(\nu)$ such that $\restr{f}{E_\nu}$ is convex and increasing, and fix a sequence $(x_n\in\supp\nu\setminus\{M_\nu^R\}:n\in\N)$ such that $x_n\nearrow M_\nu^R$. For each $n$, let $f_n$ be the (unique) function that satisfies $f_n=f$ on $[0,x_n]\cup\{M_\nu^R\}$ and is linear on $[x_n,1]$. Then $f_n\in\mathcal{F}^m$ for all $n$ by the convexity of $f$, and $f_n\to f$ pointwise on $E_\nu$. Thus, since $\inf_{E_\nu}f\leq f_n\leq f(M_\nu^R)<\infty$ on $E_\nu$ for all $n$, we can once again apply the dominated convergence theorem to deduce that $\norm{f_n-f}_{L^2(\nu)}^2=\int_{E_\nu}\,\abs{f_n-f}^2\to 0$. This shows that $\Cl\mathcal{F}_\nu^m\supseteq\{[f]_{\nu}:f\in L^2(\nu)\text{ and}\restr{f}{E_\nu}\!\text{ is convex and increasing}\}$ in this case.

Straightforward modifications of the arguments above yield the analogous conclusion when $\supp\nu\subseteq [m,1]$. This completes the proof.
\end{proof}

\deparskip
\begin{proof}[Proof of Corollary~\ref{cor:proj}]
(a) For $f\in L^2(P^X)$, it is immediate from~\eqref{eq:orth} that $f\in\psi_m^*(P)$ if and only if $[f]_{P^X}\in\mathcal{L}^2(P^X)$ is the projection of $[f_P]_{P^X}$ onto $\Cl\mathcal{F}_{P^X}^m$, which is a closed, convex subset of the Hilbert space $\mathcal{L}^2(P^X)$ by Proposition~\ref{prop:clconv}.

(b) This follows directly from the definition of $\psi_m^*(P)$ and the fact that every $f\in\mathcal{F}^m$ is bounded on $[0,1]$.

(c) Since $\Cl\mathcal{F}_{P^X}^m=\Cl\mathcal{F}_{P^X}^{\tilde{m}}$ by Proposition~\ref{prop:clconv}(b), we have $\psi_m^*(P)=\psi_{\tilde{m}}^*(P)$ by definition, and $L_m^*(P)=L_{\tilde{m}}^*(P)$ by the observation after~\eqref{eq:orth}. If there exists $f\in\psi_m^0(P)$, then setting $\tilde{f}(x):=f(M_L\vee x\wedge M_R)$ for $x\in [0,1]$ with $M_L:=\min(\supp P^X)$ and $M_R:=\max(\supp P^X)$, we have $\tilde{f}\in\psi_{\tilde{m}}^0(P)$.

(d) If condition (i) holds, then $\mathcal{F}_{P^X}^m=\Cl\mathcal{F}_{P^X}^m$ by Proposition~\ref{prop:clconv}(c). Thus, there exists $f^*\in\mathcal{F}^m$ such that $\psi^*(P)=[f^*]_{P^X}$ by (a) above, whence $f^*\in\psi_m^0(P)$.

Suppose now that condition (ii) holds, in which case $m=\tilde{m}$ and there exist a regression function $f_P$ for $P$ and $b,\varepsilon>0$ such that $\abs{f_P}\leq b$ on $(m-\varepsilon,m+\varepsilon)$. We may assume without loss of generality that $m=\max(\supp P^X)$; the case $m=\min(\supp P^X)$ is similar. Suppose for a contradiction that $\psi_m^0(P)=\emptyset$, i.e.\ that $\psi^*(P)\in\bigl(\Cl\mathcal{F}_{P^X}^m\bigr)\setminus\mathcal{F}_{P^X}^m$. In view of condition (i) in Proposition~\ref{prop:clconv}(c), this can only happen if $P^X(\{m\})=0$. By Proposition~\ref{prop:clconv}(d), we can write $\psi_m^*(P)=[f]_{P^X}$ for some $f\in L^2(P^X)$ that is convex and increasing on ${\Int(\csupp P^X)}$. Since $[f]_{P^X}=\psi_m^*(P)\notin\mathcal{F}_{P^X}^m$, the function $\restr{f}{\Int(\csupp P^X)}$ cannot be extended to an element of $\mathcal{F}^m$, so we must have $f(x)\to\infty$ as $x\nearrow m$.

Therefore, we can find $m'\in\Int(\csupp P^X)\cap (m-\varepsilon,m)$ such that $f(x)>b$ for all $x\in [m',m)$. Since 
$\restr{f}{\Int(\csupp P^X)}$ has a finite and non-negative subgradient at $m'$, there exists an increasing convex $\tilde{f}\in\mathcal{F}^m$ such that $\tilde{f}=f$ on $\Int(\csupp P^X)\cap [0,m']$, $\tilde{f}$ is linear on $[m',1]$ and $f_P\leq b<\tilde{f}\leq f$ on $[m',m)$. But this means that $\norm{\tilde{f}-f_P}_{L^2(P^X)}<\norm{f-f_P}_{L^2(P^X)}$, so $L_m^*(P)\leq L(\tilde{f},P)<L(f,P)$ by~\eqref{eq:orth}, contradicting the fact that $f\in\psi_m^*(P)$. Thus, $\psi_m^*(P)\in\mathcal{F}_{P^X}^m$, whence $\psi_m^0(P)\neq\emptyset$.

(e) If $f,g\in\psi_m^0(P)$, then $f\sim_{P^X}g$ by (a), so $f=g$ on some dense subset $S\subseteq\supp P^X$; see the first paragraph of the proof of Proposition~\ref{prop:clconv}(c). 
It follows that $f=g$ on $(\supp P^X)\setminus\{m\}$, a set on which both $f,g$ are continuous. If $f,g$ are both continuous on $[0,1]$, then $f=g$ on $\supp P^X$ by the same argument. If in addition $P^X(\{m\})>0$, then clearly $f(m)=g(m)$.

(f) The forward implication was established in (e). For the converse, suppose that $m\in\supp P^X$ and there is some $f\in\psi_m^0(P)$ that is discontinuous at $m$, so that $\lim_{x\nearrow m}f(x)<\lim_{x\searrow m}f(x)$. If $P^X(\{m\})=0$, then any $\psi_m^0(P)$ contains any $\tilde{f}\colon [0,1]\to\R$ such that $\tilde{f}=f$ on $[0,1]\setminus\{m\}$ and $\lim_{x\nearrow m}f(x)\leq\tilde{f}(m)\leq\lim_{x\searrow m}f(x)$, so the elements of $\psi_m^0(P)$ do not all agree at $m\in\supp P^X$.

(g) Suppose that $\psi_m^0(P)$ contains a continuous function $h\in\mathcal{F}^m$. For any other $f\in\psi_m^0(P)$, we know from (e) that $f=h$ on $(\supp P^X)\setminus\{m\}$. In view of the continuity of $h$ at $m$ and the assumption that $\supp P^X$ has non-empty intersection with both $(m-\varepsilon,m)$ and $(m,m+\varepsilon)$ for all $\varepsilon>0$, this forces $f(m)=h(m)$, so $f=h$ on $\supp P^X$ and $f$ is continuous.
\end{proof}

\unparskip
The proof of Proposition~\ref{prop:cont} relies on the following three key lemmas. Let the marginal distribution $P^X$ on $[0,1]$ be as in Proposition~\ref{prop:cont}, and define $M_L:=\min(\supp P^X)$, $M_R:=\max(\supp P^X)$ and $C:=[M_L,M_R]=\csupp P^X$.
\begin{lemma}
\label{lem:limsuptw}
Fix $x\in\Int C$ and $\ell\in [0,\infty)$. Let $(P_n)$ be a sequence in $\mathcal{P}$ that converges weakly to some $P\in\mathcal{P}$. Then there exists $B\equiv B(x,\ell,P)<\infty$ such that for any sequence of increasing functions $f_n\colon [0,1]\to\R$ with $\limsup_{n\to\infty}L(f_n,P_n)\leq\ell$ for all $n$, we have $\limsup_{n\to\infty}\,\abs{f_n(x)}<B$.
\end{lemma}

\deparskip
\begin{proof}
Since $x\in\Int C$, we have $P^X([0,x))\wedge P^X((x,1])>0$. Let $(f_n)$ be as above. Then for each $n$, note that since $f_n$ is increasing, we have
\begin{equation}
\label{eq:intsup}
L(f_n,P_n)
\geq\int_{[x,1]\times\R}\,\bigl(y-f_n(x)\bigr)^2\,dP_n(x,y)\geq\rbr{\frac{f_n^+(x)}{2}}^2 P_n\rbr{[x,1]\times\left(-\infty,\frac{f_n^+(x)}{2}\right]};
\end{equation}
indeed, if $f_n(x)\leq 0$, then~\eqref{eq:intsup} holds trivially, and if $f_n(x)>0$, then for all $x'\in [x,1]$ and $y'\leq f_n(x)/2$, we have $f_n(x')-y'\geq f_n(x)-y'\geq f_n(x)/2>0$. 

Now let $(f_{n_k})$ be a subsequence such that $f_{n_k}^+(x)\to\limsup_{n\to\infty}f_n^+(x)=:2s$ as $k\to\infty$. Since $P_n\cvd P$, an application of the portmanteau lemma~\citep[Lemma~2.2]{vdV98} shows that
\[\liminf_{k\to\infty}\,\rbr{\frac{f_{n_k}^+(x)}{2}}^2\,P_{n_k}\!\rbr{(x,1]\times\rbr{-\infty,\frac{f_{n_k}^+(x)}{2}}}\geq s^2\,P\bigl((x,1]\times(-\infty,s)\bigr).\]
It follows from this and~\eqref{eq:intsup} that $s^2\,P\bigl((x,1]\times(-\infty,s)\bigr)\leq\limsup_{n\to\infty}L(f_n,P_n)\leq\ell$.
Since $P\bigl((x,1]\times (-\infty,b)\bigr)\to P^X((x,1])>0$ as $b\to\infty$, we can therefore find $B'\equiv B'(x,\ell,P)<\infty$ such that $\limsup_{n\to\infty}\,f_n^+(x)=2s<B'$ for any sequence $(f_n)$ satisfying the conditions of Lemma~\ref{lem:limsuptw}. An analogous argument yields the same conclusion for $\limsup_{n\to\infty}\,f_n^-(x)$.
\end{proof}

\unparskip
For sequences of S-shaped functions, the conclusion of Lemma~\ref{lem:limsuptw} can be strengthened.
\begin{lemma}
\label{lem:limsupunif}
Fix $\ell\in [0,\infty)$ and let $(m_n)$ be a sequence in $[0,1]$ that converges to some fixed $m_0\in\Int C$. Then under the hypotheses of Lemma~\ref{lem:limsuptw}, there exists $\tilde{B}\equiv\tilde{B}(m_0,\ell,P)<\infty$ such that for any sequence $(f_n)$ with $f_n\in\mathcal{F}^{m_n}$ for all $n$ and $\limsup_{n\to\infty}L(f_n,P_n)\leq\ell$, we have\\ $\limsup_{n\to\infty}\sup_{x\in [0,1]}\,\abs{f_n(x)}<\tilde{B}$.
\end{lemma}

\deparskip
\begin{proof}
Since $m_n\to m_0\in\Int C$, we can find $a_L,a_R\in\Int C$ such that $0<a_L<m_0<a_R<1$, so that for all sufficiently large $n\in\N$, we have $f_n(a_L)\leq f_n(m_n)\leq f_n(a_R)$ and
\begin{equation}
\label{eq:limsupunif}
\frac{m_nf_n(a_L)-a_Lf_n(m_n)}{m_n-a_L}\leq f_n(0)\leq f_n(x)\leq f_n(1)\leq\frac{(1-m_n)f_n(m_n)-(1-a_R)f_n(a_L)}{a_R-m_n}
\end{equation}
for all $x\in [0,1]$. Since $\limsup_{n\to\infty}L(f_n,P_n)\leq\ell$, we have $\limsup_{n\to\infty}\,\abs{f_n(x)}\leq B(x,\ell,P)<\infty$ for $x\in\{a_L,a_R\}$ by Lemma~\ref{lem:limsuptw}, so the result follows from~\eqref{eq:limsupunif} and the fact that $m_n\to m_0$.
\end{proof}

\unparskip
Another important consequence of Lemma~\ref{lem:limsuptw} is the following. Recall that $\tilde{m}_0=\argmin_{x\in C}\,\abs{x-m_0}$. 
\begin{lemma}
\label{lem:subseqcvg}
Under the hypotheses of Proposition~\ref{prop:cont}, let $(f_n)_{n=1}^\infty$ be any sequence with $f_n\in\mathcal{F}^{m_n}$ and $\limsup_{n\to\infty}L(f_n,P_n)<\infty$ for all $n$. Then for every subsequence $(g_k)_{k=1}^\infty\equiv (f_{n_k})_{k=1}^\infty$, there is a further subsequence $(g_{k_\ell})$ and a function $g$ defined on $[0,1]$ with the following properties:

\unparskip
\begin{enumerate}[label=(\roman*)]
\item $g$ is increasing on $C\setminus\{\tilde{m}_0\}$, convex on $C\cap [0,\tilde{m}_0)$ and concave on $C\cap(\tilde{m}_0,1]$. 
\item $g_{k_\ell}\to g$ pointwise on $C$ and uniformly on closed subsets of $C\setminus\{\tilde{m}_0\}$. In particular, if $z_\ell\to z\in C\setminus\{\tilde{m}_0\}$, then $g_{k_\ell}(z_\ell)\to g(z)$.
\item $g(x)\in\R$ for all $x\in C\setminus\{\tilde{m}_0\}$, and $g(\tilde{m}_0)\in (-\infty,\infty]$ if $\tilde{m}_0>M_L$ and $g(\tilde{m}_0)\in [-\infty,\infty)$ if $\tilde{m}_0<M_R$.
\item $g\in\mathcal{F}^{m_0}$ if $m_0\in\Int C$ and $[g]_{P^X}\in\Cl\mathcal{F}_{P^X}^{m_0}\subseteq\mathcal{L}^2(P^X)$ if $P^X(\{\tilde{m}_0\})=0$. 
\item Let $Q_\ell:=P_{n_{k_\ell}}$ for each $\ell$. If $P^X(\{\tilde{m}_0\})=0$, then $\liminf_{\ell\to\infty}L(g_{k_\ell},Q_\ell)\geq L(g,P)\geq L_{m_0}^*(P)$.
\item If $P_n=P$ for all $n$, then we can ensure that the conclusions of (iii) and (iv) hold in all cases, even when the assumptions on $\tilde{m}_0$ are dropped.
\end{enumerate}

\unparskip
\end{lemma}

\deparskip
\begin{proof}
Fix a subsequence $(g_k)_{k=1}^\infty\equiv (f_{n_k})_{k=1}^\infty$. In the setting of Proposition~\ref{prop:cont}, we have $m_{n_k}\to m_0$, so if $0\leq z<m_0<w\leq 1$, then all but finitely many of the functions $g_k$ are convex on $[0,z]$ and concave on $[w,1]$. Since $\{g_k(x)\}_{k=1}^\infty$ is bounded for all $x\in\Int C$ by Lemma~\ref{lem:limsuptw}, it follows from~\citet[Theorem~10.9]{Rock97} and~Lemma~\ref{lem:convincr} that whenever $0\leq z<m_0<w\leq 1$, there is a subsequence of $(g_k)$ that converges pointwise on $C\setminus (z,w)$. By considering sequences $z_n\nearrow m_0$ and $w_n\searrow m_0$ with $z_n<m_0<w_n$ for all $n$, we deduce by a diagonal argument
that $(g_k)$ has a subsequence that converges pointwise to some $g\colon C\setminus\{\tilde{m}_0\}\to\R$ on $C\setminus\{\tilde{m}_0\}$ and uniformly on closed subsets of $C\setminus\{\tilde{m}_0\}$. 

To extend $g$ to $[0,1]\setminus\{\tilde{m}_0\}$, we set $g(x)=g(M_L)$ for all $x\in [0,M_L)$ and $g(x)=g(M_R)$ for all $x\in (M_R,1]$ if $m_0\in\Int C$, and otherwise set $g(x)=0$ for all $x\in [0,1]\setminus C$ if $m_0\notin\Int C$. Finally, by extracting a further subsequence $(g_{k_\ell})$ if necessary, we can ensure that $g_{k_\ell}(\tilde{m}_0)$ converges to some $L\in [-\infty,\infty]$ as $\ell\to\infty$, and we extend $g$ to $[0,1]$ by setting $g(\tilde{m}_0)=L$. 

(i) This follows from the construction of $g$ in the first paragraph. Note also that if $m_0\notin\Int C$, then $g$ is increasing on $C$, and $g$ is either concave or convex on $C$,
depending on whether $\tilde{m}_0=M_L$ or $\tilde{m}_0=M_R$ respectively.

(ii) By the continuity of $g$ on $C\setminus\{\tilde{m}_0\}$ and the uniform convergence established above, we deduce that if $z_\ell\to z\in C\setminus\{\tilde{m}_0\}$, then 
\begin{equation}
\label{eq:unifc}
\abs{g_{k_\ell}(z_\ell)-g(z)}\leq\abs{g_{k_\ell}(z_\ell)-g(z_\ell)}+\abs{g(z_\ell)-g(z)}\to 0.
\end{equation}
(iii) If $\tilde{m}_0>M_L$, then $L=\lim_{\ell\to\infty}g_{k_\ell}(\tilde{m}_0)\geq\lim_{x\nearrow\,\tilde{m}_0}\,\lim_{\ell\to\infty}g_{k_\ell}(x)=\lim_{x\nearrow\,\tilde{m}_0}g(x)$, so $L\in (-\infty,\infty]$. Similarly if $\tilde{m}_0<M_R$, then $L\leq\lim_{x\searrow\,\tilde{m}_0}g(x)$, whence $L\in [-\infty,\infty)$.

(iv) If $m_0\in\Int C$, then $g\in\mathcal{F}^{m_0}$ by construction.
Suppose now that $P^X(\{\tilde{m}_0\})=0$. Since $Q_\ell\cvd P$ under the hypotheses of Proposition~\ref{prop:cont}, Skorokhod's representation theorem~\citep[e.g.][Theorem~2.19]{vdV98} guarantees the existence of random vectors $(X,Y),(X_1,Y_1),(X_2,Y_2),\dotsc$ defined on a common probability space such that $(X,Y)\sim P$, $(X_\ell,Y_\ell)\sim Q_\ell$ for all $\ell$ and $(X_\ell,Y_\ell)\to (X,Y)$ almost surely. Then it follows from~\eqref{eq:unifc} that $g_{k_\ell}(X_\ell)\to g(X)$ almost surely on the event $\{X=\tilde{m}_0\}$, which has probability 1 since $P^X(\{\tilde{m}_0\})=0$ by assumption. An application of Fatou's lemma shows that
\begin{align}
\infty>\limsup_{n\to\infty}\,L(f_n,P_n)\geq\liminf_{\ell\to\infty}\,L(g_{k_\ell},Q_\ell)&=\liminf_{\ell\to\infty}\,\E\bigl(\{Y_\ell-g_{k_\ell}(X_\ell)\}^2\bigr)\notag\\
&\geq\E\Bigl(\liminf_{\ell\to\infty}\,\{Y_\ell-g_{k_\ell}(X_\ell)\}^2\Bigr)\notag\\
\label{eq:lbdfat}
&=\E\bigl(\{Y-g(X)\}^2\bigr)=L(g,P).
\end{align}
Thus, $\E\bigl(g(X)^2\bigr)^{1/2}\leq\E\bigl(\{Y-g(X)\}^2\bigr)^{1/2}+\E(Y^2)^{1/2}<\infty$, so $g\in L^2(P^X)$. By Proposition~\ref{prop:clconv}(d) and the proof of (i) above, it follows that $[g]_{P^X}\in\Cl\mathcal{F}_{P^X}^{m_0}\subseteq\mathcal{L}^2(P^X)$.

(v) Since $L(h_n,P)\to L(h,P)$ whenever $\norm{h_n-h}_{L^2(P^X)}\to 0$, it follows from~\eqref{eq:lbdfat} and the conclusion of (iv) that $\liminf_{\ell\to\infty}L(g_{k_\ell},Q_\ell)\geq L(g,P)\geq L_{m_0}^*(P)$, as required.

(vi) If $P_n=P$ for all $n$, then we can modify the argument leading up to~\eqref{eq:lbdfat} as follows: since $g_{k_\ell}\to g$ pointwise on $C$ and $P^X(C)=1$, it follows that if $(X,Y)\sim P$, then $g_{k_\ell}(X)\to g(X)$ (almost surely) and $\infty>\liminf_{\ell\to\infty}L(g_{k_\ell},P)\geq L(g,P)$, as in~\eqref{eq:lbdfat}. Thus, $g\in L^2(P^X)$ and $[g]_{P^X}\in\Cl\mathcal{F}_{P^X}^{m_0}\subseteq\mathcal{L}^2(P^X)$ as in the proof of (iv). In particular, when $\tilde{m}_0\notin\Int C$ and $P^X(\{\tilde{m}_0\})>0$, we must have $g(\tilde{m}_0)\in\R$ since $g\in L^2(P^X)$. In this case, if $m_0\in C$, then $g\sim_{P^X} h$ for some $h\in\mathcal{F}^{m_0}$, and otherwise if $m_0\notin C$, then $[g]_{P^X}\in\Cl\mathcal{F}_{P^X}^{m_0}$ by Proposition~\ref{prop:clconv}(d). We conclude as before that $\liminf_{\ell\to\infty}L(g_{k_\ell},P)\geq L(g,P)\geq L_{m_0}^*(P)$ in all cases, so the proof of Lemma~\ref{lem:subseqcvg} is complete.
\end{proof}

\deparskip
\begin{proof}[Proof of Proposition~\ref{prop:cont}]
(a) Fix $\varepsilon>0$. By Proposition~\ref{prop:clconv}(a) and the observation in the paragraph after~\eqref{eq:orth}, there exists a continuous $h_0\in\mathcal{F}^{m_0}$ such that $L_{m_0}^*(P)\leq L(h_0,P)\leq L_{m_0}^*(P)+\varepsilon$. Now for $\eta\in [-m_0,1-m_0]$, define $h_\eta\colon [0,1]\to\R$ by $h_\eta(x):=h_0(0\vee (x-\eta)\wedge 1)$. Then $h_\eta\in\mathcal{F}^{m_0+\eta}$ and $\sup_{x\in [0,1]}\,\abs{h_\eta(x)}\leq\abs{h_0(0)}\vee\abs{h_0(1)}=:B$ for all $\eta\in [-m_0,1-m_0]$. In addition, it follows from Lemma~\ref{lem:convunif} that $h_\eta\to h_0$ uniformly on $[0,1]$ as $\eta\to 0$. 
For each $n$, let $\eta_n:=m_n-m_0$, so that $h_{\eta_n}\in\mathcal{F}^{m_n}$ and 
\[L_{m_n}^*(P_n)\leq L(h_{\eta_n},P_n)=L(h_0,P_n)+\bigl(L(h_{\eta_n},P_n)-L(h_0,P_n)\bigr)\]
by the definition of $L_{m_n}^*(P_n)$. Observe that
\begin{align*}
\abs{L(h_{\eta_n},P_n)-L(h_0,P_n)}&\leq\int_{[0,1]\times\R}\;\bigl|\bigl(y-h_{\eta_n}(x)\bigr)^2-\bigl(y-h_0(x)\bigr)^2\bigr|\,dP_n(x,y)\notag\\
&\leq\int_{[0,1]\times\R}\;\abs{h_0(x)-h_{\eta_n}(x)}\,\bigl(2\abs{y}+\abs{h_0(x)}+\abs{h_{\eta_n}(x)}\bigl)\,dP_n(x,y)\notag\\
&\leq\sup_{x\in [0,1]}\,\abs{h_0(x)-h_{\eta_n}(x)}\,\int_{[0,1]\times\R}\,2(\abs{y}+B)\,dP_n(x,y).
\end{align*}
Moreover, $W_2(P_n,P)\to 0$ by assumption, so $\int_{[0,1]\times\R}\norm{w}^2\,dP_n(w)\to\int_{[0,1]\times\R}\norm{w}^2\,dP(w)$. Thus, since $\bigl(y-h_0(x)\bigr)^2\leq 2(y^2+B^2)$ and $\abs{y}\leq (1+y^2)/2$ for all $(x,y)\in [0,1]\times\R$, we deduce using Lemma~\ref{lem:wconv} and the continuity of $h_0$ that
\begin{align*}
L(h_0,P_n)=\int_{[0,1]\times\R}\,\bigl(y-h_0(x)\bigr)^2\,dP_n(x,y)&\to\int_{[0,1]\times\R}\,\bigl(y-h_0(x)\bigr)^2\,dP(x,y)=L(h_0,P)\\
\text{and\quad}\int_{[0,1]\times\R}\,(\abs{y}+B)\,dP_n(x,y)&\to\int_{[0,1]\times\R}\,(\abs{y}+B)\,dP(x,y)<\infty\notag
\end{align*}
as $n\to\infty$. Since $h_{\eta_n}\to h_0$ uniformly on $[0,1]$ as $n\to\infty$, it follows from the above that
\[\limsup_{n\to\infty}L_{m_n}^*(P_n)\leq\lim_{n\to\infty}L(h_0,P_n)=L(h_0,P)\leq L_{m_0}^*(P)+\varepsilon.\]
Since $\varepsilon>0$ was arbitrary, this yields (a).

The proofs of (b)--(i) are based on the key Lemmas~\ref{lem:limsuptw},~\ref{lem:limsupunif} and~\ref{lem:subseqcvg} above.

(b) Fix a deterministic, non-negative sequence $(\delta_n)$ with $\delta_n\to 0$ and let $(f_n)$ be any sequence with $f_n\in\psi_{m_n}^{\delta_n}(P_n)$ for all $n$, so that $L(f_n,P_n)\leq L_{m_n}^*(P_n)+\delta_n$ for all $n$. Then $\limsup_{n\to\infty}L(f_n,P_n)=\limsup_{n\to\infty}L_{m_n}^*(P_n)\leq L_{m_0}^*(P)<\infty$ by (a), so $(f_n)$ satisfies the hypotheses of Lemma~\ref{lem:subseqcvg}. In particular, Lemma~\ref{lem:subseqcvg}(v) applies to every subsequence of $(f_n)$ when $P^X(\{\tilde{m}_0\})=0$, in which case
\begin{equation}
\label{eq:liminfL}
\liminf_{n\to\infty}\,L_{m_n}^*(P_n)=\liminf_{n\to\infty}\,L(f_n,P_n)\geq L_{m_0}^*(P).
\end{equation}
(c) If in addition $P_n=P$ for all $n$, then Lemma~\ref{lem:subseqcvg}(vi) applies to every subsequence of $(f_n)$, regardless of whether or not $P^X(\{\tilde{m}_0\})=0$, so~\eqref{eq:liminfL} holds in all cases. Together with part (a) above, this establishes (c) and the fact that $L^*(P)$ and $\mathcal{I}^*(P)$ are well-defined.

(d) If $P^X(\{m\})=0$ for all $m\in [0,1]$, then (a) and (b) imply that for any sequence $(m_n')$ in $[0,1]$ converging to some $m'\in [0,1]$, we have $\lim_{n\to\infty}L_{m_n'}^*(P_n)=L_{m'}^*(P)$. In other words, the functions $m\mapsto L_m^*(P_n)$ converge continuously to $m\mapsto L_m^*(P)$ on $[0,1]$ in the sense of~\citet[Chapter~3.1.5]{Rem91}. Since continuous convergence is equivalent to uniform convergence on the compact space $[0,1]$~\citep[e.g.][pages~98--99]{Rem91}, the first part of (d) follows. This in turn implies the second assertion that $L^*(P_n)=\min_{m\in [0,1]}L_m^*(P_n)\to\min_{m\in [0,1]}L_m^*(P)=L^*(P)$ as $n\to\infty$.

(e) Let $(m_n')$ be any sequence in $[0,1]$ with $m_n'\in\mathcal{I}^{\delta_n}(P_n)$ for all $n$. For each subsequence of $(m_n')$, we may extract a further subsequence $(m_{n_k}')$ that converges to some $m'\in [0,1]$.
For each $k$, we have $L^*(P_{n_k})\leq L_{m_{n_k}'}^*(P_{n_k})\leq L^*(P_{n_k})+\delta_{n_k}$ by the definition of $\mathcal{I}^{\delta_{n_k}}(P_{n_k})$. Thus, if $P^X(\{m\})=0$ for all $m\in [0,1]$, then for any $m^*\in\mathcal{I}^*(P)$, we deduce from (a) and (b) that 
\[L^*(P)\leq L_{m'}^*(P)=\lim_{k\to\infty}L_{m_{n_k}'}^*(P_{n_k})=\lim_{k\to\infty}L^*(P_{n_k})\leq\lim_{k\to\infty}L_{m^*}^*(P_{n_k})=L_{m^*}^*(P)=L^*(P),\]
so $m'\in\mathcal{I}^*(P)$ and $L^*(P_{n_k})\to L^*(P)$ as $k\to\infty$. This shows that every subsequence of $(m_n')$ has a further subsequence that converges to an element of $\mathcal{I}^*(P)$. Since $(m_n')$ was arbitrary, this implies (e). Similarly, every subsequence of $\bigl(L^*(P_n)\bigr)$ has a further subsequence that converges to $L^*(P)$; this is another way to obtain the second part of (d).

For (f)--(i), fix any sequence $(f_n)$ with $f_n\in\psi_{m_n}^{\delta_n}(P_n)$ for all $n$. If $P^X(\{\tilde{m}_0\})=0$, then for any subsequence $(g_k)_{k=1}^\infty\equiv (f_{n_k})_{k=1}^\infty$ of $(f_n)$, we can find a further subsequence $(g_{k_\ell})$ and a function $g$ on $[0,1]$ satisfying conditions (i)--(v) in Lemma~\ref{lem:subseqcvg}. In particular, setting $m_\ell':=m_{n_{k_\ell}}$ and $Q_\ell:=P_{n_{k_\ell}}$, we deduce from (a) and Lemma~\ref{lem:subseqcvg}(v) that
\[L(g,P)\leq\liminf_{\ell\to\infty}\,L(g_{k_\ell},Q_\ell)\leq\limsup_{\ell\to\infty}\,L(g_{k_\ell},Q_\ell)=\limsup_{\ell\to\infty}\,L_{m_\ell'}^*(Q_\ell)\leq L_{m_0}^*(P),\]
where the equality above follows from the fact that $L(g_{k_\ell},Q_\ell)\leq L_{m_\ell'}^*(Q_\ell)+\delta_{m_{k_\ell}}$ for all $\ell$. We conclude from Lemma~\ref{lem:subseqcvg}(iv) and Corollary~\ref{cor:proj}(c) that $g\in\psi_{\tilde{m}_0}^*(P)=\psi_{m_0}^*(P)$. 

(f) For each closed set $A\subseteq (\supp P^X)\setminus\{\tilde{m}_0\}$, Lemma~\ref{lem:subseqcvg}(ii) asserts that $g_{k_{\ell}}\to g\in\psi_{m_0}^*(P)$ uniformly on $A$. Thus, every subsequence of $(f_n)$ has a further subsequence that converges uniformly on $A$ to an element of $\psi_{m_0}^*(P)$, and by Corollary~\ref{cor:proj}(a), all elements of $\psi_{m_0}^*(P)$ agree $P^X$-almost everywhere on $A$. Since $(f_n)$ was arbitrary, (f) follows.

(g) If $\psi_{\tilde{m}_0}^0(P)\neq\emptyset$, then by Corollary~\ref{cor:proj}(a), there exists $f^*\in\psi_{\tilde{m}_0}^0(P)$ such that $\psi_{\tilde{m}_0}^*(P)=[f^*]_{P^X}$, so $g\sim_{P^X}f^*$ in the argument before (f). Thus, in view of Lemma~\ref{lem:subseqcvg}(i), we may assume that $g\in\mathcal{F}^{\tilde{m}_0}$, so that $g\in\psi_{\tilde{m}_0}^0(P)$. By applying Lemma~\ref{lem:subseqcvg}(ii) as above, we deduce that for any closed set $A\subseteq (\supp P^X)\setminus\{\tilde{m}_0\}$, every subsequence of $(f_n)$ has a further subsequence that converges uniformly on $A$ to an element of $\psi_{\tilde{m}_0}^0(P)$. All functions in $\psi_{\tilde{m}_0}^0(P)$ agree on $A$ by Corollary~\ref{cor:proj}(e), and $(f_n)$ was arbitrary, so (g) holds.

Suppose in addition that $m_0\in\Int(\csupp P^X)$ and $P^X(\{m_0\})=0$. Then in the argument before (f), we can insist that $g\in\mathcal{F}^{m_0}$ in view of Lemma~\ref{lem:subseqcvg}(iv), so that $g\in\psi_{m_0}^0(P)$.

(h) For fixed $q\in [1,\infty)$, Lemma~\ref{lem:limsupunif} implies that there exists $\tilde{B}<\infty$ such that $\abs{g_{k_\ell}-g}^q\leq 2^{q-1}(\tilde{B}^q+\abs{g}^q)$ on $[0,1]$ for all sufficiently large $\ell$, so $\norm{g_{k_\ell}-g}_{L^q(P^X)}^q=\int_{[0,1]}\,\abs{g_{k_\ell}-g}^q\,dP^X\to 0$ by the dominated convergence theorem. In summary, every subsequence of $\bigl([f_n]_{P^X}\bigr)_{n=1}^\infty$ has a further subsequence that converges in $\mathcal{L}^q(P^X)$ to $\psi_{m_0}^*(P)$, so the entire sequence converges in $\mathcal{L}^q(P^X)$ to $\psi_{m_0}^*(P)$, as required.

(i) Under the hypotheses of (i), all elements of $\psi_{m_0}^0(P)$ are continuous by Corollary~\ref{cor:proj}(f), so we can apply Lemmas~\ref{lem:subseqcvg}(ii) and~\ref{lem:convunif} to obtain the stronger conclusion that every subsequence of $(f_n)$ has a further subsequence that converges uniformly on $\csupp P^X$ to some function in $\psi_{m_0}^0(P)$. Since elements of $\psi_{m_0}^0(P)$ agree on $\supp P^X$ by assumption and $(f_n)$ was arbitrary, (i) holds.
\end{proof}

\deparskip
\begin{proof}[Proof of Corollary~\ref{cor:cont}]
(a) By definition, we have $\psi^0(P)=\{f\in\mathcal{F}:L(f,P)=L^*(P)\}=\bigcup_{m\in\mathcal{I}^*(P)}\psi_m^0(P)$. If $\psi_m^0(P)\neq\emptyset$ for some $m\in\mathcal{I}^*(P)\setminus\csupp P^X$, then $\tilde{m}=\argmin_{x\in\csupp P^X}\abs{x-m}$ satisfies $\tilde{m}\in\mathcal{I}^*(P)\cap\csupp P^X$ and $\psi_{\tilde{m}}^0(P)\neq\emptyset$ by Corollary~\ref{cor:proj}(c). The result now follows from Corollary~\ref{cor:proj}(d).

For (b)--(d), fix a sequence $(f_n)$ with $f_n\in\psi^{\delta_n}(P_n)$ for all $n$, so that there exists a sequence $(m_n)$ in $[0,1]$ with $m_n\in\mathcal{I}^{\delta_n}(P_n)$ and $f_n\in\psi_{m_n}^{\delta_n}(P_n)$ for all $n$. By assumption, we have $P^X(\{m\})=0$ for all $m\in [0,1]$, so for each subsequence of $(f_n)$, Proposition~\ref{prop:cont}(e) ensures the existence of a further subsequence $(f_{n_k})$ such that $m_{n_k}\to m^*$ for some $m^*\in\mathcal{I}^*(P)$. Let $\tilde{m}^*:=\argmin_{x\in\csupp P^X}\,\abs{m^*-x}$. Then $L_{\tilde{m}^*}^*(P)=L_{m^*}^*(P)=L^*(P)$ by Corollary~\ref{cor:proj}(c), so $\tilde{m}^*\in\mathcal{I}^*(P)$. We are now in a position to apply Proposition~\ref{prop:cont}(f)--(i).

(b) Fix a closed set $A\subseteq (\supp P^X)\setminus\mathcal{I}^*(P)\subseteq (\supp P^X)\setminus\{\tilde{m}^*\}$. For any $f^*\in\psi_{m^*}^*(P)\subseteq\psi^*(P)$, Proposition~\ref{prop:cont}(f) implies that $\norm{f_{n_k}-f^*}_{L^\infty(A,P^X)}\to 0$. Thus, every subsequence of $(f_n)$ has a further subsequence that converges in $\norm{{\cdot}}_{L^\infty(A,P^X)}$ to an element of $\psi^*(P)$. Since $(f_n)$ was arbitrary, this yields (b).

(c) Fix a closed set $A\subseteq (\supp P^X)\setminus\tilde{\mathcal{I}}^*(P)$. 

\unparskip	
\begin{itemize}
\item If $\tilde{m}^*\notin\tilde{\mathcal{I}}^*(P)$, then $\tilde{m}^*=m^*\in\Int(\csupp P^X)$ and all elements of $\psi_{\tilde{m}^*}^0(P)=\psi_{m^*}^0(P)\neq\emptyset$ are continuous, so for any $f^*\in\psi_{m^*}^0(P)\subseteq\psi^0(P)$, Proposition~\ref{prop:cont}(i) implies that $f_{n_k}\to f^*$ uniformly on $\supp P^X\supseteq A$. 
\item If $\tilde{m}^*\in\tilde{\mathcal{I}}^*(P)$, then we still have $\tilde{m}^*\in\mathcal{I}^*(P)\cap\csupp P^X$, so $\psi_{\tilde{m}^*}^0(P)\neq\emptyset$ by assumption. Thus, for any $f^*\in\psi_{\tilde{m}^*}^0(P)\subseteq\psi^0(P)$, Proposition~\ref{prop:cont}(g) implies that $f_{n_k}\to f^*$ uniformly on $(\supp P^X)\setminus\{\tilde{m}^*\}\supseteq A$. 
\end{itemize}

\unparskip
Thus, every subsequence of $(f_n)$ has a further subsequence that converges uniformly on $A$ to an element of $\psi^0(P)$. Since $(f_n)$ was arbitrary, (c) follows.

(d) Here, $m^*\in\mathcal{I}^*(P)\subseteq\Int(\csupp P^X)$ by assumption, so for any $f^*\in\psi_{m^*}^0(P)\subseteq\psi^*(P)$, Proposition~\ref{prop:cont}(h) implies that $\norm{f_{n_k}-f^*}_{L^q(P^X)}\to 0$ for any $q\in [1,\infty)$. By the same reasoning as in (b,\,c), we obtain (d).
\end{proof}

\deparskip
\begin{proof}[Proof of Proposition~\ref{prop:consistency}]
In the definition of $P_0$, we have $\E(f_0(X)+\xi\,|\,X)=f_0(X)$ since $\xi$ is independent of $X$ and has mean 0 by (ii), so $f_0$ is a regression function for $P_0$ (in the sense of~\eqref{eq:orth} in Section~\ref{sec:proj}). It now follows from condition (iii) and Corollary~\ref{cor:proj}(d) that $\psi_m^0(P_0)\neq\emptyset$ for all $m\in\csupp P_0^X$, so in particular $\psi^0(P_0)\neq\emptyset$.

Writing $(D_n)$ for any one of the sequences of random variables in (a)--(d), we aim to prove that $D_n\cvp 0$, or equivalently that every subsequence $(D_{n_k})$ has a further subsequence that converges almost surely to 0. If it can be shown that $W_2(\Pr_n,P_0)\cvp 0$, then we can take $(D_{n_{k_\ell}})$ to be a subsequence of $(D_{n_k})$ such that $W_2(\Pr_{n_{k_\ell}},P_0)\to 0$ almost surely. Working on an event $\Omega_0$ of probability 1 on which this convergence takes place, we deduce directly from the relevant parts of Proposition~\ref{prop:cont} or Corollary~\ref{cor:cont} that $D_{n_{k_\ell}}\!\to 0$ on $\Omega_0$; note that assertions (a,\,b) follow from Proposition~\ref{prop:cont}(d,\,e) and assertions (c,\,d) follow from Corollary~\ref{cor:cont}(c,\,d).

To complete the proof, we must therefore verify that $W_2(\Pr_n,P_0)\cvp 0$ under conditions (i)--(iii). It suffices to show that 
\[\text{($\ast$)}\;\;\Pr_n\cvd P_0\text{ almost surely}\qquad
\text{and}\qquad\text{($\ast\ast$)}\;\;\int_{\R^2}\norm{w}^2\,d\Pr_n(w)\cvp\int_{\R^2}\norm{w}^2\,dP_0(w),\]
since then we can argue along subsequences of $(\Pr_n)$ as in the previous paragraph. 

($\ast$) Let $\tilde{\Pr}_n:=n^{-1}\sum_{i=1}^n\delta_{(x_{ni},\,\xi_{ni})}$ for each $n$ and define $\tilde{P}_0:=P_0^X\otimes P_\xi$. Defining the map $F_0\colon (x,z)\mapsto (x,f_0(x)+z)$ on $\R^2$, we therefore have $\Pr_n=\tilde{\Pr}_n\circ F_0^{-1}$ for each $n$ and $P_0=\tilde{P}_0\circ F_0^{-1}$. The desired convergence statement ($\ast$) follows from the following two claims.
\begin{claim}
$\tilde{\Pr}_n\cvd\tilde{P}_0$ almost surely.
\end{claim}

\deparskip	
\begin{proof}[Proof of Claim]
By~\citet[Theorem~2.3]{Bil99}, the countable set $\mathcal{R}:=\{[a_1,b_1]\times [a_2,b_2]:a_j\leq b_j\text{ and }a_j,b_j\in\Q\text{ for }j=1,2\}$ is a convergence-determining class in the sense of~\citet[page~18]{Bil99}, so it suffices to show that $\tilde{\Pr}_n(R)\to\tilde{P}_0(R)$ almost surely for each $R\in\mathcal{R}$. To this end, fix any $I_1\times I_2\in\mathcal{R}$, where $I_j=[a_j,b_j]$ is an interval with rational endpoints $a_j\leq b_j$ for $j=1,2$. By condition (i), $\Pr_n^X\cvd P_0^X$ and $P_0^X(\{a_1,b_1\})=0$, so $\Pr_n^X(I_1)\to P_0^X(I_1)$ as $n\to\infty$. For each $n$, defining $r_n:=\sum_{i=1}^n\Ind_{\{x_{ni}\in I_1\}}=n\Pr_n^X(I_1)$, we can write $\tilde{\Pr}_n(I_1\times I_2)=\Pr_n^X(I_1)\,T_n$, where
\[T_n:=\frac{\sum_{i=1}^n\Ind_{\{x_{ni}\in I_1\}}\Ind_{\{\xi_{ni}\in I_2\}}}{r_n\vee 1}\sim\frac{1}{r_n\vee 1}\Bin\bigl(r_n,P_\xi(I_2)\bigr).\]
If $P_0^X(I_1)=0$, then certainly $\tilde{\Pr}_n(I_1\times I_2)=\Pr_n^X(I_1)\,T_n\to 0=P_0^X(I_1)P_\xi(I_2)=\tilde{P}_0(I_1\times I_2)$. On the other hand, if $P_0^X(I_1)>0$, then $r_n=n\bigl(1+o(1)\bigr)P_0^X(I_1)\to\infty$, so for all $t>0$, we have $\sum_{n=1}^\infty\Pr(\abs{T_n-P_\xi(I_2)}>t)<\infty$ by Hoeffding's inequality (or some other suitable exponential tail bound for binomial random variables; see~\citet[Appendix~A.6.1]{vdVW96} for example). Thus, by the first Borel--Cantelli lemma, $T_n\to P_\xi(I_2)$ almost surely, so $\tilde{\Pr}_n(I_1\times I_2)=\Pr_n^X(I_1)\,T_n\to P_0^X(I_1)P_\xi(I_2)=\tilde{P}_0(I_1\times I_2)$ almost surely as $n\to\infty$, as required.
\end{proof}

\unparskip
\begin{claim}
If $(Q_n)$ is any sequence of probability measures such that $Q_n\cvd\tilde{P}_0$, then $Q_n\circ F_0^{-1}\cvd\tilde{P}_0\circ F_0^{-1}=P_0$ under condition (iii).
\end{claim}

\deparskip
\begin{proof}[Proof of Claim]
It follows from Skorokhod's representation theorem~\citep[e.g.][Theorem~2.19]{vdV98} that there exist random vectors $(\tilde{X},\tilde{Z}),(\tilde{X}_1,\tilde{Z}_1),(\tilde{X}_2,\tilde{Z}_2),\dotsc$ defined on a common probability space such that $(\tilde{X},\tilde{Z})\sim\tilde{P}_0$, $(\tilde{X}_n,\tilde{Z}_n)\sim Q_n$ for all $n$ and $(\tilde{X}_n,\tilde{Z}_n)\to (\tilde{X},\tilde{Z})$ almost surely. 

Since $\tilde{X}\sim P_0^X$ and $f_0$ is continuous $P_0^X$-almost everywhere under condition (iii), we have $f_0(\tilde{X}_n)\to f_0(\tilde{X})$ almost surely, so $F_0(\tilde{X}_n,\tilde{Z}_n)=(\tilde{X}_n,f_0(\tilde{X}_n)+\tilde{Z}_n)\to (\tilde{X},f_0(\tilde{X})+\tilde{Z})=F_0(\tilde{X},\tilde{Z})$ almost surely. Thus, the distribution $Q_n\circ F_0^{-1}$ of $F_0(\tilde{X}_n,\tilde{Z}_n)$ converges weakly to the distribution $\tilde{P}_0\circ F_0^{-1}=P_0$ of $F_0(\tilde{X},\tilde{Z})$, as required.
\end{proof}

\unparskip
($\ast\ast$) Let $X\sim P_0^X$ and $\xi\sim P_\xi$ be independent, so that $\bigl(X,f_0(X)+\xi\bigr)\sim P_0$. We have
\begin{equation}
\label{eq:W2consist} \int_{\R^2}\norm{w}^2\,d\Pr_n(w)=\frac{1}{n}\sum_{i=1}^n(x_{ni}^2+Y_{ni}^2)=\frac{1}{n}\sum_{i=1}^n\bigl(x_{ni}^2+f_0^2(x_{ni})+2f_0(x_{ni})\,\xi_{ni}+\xi_{ni}^2\bigr)
\end{equation}
and now consider each of the summands on the right-hand side. Since $(\Pr_n^X)$ is a sequence of probability measures on the compact set $[0,1]$ satisfying $\Pr_n^X\cvd P_0^X$, we have $W_2(\Pr_n^X,P_0^X)\to 0$, so
\begin{equation}
\label{eq:W2consist1}
\frac{1}{n}\sum_{i=1}^n x_{ni}^2=\int_{[0,1]}x^2\,d\Pr_n^X(x)\to\int_{[0,1]}x^2\,dP_0^X(x)=\E(X^2).
\end{equation}
In addition, we can apply condition (iii) and argue as in the proof of the first Claim to see that $\Pr_n^X\circ f_0^{-1}\cvd P_0^X\circ f_0^{-1}$. Since $f_0$ is bounded on $[0,1]$ by (iii), these probability measures are also supported on some compact set, so in fact $W_2\bigl(\Pr_n^X\circ f_0^{-1},P_0^X\circ f_0^{-1}\bigr)\to 0$, whence
\begin{equation}
\label{eq:W2consist2}
\frac{1}{n}\sum_{i=1}^n f_0^2(x_{ni})=\int_\R x^2\,d\bigl(\Pr_n^X\circ f_0^{-1}\bigr)(x)=\int_{\R}x^2\,d\bigl(\Pr_0^X\circ f_0^{-1}\bigr)(x)=\E\bigl(f_0^2(X)\bigr).
\end{equation}
For all $t>0$, it follows from~\eqref{eq:W2consist2}, condition (ii) and Chebsyhev's inequality that
\begin{equation}
\label{eq:W2consist3}
\Pr\biggl(\frac{1}{n}\,\biggl|\sum_{i=1}^n f_0(x_{ni})\,\xi_{ni}\biggr|>t\biggr)\leq \frac{1}{(nt)^2}\Var\biggl(\sum_{i=1}^n f_0(x_{ni})\,\xi_{ni}\biggr)=\frac{1}{(nt)^2}\sum_{i=1}^n f_0^2(x_{ni})\Var(\xi)\to 0,
\end{equation}
so $n^{-1}\sum_{i=1}^n f_0(x_{ni})\,\xi_{ni}\cvp 0=\E\bigl(f_0(X)\xi\bigr)$ by the independence of $X$ and $\xi$. Finally, by condition (ii), $\xi_{n1},\dotsc,\xi_{nn}\iid P_\xi$ for each $n$, so $n^{-1}\sum_{i=1}^n\xi_{ni}^2\cvp\E(\xi^2)$ by the weak law of large numbers. Together with~\eqref{eq:W2consist},~\eqref{eq:W2consist1},~\eqref{eq:W2consist2} and~\eqref{eq:W2consist3}, this implies that
\[\textstyle\int_{\R^2}\norm{w}^2\,d\Pr_n(w)\cvp\E(X^2)+\E\bigl(f_0^2(X)\bigr)+2\E\bigl(f_0(X)\xi\bigr)+\E(\xi^2)=\int_{\R^2}\norm{w}^2\,dP_0(w),\]
so ($\ast\ast$) holds.
\end{proof}

\deparskip
\begin{proof}[Proof of Proposition~\ref{cor:consistency}] 
The results for $(\tilde{g}_n)=(\tilde{f}_n)$ follow from Proposition~\ref{prop:consistency}. When $(\tilde{g}_n)=(\hat{f}_n^{m_0})$, we again have $W_2(\Pr_n,P_0)\cvp 0$ under conditions (i)--(iii) in Proposition~\ref{prop:consistency}, so assertions (a,\,b,\,c,\,d) follow from Proposition~\ref{prop:cont}(e,\,g,\,i,\,h) in this case.
\end{proof}

\umparskip
\section{Auxiliary results and examples for Section~\ref{sec:proj2}}
\label{sec:contproof}
The proofs in Section~\ref{sec:projproofs} make use of two straightforward results on the convergence of sequences of S-shaped functions.
\begin{lemma} 
\label{lem:convincr}
Suppose that $(f_n)_{n=1}^\infty$ is a sequence of increasing convex functions on $[0,1)$ such that $\lim_{n\to\infty}f_n(x)$ exists for all $x\in (0,1)$. Then $\lim_{n\to\infty}f_n(0)$ exists and the function $f\colon [0,1)\to\R$ defined by $f(x):=\lim_{n\to\infty}f_n(x)$ is convex and increasing. 
Moreover, $f_n\to f$ uniformly on $[0,w]$ for every $w\in [0,1)$.
\end{lemma}

\deparskip
\begin{proof}
Since $\restr{f_n}{(0,1)}$ is convex and increasing for all $n$, the same is true of the pointwise limit $\restr{f}{(0,1)}$. Thus, $l:=\lim_{t\,\searrow\,0}f(t)$ exists and is finite, and we now show that $f_n(0)\to l$ as $n\to\infty$. Since each $f_n$ is convex on $[0,1)$, we have $f_n(0)\geq 2f_n(t)-f_n(2t)$ for all $t>0$ and $n\in\N$, so \[\liminf_{n\to\infty}f_n(0)\geq\lim_{t\,\searrow\,0}\,\lim_{n\to\infty}\,\bigl(2f_n(t)-f_n(2t)\bigr)=\lim_{t\,\searrow\,0}\,\bigl(2f(t)-f(2t)\bigr)=l.\]
On the other hand, $f_n$ is increasing on $[0,1)$ for all $n$, so \[\limsup_{n\to\infty}f_n(0)\leq\lim_{t\,\searrow\,0}\,\lim_{n\to\infty}f_n(t)=\lim_{t\,\searrow\,0}f(t)=l,\]
which means that $f(0)=l$, as required. Consequently, $f$ is convex and increasing on $[0,1)$. For the final assertion of the lemma, we extend $f,f_1,f_2,\dotsc$ to increasing convex functions on $(-\infty,1)$ by setting $f(x)=f(0)$ and $f_n(x)=f_n(0)$ for all $x<0$ and $n\in\N$. It follows from what we have already shown that $f_n\to f$ pointwise on $(-\infty,1)$, so in fact $f_n\to f$ uniformly on compact subsets of $(-\infty,1)$ by~\citet[Theorem~10.8]{Rock97}. This yields the desired conclusion. 
\end{proof}

\unparskip
\begin{lemma}
\label{lem:convunif}
If $m\in (0,1)$ and $(f_n)_{n=1}^\infty$ is a sequence of functions in $\mathcal{F}^m$ that converges pointwise on $[0,1]\setminus\{m\}$ to some continuous $f\in\mathcal{F}^m$, then $f_n\to f$ uniformly on $[0,1]$.
\end{lemma}

\deparskip
\begin{proof}
First note that $f_n\to f$ pointwise on $[0,1]$. Indeed, \[f(z)=\lim_{n\to\infty}f_n(z)\leq \liminf_{n\to\infty}f_n(m)\leq\limsup_{n\to\infty}f_n(m)\leq\lim_{n\to\infty}f_n(w)=f(w)\]
whenever $z<m<w$, and since $\lim_{z\nearrow\,m}f(z)=f(m)=\lim_{w\searrow\,m}f(w)$ by continuity, it follows that $f_n(m)\to f(m)$. We now show that $f_n\to f$ uniformly on $[0,1]$ by a standard argument: 
fix $\varepsilon>0$ and note that since $f$ is continuous and increasing on $[0,1]$, we can find $0=z_0<z_1<\dotsc<z_{k-1}<z_k=1$ such that $f(z_i)-f(z_{i-1})<\varepsilon$ for all $i\in [k]$. Since each $f_n$ is increasing, we see that if $x\in [z_{i-1},z_i]$, then
\[f_n(z_{i-1})-f(z_{i-1})-\varepsilon<f_n(z_{i-1})-f(z_i)\leq f_n(x)-f(x)\leq f_n(z_i)-f(z_{i-1})<f_n(z_i)-f(z_i)+\varepsilon.\]
Since $f_n\to f$ pointwise, $\limsup_{n\to\infty}\sup_{x\in [0,1]}\,\abs{f_n(x)-f(x)}<\lim_{n\to\infty}\max_{1\leq i\leq k}\,\abs{f_n(z_i)-f(z_i)}+\varepsilon=\varepsilon$. This holds for all $\varepsilon>0$, so the result follows.
\end{proof}

\unparskip
The weak convergence result below is stated as Lemma~4.5 in~\citet{DSS11} and proved here for completeness.
\begin{lemma}
\label{lem:wconv}
Let $P,P_1,P_2,\dotsc$ be probability measures on $\R^d$ such that $P_n\cvd P$. If $h$ is a non-negative, continuous function on $\R^d$ such that $\int_{\R^d}h\,dP_n\to\int_{\R^d}h\,dP<\infty$, then $\int_{\R^d}f\,dP_n\to\int_{\R^d}f\,dP$ for any continuous function $f\colon\R^d\to\R^k$ such that $\norm{f}/(1+h)$ is bounded.
\end{lemma}

\deparskip
\begin{proof}
We can restrict attention to the case $k=1$ since the component functions can be considered separately when $k>1$. Let $C>0$ be such that $\abs{f}\leq C(1+h)$ pointwise. For a Borel measure $Q$ on $\R^d$ and a $Q$-integrable function $h$, we write $Q(h)$ as shorthand for $\int_{\R^d}h\,dQ$. 

Since $P_n\cvd P$, we have $\liminf_{n\to\infty}P_n(g)\geq P(g)$ for all non-negative, continuous $g\colon\R^d\to\R$ by the portmanteau lemma~\citep[e.g.][Lemma~2.2(iv)]{vdV98}.
Thus, since $f^+:=f\vee 0$ and $C(1+h)-f^+$ are non-negative and continuous, we have $\liminf_{n\to\infty}P_n(f^+)\geq P(f^+)$ and $\liminf_{n\to\infty}P_n\bigl(C(1+h)-f^+\bigr)\geq P\bigl(C(1+h)-f^+\bigr)$. Moreover, $P_n\bigl(C(1+h)\bigr)\to P\bigl(C(1+h)\bigr)$ by assumption, so in fact $P_n(f^+)\to P(f^+)$. A similar argument shows that $P_n(f^-)\to P(f^-)$, where $f^-:=(-f)\vee 0$, so we indeed have $P_n(f)\to P(f)$.
\end{proof}

\unparskip
We conclude this subsection with a series of related examples which illustrate that some of the assertions of Proposition~\ref{prop:cont} and Corollary~\ref{cor:cont} do not hold in general if the associated technical conditions are not satisfied. 
\begin{example}
\label{ex:1}
We first consider situations where either $P^X(\{\tilde{m}_0\})>0$ or $m_0\notin\Int(\csupp P^X)$. In each of the following, we construct $P\in\mathcal{P}$ and a sequence $(P_n)$ in $\mathcal{P}$ with $W_2(P_n,P)\to 0$, where $P=(1-\eta)Q+\eta\tilde{Q}$ and $P_n=(1-\eta_n)Q_n+\eta_n\tilde{Q}_n$ for suitable $Q,Q_n,\tilde{Q}_n\in\mathcal{P}$ and $\eta,\eta_n\in [0,1]$. For $w\in\R^d$ with $d\in\N$, we write $\delta_w$ for a point mass at $w$.

\unparskip
\begin{enumerate}[label=(\alph*)]
\item Fix $m_0\in (0,1]$ and $m\in (0,m_0]$. Let $Q_n:=\frac{1}{2}\delta_{(0,0)}+\frac{1}{2}\delta_{(m(1-1/n),\,0)}$, $Q:=\frac{1}{2}\delta_{(0,0)}+\frac{1}{2}\delta_{(m,0)}$, $\tilde{Q}_n=\tilde{Q}:=\delta_{(m,1)}$ and $\eta_n=\eta:=1/3$ for all $n$. Then $P_n=\frac{1}{3}\delta_{(0,0)}+\frac{1}{3}\delta_{(m(1-1/n),\,0)}+\frac{1}{3}\delta_{(m,1)}\cvd\frac{1}{3}\delta_{(0,0)}+\frac{1}{3}\delta_{(m,0)}+\frac{1}{3}\delta_{(m,1)}=P$ and $P^X=\frac{1}{3}\delta_0+\frac{2}{3}\delta_m$, so $m_0\notin\Int(\csupp P^X)$, $\mathcal{I}^*(P)=[0,1]\not\subseteq\Int(\csupp P^X)$, $\tilde{m}_0=m$ and $P^X(\{\tilde{m}_0\})>0$. Since $P,P_1,P_2,\dotsc$ are supported on the compact set $[0,1]^2$, we automatically have $W_2(P_n,P)\to 0$. 

We claim that Propositions~\ref{prop:cont}(b,\,d,\,h) and Corollary~\ref{cor:cont}(d) do not apply here. Indeed, for each $n$, we have $f_n(m)=1$ for all $f_n\in\psi_{m_0}^0(P_n)=\psi^0(P_n)$ and $L^*(P_n)=L_{m_0}^*(P_n)=0$, whereas $f^*(m)=1/2$ for all $f^*\in\psi_{m_0}^*(P)=\psi^*(P)$ and $L^*(P)=L_{m_0}^*(P)=1/6>0=\lim_{n\to\infty}L_{m_0}^*(P_n)$. Thus, if $f_n\in\psi_{m_0}^0(P_n)$ for all $n$, then $\abs{f_n(m)-f^*(m)}=1/2$ for all $n$ and $f^*\in\psi^*(P)$, so $\inf_{f^*\in\psi^*(P)}\norm{f_n-f^*}_{L^q(P^X)}\nrightarrow 0$ for $q\in [1,\infty)$.
\item This is a variant of (a) with $m=m_0\in\Int(\csupp P^X)=(0,1)$ but $P^X(\{m_0\})>0$. Let $Q_n,Q$ and $\eta_n,\eta$ be as in (a) but instead define $\tilde{Q}_n=\tilde{Q}:=\frac{1}{2}\delta_{(m_0,1)}+\frac{1}{2}\delta_{(1,1)}$. Then $P^X=\frac{1}{3}\delta_0+\frac{1}{3}\delta_{m_0}+\frac{1}{3}\delta_1$ and all the deductions in the second paragraph of (a) remain valid here.
\item Fix $m_0\in (0,1]$ and $m\in (0,m_0]$. For $n\in\N$, let $Q,Q_n$ be the uniform distributions on $\{(x,0):0\leq x\leq m\}$ and $\{(x,0):0\leq x\leq m(1-1/n)\}$ respectively. Moreover, let $\tilde{Q}:=Q$, $\eta:=0$, $\tilde{Q}_n:=\delta_{(m,n^2)}$ and $\eta_n:=n^{-5}$ for all $n$. Then $P_n=(1-\eta_n)Q_n+\eta_n\tilde{Q}_n$ converges in $W_2$ to $P=Q$. Indeed, for any continuous $f\colon\R^2\to\R$ such that $\abs{f(x,y)}\leq 1+x^2+y^2$ for all $(x,y)\in\R^2$, we have $\eta_n f(m,n^2)=O(1/n)$, so $\int_{\R^2}f\,dP_n=(1-\eta_n)\int_0^{m(1-1/n)}f(x,0)\,dx+\eta_n f(m,n^2)\to\int_0^m f(x,0)\,dx=\int_{\R^2}f\,dP$ as $n\to\infty$.

Here, $\tilde{m}_0=m$ and $P^X(\{m'\})=0$ for all $m'\in [0,1]$, but $m_0\notin\Int(\csupp P^X)$, $\mathcal{I}^*(P)=[0,1]\not\subseteq\Int(\csupp P^X)$, and Proposition~\ref{prop:cont}(h) and Corollary~\ref{cor:cont}(d) do not hold. Indeed, for each $n$, let $f_n\in\psi_{m_0}^0(P_n)=\psi^0(P_n)$ be such that $f_n=0$ on $[0,m(1-1/n)]$, $f_n(m)=n^2$ and $f_n$ is linear on $[m(1-1/n),1]$. Since $f^*=0$ $P^X$-almost everywhere on $[0,m]$ for all $f^*\in\psi_{m_0}^*(P)=\psi^*(P)$, we have $\inf_{f^*\in\psi^*(P)}\norm{f_n-f^*}_{L^q(P^X)}=\inf_{f^*\in\psi^*(P)}\norm{f_n}_{L^q(P^X)}\geq\inf_{f^*\in\psi^*(P)}\norm{f_n}_{L^1(P^X)}=n\to\infty$ for all $q\in [1,\infty)$.
\end{enumerate}

\unparskip
\end{example}

\hfparskip
In Proposition~\ref{prop:cont}(i), the assumption that $\psi_{m_0}^0(P)\neq\emptyset$ and all elements of $\psi_{m_0}^0(P)$ agree on $\supp P^X$ is clearly necessary for the conclusion to hold (not least when $P_n=P$ for all $n$). When this condition is not satisfied, some elements of $\psi_{m_0}^0(P)$ are discontinuous by Corollary~\ref{cor:proj}(e). If in addition $m_0\in\Int(\csupp P^X)\cap\supp P^X$ and $P^X(\{m_0\})=0$, then Proposition~\ref{prop:cont}(i) fails. This is demonstrated by the next example, which is a modification of Example~\ref{ex:1}(b).
\begin{example}
\label{ex:2}
Fix $m\in (0,1)$ and for $n\in\N$, let $f_n\colon [0,1]\to [0,1]$ be such that $f_n=0$ on $[0,m(1-1/n)]$, $f_n=1$ on $[m,1]$, and $f_n$ is continuous and linear on $[m(1-1/n),m]$. For each $n$, let $P_n$ be the distribution supported on $\{(x,f_n(x)):x\in [0,1]\}$ for which the corresponding marginal distribution $P_n^X$ is the uniform distribution on $D:=[0,m]\cup [m',1]$ for some $m'\in [m,1)$. Then $(P_n)$ converges in $W_2$ to the uniform distribution $P$ on $\{(x,\Ind_{[m,1]}(x)):x\in D\}$.

Not all elements of $\psi^0(P)=\psi_m^0(P)$ agree at $m\in\supp P^X$; for example when $m'=m$ and $D=[0,1]$, the functions in $\psi_m^0(P)$ agree with $\Ind_{[m,1]}$ on $[0,1]\setminus\{m\}$ but can take any value in $[0,1]$ at $m$. We have $P^X(\{m\})=0$ and $m\in\Int(\csupp P^X)$, so the other conditions of Proposition~\ref{prop:cont}(i) are satisfied, but $\psi^0(P_n)=\psi_m^0(P_n)=\{f_n\}$ for all $n$ and $(f_n)$ does not have a uniform limit on $[0,m]\subseteq\supp P^X$.
\end{example}

\unparskip
In Example~\ref{ex:2}, note in addition that while all the conditions of Corollary~\ref{cor:cont}(c) are met, the associated convergence statement cannot be extended to $A=\supp P^X$. In general, Proposition~\ref{prop:cont}(g) also does not hold with $A=\supp P^X$ under the stated conditions $P^X(\{\tilde{m}_0\})=0$ and $\psi_{\tilde{m}_0}^0(P)\neq\emptyset$.
To see this, we can instead take $D=[0,m]$ in Example~\ref{ex:2} and fix $m_0\in (m,1]$, so that $\tilde{m}_0=m$, $P^X(\{\tilde{m}_0\})=0$ and $0\in\psi_{\tilde{m}_0}^0(P)$. Then there is a sequence $(g_n)$ with $g_n\in\psi_{m_0}^0(P_n)$ for all $n$ but no uniform limit on $\supp P^X=[0,m]$; for example, let $g_n$ be such that $g_n=0$ on $[0,m(1-1/n)]$, $g_n(m)=1$ and $g_n$ is linear on $[m(1-1/n),1]$.

\end{document}